\documentclass[12pt,oneside]{extbook}
\usepackage{arxiv}
\pdfoutput=1

\usepackage[
  colorlinks    = true,
  allcolors     = blue,
  plainpages    = false,
  pdfpagelabels = true]{hyperref}       

\usepackage{natbib}
\usepackage{doi}
\setcitestyle{authoryear}
\usepackage{etoolbox}
\makeatletter

\patchcmd{\NAT@citex}
  {\@citea\NAT@hyper@{%
     \NAT@nmfmt{\NAT@nm}%
     \hyper@natlinkbreak{\NAT@aysep\NAT@spacechar}{\@citeb\@extra@b@citeb}%
     \NAT@date}}
  {\@citea\NAT@nmfmt{\NAT@nm}%
   \NAT@aysep\NAT@spacechar\NAT@hyper@{\NAT@date}}{}{}

\patchcmd{\NAT@citex}
  {\@citea\NAT@hyper@{%
     \NAT@nmfmt{\NAT@nm}%
     \hyper@natlinkbreak{\NAT@spacechar\NAT@@open\if*#1*\else#1\NAT@spacechar\fi}%
       {\@citeb\@extra@b@citeb}%
     \NAT@date}}
  {\@citea\NAT@nmfmt{\NAT@nm}%
   \NAT@spacechar\NAT@@open\if*#1*\else#1\NAT@spacechar\fi\NAT@hyper@{\NAT@date}}
  {}{}

\makeatother

\usepackage{url}
\usepackage{graphicx}
\usepackage{lmodern}
\usepackage[T1]{fontenc}
\usepackage{microtype}
\usepackage[dvipsnames]{xcolor}

\newcommand{\ignorethis } [1] { }

\usepackage{color}

\newcommand{\mycomment}[1]{}

\newcommand{\sF}{\mathsf{F}}
\newcommand{\sB}{\mathsf{B}}

\newcommand{\chapnum    } [1] {\ref{#1}}
\newcommand{\appnum     } [1] {\ref{#1}}
\newcommand{\sectnum    } [1] {\ref{#1}}

\newcommand{\fignum     } [1] {\ref{fig:#1}}
\newcommand{\eqnnum     } [1] {\ref{#1}}
\newcommand{\chap       } [1] {Chapter~\chapnum{#1}}

\newcommand{\app        } [1] {Appendix~\appnum{#1}}

\newcommand{\sect       } [1] {Section~\sectnum{#1}}

\newcommand{\fig        } [1] {Figure~\fignum{#1}}

\newcommand{\eqn        } [1] {Equation~\eqnnum{#1}}

\newcommand{\defn       } [1] {Definition~\ref{#1}}

\newcommand{\etal       }     {{\it et~al.}}

\newcommand{\eg         }     {{e.g.,}}

\newcommand{\ie         }     {{i.e.,}}

\newcommand{\Reals      }     {{\textrm{I\kern-0.18em R}}}

\renewcommand{\vec      } [1] {{\mathbf{#1}}}

\newcommand{\change     } [1] {\mbox{{\footnotesize $\Delta$} \kern-3pt}#1}

\newcommand{\abs        } [1] {{| #1 |}}

\newcommand{\norm       } [1] {{\| #1 \|}}

\newcommand{\dotPabs    } [2] {\abs{ #1 \cdot #2}}

\usepackage{amssymb} 
\usepackage{amsmath} 
\usepackage{mathptmx}
\usepackage[utf8]{inputenc}
\usepackage{import}
\usepackage[small,bf]{caption}
\usepackage{subcaption}
\usepackage[explicit]{titlesec}
\definecolor{gray75}{gray}{0.75}
\titleformat{\chapter}[display]%
  {\scshape\Huge\bfseries}%
  {\vspace{-8em}\raggedleft{%
    {\color{gray75}%
        \rule[-5pt]{2pt}{5cm}}\enskip%
    {\color{gray75}
        \fontsize{60}{60}\selectfont\thechapter}%
    }%
  }%
  {-1em}%
  {\parbox[b]{\dimexpr\textwidth-3em\relax}{\raggedright#1}}%
  [\vspace{2em}]

\title{Line Drawings from 3D Models: \protect\newline A Tutorial}

\author{
Pierre B\'enard \\
LaBRI (UMR 5800, CNRS, Univ. Bordeaux)\\
Inria Bordeaux Sud-Ouest\\
\texttt{pierre.benard@labri.fr}~\orcid{0000-0002-2846-1955}
\And
Aaron Hertzmann \\
Adobe Research\\
\texttt{hertzman@dgp.toronto.edu}~\orcid{0000-0001-9667-0292}
}

\date{}

\usepackage{mwe}
\usepackage[c]{esvect}
\usepackage{ccicons}
\usepackage{transparent}

\begin{document}

\maketitle

\begin{abstract}
  This tutorial describes the geometry and algorithms for generating line drawings from 3D models, focusing on occluding contours.  

  The geometry of occluding contours on meshes and on smooth surfaces is described in detail, together with algorithms for extracting contours,  computing their visibility, and creating stylized renderings and animations. Exact methods and hardware-accelerated fast methods are both described, and the trade-offs between different methods are discussed.  The tutorial brings together and organizes material that, at present, is scattered throughout the literature. It also includes some novel explanations, and implementation tips.
   
  A thorough survey of the field of non-photorealistic 3D rendering is also included, covering other kinds of line drawings and artistic shading.
\end{abstract}

\setcounter{tocdepth}{1}
\tableofcontents

\chapter{Introduction}
\label{sec:intro}

Humans have been drawing pictures since the days of prehistoric cave painting. 
Various forms of line drawing have been developed since then, including Egyptian hieroglyphs, medieval etching, and industral-era printmaking.
Nowadays, line drawing and outline drawing methods are used throughout cartoons and comics, architectural rendering, instructional tutorials, and many other settings.
Drawing is the starting point for many kinds of tasks, for
everyone from children making pictures to professional architects sketching ideas.
Drawing seems to be fundamentally connected to how we represent the world visually.

While most computer graphics focuses on realistic visual simulation, over the past few decades, line drawing algorithms have also matured.  We now have the ability to automatically create reasonable line drawings from 3D geometry, much like photorealistic rendering. These algorithms provide deep insight into the geometry and topology of line drawings, which can be surprisingly subtle, given how simple line drawing might seem. Versions of these algorithms have been used throughout art, entertainment, and visualization. User evaluation has shown that these algorithms, indeed, accurately describe important aspects of how artists draw lines. This shows that these algorithms can contribute to a scientific understanding of art.

This tutorial provides a detailed guide to the mathematical theory and computer algorithms for line drawing of 3D objects.  
We focus on the curves known as \emph{occluding contours} or, simply, \emph{contours}.  
These are the most important curves for line drawing of 3D surfaces. They have a rich theory around them, and, once one understands this theory, understanding how other curves operate is much simpler. 
We describe the different algorithms required to compute and render these curves, together with references to the literature.
We also explain boundary curves and surface-surface intersection curves, since these are straightforward to include and often important.
We also discuss open research problems in contour rendering.


In addition, we survey of other topics in 3D non-photorealistic rendering, with extensive pointers to the literature, including: other types of curves, stroke rendering, and non-photorealistic shading.
We do not cover the complementary topic of image-based non-photorealistic rendering; for a survey of image-based methods, we refer the reader to the book by \citet{Rosin:2013}.

The theory of line drawing is currently scattered about and incomplete in research papers. The  algorithms for line drawing include many subtleties that are not described in the literature, and many pitfalls await the coder attempting them.  There remain some important open problems, but these gaps are not obvious from the literature.  This tutorial is meant to address these issues.

We believe that these topics ought to be known by anyone interested in understanding the curves in visual representational art.   It is one where computer graphics can make a unique contribution. Arguably, the algorithmic simulation of line drawing is a crucial step in understanding visual art.


\section{Occluding contours}

This tutorial focuses on the curves known as \emph{occluding contours} or, simply \emph{contours}. 
In some computer graphics research, these have been called \emph{silhouettes}, though the \emph{silhouettes} are technically a separate set of curves, as we will describe below.

The occluding contours of a simple 3D object are shown in \fig{contours}.  As a first definition, suppose we have a 3D object that we wish to render from a specific viewpoint. The occluding contours are \emph{surface curves that separate visible parts from invisible parts}. By rendering the visible portions of these 3D curves together with the object, we get a basic line drawing (\fig{contour_basic}).

\begin{figure}
	\centering
	\small
	\begin{subfigure}[b]{0.45\linewidth}
		\includegraphics[width=\linewidth]{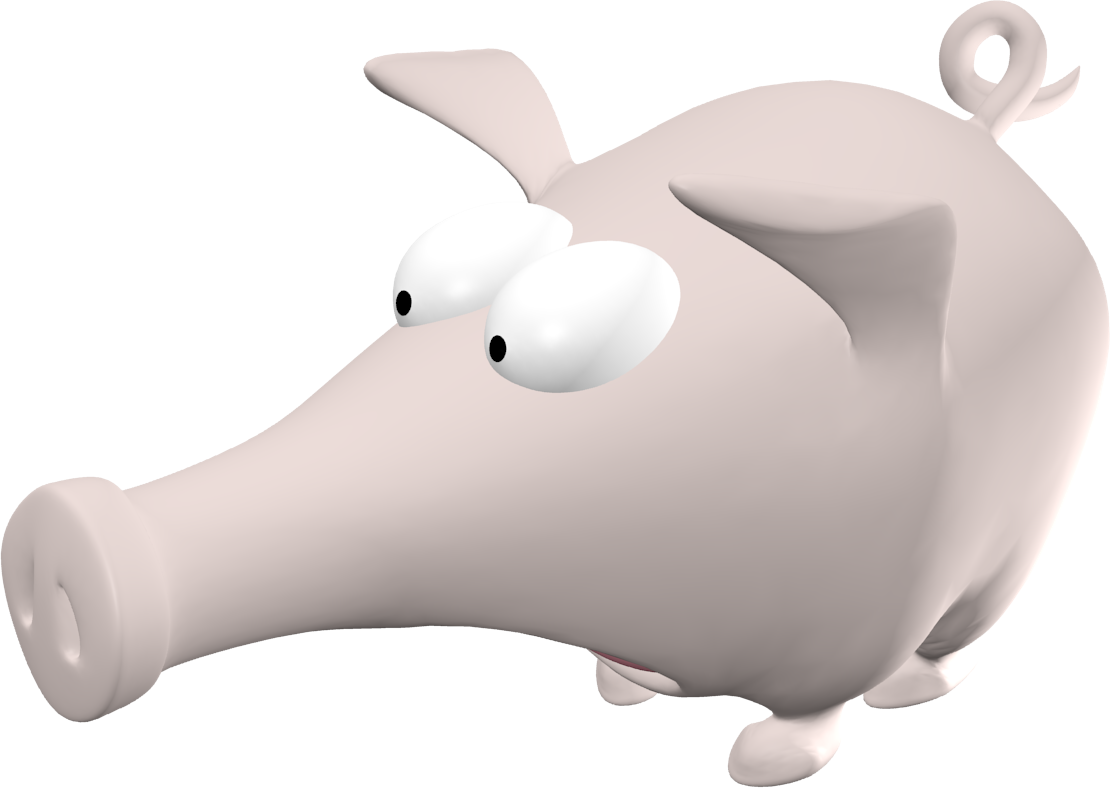}
		\caption{3D object with diffuse shading}
  \end{subfigure}
  \quad
	\begin{subfigure}[b]{0.45\linewidth}
		\includegraphics[width=\linewidth]{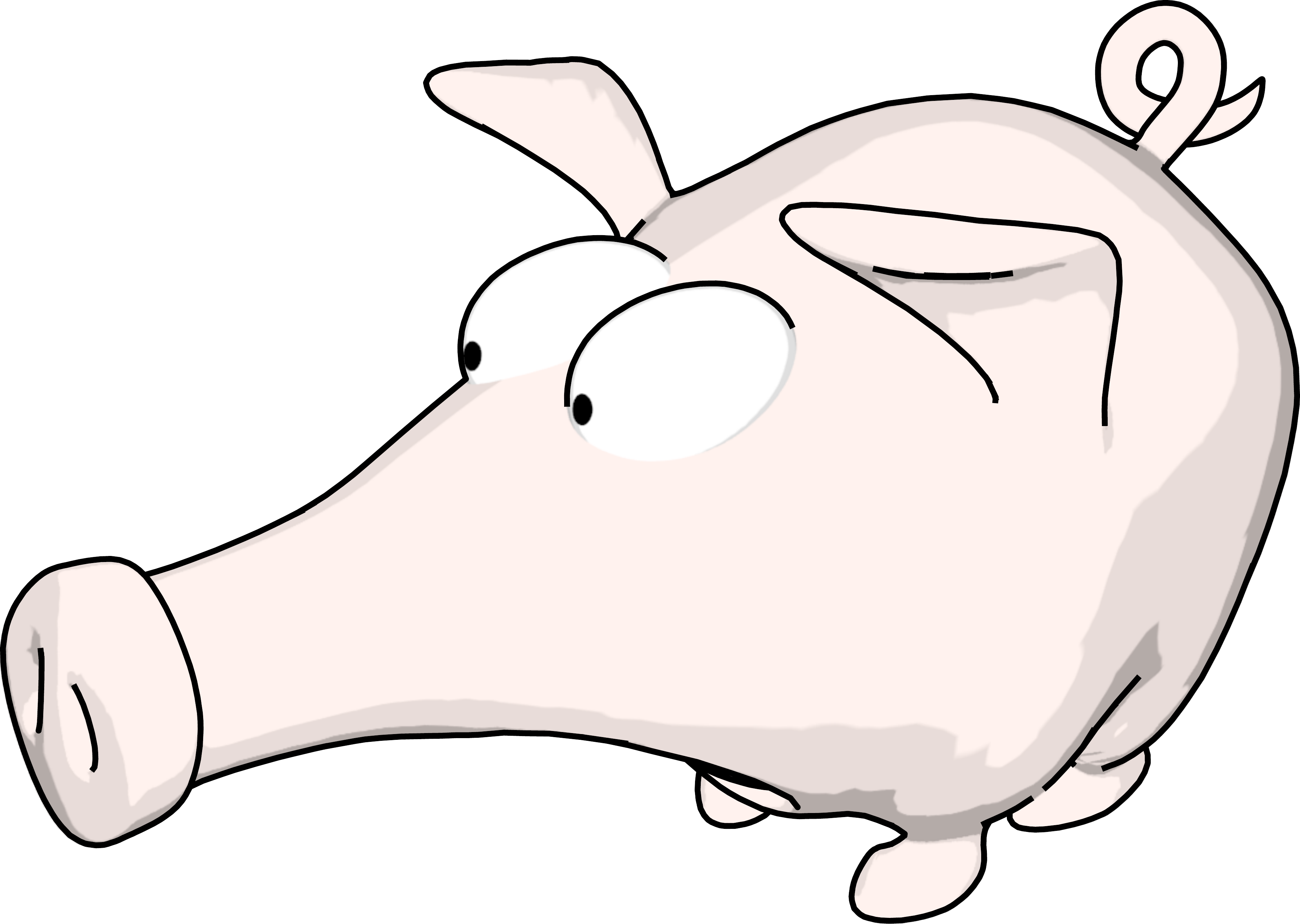}
		\caption{Contours with toon shading} \label{fig:contour_basic}
  \end{subfigure}
  \begin{subfigure}[b]{0.5\linewidth}
    \vspace{1em}
    \def\svgwidth{\textwidth}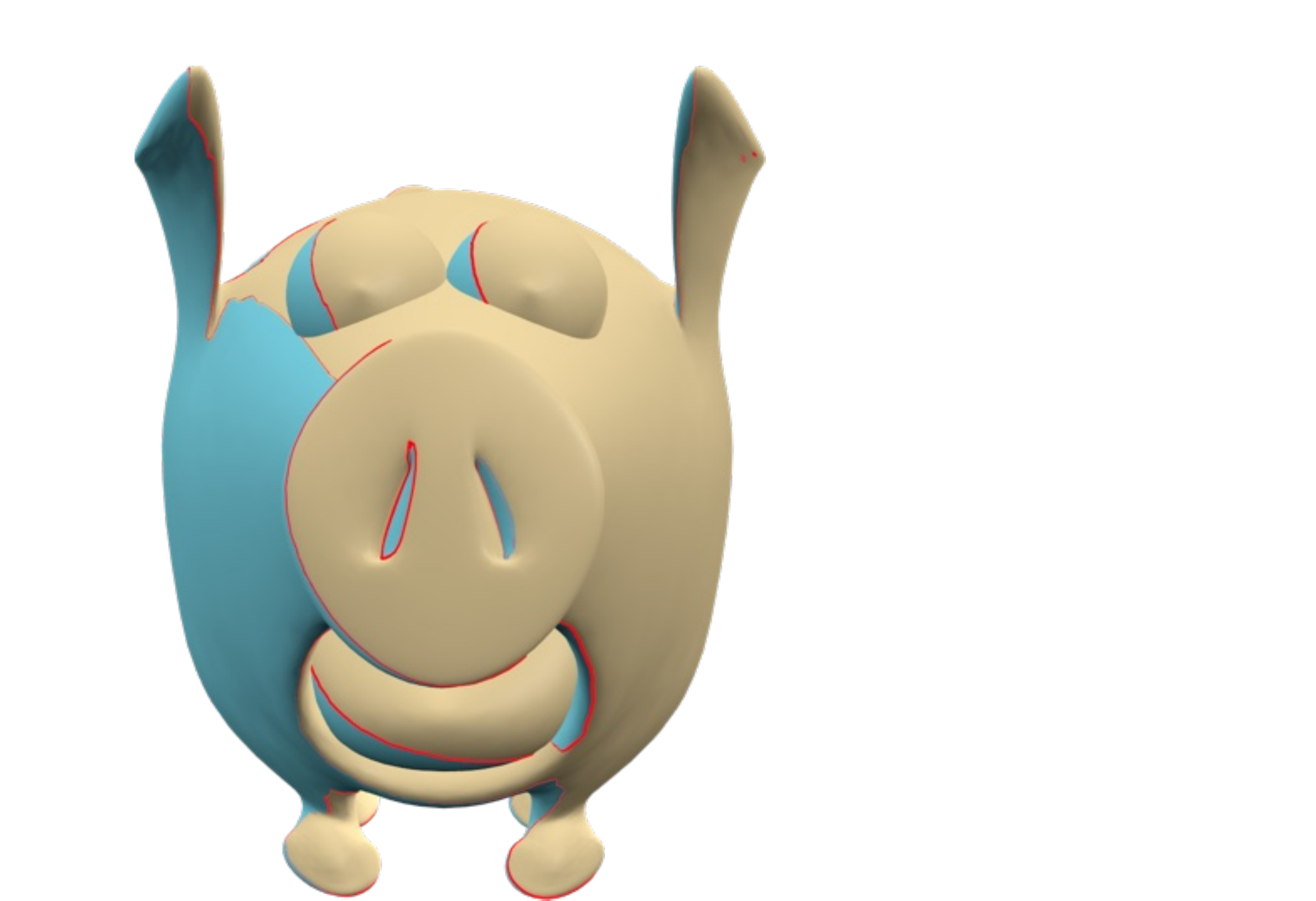
    \caption{Side view of the scene} \label{fig:contour_facing}
  \end{subfigure}
  \begin{subfigure}[b]{0.45\linewidth}
    \includegraphics[width=\linewidth]{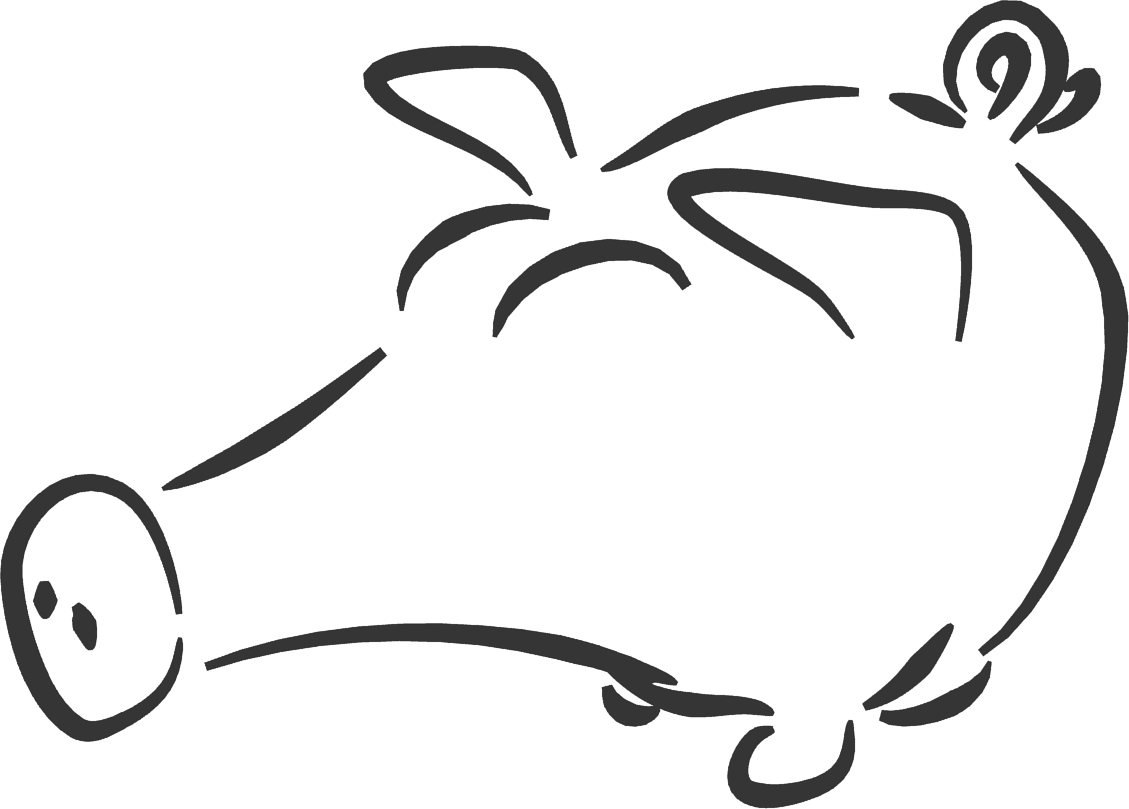}
    \caption{Stylized curves} \label{fig:contour_strokes}
  \end{subfigure}
  \caption{\textbf{Occluding contours} --- The occluding contours of the 3D model ``Origins of the Pig'' \ccCopy~Keenan Crane, shown in~\textbf{(a)} with diffuse shading, are depicted in~\textbf{(b)} composited with toon shading to produce a cel-like drawing. As illustrated in~\textbf{(c)} from a side view, they delineate the frontier between the front and back parts of the surface when seen from the camera. These contour curves can be further process to produce stylized imagery, such as the calligraphic brush strokes in~\textbf{(d)}.
  \label{fig:contours}}
\end{figure}

There are many different contour detection and rendering algorithms, and some significant tradeoffs between them. The most important tradeoff is between simple algorithms that produce approximate results, and more complex algorithms that give more precision, control, and stylization capabilities. Just rendering reasonable-looking contours as solid black lines is very straightforward for a graphics programmer. These most basic algorithms can be implemented in a few additional lines of code in an existing renderer, and have been implemented in many real-time applications, including many popular video games (one of the earliest was \textit{Jet Set Radio} in 2000). However, if we wish to stylize the curves, for example, by rendering the curves with sketchy or calligraphic strokes (\fig{contour_strokes}), things become more difficult. With a bit of perseverance, renderings with distinctive and lovely styles can be created. At times, these renderings may still contain topological artifacts that are not suitable for very high-end production. High-quality algorithms that remove these artifacts are more complex; in the extreme, no provably-correct algorithm exists for this problem. However, there are a number of partial solutions that are good enough to be used in many circumstances, and we discuss in detail the issues involved.

Note that, formally, the occluding contour and occluding contour generator are separate curves in 2D and 3D. However, we will frequently use the term ``contour'' to refer to each of them, since the correct terms are very cumbersome, and the meaning of ``contour'' is usually obvious from context.

\begin{figure}
  \centering
  \small
	\begin{subfigure}[b]{0.45\linewidth}
		\includegraphics[width=\linewidth]{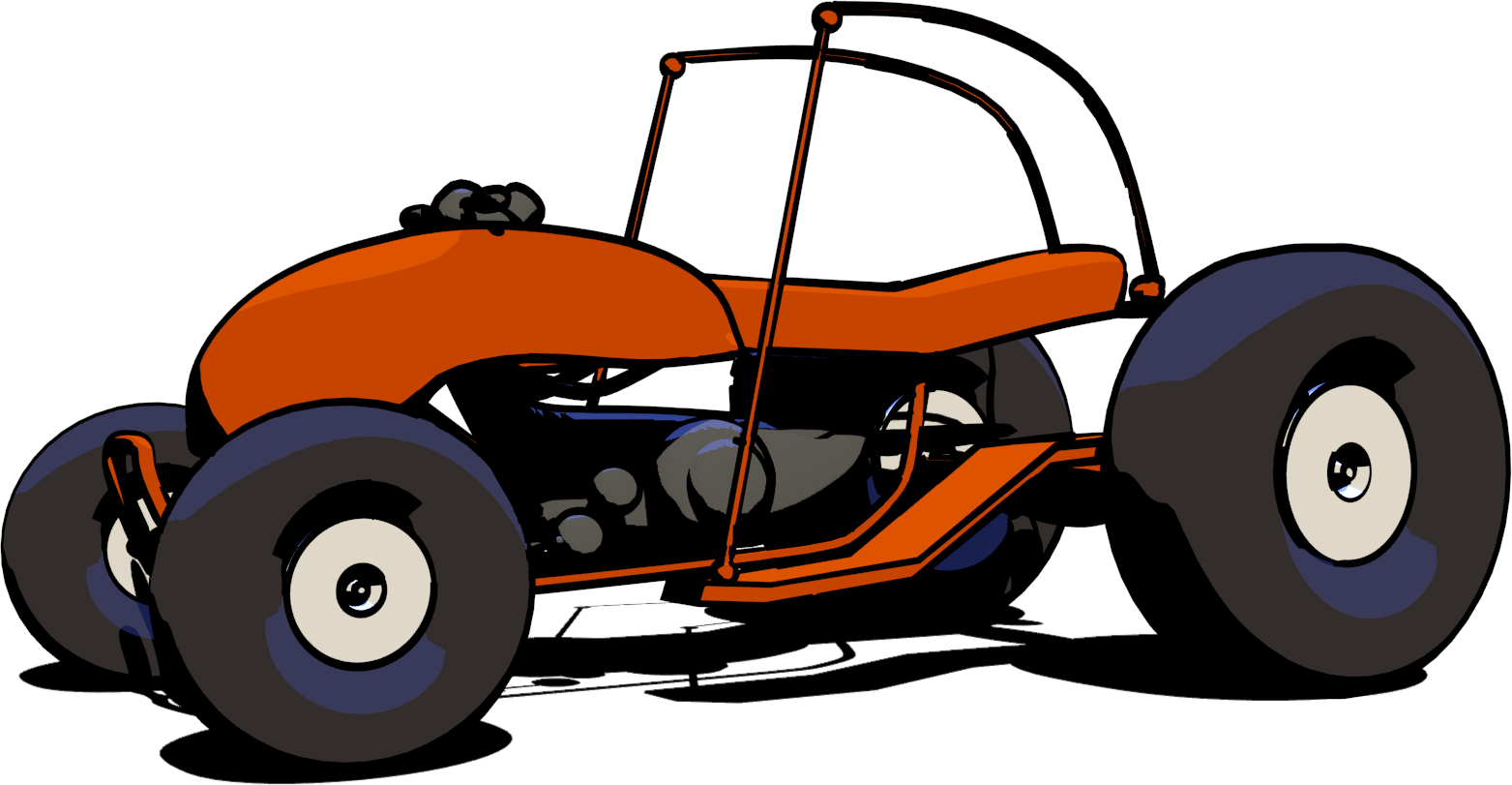}
		\caption{Buggy by Rylan Wright \ccby}
	\end{subfigure}
	\hspace{0.25pt}
	\begin{subfigure}[b]{0.5\linewidth}
		\includegraphics[width=\linewidth]{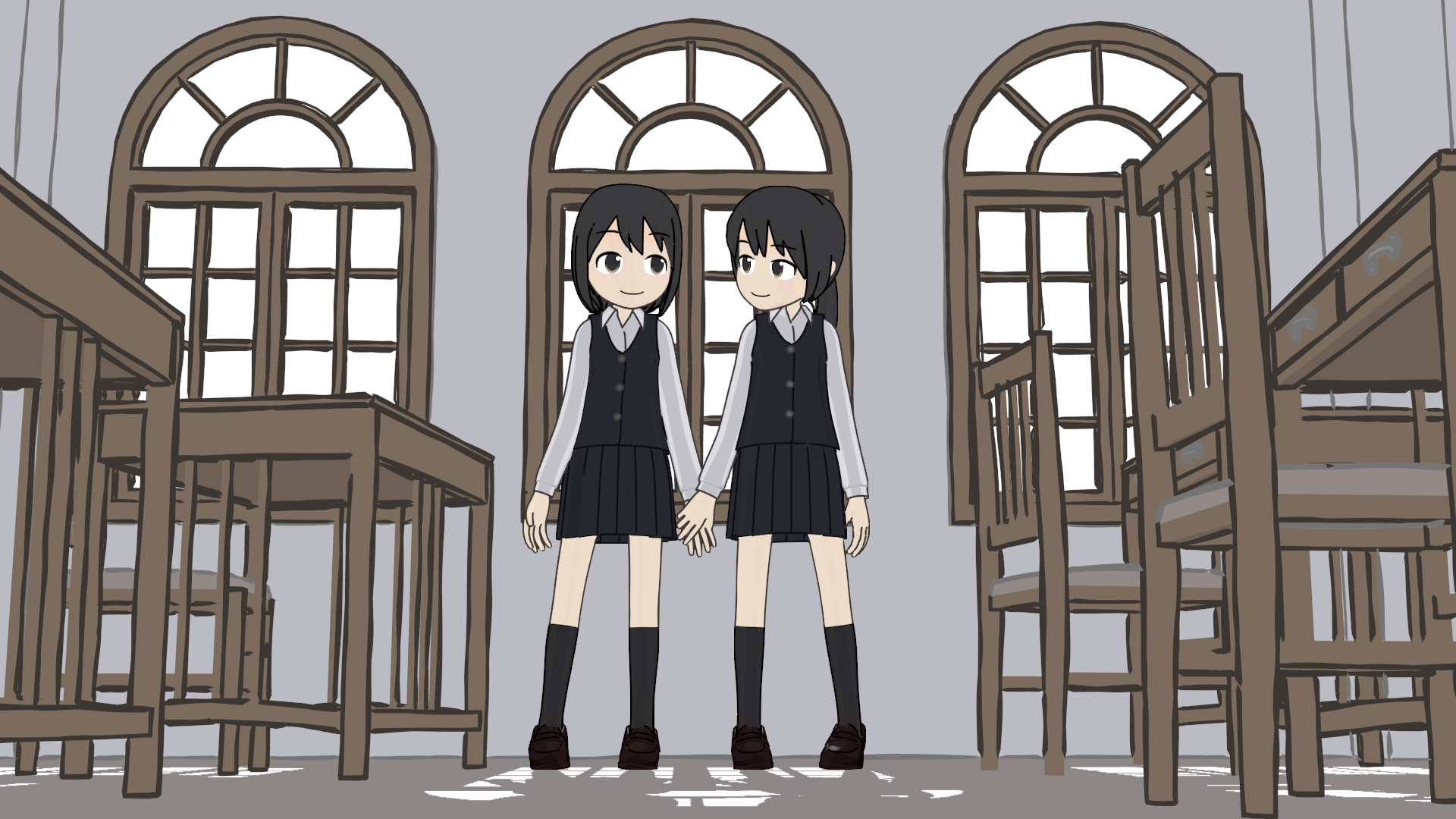}
		\caption{Anime by mato.sus304 \ccbysa}
	\end{subfigure}
	\par\vspace{0.5em}
  \begin{subfigure}[b]{0.7\linewidth}
		\includegraphics[width=\linewidth]{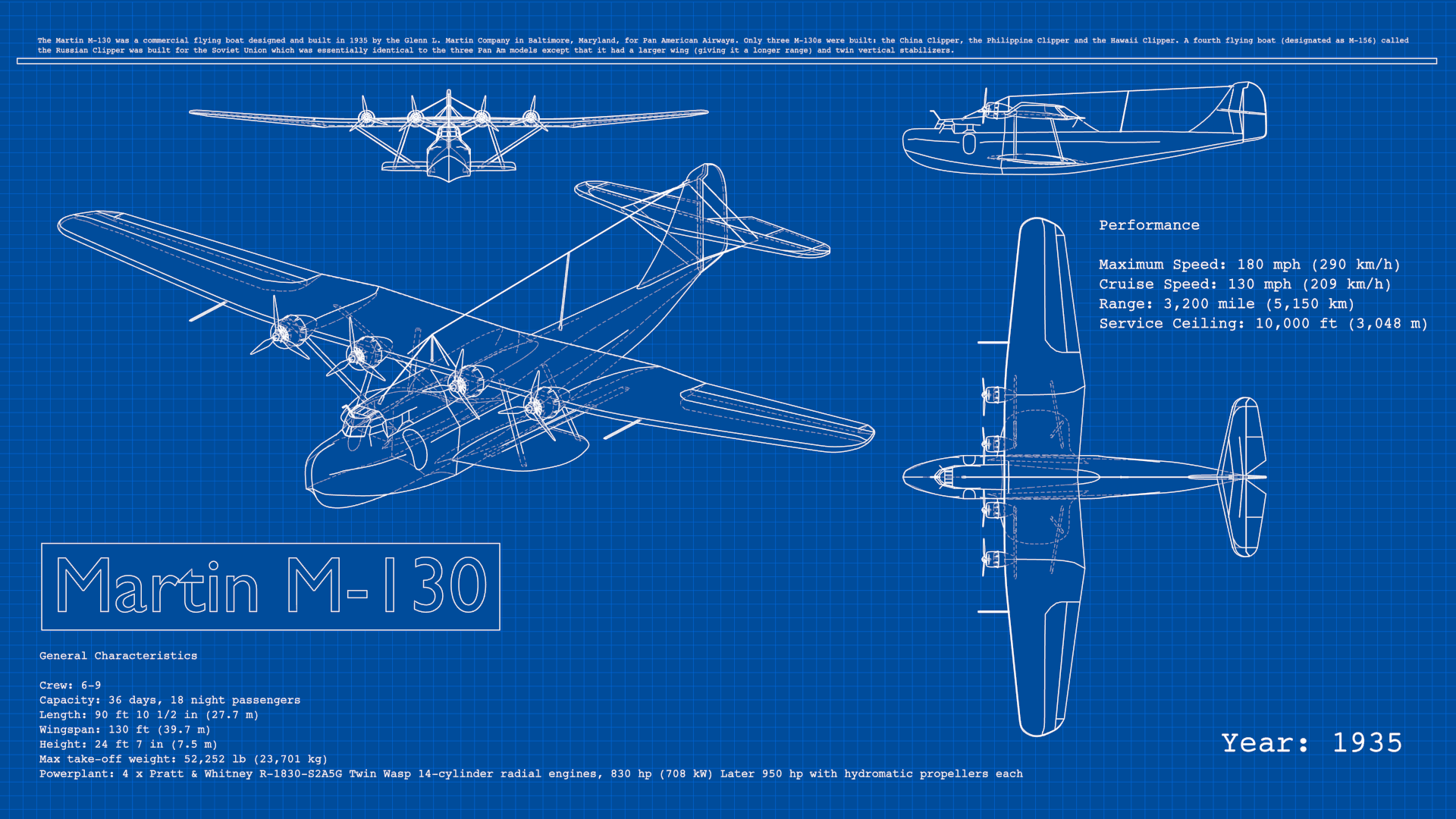}
		\caption{Martin M-130 blueprint by LightBWK \cczero}
	\end{subfigure}
	\par\vspace{0.5em}
  \begin{subfigure}[b]{0.9\linewidth}
    {\sf\def\svgwidth{\textwidth}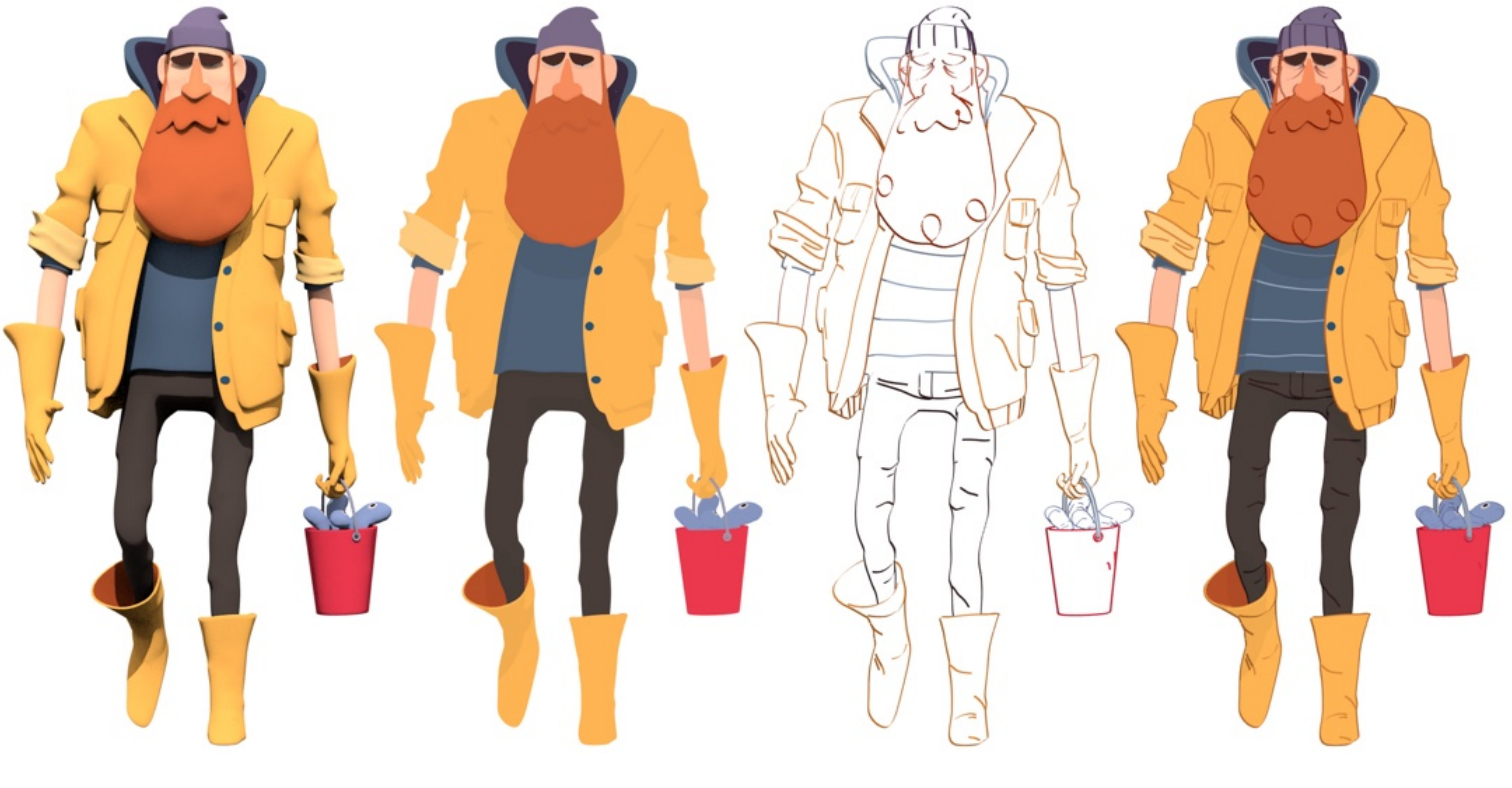}
		\caption{Ryner by Lucas Gogol \cczero}
	\end{subfigure}
	\caption{\textbf{Artworks created by artists using Blender Freestyle} --- Each of these is a non-photorealistic rendering, using the techniques described in this tutorial in different ways.
	}
	\label{fig:artworks}
\end{figure}

\section{How to use this tutorial}

This tutorial is two things: a detailed tutorial of the core contour algorithms, and a high-level survey of nearly all of 3D non-photorealistic rendering. We cover some core topics more thoroughly than any previous publication, and, for other topics, we mainly provide pointers to the literature.

Hence, reading the tutorial directly will give a good overview of the field, but one may skim through survey sections. Alternatively, a practitioner may wish to jump directly to the algorithms relevant to their task. 

Generally speaking, real-time image-based methods, especially based on graphics shaders, offer the best real-time performance and have been used in many games. These are described in the next Chapter, and pointers to further reading are provided there. This chapter can also help build intuitions for all readers.

The core chapters of the tutorial focus on contour detection and visibility on 3D models. We start with 3D mesh representations, and then apply the same ideas to different smooth surface representations in the subsequent chapters. 

We then cover the core topic of detecting contours on meshes (Chapter \ref{chap:mesh_contours}) and computing their visibility (Chapter \ref{chap:visibility}).  Contour detection and visibility on meshes is the most basic and well-understood problem, and we go into the most detail in algorithms here. We describe fast, approximate hardware based visibility in Chapter \ref{chap:fast_visibility}.

While it may be tempting to use mesh algorithms for smooth surfaces, in Chapter \ref{chap:smooth_as_meshes}, we explain some of the problems with doing so.  We then describe a method called Interpolated Contours that provides a compromise position, being almost as simple as mesh contours to implement, with relatively few inconsistencies.

We then discuss true contours on parametric surface representations (\chap{chap:smooth_contours}). Understanding these curves involves some differential geometry (reviewed in Appendix \ref{app:diff_geom}), and the resulting mathematics and theory is rather elegant. We describe detection and visibility algorithms, which are adapted from the mesh algorithms.  We describe the different strategies that have been applied to this problem and how they compare.
In the following chapter, we then discuss these algorithms as applied to implicit smooth surface representations (Chapter \ref{chap:implicit_surfaces}).

Finally, we discuss stylized rendering and animations algorithms (Chapter \ref{chap:rendering}), and conclude with a discussion of the state of research and applications in 3D non-photorealistic rendering (Chapter \ref{chap:conclusion}).

\section{The importance of visualizations}

Although we have done our best to explain contours in text, they can take some time to wrap your head around. Understanding how the 2D curves and 3D curves relate in an image like Figure \ref{fig:contours} can be challenging. It is worthwhile spending time with these figures, perhaps starting with simpler examples like different views of a torus, to understand how the 2D and 3D shapes relate, what the curves look like at singular points, and so on.

We provide an interactive viewer at \url{https://benardp.github.io/contours_viewer/}. Experimenting with this viewer can help give intuitions on contours.

Even better is to use, or build, a 3D visualization. If you implement a 3D contour rendering system, it is essential to also implement visualizations that let you zoom into the 2D drawing and rotate around the 3D model.  In each view, you should be able to render the different types of curves and singularities and their attributes. These visualizations are essential for deep understanding of these curves, as well as for debugging and algorithm development. You can start with simple 3D drawings, e.g., rendering contour edges on the 3D model, and coloring mesh faces according to facing direction as in \fig{contour_facing}. As your system becomes more sophisticated, you may eventually have visualizations like those in \fig{viewgraph}.

These visualizations are also useful in making certain design choices.  As we discuss, there is no current foolproof system for smooth contour rendering, and so there are some choices to be made, e.g., selecting heuristics.  Good visualizations can also be helpful in understanding how different heuristics behave.

\section{The science and perception of art}

The algorithms described in this tutorial provide a new level of insight and understanding into the science of art \citep{Hertzmann:2010}.
For centuries, artists, historians, philosophers, and scientists have sought a formal understanding of visual art: how do we make it, and how do we perceive it?  One of the first generative tools in art was the development of linear perspective during the Italian Renaissance.  The theory of occluding contours, which is the main subject of this tutorial, originated in perceptual psychology and computer vision, and was developed into the sophisticated algorithms we described here by computer graphics researchers. 

In the modern era, scientific discussions of visual art have been attached to the perceptual theories of the day.  Ernst \citet{Gombrich}, one of the most famous writers on the subject, wrote nearly an entire book arguing against the theory that art is just a replication of a fixed set of culturally-learned symbols. Rudolf \citet{Arnheim} wrote in the era of Gestalt psychology and thus applied Gestalt-like qualitative rules to drawings. John \citet{Willats} described a denotational semantics to describe different kinds of realistic styles --- expanded by \citet{WillatsDurand} to include insights from computer graphics.

Non-photorealistic rendering provides a generative theory for how artists create representational art. Like any theory, it does not cover every case or describe every phenomenon accurately, nor does it say anything about cultural, psychological, or other outside factors in the work. However, this generative theory provides considerable potential insight into how art is made.

We can compare the generative theory to the world before and after Newtonian mechanics. Before Newton, philosophers like Aristotle could make qualitative observations about how objects move (e.g., ``heavy objects like to fall'') but could make no real predictions. Newtonian mechanics is predictive, it generates insights, and leads to real understanding. Likewise, understanding the generative model of representational art provides a potentially compact way to understand many phenomena.

Two landmark studies validate and justify the use of line drawing algorithms developed in the non-photorealistic rendering literature.
\citet{Cole:2008} undertook a careful study of how artists depict 3D objects. They asked a collection of art students to illustrate several 3D models with line drawings, and compared how the artists' drawings related to the line drawing algorithms in this section (\fig{correlation}). They showed that roughly 80\% of a typical drawing could be explained by existing algorithms. This study helps show which of these algorithms are most useful, while also highlighting gaps in the literature.  In a follow-up paper, \citet{Cole:2009} showed that line drawing algorithms are also very effective at conveying 3D shape.

\begin{figure}
	
	\includegraphics[width=0.32\linewidth]{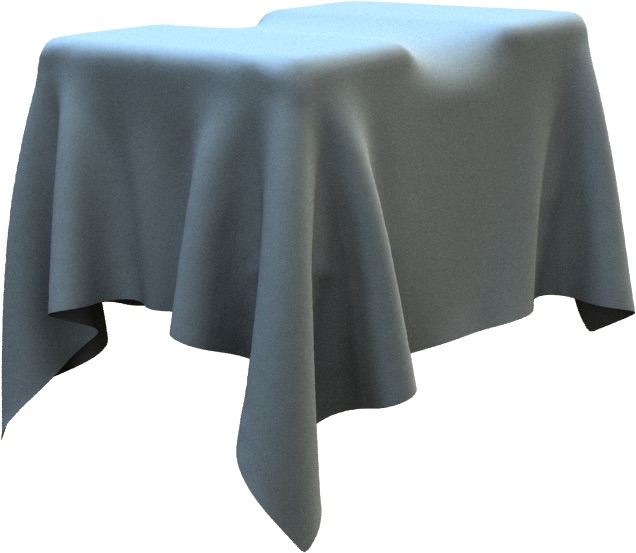}\hspace{0.25pt}
	\includegraphics[width=0.32\linewidth]{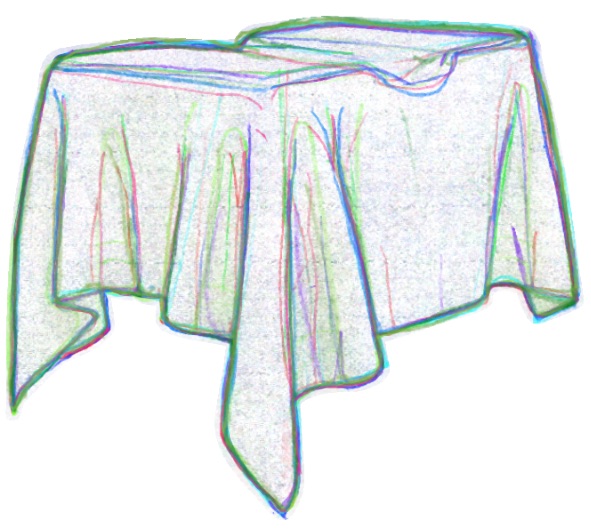}\hspace{0.25pt}
	\includegraphics[width=0.32\linewidth]{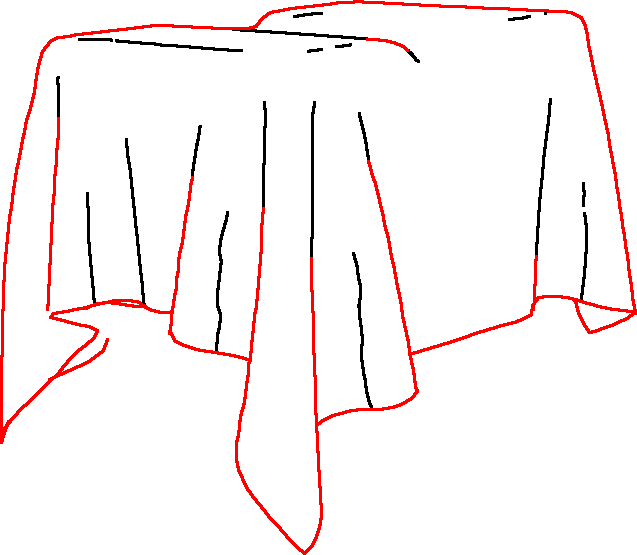}
	\caption{\textbf{Correlation between hand-drawn lines and contours} --- A 3D model rendered from a given viewpoint and illumination (left) has been hand-drawn by ten artists (center). Observe how consistent the drawings are, especially near the contours of the shape. The contours (in red) and suggestive contours (in black) extracted from the 3D model are depicted in the right. Images taken from the ``Javascript Drawing Viewer''\protect\footnotemark~of \citet{Cole:2008}.
	 \label{fig:correlation}
}
\end{figure}

\footnotetext{\url{http://gfx.cs.princeton.edu/proj/ld3d/lineset/viewer/index.html}}

\section{Survey of feature curves} \label{sec:survey}

This section surveys other important types of surface curves for line drawing, together with pointers to the relevant literature. The remaining chapters focus solely on contour, boundary, and surface-intersection curves.


Most of these curves have been developed both for artistic use and for visualization purposes \citep{Lawonn:2018}. However, some types of curves, such as ridges and valleys, seem useful for visualization without mimicking conventional artist curves as well.

\paragraph{Visibility-indicating curves.}
\textit{Contours} indicate where parts of the surface become visible and invisible, and also indicate where visibility changes. There are a few other important curves that are important for similar reasons.

\textit{Boundary curves} are simply the boundaries of the surface. Closed surfaces do not have boundaries. These curves are usually rendered when visible. Boundary curves can indicate change of visibility for curves that they intersect in image space, so they are important to handle, and we include them in the discussions of our algorithms in this tutorial.

\textit{Surface intersection} curves occur when two different sections of surface intersect. These do not occur in the clean models often used in computer graphics research. However, in professional 3D computer animation applications, modelers frequent connect different object parts this way (Figure \ref{fig:red_intersections}).  These curves can be detected with standard computer graphics algorithms, and are important to extract since they can indicate changes of visibility with curves that they intersect on the surface. We do not discuss them any further in this tutorial.

\begin{figure}
	\centering
	\small
  \textbf{(a)}\hspace{-10pt} \includegraphics[width=1.75in]{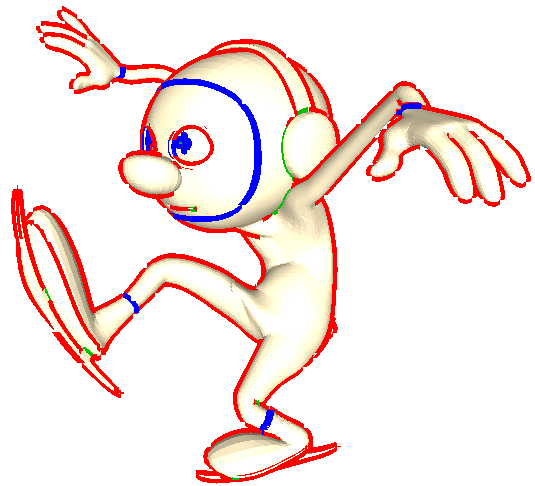}  
  \hspace{-10pt}\textbf{(b)}
	\includegraphics[width=1.75in]{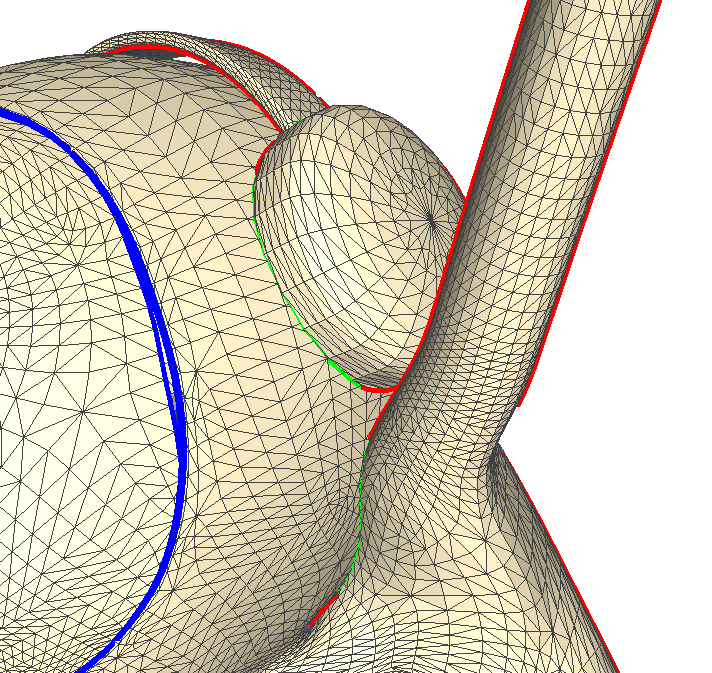}
	\hspace{-6pt}\textbf{(c)}~\includegraphics[width=1.8in]{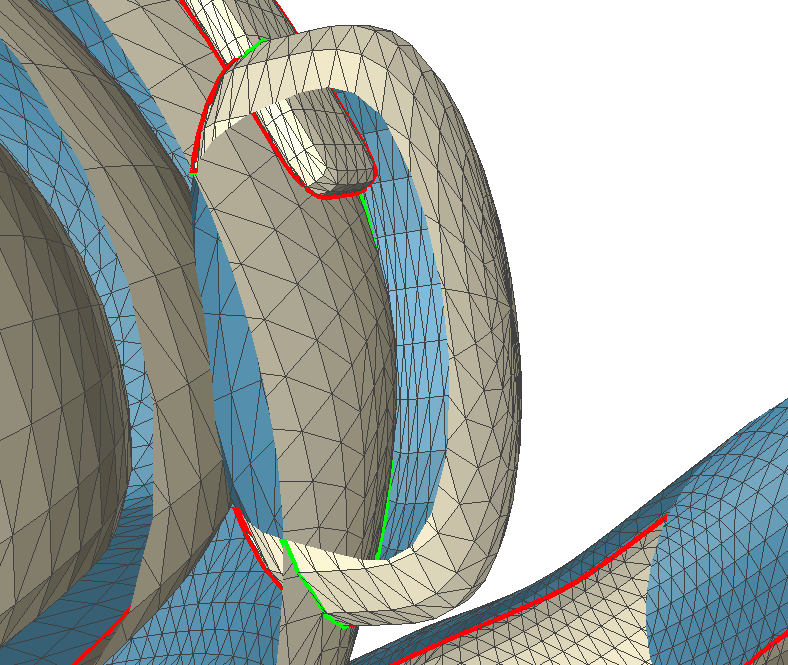}
	\caption{\textbf{Surface-surface intersection curves} (from \citet{Benard:2014}) ---
		Professionally-modeled surfaces include many intersections between surface, such as this ice-skating character.
		Surface intersection curves are shown in green, occluding contours in red, and boundaries in blue. Observe how the ear muffs intersects the headband and the hoodie; the shoulder also happens to intersect the hoodie in this animation frame. \textbf{(a,b)} Original surface. \textbf{(c)} Cross-section from a different viewpoint.
		(``Red'' character created at Pixar by Andrew Schmidt, Brian Tindall, Bernhard Haux and Paul Aichele, based on the original design of Teddy Newton.)
		\label{fig:red_intersections}}
\end{figure}

All other curves are essentially surface ``decorations''; computing them is optional for visibility computations.  They typically visualize the surface curvature rather than its outlines.

%

\paragraph{Contour generalizations.}
Perhaps the next most significant set of curves are those that generalize contours. These curves were first introduced by DeCarlo \etal~(\cite*{DeCarlo:2003,DeCarlo:2004}), who described a mathematically-elegant generalization of contours and the algorithms needed to render them.  Several other variants inspired by this idea were proposed, including apparent ridges \citep{Judd:2007}. 

The abstracted shading method \citep{Lee:2007} demonstrated how these and lighting-based variants could be computed in image-space.
Other variants based on image-space processing include Laplacian Lines \citep{Zhang:2009} and DoG lines (Zhang \etal~\cite*{Zhang:2012,Zhang:2014}). In addition to speed, image-space lines have the advantage that they automatically remove clutter as a function of image-space line density, although, like all image-based methods, they potentially lose some fine-scale precision and control. 

\fig{contour_generalization} shows some examples of these contour generalizations. 
Including some form of these curves seems essential for capturing how artists depict surfaces; these curves were essential in the study of \citet{Cole:2008}.

These curves have also been generalized to include highlights that illustrate shading on an object \citep{DeCarlo:2007}. \citet{DeCarlo:2012} provides a thorough survey and comparison of these different types of contour generalizations.

\begin{figure}
	
	\includegraphics[width=0.32\linewidth]{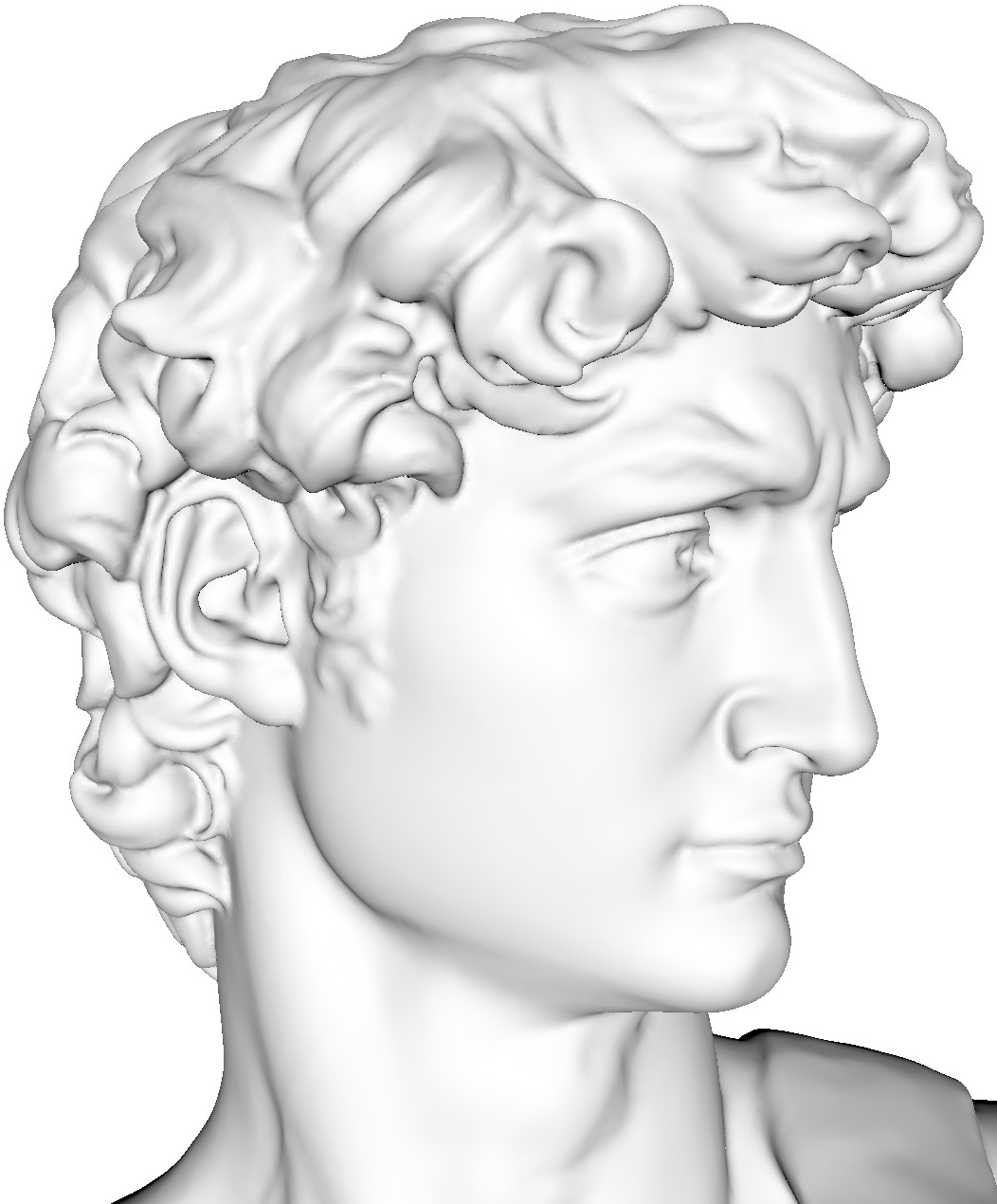}
	\includegraphics[width=0.32\linewidth]{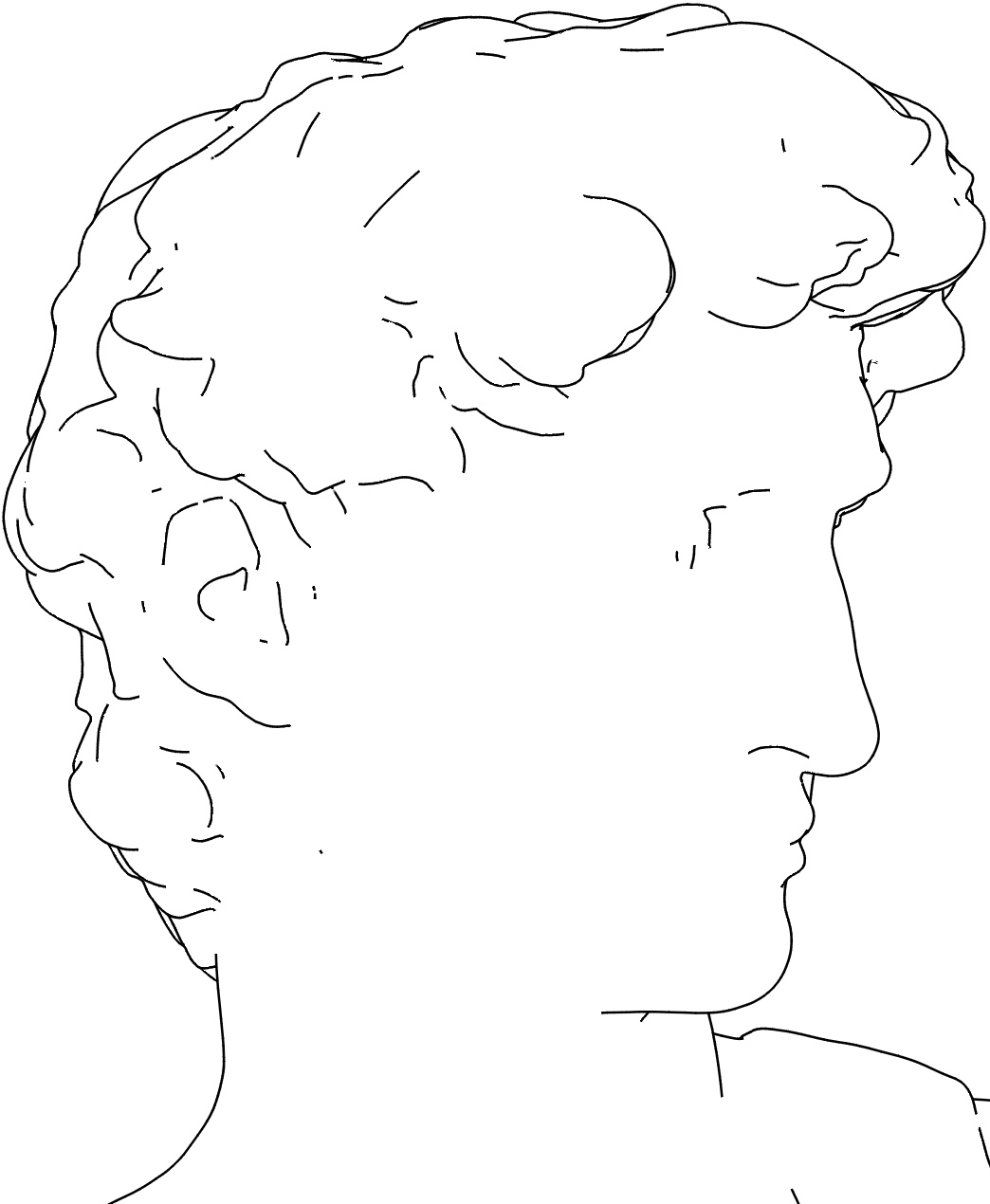}
	\includegraphics[width=0.32\linewidth]{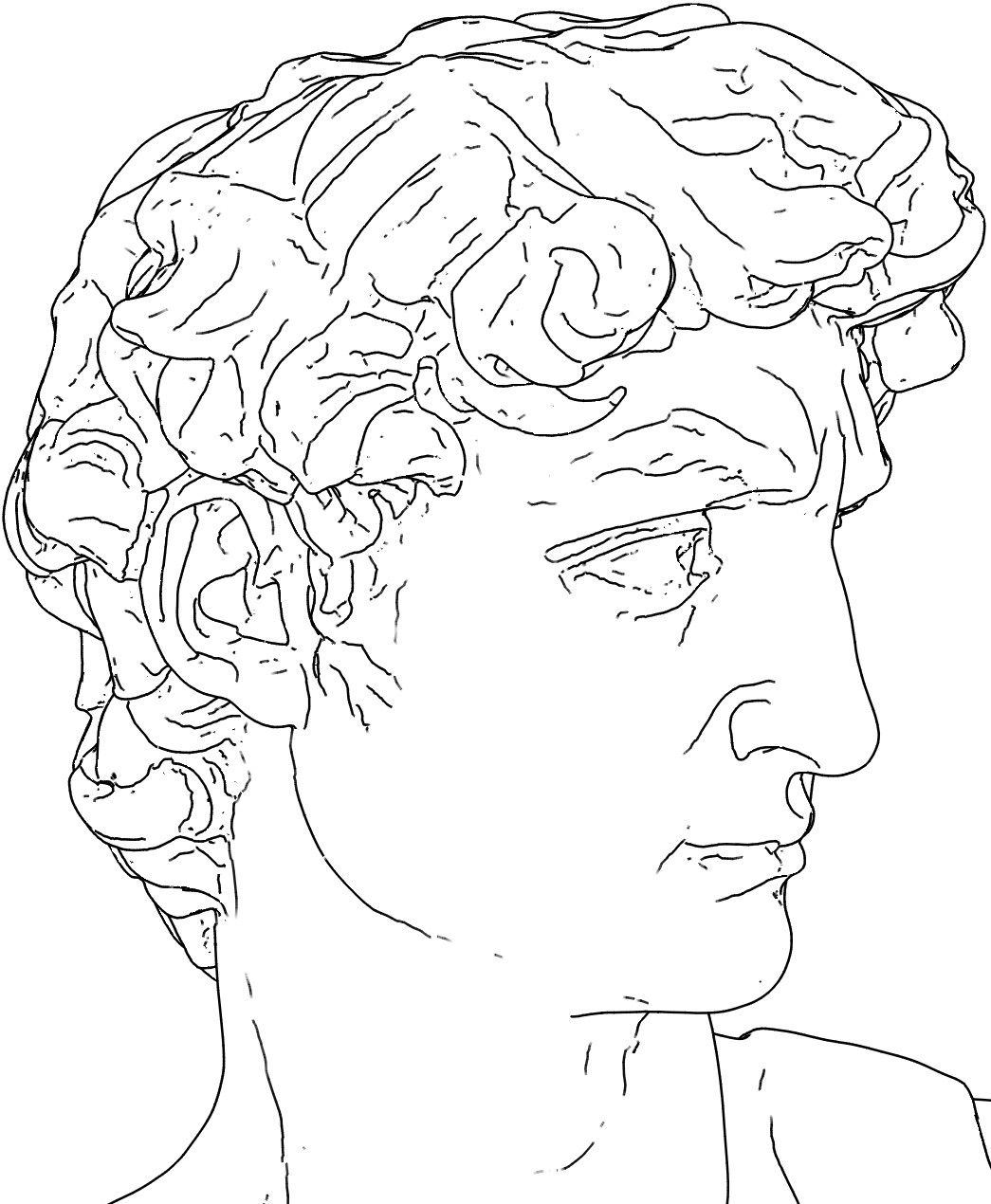}\\[0.5em]
	\includegraphics[width=0.32\linewidth]{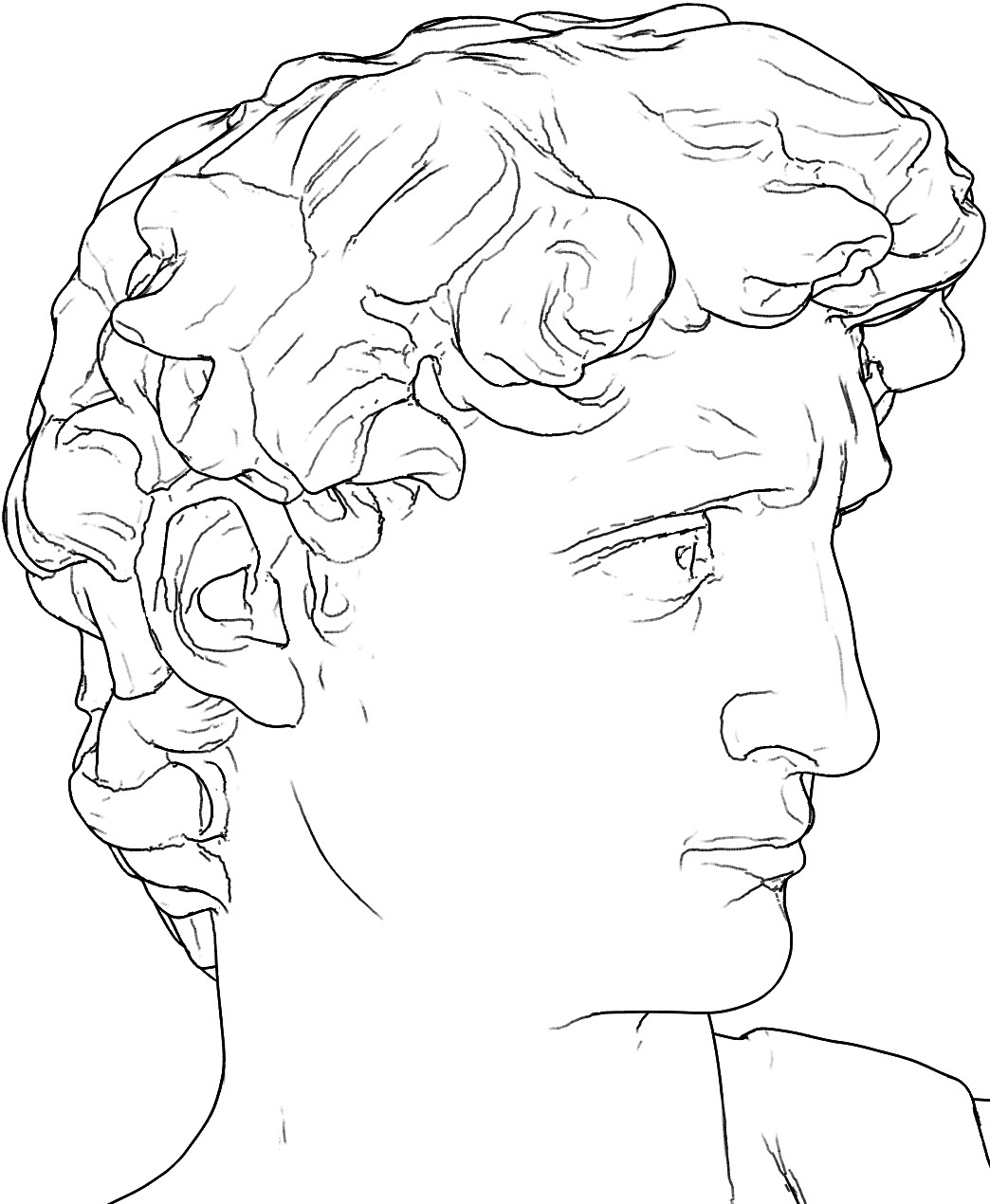}
	\includegraphics[width=0.32\linewidth]{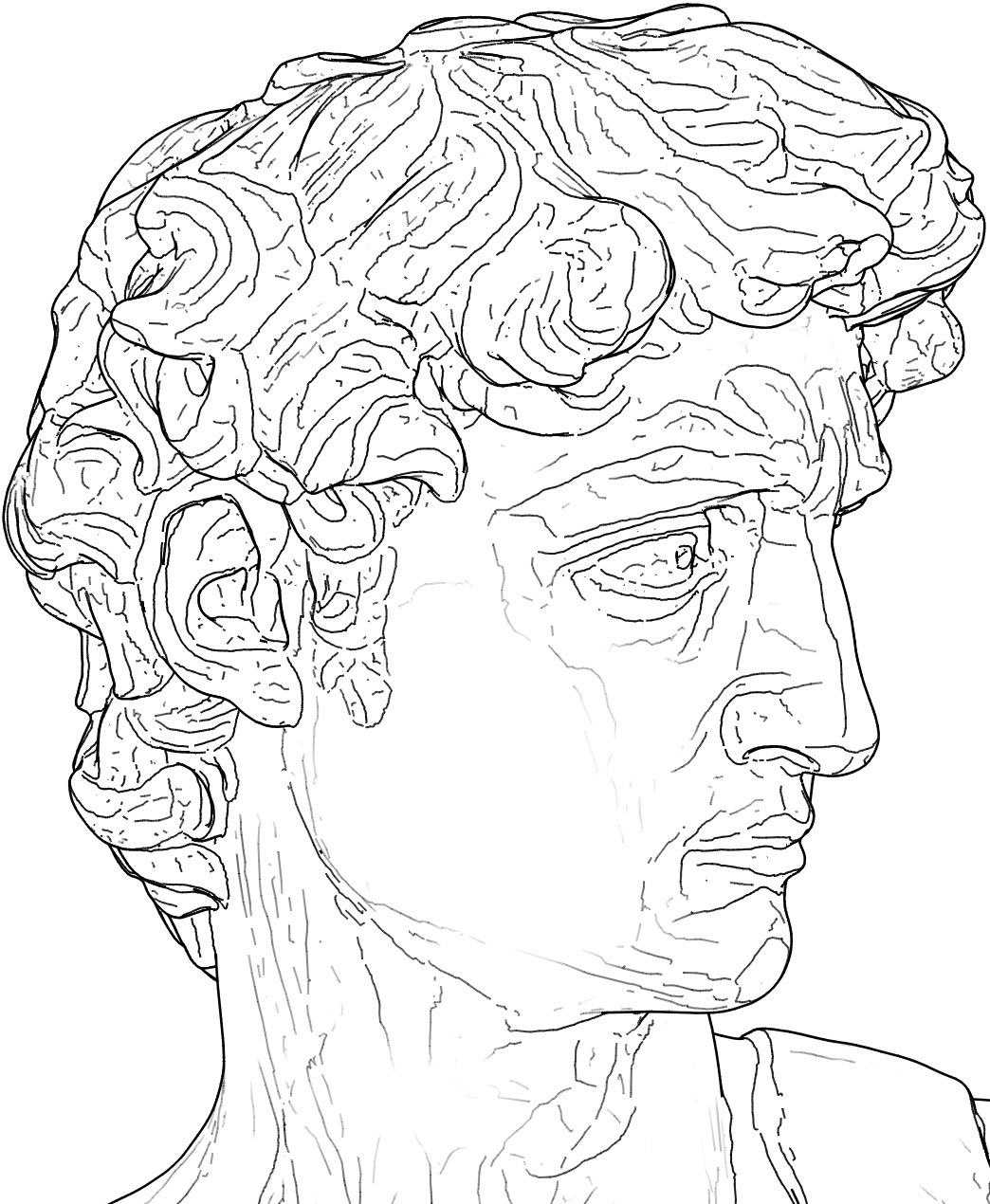}
	\includegraphics[width=0.32\linewidth]{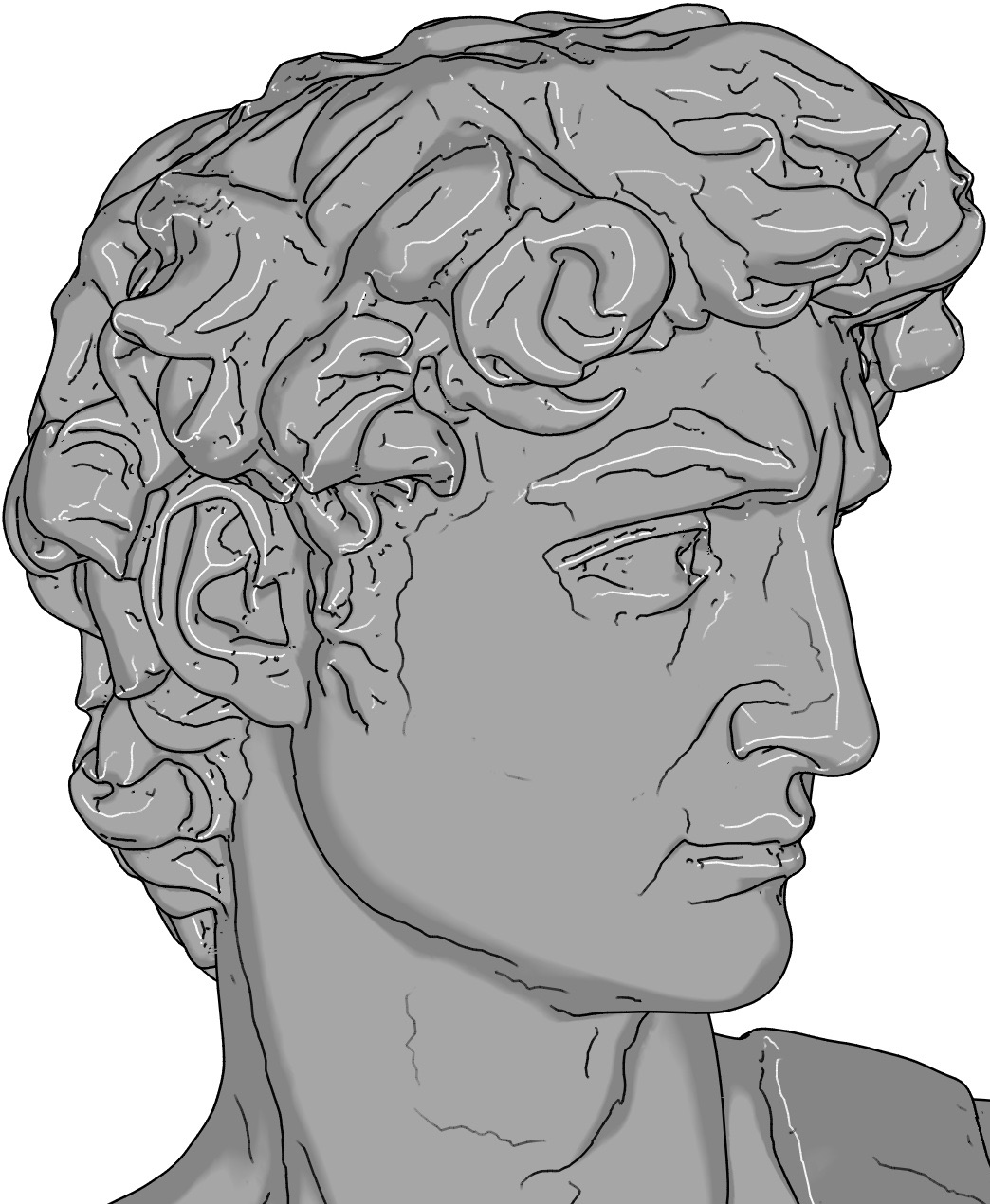}
	\caption{\textbf{Feature curve examples} --- From left to right, top to bottom: diffuse rendering of the 3D scanned David model by ``Scan The World'' (\url{http://mmf.io/o/2052}), occluding contours (OC), OC + suggestive contours (SC) \citep{DeCarlo:2004}, OC + apparent ridges \citep{Judd:2007}, OC + ridges \& valleys \citep{Rusinkiewicz:2004}, and OC + SC + principal highlights + toon shading \citep{DeCarlo:2007}.
		Images generated with ``qrtsc'' \citep{qrtsc}.
	}\label{fig:contour_generalization}
\end{figure}

\paragraph{Surface features/properties.}
Some intrinsic properties of the surface can be drawn, such as sharp creases on smooth surfaces \citep{Saito:1990}, as well as changes in shading \citep{Xie:2007}. When objects have assigned texture and materials, one may wish to draw the material boundaries or the texture itself.

\begin{figure}
	\centering
	\begin{tabular}{c|c}
		\includegraphics[width=0.21\linewidth]{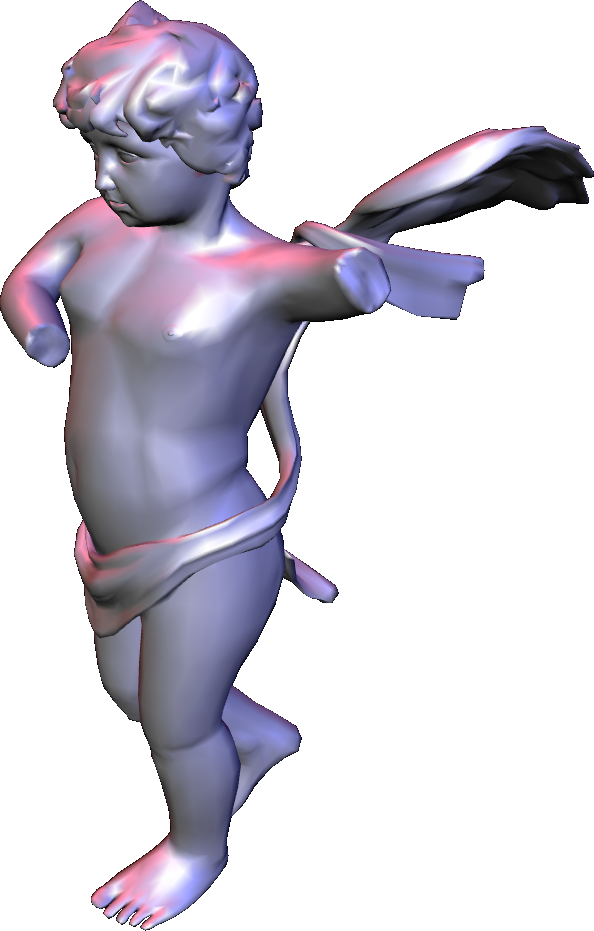}
		\includegraphics[width=0.21\linewidth]{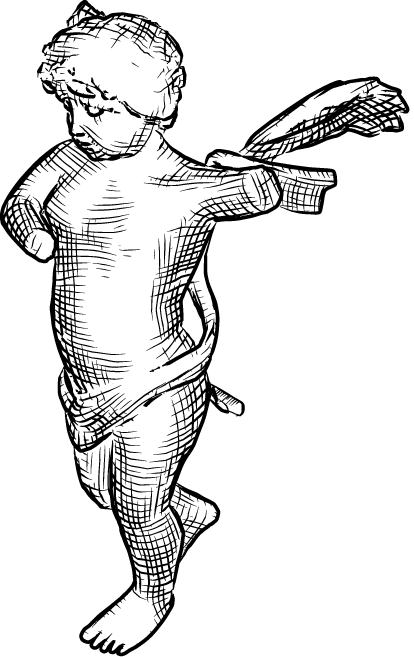} \hspace{1pt} & \hspace{1pt}
		\includegraphics[width=0.23\linewidth]{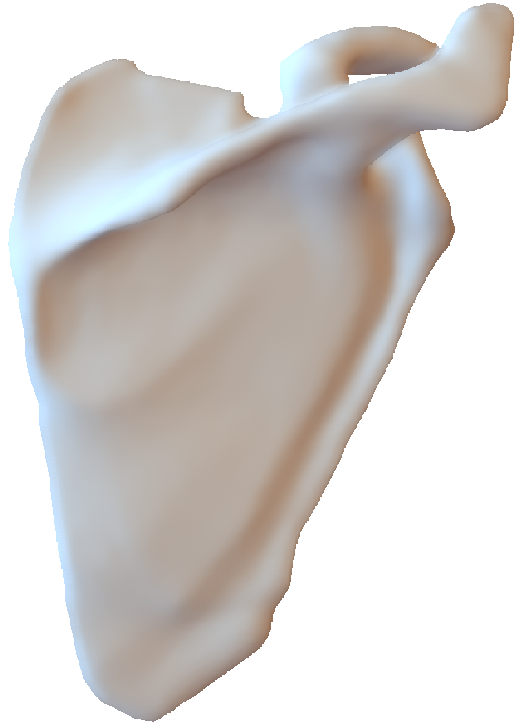}
		\includegraphics[width=0.23\linewidth]{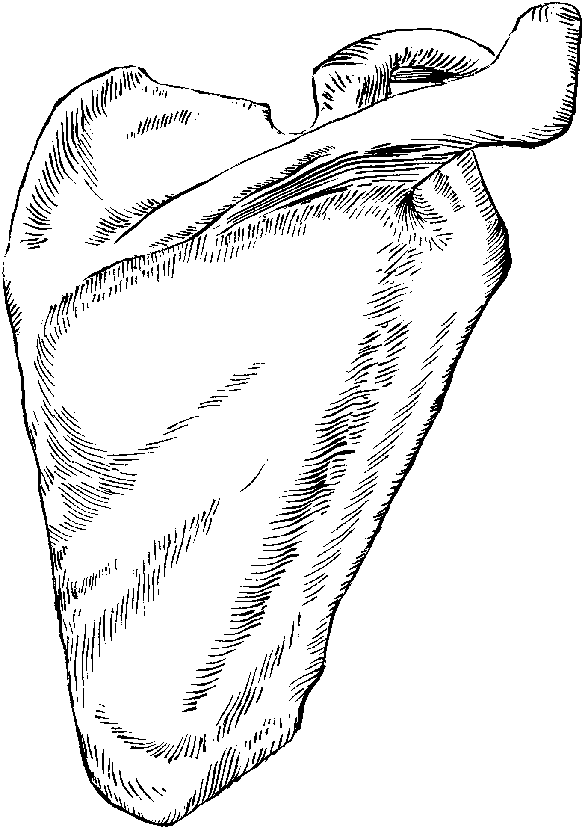}
	\end{tabular}
	\caption{\textbf{Hatching} --- Left, 3D Cupid model and hatching result obtain with the automatic method of \citet{Hertzmann:2000}; right, shoulder blade model and user-edited result of \citet{Gerl:2013}.}\label{fig:hatching}
\end{figure}

\paragraph{Hatching.}
Hatching strokes illustrate surface shape in line drawings.
\citet{Winkenbach:1994} use manually-authored hatching textures and orientations.
For more automation, one can use the iso-parametric curves of parametric surfaces \citep{Elber:1995a,Winkenbach:1996}. However, these lines depend on how the shape was authored, and do not generalize to other types of surfaces. \citet{Elber:1998} explored many different possible hatching directions, including principal curvature directions, texture gradients, and illumination gradients.
Principal curvature-based hatching is supported by perceptual studies suggesting that human perceive hatching strokes as curvature directions \citep{Mamassian:1998}.
\citet{Hertzmann:2000} refine principal curvature hatching for umbilic regions (\fig{hatching}, left). \citet{Singh:2010} describe hatching strokes that follow shading gradients.
Since artists draw different types of hatching curves in different situations,
\citet{Kalogerakis:2012} combine these ideas, describing a machine learning system for learning hatching directions, identifying which hatching rules are used in which parts of a 3D surface. \citet{Gerl:2013} additionally offer interactive tools to let the user dynamically control the placement and orientation of hatches (\fig{hatching}, right).

\paragraph{Surface extrema.}
Extremal curves, such as ridges and valleys, generalize the notion of ridges and valleys in terrain maps, identifying curves of locally maximal or minimal curvature. These types of curves are a visualization technique that can be useful in understanding surface shape; they do not typically correspond to artist-drawn curves otherwise.

Numerous algorithms have been developed to extract ridges and valleys from various types of geometric models \citep{Interrante:1995,Thirion:1996,Pauly:2003,Rusinkiewicz:2004,Ohtake:2004,Yoshizawa:2007,Vergne:2011}.  A variant, called Demarcating Curves \citep{Kolomenkin:2008,DeCarlo:2012} can help visualize shapes of different regions on a surface.

\section{Brief history of 3D Non-Photorealistic Rendering}

The earliest 3D computer graphics algorithms were hidden-line rendering algorithms \citep{Roberts:1963}, including methods that we discuss in this tutorial \citep{Appel:1967,Weiss:1966:VPI:321328.321330}. While the mainstream of computer graphics focused on photorealistic imagery, a few works aimed at adding artistic stroke textures to architectural drawings and technical illustrations\footnote{Many works are being omitted from this history. A much more comprehensive bibliography, up to 2011, can be found here: \url{https://www.npcglib.org}.}, e.g., \citep{Dooley:1990,Yessios:1979}; meanwhile a number of 2D computer paint programs were developed as well. Many of these papers argued for the potential virtues of hand-drawn styles in technical illustration.   

In 1990, the flagship computer graphics conference SIGGRAPH held a session entitled ``Non Photo Realistic Rendering,'' which seems to be the first usage of this term. In this session, two significant papers for the field were presented. \citet{Saito:1990} introduced depth-buffer based line enhancements (Chapter \ref{chap:image_space}), which started to create cartoon-like renderings of smooth objects by emphasizing contours and other feature curves. \citet{Haeberli:1990} introduced a range of artistic 2D image-processing effects; these papers together demonstrated a significant step forward in the quality and generality of non-photorealistic effects.  

\citet{Winkenbach:1994} demonstrated the first complete line-drawing algorithm from 3D models, including contours and hatching.
Their work was seminal in that their method automatically produced beautiful results from 3D models; one could, for the first time, be fooled into thinking that these images were really drawn by hand.
Perhaps even more importantly, their work  provided a model for one could develop algorithms by careful study of artistic techniques in textbooks and illustrations.

\citet{Meier:1996} demonstrated the first research paper focusing on 3D non-photorealistic animation, describing the problem of temporal coherence for animation.  Between the beautiful images of Winkenbach and Salesin~(\cite*{Winkenbach:1994,Winkenbach:1996}) and beautiful animations of \citet{Meier:1996}, non-photorealistic rendering was firmly established as an important research direction.

Research activity at SIGGRAPH increased significantly, and the inaugural NPAR symposium on Non-Photorealistic Animation and Rendering met in 2000, sponsored by the Annecy Animation Festival in France and chaired by David Salesin and Jean-Daniel Fekete.
Through the following decade, many improvements and extensions to the basic ideas were published, and, occasionally, techniques like toon shading and contour edges appeared in video games and movies. \citet{DeCarlo:2003} described Suggestive Contours, which substantially improved the quality of line renderings, while making deep connections to perception and differential geometry, notably the work of \citet{Koenderink:1984}. Several systems were created to help artists design artistic rendering styles, such as WYSIWYG NPR \citep{Kalnins:2002} and a procedural NPR system called Freestyle \citep{Grabli:2010}. 
 Cole \etal~(\cite*{Cole:2008,Cole:2009}) performed the scientific studies described in Section \ref{sec:survey} demonstrating that line drawing algorithms were quite good at capturing how artists draw lines.

Since then, research in 3D non-photorealistic rendering has tapered off, despite the presence of several significant open problems. In contrast, interest in image stylization has recently exploded, due to developments in machine learning.  
Still, 3D non-photorealistic rendering continues to appear in a few games and movies here and there.  
This tutorial aims, in part, to summarize the field and highlight open problems, to help researchers and practitioners make progress in this field in order to enable them to be more widely used.  We discuss future prospects for the field in the Conclusion (Chapter \ref{chap:conclusion}).


\chapter{Image-Space Curves}
\label{chap:image_space}

We begin in this chapter by describing simple image-space algorithms. This will provide some informal examples of contours. We will then formally define these curves in the next section.

Most computer graphics rendering systems, including OpenGL, have a way to extract a depth buffer from a rendered scene, such as the one shown in \fig{pig_depth}. The gray level of each pixel shows the distance of the object from the viewer.

\begin{figure}
	\centering
	\begin{subfigure}[b]{0.45\linewidth}
		\includegraphics[width=\linewidth]{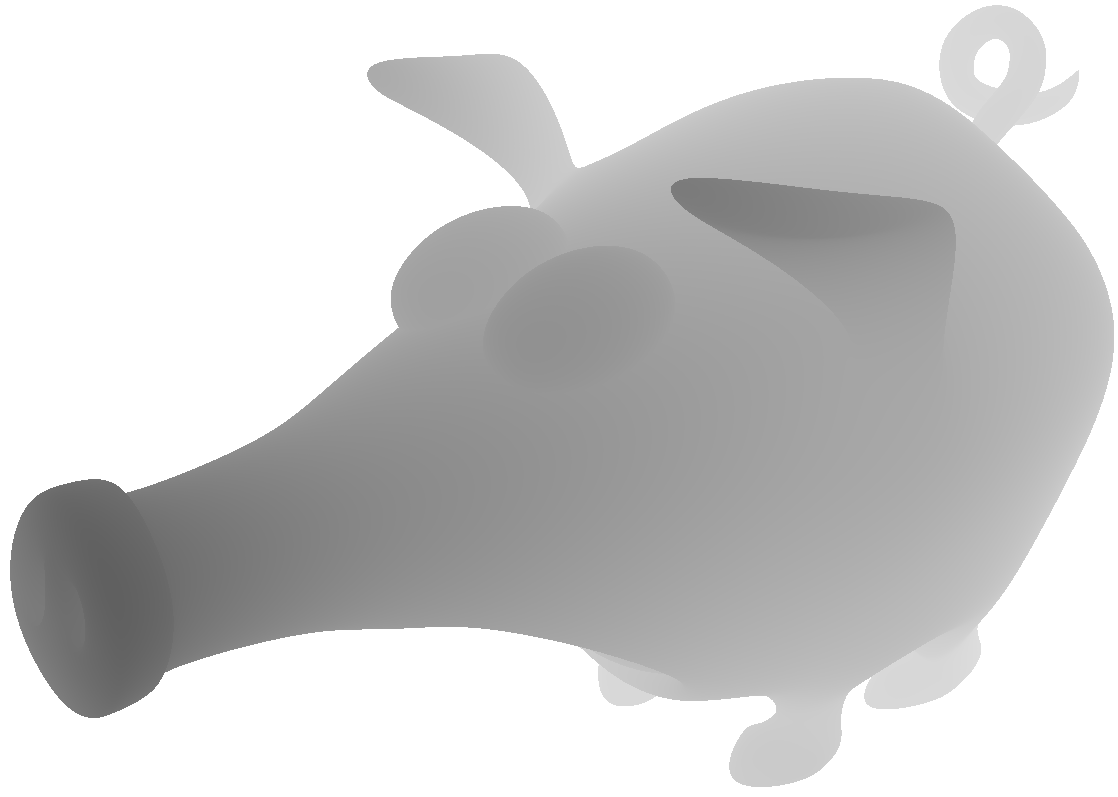}
		\caption{linearized depth}\label{fig:pig_depth}
	\end{subfigure}
	\quad
	\begin{subfigure}[b]{0.45\linewidth}
		\includegraphics[width=\linewidth]{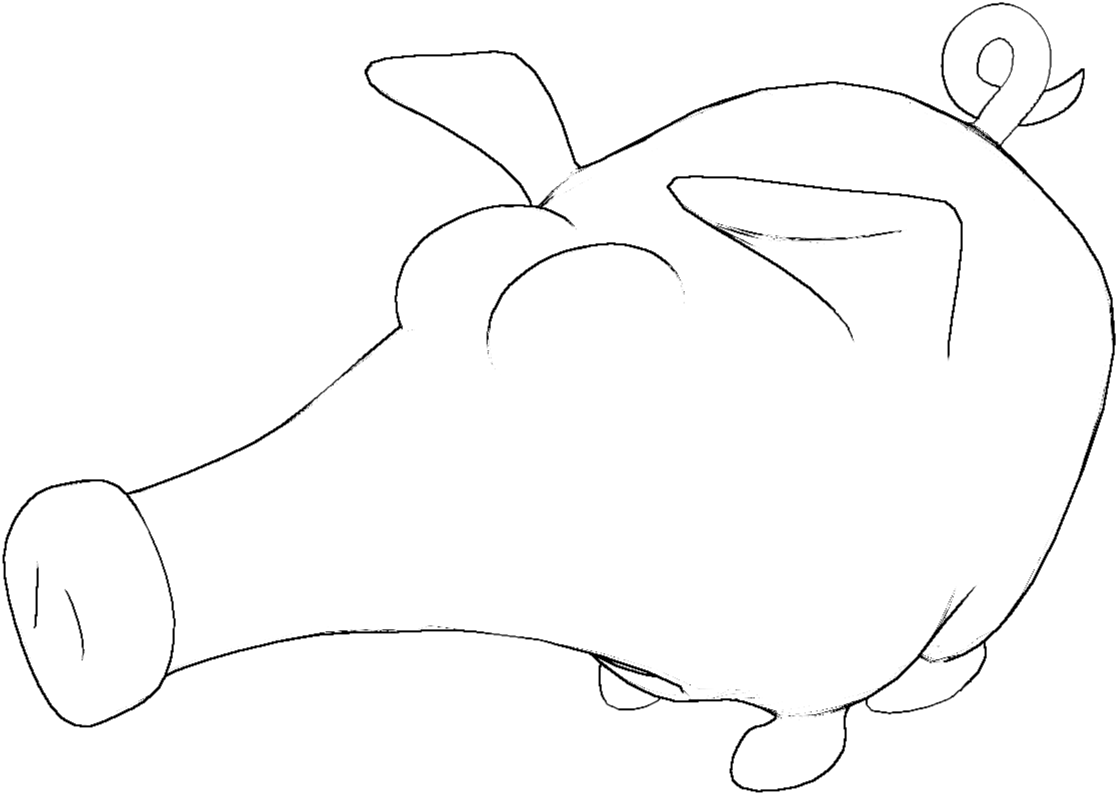}
		\caption{depth discontinuities}\label{fig:pig_edges}
	\end{subfigure}
	\begin{subfigure}[b]{0.45\linewidth}
		\includegraphics[width=\linewidth]{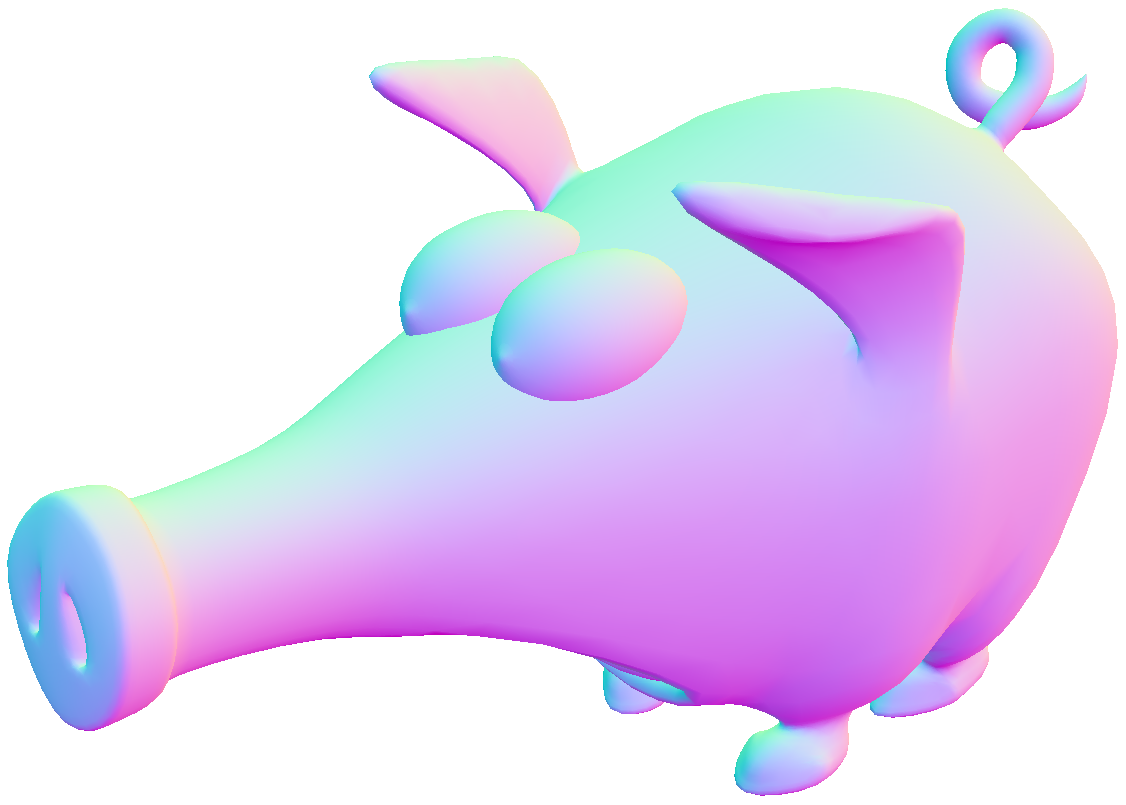}
		\caption{image-space normals}\label{fig:pig_normals}
	\end{subfigure}
	\quad
	\begin{subfigure}[b]{0.45\linewidth}
		\includegraphics[width=\linewidth]{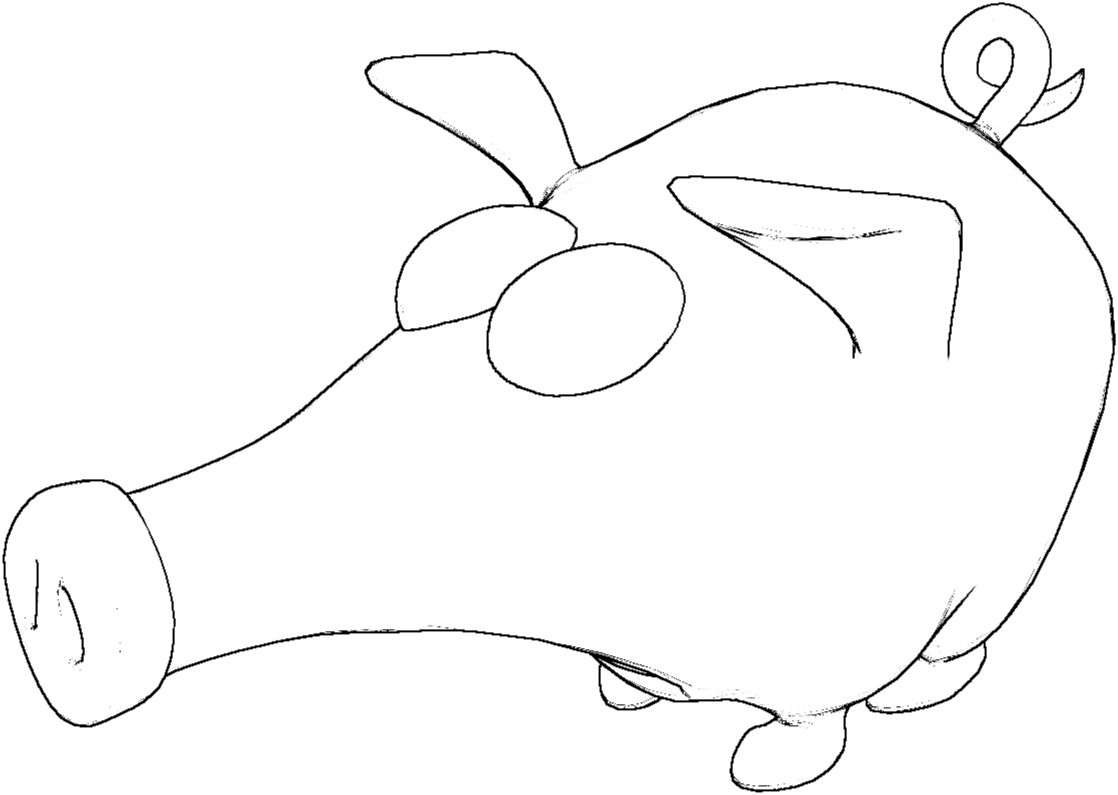}
		\caption{depth + normals discontinuities}\label{fig:pig_normals_edges}
	\end{subfigure}
		\begin{subfigure}[b]{0.45\linewidth}
		\includegraphics[width=\linewidth]{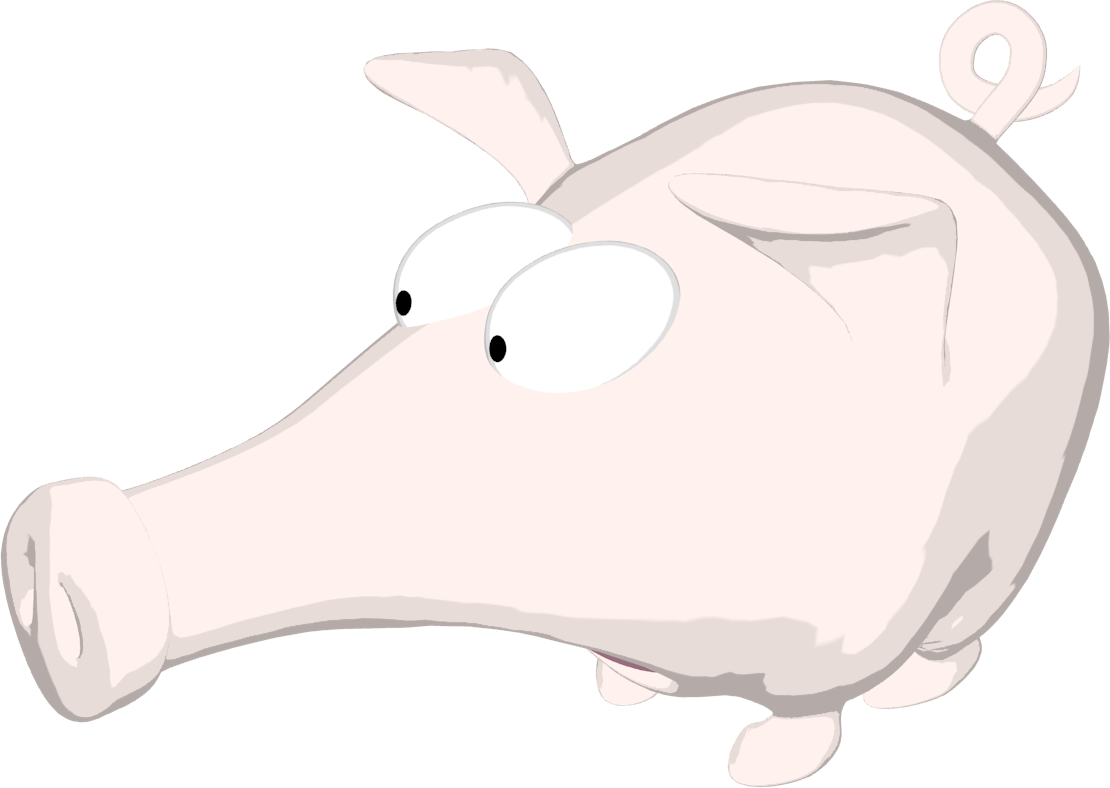}
		\caption{toon shading}\label{fig:pig_colors}
	\end{subfigure}
	\quad
	\begin{subfigure}[b]{0.45\linewidth}
		\includegraphics[width=\linewidth]{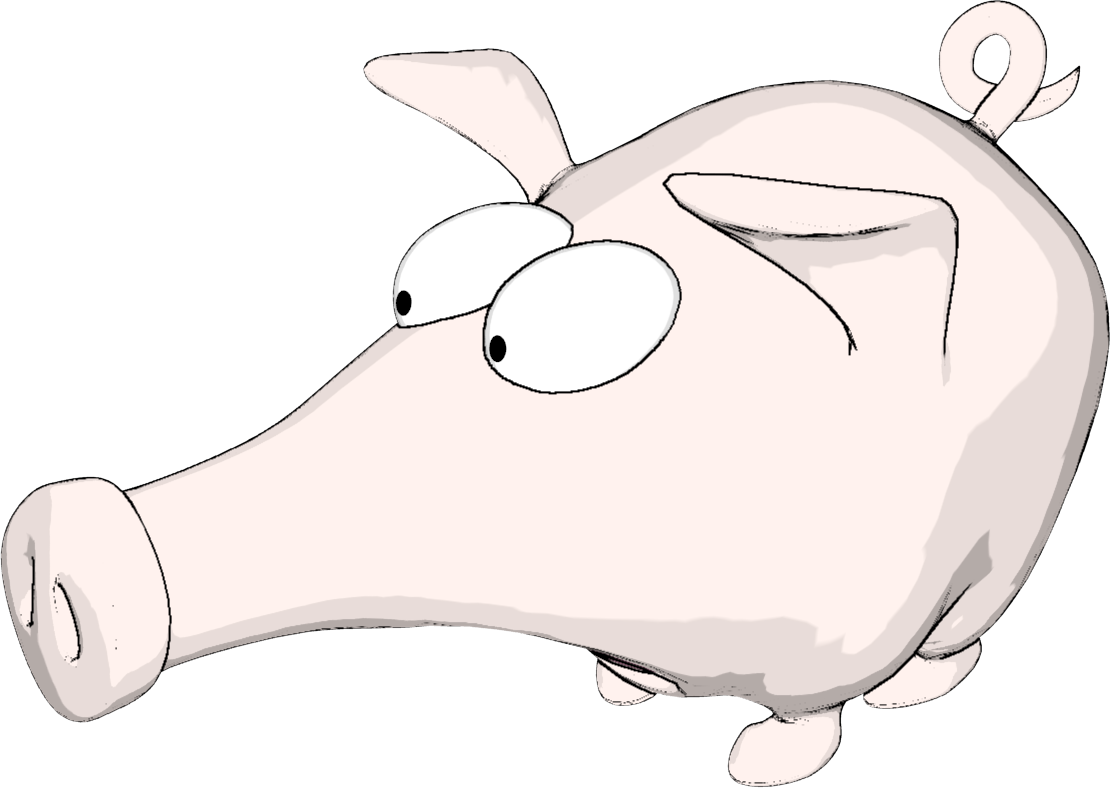}
		\caption{egdes composited with toon}\label{fig:pig_comp}
	\end{subfigure}
	\caption{\textbf{Image-space edges} --- The depth buffer of the scene (``Origins of the Pig'' \ccCopy~Keenan Crane) is obtained by rasterization and linearized \textbf{(a)}, depth discontinuities are then extracted by filtering, here, with a Laplacian of Gaussian filter \textbf{(b)}; normals discontinuities can also be considered \textbf{(c)} to extract creases; the final edges can eventually be re-composited with the color buffer \textbf{(f)}. Images computed with BlenderNPR Edge Node plugin~\citep{BNPREdge}. The pupils are added separately, as materials on the surface.}
	\label{fig:image_space}
\end{figure}
Depth discontinuities in this image correspond to contours. To find them, we can apply an edge detection filter to this image \citep{Saito:1990}, producing the image in \fig{pig_edges}: these are the occluding contours of the surface.
The key assumption of this method is that depth variations between adjacent pixels are small for continuous smooth surfaces, but become large at occlusions. 
We can also compute a separate normal map image \citep{Decaudin:1996} shown in \fig{pig_normals}, and compute its edges, which adds edges at surface creases (\fig{pig_normals_edges}).

Image-space algorithms work by performing image processing operations on buffers like these ones. They are simple to implement and can run in real-time on graphics hardware. However, they provide limited control over stylization. For example, one cannot easily draw a pencil stroke over the outlines, because there is no explicit curve representation; they are just pixels in a buffer. Furthermore, they can be incorrect, for example, missing contours at small discontinuities or falsely detecting them for highly foreshortened surfaces. 

These kinds of edges were used in the video game ``Borderlands''; some of the issues involved in getting them to work are described by \citet{Thibault:2010}. 




We now describe the depth edge detection algorithm in more detail. A standard choice from image processing is the Sobel filter, which computes approximate depth derivatives (2D gradients) by discrete convolution of the depth buffer $D$ with the following kernels:
\[
	S_x = \left[ \begin{array}{ccc}
		-1 & 0 & 1 \\
		-2 & 0 & 2 \\
		-1 & 0 & 1 
	\end{array}\right]
\hspace{1cm}
	S_y = \left[ \begin{array}{ccc}
		-1 & -2 & -1 \\
		0 & 0 & 0 \\
		1 & 2 & 1 
	\end{array}\right].
\]
The edge image is then obtained by computing their magnitude:
$$G(x,y) = \sqrt{(D(x,y) \otimes S_x)^2 + (D(x,y) \otimes S_y)^2},$$ 
and thresholding it by a user-defined threshold $\tau$:
\[
	Edge(x,y) = 
	\begin{cases} 
		1 & \quad \text{if } G(x,y) \geq \tau \\
		0 & \quad \text{if } G(x,y) < \tau
	\end{cases}
\]
The results are demonstrated in Figure \ref{fig:image_space}.
The $3 \times 3$ Sobel kernels are the most computationally efficient, but they tend to produce noisy results. As suggested by \citet{Hertzmann:1999}, they can favorably be replaced by the ``optimal'' $5 \times 5$ kernels of \citet{Farid:1997}.
%
Alternatively, second-order derivatives can also be considered, using, for instance, the Laplacian-of-Gaussian filter, or the separable approximation provided by the Difference-of-Gaussians filter~\citep{Marr:1980} and its artistic extensions~\citep{Winnemoller:2012}.

Note that the depth edge image contains not just contours, but also object boundaries. The normal edge image often includes contours and boundaries, as well as surface-intersection curves. Distinguishing these types of curves, if desired, would be difficult.

As noted by \citet{Deussen:2000}, GPU depth buffers store non-linear depth values in screen-space, hence depth gradients for remote objects correspond to much larger differences in eye coordinates. If this effect is not desirable, the depth value $d \, (d \in [0..1])$ first needs to be linearized according to the camera near $z_0$ and far $z_1$ clipping plane distances:
\begin{align*}
	z &= \frac{\frac{z_0 z_1(d_1-d_0)}{z_1-z_0}}{d-\frac{(z_1+z_0)(d_1-d_0)}{2(z_1-z_0)}-\frac{d_1+d_0}{2}}
\end{align*}
where $d_0$ and $d_1$ are the minimal and maximal values represented in the depth buffer. Alternatively, with modern graphics hardware, the linear camera z-value can directly be written into an offscreen buffer.

\section{Discussion and extensions}

Image-space algorithms only depend on the final image resolution, which is usually an advantage performance-wise, making this approach popular for real-time applications such as games \citep{Thibault:2010}. They naturally omit tiny, irrelevant details. However the results may not consistent and predictable when the resolution of the image changes.

\paragraph{Additional buffers.}
One limitation of depth-buffer algorithms is that they cannot detect edges between objects that are close in depth, such as a foot contact on the ground. However, they can be easily extended to other line definitions by rendering and filtering a \emph{G-buffer} containing, for instance, per-pixel object IDs and surface normals (\fig{pig_normals}). The former solves the depth ambiguity, whereas first-order normal discontinuities correspond to creases~\citep{Saito:1990,Decaudin:1996,Hertzmann:1999,Nienhaus:2004}, and second- and third-order screen-space tensors allow to extract view-dependent ridges and valleys as well as demarcating curves~\citep{Vergne:2011}.

\begin{figure}
	\centering
		\includegraphics[width=0.45\linewidth]{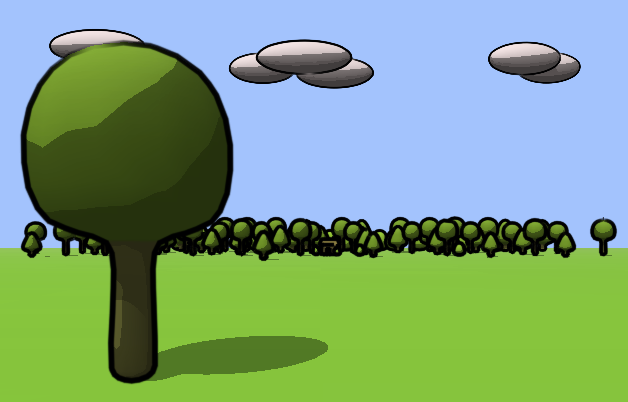}
		\quad
		\includegraphics[width=0.45\linewidth]{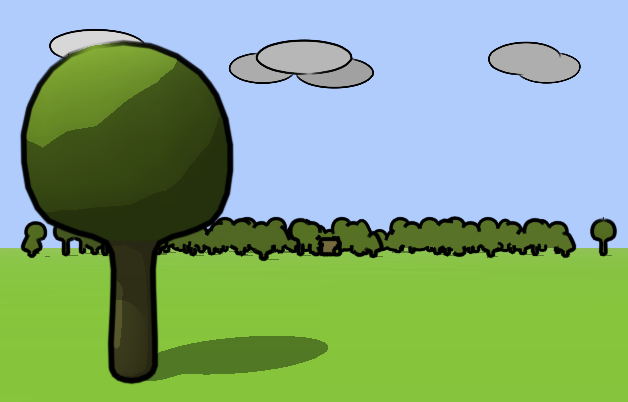}
		\caption{\textbf{Region segmentation} \citep{Kolliopoulos:2006} --- Toon rendering of a forest scene with no segmentation (left) exhibiting cluttering in the background. With segmentation (right), many of the background trees are grouped together. Contours are only drawn near segment boundaries, resulting in a cleaner image.}\label{fig:segment}
\end{figure}

\citet{Kolliopoulos:2006} render scenes by hardware ID buffers to determine pixel-wise object visibility. The scene is adaptively grouped into regions using a segmentation algorithm; the user may determine the grouping parameters so that small objects are grouped together. This is similar to computing planar maps, which will be discussed in more detail in Section \ref{sec:planar_map}. These planar maps are then stylized in image space (\fig{segment}).

\paragraph{Stroke stylization.}
These filtering techniques produce a set of disconnected pixels. Hence, modifying the appearance of strokes first requires extracting approximate curves from the buffer. 
One solution is to fit parametric curves to the edge image using vectorization algorithms (\eg{} \citep{Favreau:2016,Bo:2016}), but this tends to introduce inaccuracies and is often too slow for real-time applications. 

To create sketchy drawings, \citet{Curtis:1998} proposed particles that trace small line segments in the vicinity of the extracted contours. These particles are guided by a density image and a force field which can be obtained by calculating unit vectors perpendicular to the depth buffer's gradient. Although this technique is appealing for its dynamic behavior, the range of style that it can achieve is somewhat limited, and the particle simulation is computationally expensive. 

To produce lines of controllable thickness, \citet{Lee:2007} proposed to fit a simple analytic profile (degree-2 polynomial) to every pixel of a luminance image, viewed as a height field. This profile locally describes the shape of the closest illumination ridge (or valley). The thickness and opacity of the lines can then be computed based on the distance to the ridge or valley line and its first principal curvature. \citet{Vergne:2011} generalized this idea in two ways: first by fitting profiles to various surface features, and then by convolving these profiles with a brush footprint to produce various stylization effects (\fig{IB}). This method achieves real-time performance on the GPU and exhibits a natural coherence in animation. However, it is limited to brush styles which are independent of the contour's arc-length; for example, it does not support a brush stroke texture that curves around the object.

\begin{figure}
	\centering
		\includegraphics[width=0.45\linewidth]{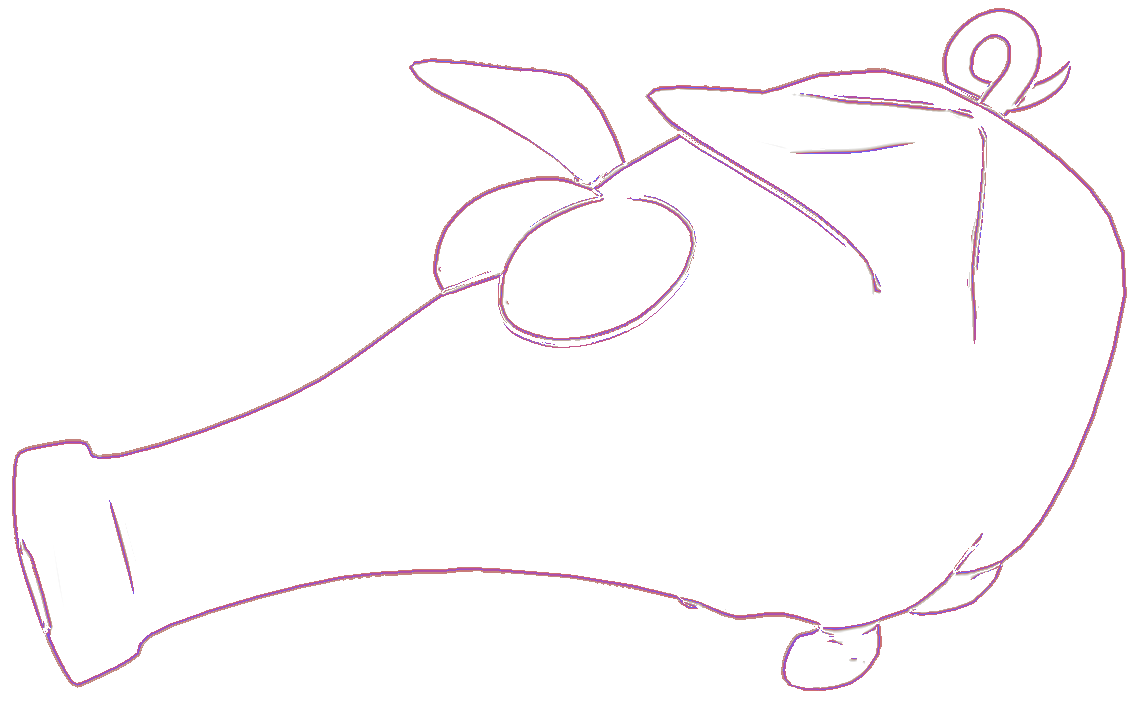}
		\quad
		\includegraphics[width=0.45\linewidth]{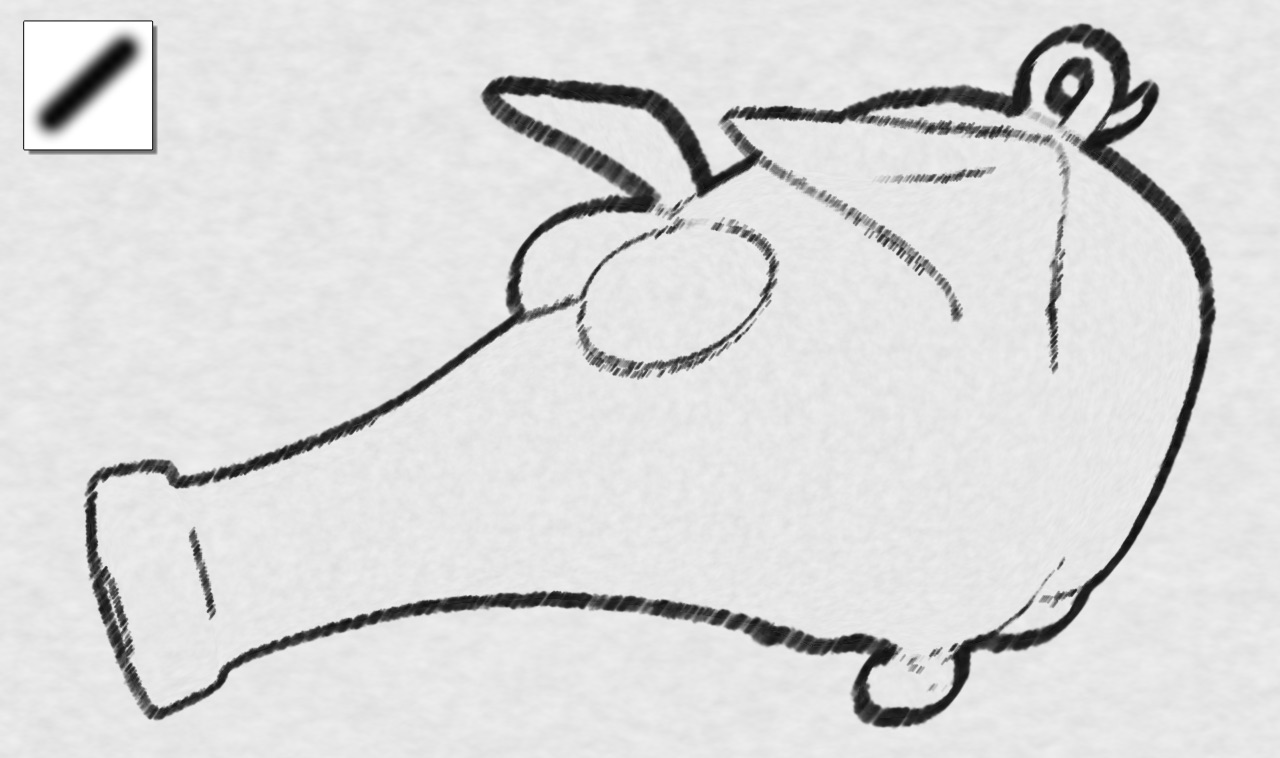}
		\caption{\textbf{Implicit brushes} \citep{Vergne:2011} --- Surface feature profiles (left) are extracted in image-space and fitted with polynomials; they are then convolved with a brush footprint (inset) to produce stylized lines (right).}\label{fig:IB}
\end{figure}

\paragraph{Raytracing framework.}
With a raytracer, a G-buffer can still be computed by casting a ray per pixel and storing the relevant information (\eg~distance to the camera, normal, etc.) at the closest hit point \citep{Leister:1994,Bigler:2006}. To avoid explicitly storing this buffer and allow the user to control the line width, \citep{Choudhury:2009} developed a method inspired by cone-tracing. For each per-pixel ray, they sample a set of concentric ``probe'' rays in an screen-space disc whose radius corresponds to the half line width, that hey call a ray stencil (\fig{samples}). Then, they compute an edge strength metric based on the proportion of probe samples falling on the same primitive as the central ray. The final pixel color is modulated by this edge factor, producing naturally anti-aliased lines.
\citet{Ogaki:2018} both simplify and extend this approach to better deal with line intersections, and allow line thickness and color variations. They also support drawing lines in specular reflections and refractions (\fig{pig_glass}), at the price of storing the tree of reflections and refractions events associated with every pixel ray.

\begin{figure}
	\centering
	\begin{subfigure}[b]{0.35\linewidth}
		\includegraphics[width=\linewidth]{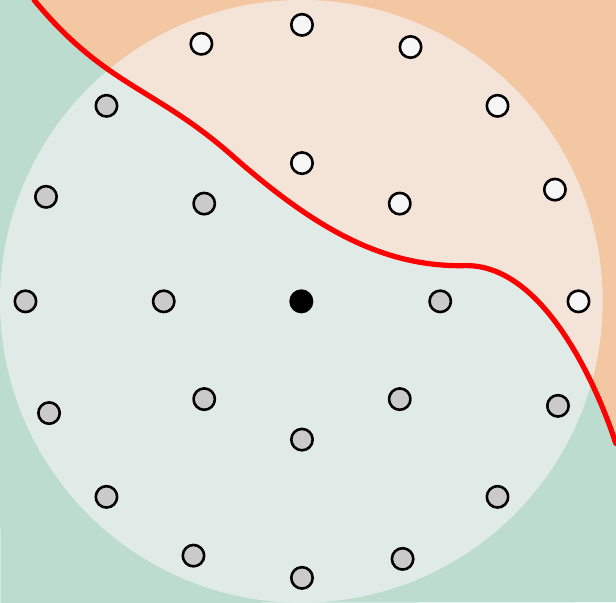}
		\caption{Ray stencil in screen-space \citep{Choudhury:2009}}
		\label{fig:samples}
	\end{subfigure}
	\qquad
	\begin{subfigure}[b]{0.55\linewidth}
		\includegraphics[width=\linewidth]{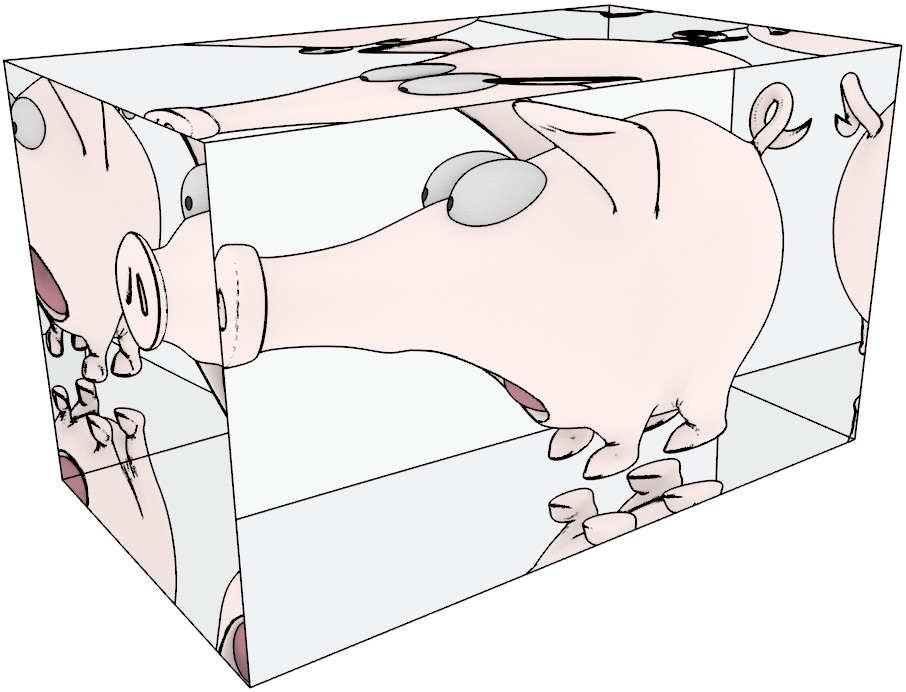}
		\caption{Image generated with Arnold Toon shader \citep{Ogaki:2018}.}
		\label{fig:pig_glass}
	\end{subfigure}
	\caption{\textbf{Ray-traced feature lines} --- \textbf{(a)} Around a central ray (black dot) a stencil of rays (grey and white dots) is cast to estimate the foreign primitive area, \ie~the proportion of samples intersecting a different primitive than the central one (orange vs. green surfaces).  \textbf{(b)} Image-space depth, ID and normals discontinuities extracted taking into account reflections and refractions.} \label{fig:raytracing}
\end{figure}


\chapter{Mesh Contours: Definition and Detection}
\label{chap:mesh_contours}

This chapter formally introduces the occluding contours of polyhedral meshes.  
We begin with some basic definitions of the mesh and viewing geometry, and then give formal definitions of contours. We then describe a range of extraction algorithms for faster detection. 
The following chapter will then discuss visibility computations.

Extracting contours from meshes can allow exact computation of the contour topology, allowing for more sophisticated curve stylization algorithms, while also fixing potential problems with the algorithms from the previous chapter.

\section{Meshes}

A polyhedral mesh comprises a list of vertices and a list of faces, each face containing three or more vertices (\fig{mesh}). Faces meet in mesh edges, each edge connecting two vertices. The mesh \emph{connectivity}, describes the incidence relations among those mesh elements, \eg{} adjacent vertices and edges of a face. The mesh \emph{geometry} specifies 3D position of each vertex: $\vec{p} = [p_x, p_y, p_z]^\top$. 
In computer graphics, most polyhedral meshes are either triangular or quadrilateral meshes (\fig{mesh}). In this tutorial, we focus on triangular meshes, although the definitions and algorithms presented below generalize to any polyhedral meshes with planar faces. Non-planar faces, such as non-planar quad faces, need to be subdivided into planar faces.

The normal of a face is the vector orthogonal to all edges of the face, which can be computed by the cross-product: $\vec{n} = (\vec{p}_3 - \vec{p}_1) \times (\vec{p}_2 - \vec{p}_1)$, where $\vec{p}_{1:3}$ are any three vertices of the face taken in clockwise order.
Note that the normal orientation depends on the order of the vertices (e.g., swapping vertices 1 and 2 would reverse the normal direction). The ordering of the three vertices is usually encoded in the mesh data structure.

We further assume that the mesh is \emph{manifold}, \ie{} (1) each edge is incident to only one or two faces,
and (2) the faces incident to a vertex form a fan, either open or closed.

\begin{figure}
	\centering
	\small
	\begin{subfigure}[b]{0.30\linewidth}
		\def\svgwidth{\hsize}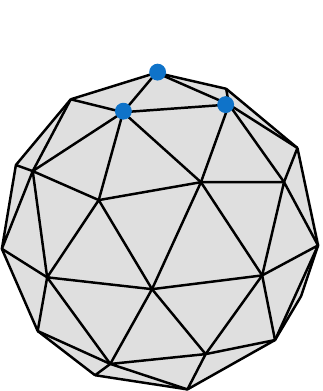\caption{triangle mesh}\label{fig:trimesh}
	\end{subfigure}
	\qquad
	\begin{subfigure}[b]{0.60\linewidth}
		\def\svgwidth{\hsize}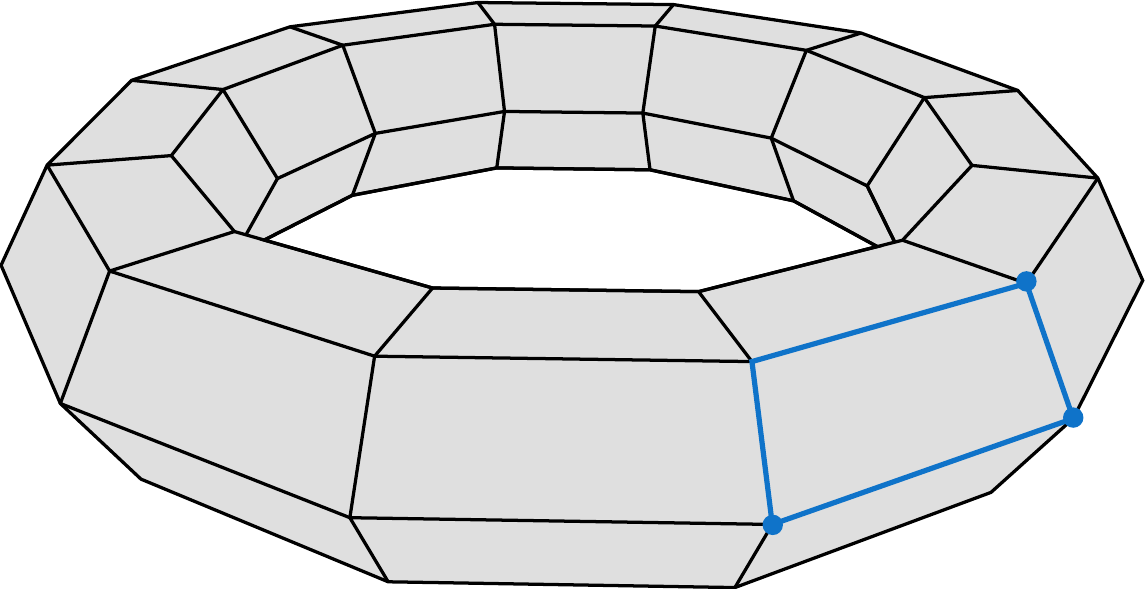\caption{quad mesh}\label{fig:quadmesh}
	\end{subfigure}
	\caption{\textbf{Polygonal meshes} --- Since each face of a polygonal mesh is planar, its normal $\vec{n}$ can be computed as $(\vec{p}_3 - \vec{p}_2) \times (\vec{p}_2 - \vec{p}_1)$, where $\vec{p}_{1:3}$ are any three vertices of the face.} \label{fig:mesh}
\end{figure}

\section{Camera viewing}
The polyhedral mesh will be projected by either orthographic (parallel) or perspective (central) projection. 
For orthographic projection in the \emph{view direction} $\vec{v}$, a given scene point $\vec{p}$ is projected to the image plane by intersecting the line that passes through $\vec{p}$ in the direction $\vec{v}$ --- called the the \emph{visual ray} --- with the image plane (\fig{orthographic}). The point $\vec{p}$ is visible if the visual ray does not intersect any other surface point before reaching the image plane; otherwise it is invisible.

For perspective projection, the camera is defined by the position of its center $\vec{c}$ and an image plane (\fig{perspective}). In this case, the visual ray is the line from $\vec{c}$ to the scene point $\vec{p}$; the corresponding view direction $(\vec{c} - \vec{p})$ is not constant anymore; it depends on $\vec{p}$. The projection of $\vec{p}$ remains the intersection of the visual ray with the image plane. 

\begin{figure}
	\small
	\begin{subfigure}[b]{0.45\linewidth}
		\def\svgwidth{\hsize}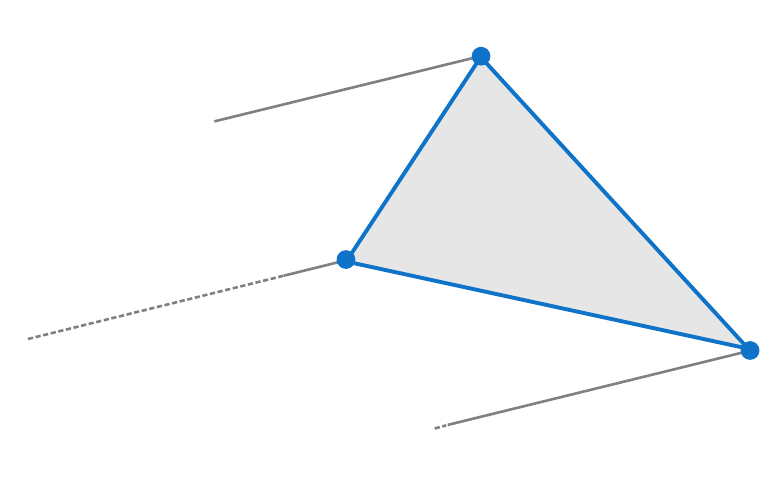\caption{orthographic projection}\label{fig:orthographic}
	\end{subfigure}
	\qquad
	\begin{subfigure}[b]{0.45\linewidth}
		\def\svgwidth{\hsize}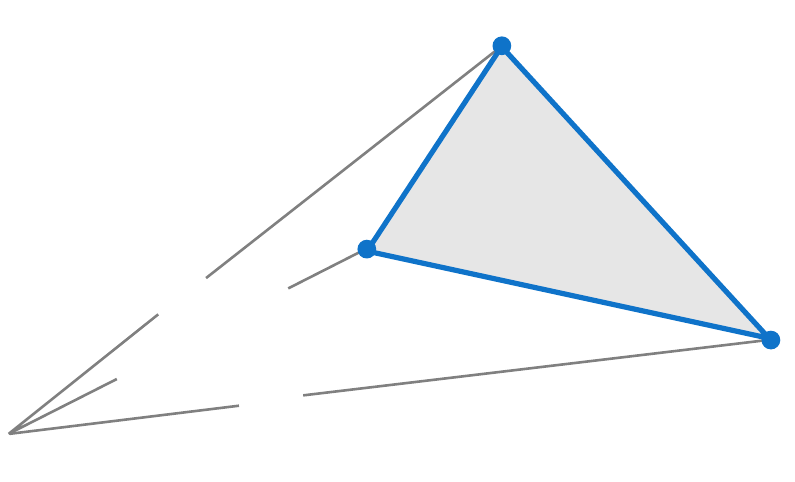\caption{perspective projection}\label{fig:perspective}
	\end{subfigure}
	\caption{\textbf{Projections} --- The triangular face formed by the vertices $\vec{p}_1$, $\vec{p}_2$ and $\vec{p}_3$ is viewed \textbf{(a)}~under orthographic projection along the view direction $\vec{v}$, and \textbf{(b)}~under perspective projection of center $\vec{c}$.} \label{fig:projections}
\end{figure}

The mesh \textit{boundary} is the set of edges where each edge is adjacent to only one mesh face. A surface is \textit{closed} if it has no boundary, otherwise it is \textit{open}.

\section{Front faces and back faces}
\label{sec:front_back}

We assume that the mesh is \emph{orientable}. Informally, this requires that all adjacent pairs of faces have consistent normal directions, facing the same side of the surface. This gives the surface a consistent notion of ``inside'' and ``outside'', and rules out esoteric surfaces like the M\"{o}bius strip and the Klein bottle (\fig{klein_bottle}).  For an open surface, we can think of the surface as a subset of an orientable closed surface that will only be seen from certain viewpoints.

\begin{figure} 
	\centering
	\small
	\includegraphics[width=0.4\linewidth]{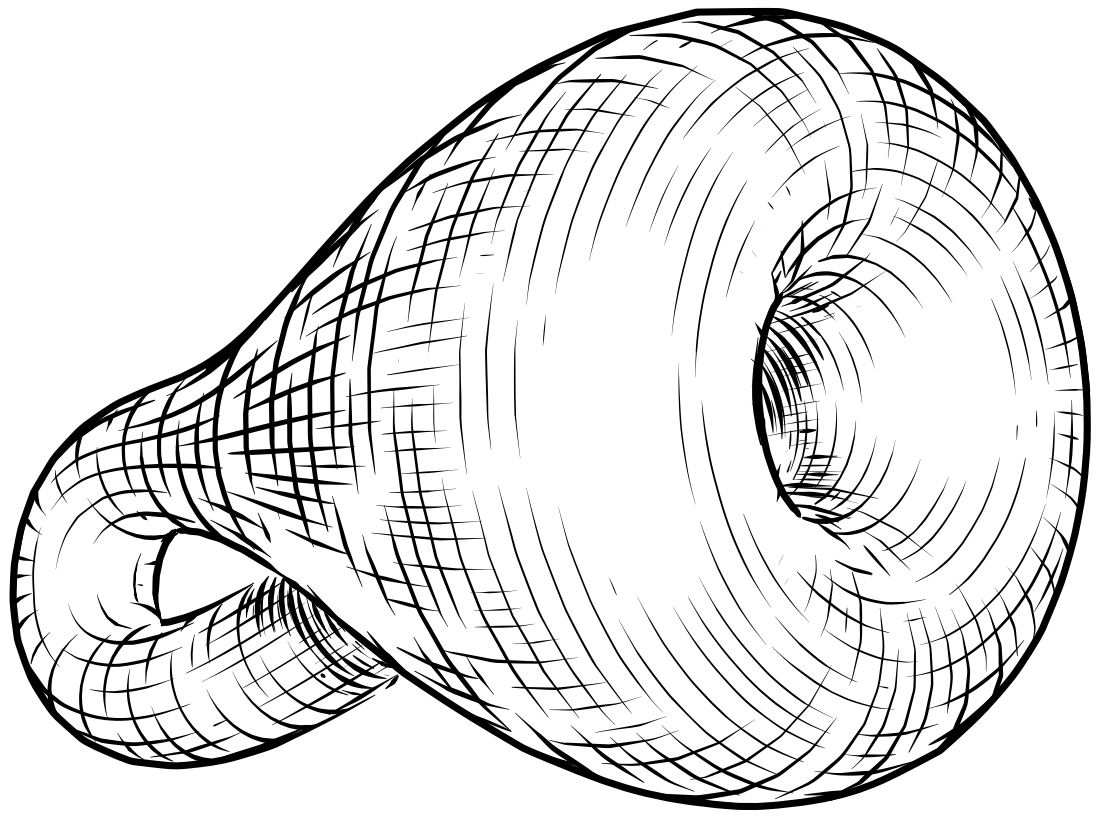}
	\qquad
	\includegraphics[width=0.45\linewidth]{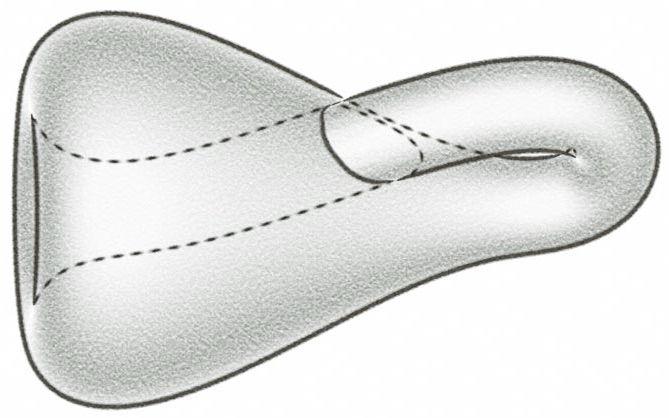}
	\caption{\textbf{Contour renderings of a non-orientable surface}, the Klein bottle, left by \citet{Hertzmann:2000} and right by \citet{Kalnins:2003}. Right image generated with ``Jot''~\citep{jot}.}
	\label{fig:klein_bottle}
	\end{figure}

More formally, orientability can be determined in a mesh data structure as follows. Each triangle in the data structure represents its vertices in a cyclic ordering. Two adjacent faces are consistent if the two vertices of their common edge are in opposite order.

A face is \emph{front-facing} if the camera position lies on the side of the face pointed to by the face's normal, \ie~$(\vec{c} - \vec{p}) \cdot \vec{n} > 0$ (Figure \ref{fig:ndotv}). It is \emph{back-facing} if the camera lies on the other side of this plane.  In orthographic projection, a face is front-facing when $\vec{v} \cdot \vec{n} < 0$.


\begin{figure}
	\centering
	\small
	\def\svgwidth{0.7\textwidth}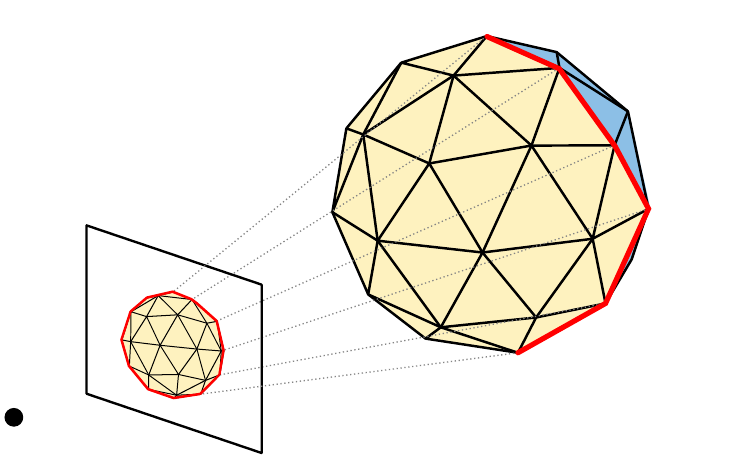\caption{\textbf{Front faces, back faces, and contour} --- The front faces are shown in yellow, and are visible to the camera. The back faces are in blue, and are not visible to the camera. The contour generator separates the front facing regions from the back-facing regions of the surface. The apparent contour is the visible projection of this curve onto the image plane.}\label{fig:contour}
	\label{fig:ndotv}
\end{figure}

 We assume that only front-faces may be visible; back-faces must always be invisible.  This occurs in two ways.  First, when a closed mesh with outward-facing normals is viewed from the outside, the back-faces must all be occluded by front-faces. 
Second, in a professional animation setting, objects are often modeled with open surfaces, but with the camera movements constrained so that only the front facing regions will be visible.


\section{Mesh contours and boundaries}

The general definition of occluding contours, for all surfaces, is as follows.

\begin{definition}[occluding contour generator]\label{def:image-space-contours}
	For a given viewpoint, the \emph{occluding contour generator} is a curve on the surface that delineates the frontier between what is locally visible and invisible, that is, between front- and back-facing surface regions.
\end{definition}
As such, they mark any depth discontinuity either between the surface and the background in the image, or where parts of an object pass in front of itself.

\begin{definition}[occluding contour]\label{def:occluding_contour}
	For a given viewpoint, the \emph{occluding contour} (or \textit{apparent contour}) is the visible 2D projection of the occluding contour generator.
\end{definition}


These definitions, as applied to meshes, are:

\begin{definition}[mesh contour generator]\label{def:mesh_contours}
 The collection of all mesh edges that connect front-faces to back-faces are together called the \emph{occluding contour generator}~\citep{Marr:1977}. The visible projection of the contour generator onto the image plane is called the \emph{occluding contour}, or, \emph{apparent contour}.
\end{definition}

Despite the different terminology here, we will often simply use the term ``contour'' to refer to these different curves, where the meaning is obvious from context, simply because terms like ``occluding contour generator'' are rather cumbersome.

In the literature, there is considerable variation in how these terms are used. 
The \emph{silhouette} is the subset of the contour that separates an object from the background behind it. In the computer graphics literature, the word ``silhouette'' was often used to mean ``contour'', especially prior to 2003. 
Koenderink uses the term ``rim'' to refer to the occluding contour generator.  Some authors use the term ``contour'' to refer to any image curve.


\section{Generic position assumption}

We further assume that the mesh is in \emph{generic position}, which is a helpful trick for avoiding many tedious technicalities.

Loosely speaking, the generic position assumption implies that the mesh does not have any specific ``weird'' connectivity, nor does the camera's view of the mesh --- this frees us from handling many possible special cases. 

More precisely, in generic position, any relevant topological properties of the mesh and camera together are robust to infinitesimal perturbations. 

For example, it is theoretically possible for a face of the mesh to be exactly edge-on: $(\vec{c} - \vec{p}) \cdot \vec{n} = 0$. This is a face which is neither front-facing nor back-facing.  Correctly drawing the contours through this face, with correct contour topology for stylization, would require some extra effort on top of the basic algorithms we will describe. However, if we added an infinitesimal amount of random noise to $\vec{c}$ or any of the vertices $\vec{p}$, then the face would no longer be edge-on (with probability 1). Other non-generic cases that can cause problems include degenerate edges (adjacent vertices have the exact same coordinates), and coincident geometry (e.g., two distinct triangles lie in the same plane and overlap).

In general, handling non-generic cases like these require extra effort to implement, and they would be largely unenlightening to explain in this tutorial.  Even enumerating potential non-generic cases could be quite tedious and difficult.  Furthermore, the research literature has largely ignored non-generic cases. 

For real-valued geometry that is randomly-positioned, violations of the generic position assumption are zero-probability events. Even in floating-point computations, the odds of the assumption being violated are vanishingly rare.  
In some cases, genericity violations may be intentional, such as in mechanical illustration and industrial design applications, where edge-on faces are common. For these applications, some specific non-generic cases would need to be handled. 

A simple fix to violations of generic position is to randomly add a tiny random number to each vertex coordinate of the mesh; by definition, this will cure all non-generic cases.  For example, the edge-on face described above would become either front-facing or back-facing.
More principled handling is potentially application-dependent, and we do not discuss it further in this tutorial.

\section{Contours are sparse} \label{sec:sparse}

As noted by \citet{Markosian:1997,Kettner1997,Sander:2000,McGuire:2004a}, contour edges only represent a tiny percentage of the total number of mesh edges. For a reasonable polyhedral approximation of a smooth surface, \citet{Glisse:2006} showed that the contour length, averaged over all viewpoints, is in the order of $\sqrt{n}$ where $n$ is the number of faces in the mesh. In practice, for general man-made triangular meshes, \citet{McGuire:2004a} measured empirically a trend closer to $n^{0.8}$.  For example, the Buddha model has over 1 million faces but only around 50k contour edges on average from different views.  Of these, a large fraction are surely concave, and thus can trivially be marked as always invisible, as discussed in Section \ref{sec:concave_edge}.

Additionally, edges that are more convex are more likely to be contours than edges that are flatter \citep{Markosian:1997}. For example, a nearly-flat edge only becomes a contour from a narrow range of camera positions, unlike a very sharply convex edge.


\section{Extraction algorithms}
\label{sec:extraction}


We now survey different algorithms for detecting the set of contour edges on a mesh.  These range from the basic brute-force procedure, to more sophisticated data structures and algorithms.  Computing the apparent contours further requires determining the visibility of the contour generators; solutions to this challenging problem will be presented in \chap{chap:visibility} and \chap{chap:fast_visibility}.

\subsection{Brute force extraction}

The basic brute force algorithm directly stems from \defn{def:mesh_contours}. For a given viewpoint, the algorithm consists in iterating over every mesh edge, computing the normals of its two adjacent faces, and checking whether their dot-products with the view direction have opposite the signs. For perspective projection, any position on the edge can be used to define the view direction; the first vertex of the edge is commonly chosen. 

Representing the mesh with a half-edge data structure~\citep{Campagna:1998} makes these operations easier to implement.
To avoid redundant calculations, the face normals are usually precomputed and stored as face attributes. 
The iteration over the mesh edges must be performed every time the camera or object position changes, which is very expensive for complex models. \fig{mesh_contours} shows three results of this algorithm; the hidden contours are illustrated with dotted lines.

\begin{figure}
	\centering
	\small
	\includegraphics[width=\linewidth]{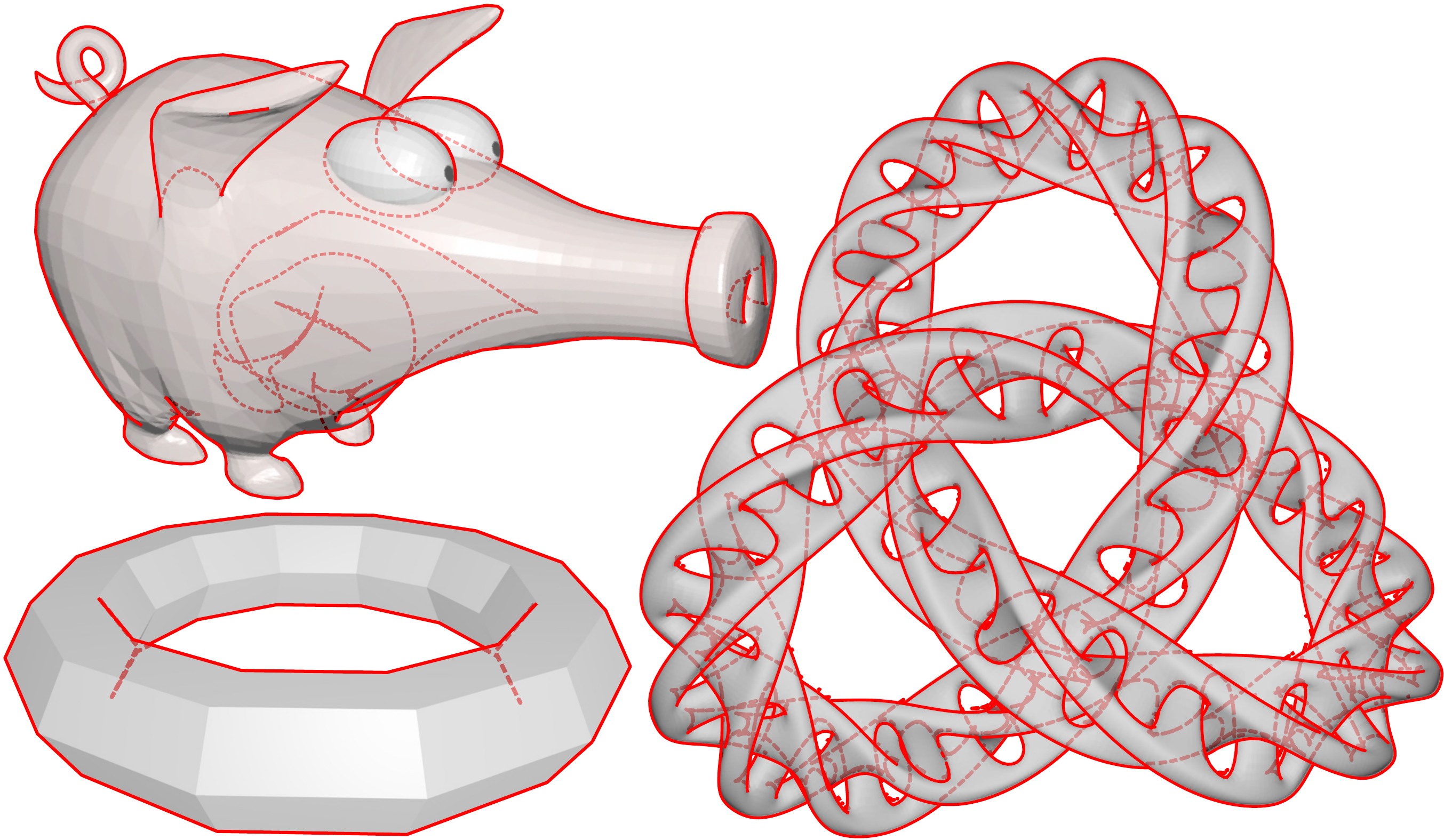}
	\caption{\textbf{Results} --- Mesh contours extracted from a low-resolution torus (bottom left), the ``Origins of the Pig'' \ccCopy~Keenan Crane (top left), and the  ``Moebius Torus Knot'' \ccCopy~Francisco Javier Ortiz V{\'a}zquez (right). Hidden contours are depicted with dotted lines (visibility algorithms are discussed in the next Chapter).}\label{fig:mesh_contours}
\end{figure}

\subsection{Pre-computation for static meshes} 

\label{sec:dataStructures}

In many applications, parts of the 3D scene are static, or rigid. In such a case, a data-structure can be built for each static mesh during an advance pre-process, so that the search is significantly accelerated at rendering time.  At run-time, contour extraction can be a function of the number of contour edges, rather than the number of mesh edges, yielding a substantial time savings due to the sparsity of contours (Section \ref{sec:sparse}).

\paragraph{Orthographic dual space.} 
For orthographic projection, \citet{Benichou:1999} and \citet{Gooch:1999} proposed a dual space for fast contour detection.
The dual space is a 3D coordinate system $\vec{s}=(s_1,s_2,s_3)$.
Each mesh face is mapped to a single point $\vec{s}$ in the dual space, with coordinates given by the face normal: $\vec{s}_i=\vec{n}_i = (n_x, n_y, n_z)$. Likewise, the orientation function, based on view direction $\vec{v}=(v_x,v_y,v_z)$, is mapped to a plane in the dual space: 
\begin{align}
g(\vec{s}) &= \vec{v} \cdot \vec{s}= v_x s_1 + v_y s_2 + v_z s_3 = 0
\end{align}
For front faces, $g(\vec{s})>0$, and $g(\vec{s})<0$ for back faces. A mesh edge between faces $(i,j)$ on the original surface corresponds in the dual space to a line segment $\overline{\vec{n}_i \vec{n}_j}$. In the dual space, when the orientation plane intersects a line segment, this line segment must correspond to a contour edge on the original surface. Hence, contour detection is reduced to intersecting a plane with a set of line segments, which can be accelerated by standard geometric data-structures, such as octrees or BSP trees.

In implementation, the 3D space does not need to be represented; a 2D space is sufficient.  Specifically, one can observe that the normals can be arbitrarily scaled without changing the results.  Scaling each point to have unit norm projects the points onto the Gaussian sphere, and, within the Gaussian sphere, line segments become arcs (\fig{gauss_sphere}). For computation, all points can be projected onto a unit cube \citep{Benichou:1999}, or a hierarchy of platonic solids \citep{Gooch:1999}. Hence, arcs are transformed into line segments, reducing the 3D intersection test to a set of 2D intersection tests, that can be further accelerated with standard 2D data-structures. 

\begin{figure}
	\small
	\begin{subfigure}[b]{0.47\linewidth}
		\def\svgwidth{\hsize}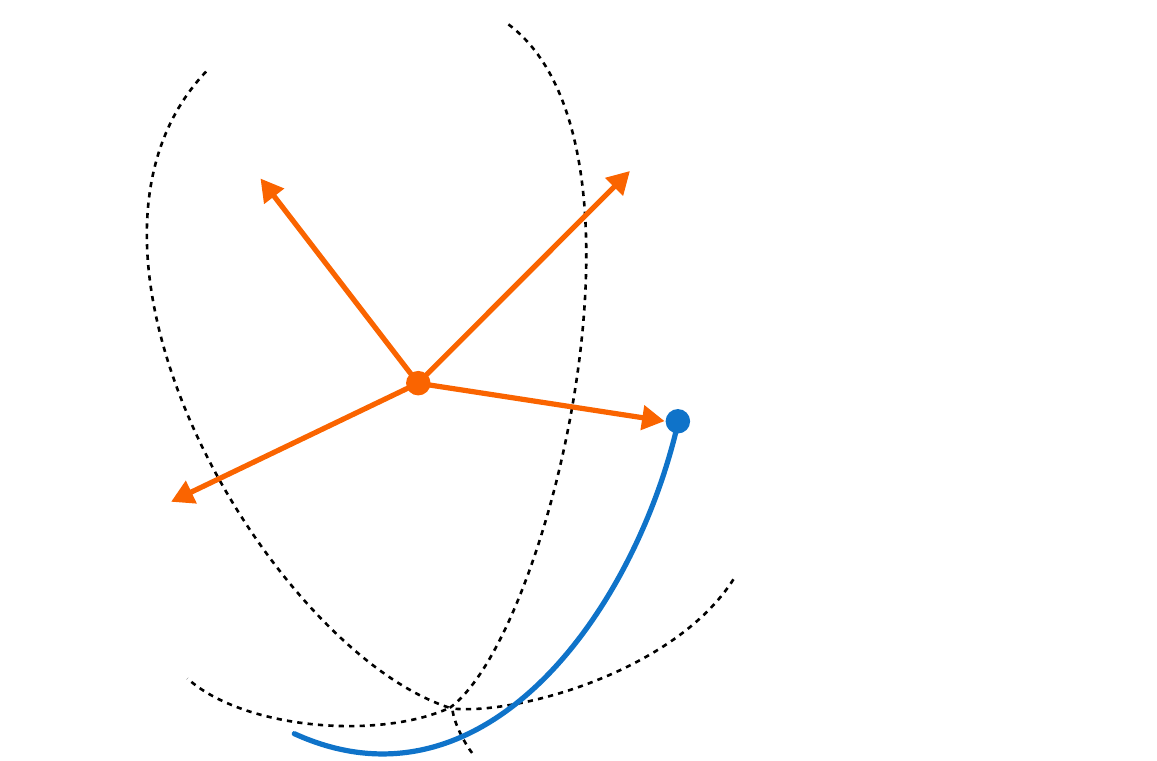\caption{Gaussian sphere}\label{fig:gauss_sphere}
	\end{subfigure}
	\quad
	\begin{subfigure}[b]{0.47\linewidth}
		\def\svgwidth{\hsize}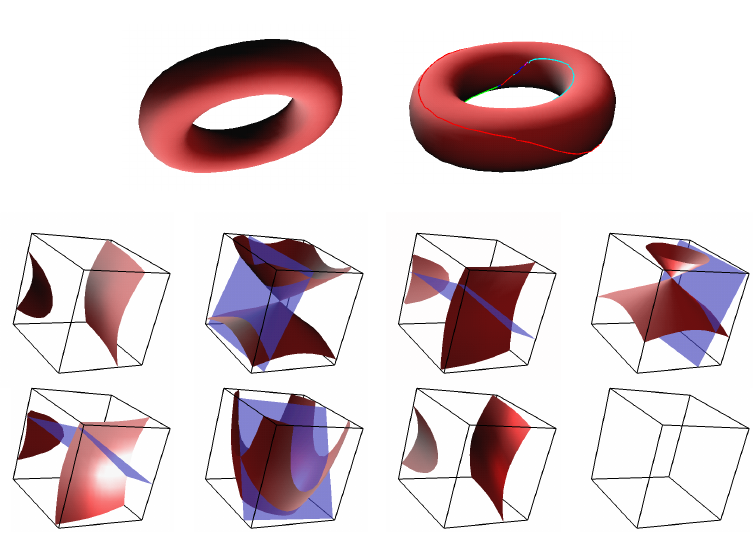\caption{4D dual space}\label{fig:dual_space}
	\end{subfigure}
	\caption{\textbf{Dual spaces} --- Preprocessing by \textbf{(a)} projecting the normals of two adjacent faces onto the Gaussian sphere, or \textbf{(b)} constructing a representation of the mesh in 4D space based on the position and tangent planes of its vertices. At runtime, finding the contour edges consists in computing the intersection of the dual viewing plane (in blue) with \textbf{(a)} circular arcs or \textbf{(b)} the dual surface, which can be further accelerated by space-partitioning data-structures.} \label{fig:acceleration}
\end{figure}


\paragraph{Perspective dual space.}
The above approach can be generalized to perspective projection \citep{Hertzmann:2000}. In this case, a 4D dual space is used conceptually, but a 3D dual space is used in practice.  

A mesh face with position $\vec{p} = (p_x,p_y,p_z)$ and normal $\vec{n} = (n_x,n_y,n_z)$ is mapped to a dual point $\vec{s}=(s_1,s_2,s_3,s_4)=(-n_x,-n_y,-n_z,\vec{p}\cdot\vec{n})$.  (Any point on the face may be used.) Given the camera center $\vec{c}$, the orientation function is mapped to a dual hyperplane: $$g(\vec{s})=(c_x,c_y,c_z,1)\cdot{s}= 0.$$  
Hence, front-faces have $g(\vec{s})>0$ and back-faces have $g(\vec{s})<0$.
A mesh edge between faces $i$ and $j$ corresponds to a dual line segment $\overline{\vec{s}_i \vec{s}_j}$. Any line segment that intersects the dual hyperplane corresponds to a mesh contour.  Hence, finding all contour edges reduces to a 4D hyperplane intersection with a set of line segments. Orthographic cameras can also be handled in this dual space with $g(\vec{s})=[-v_x,-v_y,-v_z,0]\cdot{s}=0$.

As in the orthographic case, the dual points can be scaled arbitrarily without changing the results.  Hence, a 3D space can be used.
\citet{Hertzmann:2000} normalize each dual point $\vec{s}$ using the $l_\infty$ norm --- effectively projecting it on the surface of the unit hypercube. The surface of the unit hypercube can be represented as eight octrees.
Each dual point is stored in one of the octrees. At runtime, 
the viewpoint is converted into a dual plane, and the dual plane is intersected with the eight octrees. The expected complexity is linear to the number of contour edges.

\begin{figure}
	\centering
	\small
	\begin{subfigure}[b]{0.74\linewidth}
		\def\svgwidth{\hsize}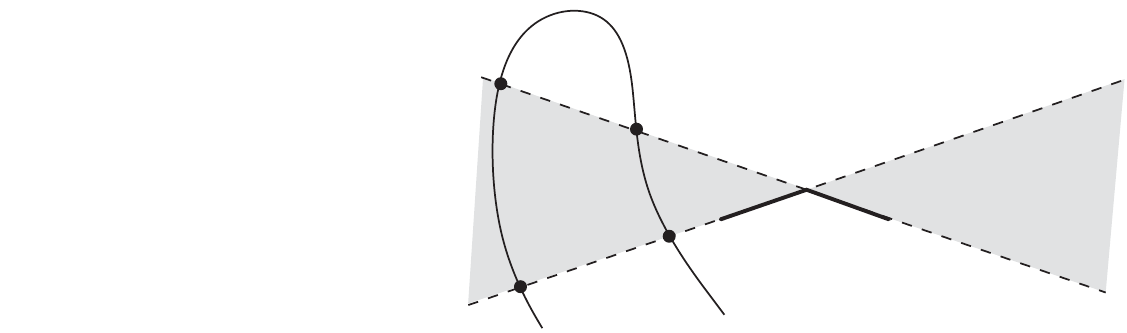\caption{Time intervals}\label{fig:time_intervals}
	\end{subfigure}
	\begin{subfigure}[b]{0.24\linewidth}
		\def\svgwidth{\hsize}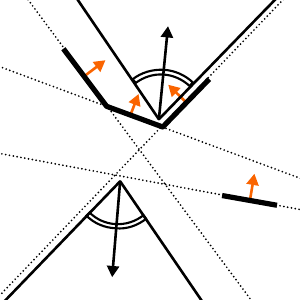\caption{Anchor cones}\label{fig:cone_trees}
	\end{subfigure}
	\caption{\textbf{Spatial partitioning} --- \textbf{(a)} A given edge is on the contour generator if the viewpoint trajectory $\vec{c}(t)$ is inside the intersection (in grey) of one positive and one negative half-space defined by the face supporting plane $\pi_1$ and $\pi_2$, \ie{} during the time intervals $[t_0,t_1]$ and $[t_2,t_3]$ in this example. \textbf{(b)} The front and back-facing anchored cones are defined by their center $\vec{o}_{f|b}$, an opposite normal $\vec{n}_f=\vec{n}_b$ and a common half opening angle $\theta$, here visualized in 2D for four oriented segments.}
\end{figure}

\paragraph{Animation.} 
These data structures can further be exploited to accelerate detection during animation.
In dual space representations, when the camera makes small moves, it is possible to only visit a small portion of the dual space. \citet{Pop:2001} and \citet{Olson:2006} describe incremental methods that are able to update an existing set of contour edges when the camera moves.

If the viewpoint trajectory is known in advance and can be represented by a polynomial curve $\vec{c}(t)$ of degree $d$, \citet{Kim:2005} showed that there are at most $d+1$ time-intervals $[t_i,t_{i+1}]$ at which an edge can be a contour. Those intervals can thus be pre-computed for each edge, by intersecting the polynomial curve with the supporting planes $\pi_1$ and $\pi_2$ of the edge's adjacent faces (\fig{time_intervals}), and stored in an array or a tree data-structure. At runtime, the contour edges can then easily be updated incrementally during the camera motion along the prescribed trajectory. The incremental update mechanism of \citet{Pop:2001,Olson:2006} is more computationally demanding, but it is not constrained to a fixed camera path.

\paragraph{Cone trees.} \citet{Sander:2000} proposed accelerating contour extraction using a forest of search trees constructed over the mesh edges. Taking inspiration from previous work on back-face culling, their key idea is to build, for each edge of the mesh, a hierarchy of face clusters. At runtime, clusters whose faces are all front-facing or all back-facing can be fully discarded. To conservatively decide in constant time whether a cluster is front- or back-facing, \citet{Sander:2000} compute and store two open-ended anchored cones per cluster: one cone inside which any viewpoint would make the face cluster entirely front-facing, and another cone making the cluster back-facing (\fig{cone_trees}). They demonstrated that, experimentally, this approach also has linear complexity with respect to the number of contour edges.  However, their data structure construction can be extremely slow.

\subsection{Randomized search} 
\label{sec:randomized}

The above data structures are not useful for deforming meshes.  The following randomized algorithm, proposed by \citet{Markosian:1997}, works for any mesh, though it is not guaranteed to detect all edges. 


The method first selects a few mesh edges at random. Because contours are sparse, the probability of finding a first contour edge is rather low. However, since mesh contours form continuous chains of edges on the surface, once a first contour edge has been found, spatial coherence can be leveraged to explore adjacent edges in an advancing front manner and trace the full contour loop. By further assigning to each edge a probability inversely proportional to the exterior dihedral angle $\alpha$ (in radians) between its adjacent faces, the chance of finding contour edges is increased since, given a random view direction, the probability that an edge is a contour is $\alpha/\pi$. Derivations for this probability can be found in \citet{McGuire:2004a} (perspective case) and \citet{Elber:2006} (orthographic case).


In addition, for small viewpoint changes, \citet{Markosian:1997} observed that temporal coherence can also be leveraged by re-seeding the search in the new frame from the previous frame's contour, and by searching for contour edges in its vicinity, moving towards (resp.~away) from the camera if the edge is adjacent to back-faces (resp.~front-faces). 

This approach does not guarantee that all contour edges will be found, but it will usually detect the longest mesh contours. 
If the algorithm samples edges without replacement, then it will converge to the correct solution once it has visited every edge.


\chapter{Mesh Curve Visibility}
\label{chap:visibility}

Once we have found the curves on a mesh, we need to determine which portions of them are visible.  In doing so, we will also build a data structure, called the \textit{view graph}, that represents the topology of the visible curves.


This chapter introduces the algorithms used for efficiently computing correct visibility for edges on the surface.
This question is related to the more general hidden-line removal problem, which dates back to the earliest ages of Computer Graphics, at the beginning of the sixties. \citet{Roberts:1963,Weiss:1966:VPI:321328.321330} devised the first known solutions to this problem, using brute-force ray tests.
\citet{Appel:1967} introduced Quantitative Invisibility as a way to greatly decrease the number of ray tests required, and improve accuracy. 

It can be tempting to implement many of these algorithms with heuristics. However, if not implemented carefully, the visibility operations here can be very sensitive. Our goal is to compute global curve topology, and depending on implementation, the visibility of a large curve may depend on a single visibility test somewhere on the curve. If this visibility test is erroneous, an entire curve from the drawing may disappear.  Hence, it is important to formulate these algorithms to carefully track curve visibility and topology, rather than using heuristics. Even with mathematically correct operations, numerical instability can also cause errors.  Techniques for robust visibility computation are discussed in Appendix \ref{app:numerical}.



\section{Ray tests} \label{sec:ray_casting}

For a perspective camera, a point $\vec{p}$ on the surface is visible from the camera center $\vec{c}$ if the line segment $\overline{\vec{pc}}$ intersects the image plane and does not intersect any other surface point (\fig{occlusions}). Determining visibility this way is called a \emph{ray test}, since it amounts to casting a ray from $\vec{c}$ and checking if the tripling is the first object hit. (For an orthographic camera, the test involves a line segment from $\vec{p}$ to the camera plane along the ray $-\vec{v}$.) 
Ray tests can be accelerated by spatial subdivision data structures, such as a 3D grid or a bounding volume hierarchy \citep[Chapter~4]{Pharr:2016}.


\begin{figure}
	\centering
	\small
	\def\svgwidth{0.85\textwidth}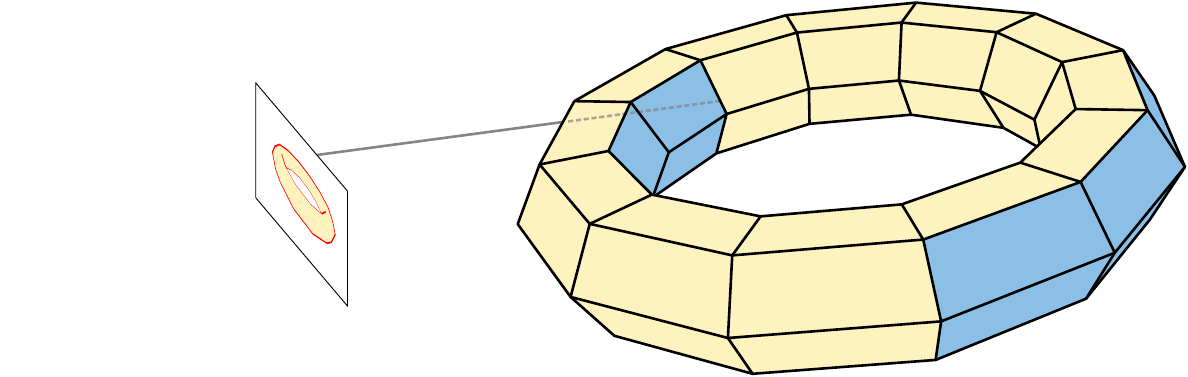\caption{\textbf{Visibility, ray tests, and convex/concave contours} --- 
	A point is visible if the line segment from the camera to the point does not intersect any other surface point. Determining this is called a ray test.
	In the example here, the segment from $\vec{p}_1$ to the camera $\vec{c}$ intersects another part of the surface, so $\vec{p}_1$ is not visible. The segment from $\vec{p}_2$ does not intersect the surface so it is visible.
	One can avoid computing one of these ray tests: $\vec{p}_1$ lies on a concave contour point, so it must be invisible. The other point, $\vec{p}_2$, is on a concave contour, so a ray test is necessary to determine if it is visible.
\label{fig:occlusions}
	}\label{fig:occlusions}
\end{figure} 

In principle, the apparent contour could be rendered by separately testing the visibility of many points on the contour generator, and connecting the visible points. However, as noted by \citet{Appel:1967}, this would be both computationally expensive, because it would require testing a large number of points between which the visibility does not change, and inaccurate, since it would miss the points where curves transition between visible and invisible. 
Instead, we will use techniques to propagate visibility on the surface.







\section{Concave and convex edges}
\label{sec:concave_edge}

We can classify mesh edges as to whether they are concave and convex, which provides an additional visibility constraint, and will be helpful for identifying singularities in the next section \citep{Markosian:1997}. 
The content of this section is new for this tutorial, building on \citep{Markosian:1997,Koenderink:1984}. 

A mesh edge is concave if the angle between the front-facing sides of its two faces is less than $\pi$. It is convex if the angle is greater than $\pi$.  (Note that this is the angle on the outside of the surface, and so it is different from the dihedral angle.)  Equivalently, when the edge is convex, each face is on the back-facing side of the other face.

Contours on concave edges must always be invisible:  if a concave edge is viewed at a grazing angle --- where a contour appears --- the edge is hidden inside the surface from that viewpoint.  Only convex edges can produce visible contours.
Figure \ref{fig:occlusions} shows examples of convex/concave edges, and how they can be invisible or visible. We provide a new, formal proof of this in Appendix \ref{app:convex}.

As a result, ray tests are never necessary for contours on concave edges.  Additionally, for static surfaces, concave edges can be omitted from any detection data structure (Section \ref{sec:dataStructures}), if hidden lines will not be rendered.

Algorithms for determining whether an edge is convex or concave are given in Appendix~\ref{app:numerical}.

\section{Singular points}
\label{sec:singular}

\begin{figure}
	\centering
	\small
	\setlength{\tabcolsep}{6pt}
	\begin{tabular}{ccc}
		\def\svgwidth{0.27\textwidth}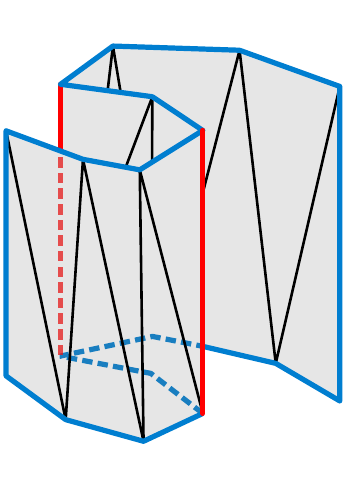 &
		\def\svgwidth{0.35\textwidth}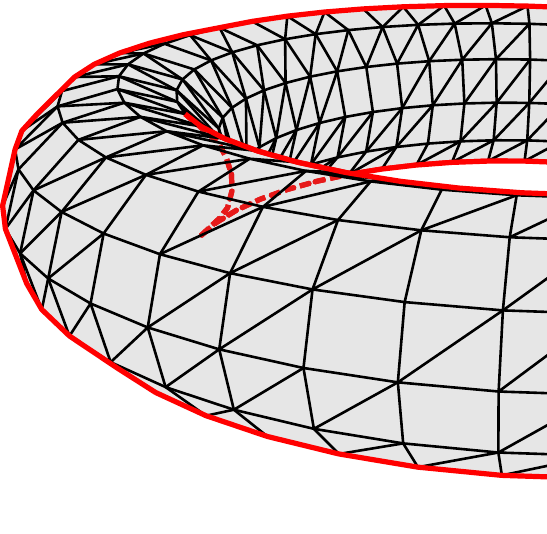 &
		\def\svgwidth{0.27\textwidth}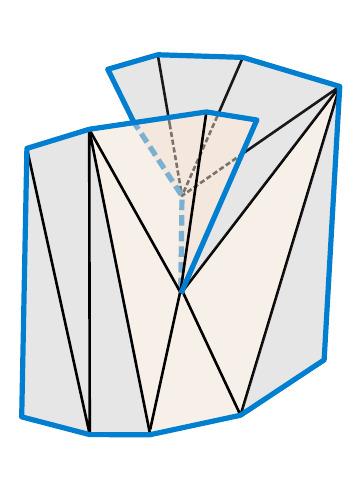
	\end{tabular}
	\caption{\textbf{Singular points} --- From the camera viewpoint, {\large\textcircled{\footnotesize 1}} T-junctions at image-space intersections, {\large\textcircled{\footnotesize 2}} Y-junctions between a contour generator and two boundary edges, {\large\textcircled{\footnotesize 3}} contour generator curtain folds, {\large\textcircled{\footnotesize 4}} boundary curtain-folds.  Contour generator edges are drawn in red, and boundaries in blue. Curtain folds are visualized in more detail in Figure~\ref{fig:curtain_fold}.
	}\label{fig:singular_points}
\end{figure}


The contour curves are the set of contour edges, and the boundary curves are formed from the set of boundary edges (\fig{singular_points}).  There are a few types of points on these curves where visibility may change. We call these points \textit{singular points}, or \textit{singularities}.  Two points that are connected by a curve that does not pass through any singularities must have the same visibility.  So we can perform a ray test at one curve point, and then propagate the result each direction along the curve until reaching singularities.  This will drastically reduce the number of required ray tests.   Singularities indicate places where visibility \textit{might} change.

These singularity data structures are also be used to record the 2D topology of the set of curves, \ie{} which curves connect to which, which will later be useful for stylization of the line drawing.

There are only a few different types of singular points:

\begin{enumerate}
\item[{\Large\textcircled{\footnotesize 1}}]
The visibility of a mesh curve may change
at an \textbf{image-space intersection}. Specifically, when a contour or boundary edge overlaps another curve in image space (\fig{singular_points} {\large\textcircled{\footnotesize 1}}), it indicates that part of the far curve is obscured by the surface closer to the camera. This splits the far curve into two segments, one of which must be invisible, and creates a T-junction between the near and far edges in image space. (The other segments may be invisible as well, if some other part of the surface occlude them.) 
Note, that while all intersections on the 3D surface are also 2D intersections, it is more robust to detect and handle them as a separate case, below.

\item[{\Large\textcircled{\footnotesize 2}}]
Curve visibility may change when two curves
 \textbf{intersect on the 3D surface}. For example, a contour generator may intersect a boundary curve (\fig{singular_points} {\large\textcircled{\footnotesize 2}}). This intersection can only occur at a mesh vertex, producing a Y-junction in image space \citep{Grabli:2010} if the three curve segments are visible, or may appear to form a continuous curve if one segment is hidden by the surface. 
In this case, the contour generator or the boundary curve may change visibility at the vertex. 
Surface intersection curves, on the other hand, lie within mesh faces, and thus intersect the contour generator within mesh edges.

\item[{\Large\textcircled{\footnotesize 3}} - {\Large\textcircled{\footnotesize 4}}]
\textbf{Curtain folds} \citep{Blinn:1978} occur where the surface sharply folds back on itself in image space, causing the curve to become occluded by local geometry (\fig{curtain_fold}). 
More precisely, \textbf{a curtain fold occurs at a vertex connecting two curve edges, when one of the adjacent curve edges is occluded by another face connected to the vertex, and the other curve edge is not}. This definition can be used directly to identify boundary curtain folds. On contour generators, a simpler rule can be used:
\textbf{a curtain fold occurs at any vertex where a convex contour edge meets a concave contour edge}, since concave contour edges must be invisible, and convex contour edges are locally visible.   Details on detecting boundary curtain folds are given in Appendix \ref{app:numerical}.
Curtain folds are not important on other types of curves.

\begin{figure}
	\centering
	\small
	\begin{tabular}{c|c}
		\def\svgwidth{0.41\textwidth}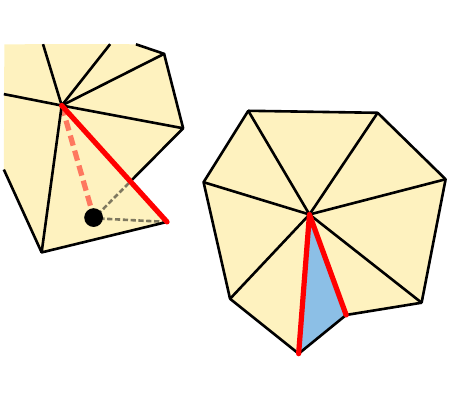 &
		\def\svgwidth{0.5\textwidth}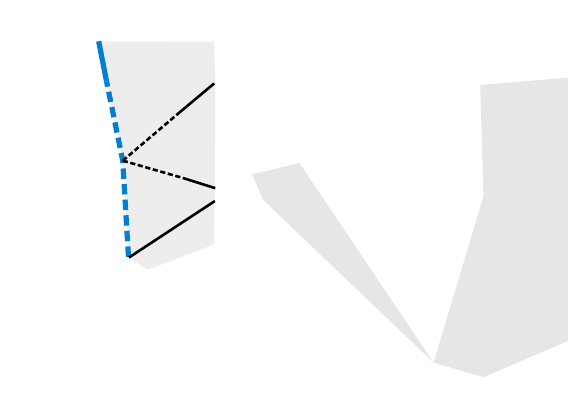
	\end{tabular}
	\caption{\textbf{Curtain folds} --- A vertex $\vec{p}$ is a curtain fold if \textbf{(a)} it connects a \emph{convex} contour generator edge to a \emph{concave} edge, or \textbf{(b)} the edge $\vec{pe}$ is occluded by a face of the one-ring neighborhood of $\vec{p}$ (here, the triangle $\vec{pqr}$ in brown).  (The former  case \textbf{(a)} is a special case of the latter \textbf{(b)}.)
	}\label{fig:curtain_fold}
\end{figure}

\item[{\Large\textcircled{\footnotesize 5}}]
A vertex may also connect more than two contour generator edges, in which case we call this vertex a \textbf{bifurcation} (Figure \ref{fig:bif}). 
In this case, there are no constraints on how the visibility can change; any adjacent edge of such a vertex can be visible or invisible.
Boundaries can also exhibit bifurcations.
 \end{enumerate}

 \begin{figure}
 \centering
 \includegraphics[width=0.85\linewidth]{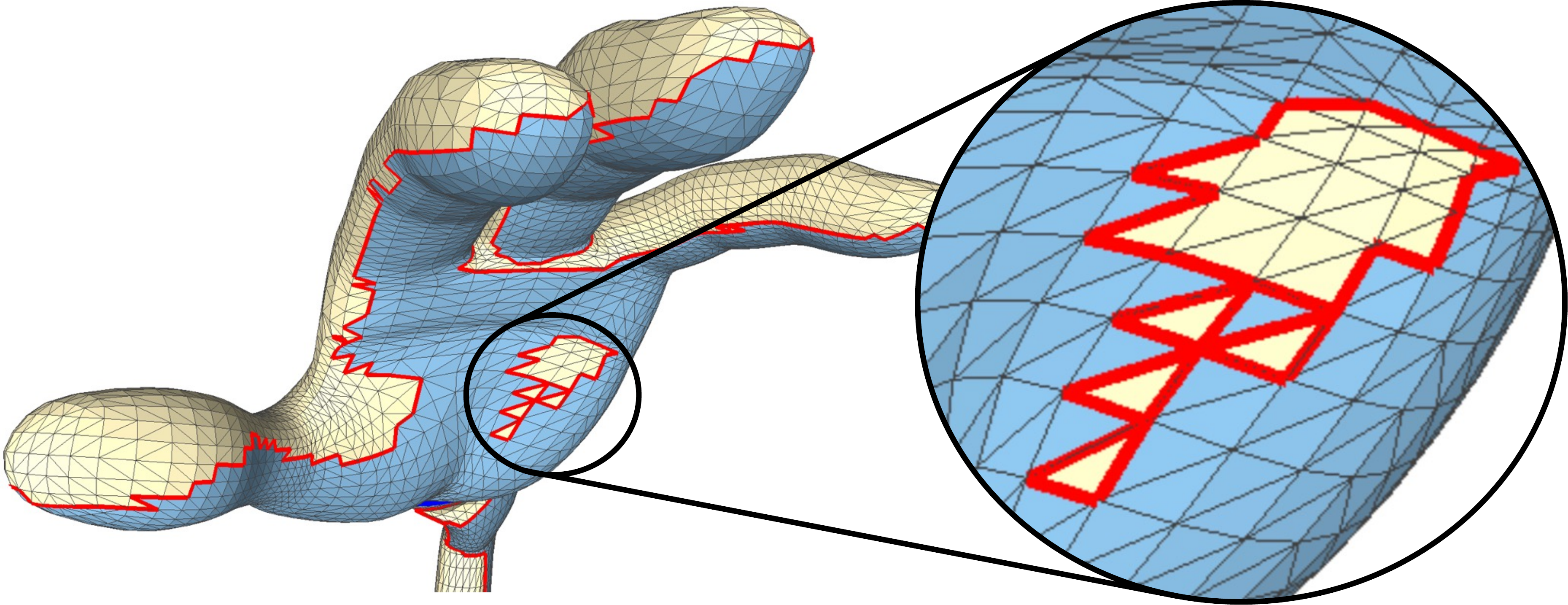}
 \caption{\textbf{Bifurcations in the contour generator} occur when more than two contour generator curves meet at a vertex. Closeup view on the mesh contours extracted from the smooth surface shown in \fig{red_intersections}. ``Red'' \ccCopy Disney/Pixar
 }\label{fig:bif} 
 \end{figure}

The above four cases are the only kinds of singular points for any surface curves. In implementation, different types of curves are handled separately, e.g., intersections on the surface need to be implemented with different cases for different kinds of intersecting curve, and intersections at vertices are handled separately from intersections within faces.

\section{Visibility for other curve types}

The above discussion is mainly for visibility-indicating curves: contours, boundaries, and surface-surface intersections.
Computing visibility for other surface curves is generally simpler.  The cases are the same: visibility may change when overlapped by a contour generator or boundary, or when intersecting a contour/surface-intersection on the surface.  Curtain folds occur only for contours and boundaries, and not other curves. Convex/concave determination is only useful for contours and not other curves.  Furthermore, ray tests for other curves are generally more numerically stable, if they are not themselves near contours.


Since we have assumed that back-faces are always invisible, any curve that lies within a back-face must also be invisible, as must an edge connected only to back-faces. For example, a boundary edge on a back-face is always invisible.






\section{View Graph data structures}\label{sec:view_graph}

In order to propagate visibility, we will build a data structure called a View Graph (\fig{pig_viewmap}). Later, this View Graph will be used to represent the image curves for stylization.

\begin{figure}
	\centering
	\small
	\def\svgwidth{\linewidth}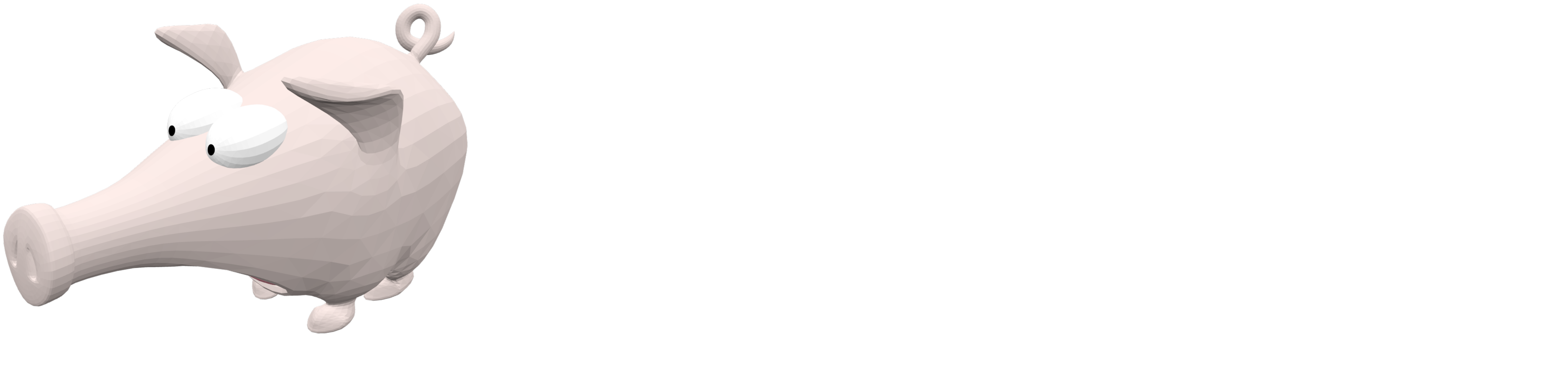\caption{\textbf{View Graph} of \fig{contours} --- The View Graph structures the extracted contour curves by storing their topology and geometry. It serves as a support for both determining their visibility efficiently and generating stylized strokes.}\label{fig:pig_viewmap}
\end{figure}

The View Graph stores the complete topology and geometry of the curves, in both 2D and 3D (\fig{viewgraph} and \ref{fig:fishtail}). It is composed of the line segments of each curve, and the singularities that connect them.

\begin{figure}
	\centering
	\small
	\begin{subfigure}[b]{0.8\linewidth}
		\def\svgwidth{\linewidth}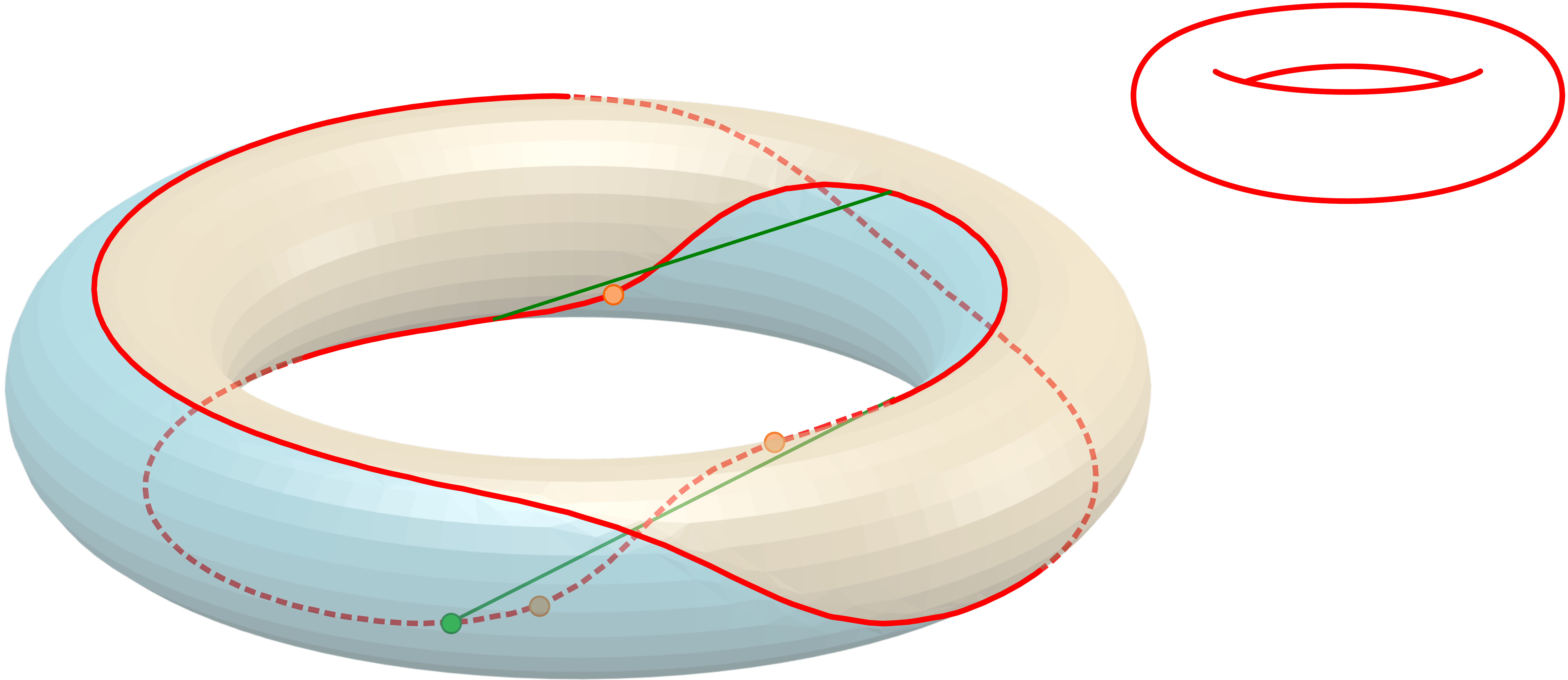
		\caption{View Graph of a torus}
	\end{subfigure}
	\par\vspace{1em}
	\begin{subfigure}[b]{0.9\linewidth}
		\def\svgwidth{\linewidth}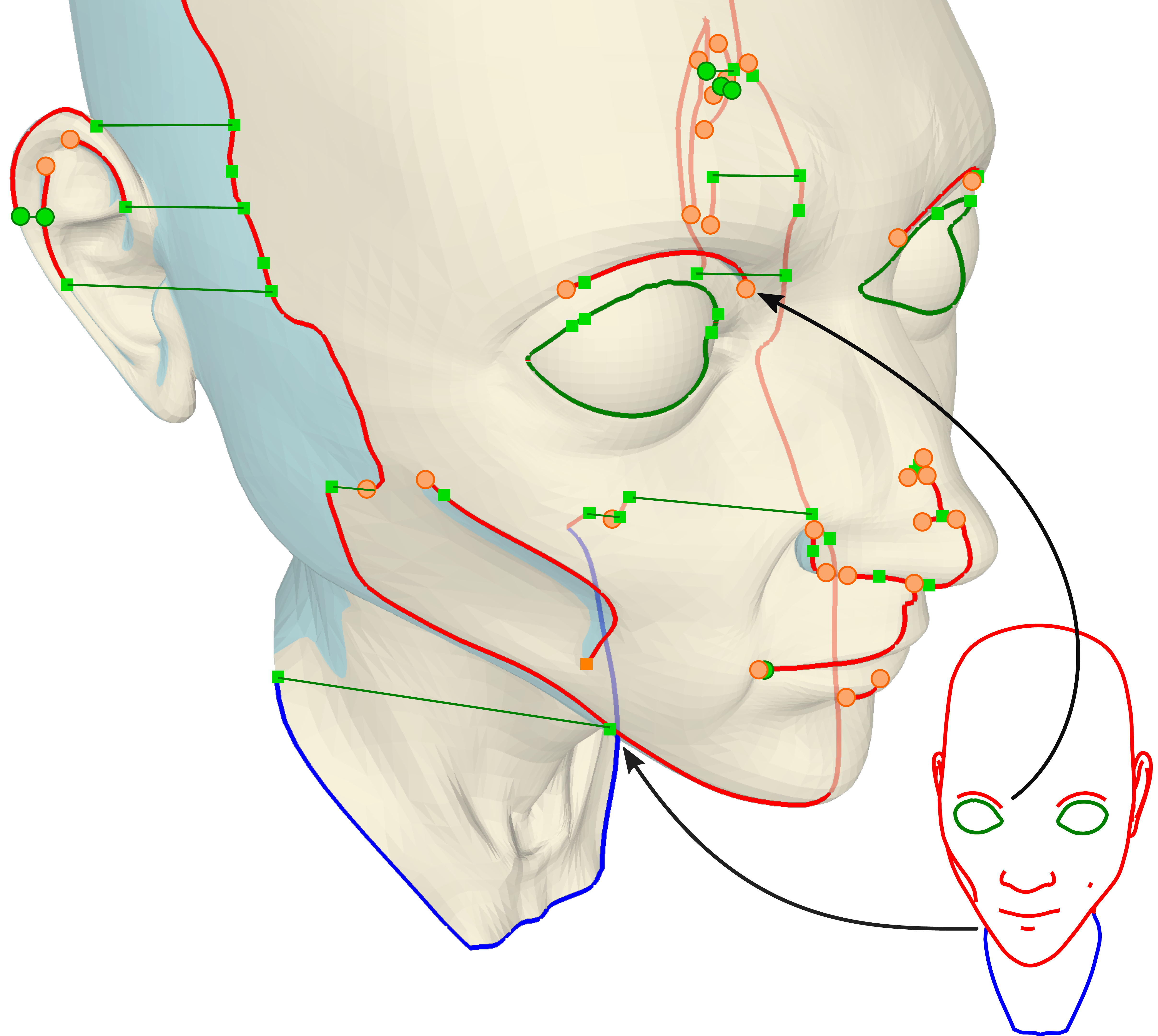
		\caption{View Graph of the ``Angela'' model \ccCopy~Chris Landreth}
	\end{subfigure}
	\caption{
		\textbf{View Graph} from \citep{Benard:2014} --- Line segments (contour generators, boundaries and surface
		intersections) are combined into chains that terminate at singular points (curtain folds, image-space intersections). This network of chains is called the View Graph. The graph on the right shows only the visible chains.} \label{fig:viewgraph}
\end{figure}

\begin{figure}
	\centering
	\small
	\begin{subfigure}[b]{0.28\linewidth}
		\def\svgwidth{\linewidth}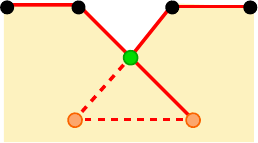
		\caption{Camera View}
	\end{subfigure}
	\quad
	\begin{subfigure}[b]{0.4\linewidth}
		\def\svgwidth{\linewidth}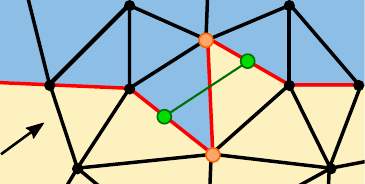
		\caption{Overhead view} 
	\end{subfigure}
	\quad
	\begin{subfigure}[b]{0.21\linewidth}
		\includegraphics[width=\linewidth]{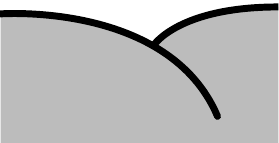}
		\caption{Rendering}
	\end{subfigure}
\caption{
	\textbf{The View Graph around a  ``fishtail'', a common contour shape} --- Studying renderings like these, and how the 2D figure relates to the 3D drawing, is very helpful in understanding a specific rendering. \textbf{(a)} The ``fishtail'' shape in image space includes two cusps, and a partially-occluded contour.
	\textbf{(b)} The overhead view shows the contour's path over the surface.  The contour separates front-facing and back-facing, and curtain fold cusps appear when the path switches direction. The curtain folds also separate convex from concave contours, which are always invisible.
	\textbf{(c)} A stylized rendering of this path produces this overlap drawing. This shape can occur at a large scale (e.g., a pair of hills or a puffy cloud), or at a subpixel level. In the latter case, we may wish to trim the extra bit of curve, as discussed in Section~\ref{sec:simplification}.
}
\label{fig:fishtail}
\end{figure}

The View Graph is built implicitly, in the algorithm described in the next section.  Each edge on the surface is converted to a \emph{line segment} data structure. Each line segment stores both the 2D and 3D positions of its endpoints, and the curve type (e.g., contour, boundary, etc.). The segment stores a pointer to the mesh face or edge that it lies on; each endpoint stores the vertex or edge it came from, if appropriate.  

Each line segment stores ``head'' and ``tail'' pointers, like a doubly-linked list.  The ``head'' pointer points either to the next edge in the list, or else to a singularity object; likewise, the tail pointer points the other way.  Each line segment also records whether or not it is visible; initially, visibility for all segments is marked as ``unknown.''

Each singularity records its type, and information specific to the singularity, e.g., a curtain fold points to the near and far line segments that it connects; an image-space intersection records the four line segments (near and far) that it connects.

Since any edge may have arbitrary numbers of overlapping curves, line segments may be broken repeatedly into smaller and smaller segments during construction of the View Graph.


As a simple example, once visibility is computed, one can draw a curve by starting at an arbitrary visible segment, and following pointers forward and backward, and continuing through singularities when possible, stopping only when the curve becomes invisible.  Concatenating the 2D positions visited along the way yields a curve to draw. 

\section{Curve-based visibility algorithms}\label{sec:vis_algo}

The  basic visibility algorithm, is as follows:
\begin{enumerate}
	\item \textbf{Detect all edges and project them to the image plane.} Each edge stores both the 2D and 3D coordinates of its vertices.
	\item \textbf{Optionally, mark locally-invisible curves}: mark concave contour edges as invisible.  Mark curves that lie entirely on, or adjacent to, back-faces as invisible.
	\item \textbf{Insert a singularity at each curtain fold vertex}.
	\item \textbf{Detect intersections on the surface}, i.e., when a boundary and contour edge pass through the same vertex, by iterating over all boundary vertices. Insert a singularity at the intersection point. Intersections involving two non-visibility-indicating curves can be ignored. 
	\item \textbf{Compute image-space intersections} between all pairs of edges; this can be done using a sweep-line algorithm in $O(n \log(n))$ time with $n$ edges, \eg{} \citep{Bentley:1979}.  Intersections where the near curve is not visibility-indicating can be ignored.  Intersections on the surface should be ignored, since they are handled in the previous step. 
	\textbf{Split the edges} at the intersection point and insert a singularity.  The near edge does not need to be split since its visibility will not change at the intersection; doing so may still be useful for later stylization.
	\item \textbf{For each edge where visibility is not yet marked, determine visibility} using a ray test to the center of the edge. Optionally, \textbf{propagate visibility} to adjacent edges, as described in the next section.
\end{enumerate}
As a reminder, the visibility-indicating curves are contours, boundaries, and surface intersections.   The above computations can be sped up by combining steps, e.g., the first four steps can all be performed with a single iteration over the mesh (or over the edges, for static meshes with one of the data structures of Section \ref{sec:dataStructures}).

The above algorithm assumes that all curves lie on edges. For curves within edges, such as surface-intersections and hatching curves, the same basic procedure is used as well.

In the above computations, if invisible edges will never been drawn, then steps involving edges already known to be invisible can be skipped, to speed up computations.  For example, concave contour edges can be omitted from the image-space intersection step.  However, invisible edges are necessary for hidden-line rendering, and useful for QI propagation (described in the next Section); they are also very useful for visualization and debugging.

\citet{Markosian:1997} also point out that curves adjacent to the scene's bounding box in image space must be visible. This is only useful for situations such as viewing only a single object in isolation, as opposed to entering a full 3D environment. A simple way to use this observation is to  
find the contour or boundary points with the maximum and mininimum $x$ and $y$ values; those points must be visible.

\paragraph{Visibility propagation.}
In the most basic version of the above algorithm, we perform a ray test for each edge that is not a concave contour. However, ray tests are computationally expensive, and we would like to perform as few of them as possible.

When two curve edges are connected at a shared endpoint, and there is no singularity, then the two edges must have the same visibility. Using this observation, we can propagate visibility after each ray test, following connections between edges until reaching a singularity. This simple propagation substantially reduces the number of ray tests required.

\paragraph{Implementation choices and numerics.}

In practice, there are many different ways to implement the algorithms in this section. One might first implement the vanilla ray-test algorithm. One then might implement grouping sequences of edges into singularity-free \textit{chains}, and do one ray test per chain, then implement visibility propagation between chains. One can then add in additional constraints, e.g., concave contours must be invisible, and far edges of intersections must be invisible. At each phase of implementation, the algorithm should work correctly; each additional piece then accelerates the computation.  On the other hand, implementing chains in the visibility pipeline adds considerable implementation complexity for a questionable amount of benefit.

There are multiple constraints on visibility that can be exploited for debugging. For example, if a ray test marks a concave contour as visible, then there is a bug or numerical error in either the concavity test or the ray test.  

In practice, any of these tests can be corrupted by numerical errors.  One heuristic is to ignore tests that are close to a threshold, or to vote among multiple tests (e.g., multiple ray tests at different points on a chain). Numerical issues are discussed more in Appendix \ref{app:numerical}.

\section{Quantitative Invisibility} 

We can propagate visibility information even further --- and thus reduce the number of ray tests --- by using the concept of \emph{Quantitative Invisibility} (QI)~\citep{Appel:1967,Markosian:1997}. 
The QI value of a point is the number of occluders of the  point. A visible point has QI of zero (\eg{} point $\vec{p}_2$ in \fig{occlusions}). A point blocked by two surfaces has QI of two. In practice, the QI of a point $\vec{p}$ can be computed by counting the number of mesh faces that intersect the line segment $\overline{\vec{pc}}$, excluding the face containing $\vec{p}$.

\begin{figure}
	\centering
	\small
	\def\svgwidth{\textwidth}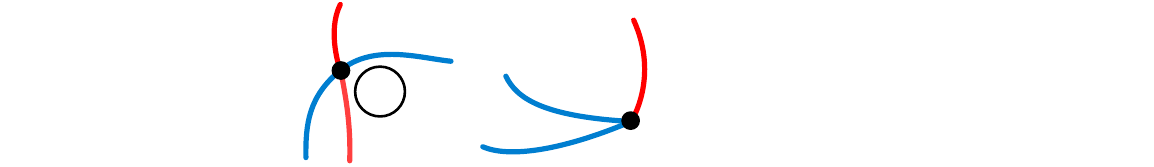
	\caption{\textbf{View Graph \& QI propagation} --- QI values can be propagated at image-space intersections and curtain folds.
	{\large\textcircled{\footnotesize 1}} T-junctions at image-space intersections, {\large\textcircled{\footnotesize 2}} Y-junctions between a contour generator, two boundary edges, {\large\textcircled{\footnotesize 3}} contour generator curtain folds, {\large\textcircled{\footnotesize 4}} boundary curtain folds.
	}\label{fig:view_graph}
\end{figure}

Each type of singularity in the View Graph imposes different constraints on the QI. QI values may be propagated over the surface by the following rules:

\begin{itemize}
\item \textbf{Front/back, convex/concave}: Points on back-faces, and concave contour edges, must have QI greater than zero. 

\item \textbf{The image-space bounding box of the scene must be visible}. For example, if a single object is viewed in isolation, all of the outer edges (the outer silhouette) must be visible. A practical test is to find all 2D edges with minimum and maximum $x$ and $y$ coordinates.

\item \textbf{Image-space intersection}: suppose that the nearest edge to the camera is a contour with QI value $q$.  The occluded far edge must have QI of $q+2$, and the other side must have QI of $q$ (\fig{view_graph}~{\large\textcircled{\footnotesize 1}}). If the occluding curve is a boundary, then the occluded far edge must have QI of $q+1$.

\item \textbf{Intersection on the surface where a contour terminates at a boundary}.
The boundary curve and the near contour generator must have the same QI of $q$. The far boundary edge may have QI of either $q$ or $q+1$ (\fig{view_graph}~{\large\textcircled{\footnotesize 2}}).  The specific value can be determined by a local overlap test, similar to the boundary curtain fold detection test.

\item \textbf{At a curtain fold:} if the near edge has a QI of $q$, the far edge will have a QI of at least $q + 1$ (\fig{view_graph}~{\large\textcircled{\footnotesize 3}},{\large\textcircled{\footnotesize 4}}).  However, the far edge's QI could be higher in some exotic, unusual cases.  For example, in a boundary curtain fold where the one-ring neighborhood spirals multiple times around the vertex like a fusilli pasta, the QI could increase more than 1. Hence, a local overlap test is necessary to count how many triangles in the vertex's one-ring overlap the far edge.



\item \textbf{Other cases where multiple curves meet, including  bifurcations}: the differences in QI between the adjacent curves can be determined using local overlap tests, similar to the boundary curtain fold detection test. For example, two curves that meet at a bifurcation that are not occluded by any triangles in the one-ring must have the same QI.

\end{itemize}

Hence, the resulting algorithm begins by first building the View Graph. The QI for most edges is initially marked as ``unknown,'' though some edges can also be marked as ``invisible'' ($q>0$), such as concave contours.
 A single ray test is performed at some edge with unknown QI. By propagating this value through the view map, the QI can be determined for every edge in this edge's connected component. Hence, at most one ray test is necessary for each connected component. It is possible to determine some connected components' visibility without any ray tests at all. Propagating lower-bounds on QI increases the number of cases where this works. For example, a concave edge has $q>0$; if it is the near edge at a curtain fold, then the far edge has $q>1$.

\section{Planar Maps} \label{sec:planar_map}

The Planar Map is a generalization of the View Graph that provides a more complete representation for artistic rendering: it represents not just the curves in a drawing, but also the regions between them.  Given a Planar Map, one could theoretically stylize regions, strokes and their relationships in a more coherent way.

The Planar Map is a concept originally from graph theory.  
Given a set of 2D curves $\mathcal{C}$, the Planar Map corresponds to the \emph{arrangement} $\mathcal{A}(\mathcal{C})$ of those curves, that is the partition of the plane into 0-, 1- and 2-dimensional cells (\ie{} vertices, edges and faces) induced by the curves in $\mathcal{C}$. This partitioning is coupled with an incidence graph which allows navigation between adjacent cells.

Intuitively, the Planar Map is constructed by the following procedure.
Specifically, all mesh faces are projected into the image plane. Faces that are completely occluded are discarded; faces that are partially occluded are subdivided in their visible and invisible parts, and their image-space adjacency information is updated.


\begin{figure}
	\centering
	\small
	\def\svgwidth{\linewidth}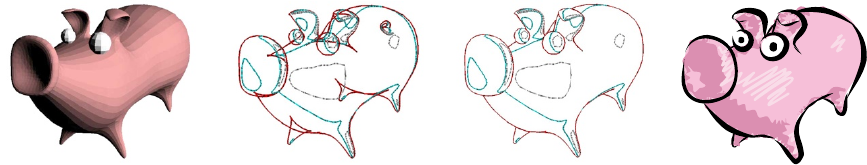\caption{\textbf{Planar Map} (from \citet{Eisemann:2008}) --- Starting from a 3D model \textbf{(a)}, contours and isophotes are first extracted \textbf{(b)} and their visibility is computed by constructing a Planar Map \textbf{(c)} that yields a base vector depiction for stylization \textbf{(d)}.
Isophotes are curves with constant shading, i.e., $(\vec{c}-\vec{p})\cdot \vec{n}=\alpha$ for some constant $\alpha$. 
	}\label{fig:clipartmesh}
\end{figure}

If $\mathcal{C}$ contains the projection into the image plane of the contour generator and boundary curves, then the Planar Map $\mathcal{A}(\mathcal{C})$ corresponds to a generalization of the view graph presented in \sect{sec:view_graph}, one that includes the 2D regions bordered by the curves. By only keeping the closest cells to the camera, visibility can be determined (\fig{planar_map}).

\citet{Winkenbach:1994} introduced the use of Planar Maps for stylized rendering.  They used a 3D BSP tree to compute the visibility of the mesh faces~\citep{Fuchs:1980}, and a 2D BSP tree to build a partition of the image plane according to the visible faces. From this 2D BSP tree, they construct the Planar Map to have direct access to 2D adjacency information. They showed that this representation allows one not just to stylize contours, but to stylize the regions between the contours.  

Computing visibility with BSP trees is both very expensive and numerically sensitive, even for simple models. A modern implementation of 2D arrangements is offered by the CGAL library~\citep{Fogel:2012}.

\citet{Eisemann:2008} compute an approximate Planar Map by first computing the View Graph, and then joining regions between curves, identifying 3D correspondence using hardware buffers (\fig{clipartmesh}). This method is sufficient when a precise mapping between geometry and image space is not needed.

\section{Non-orientable surfaces}
\label{sec:non-orientable}

As stated in Section \ref{sec:front_back}, this tutorial assumes that all surfaces are orientable, and that only front-faces may be visible. It is also possible to generalize these algorithms to handle non-orientable surfaces (Figure \ref{fig:klein_bottle}), though with some additional complexity.  This section outlines some of modifications to the algorithms of the last two chapters, though not all details will be spelled out.

We begin with contour detection between two triangles $\triangle \vec{a} \vec{b} \vec{c}$ and $\triangle \vec{a} \vec{b} \vec{d}$. We cannot directly use the front/back-facing test for contours, because facing direction is not defined on non-orientable surfaces.

For each triangle, there are two possible normals. For the first triangle, the possible normals are $$\vec{n}_1 = \pm \frac{(\vec{b}-\vec{a})\times (\vec{c}-\vec{a})}{|| (\vec{b}-\vec{a})\times (\vec{c}-\vec{a}) ||}.$$ The possible normals for the second triangle are similarly plus or minus the face normal.  

In order to determine whether an edge is a contour, we must determine a locally-consistent pair of normals, that is $\hat{\vec{n}}_1$ which is either $\vec{n}_1$ or $-\vec{n}_1$, and $\hat{\vec{n}}_2$ which is either $\vec{n}_2$ or $-\vec{n}_2$. Consistent and inconsistent cases are visualized in Figure \ref{fig:invalid_edge}.
The pair of normals is consistent as long as both normals agree as to whether the edge is convex or concave; see Appendix \ref{app:convex}.

\begin{figure}
	\centering
	\small
	\textbf{(a)}~
	\def\svgwidth{0.22\textwidth}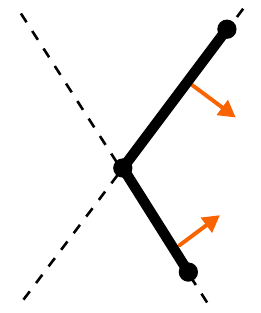	
	\qquad
	\textbf{(b)}~
	\def\svgwidth{0.22\textwidth}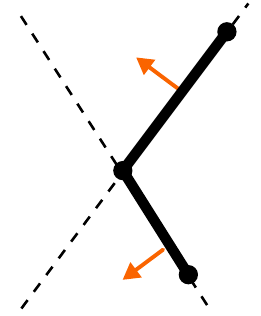	
	\qquad
	\textbf{(c)}~
	\def\svgwidth{0.22\textwidth}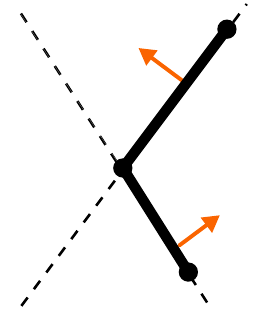	
	\caption{\textbf{Valid/invalid configurations} ---
	Two adjacent triangles, shown in cross-section, with their assigned normals $\vec{\hat{n}}_1$ and $\vec{\hat{n}}_2$. The cross-section is some 3D plane perpendicular to the edge between the two triangles.
	\textbf{(a)} and \textbf{(b)} are valid configurations of the normals and \textbf{(c)} is invalid.
	}
	\label{fig:invalid_edge}
\end{figure}

Once we have computed the locally-oriented normals $\hat{\vec{n}}_1$ and $\hat{\vec{n}}_2$, we can use the local sign test to determine if the edge is a contour, checking if the sign of $(\vec{a} - \vec{c}) \cdot \hat{\vec{n}}_1$ is the same as the sign of $(\vec{a} - \vec{c}) \cdot \hat{\vec{n}}_2$.

These local orientations cannot be reused when looking at other faces; when determining whether $\triangle \vec{a} \vec{b} \vec{c}$ has a contour with one of its other neighbors, a local pair of normals must be computed for this pair of edges.

In general, any steps of the visibility algorithm that rely on the definition of front-facing or back-facing (or of convex or concave edges) must (a) be skipped if they are optional; (b) compute locally-consistent orientations before use; or (c) be replaced with a more general computation. For example, to determine if a vertex is a contour curtain fold, one could either compute locally-consistent orientations for the vertex's entire one-ring, or one could use the image-space self-overlap rule instead. One cannot assume that back-faces are invisible, as one can with oriented surfaces.  

\chapter{Fast Hardware-Based Extraction and Visibility}
\label{chap:fast_visibility}

 This chapter describes algorithms that use graphics hardware, such as multipass rendering and Graphics Processing Units (GPUs), to perform real-time contour detection and visibility.  This improves performance over the CPU algorithms of the past few chapters, which can be very slow, especially for very complex geometry. Hardware-based methods can be very fast; in return, they do not guarantee correctness in all cases.


The earliest hardware methods directly produce visible contours using two rendering passes on the graphics card (\sect{sec:twopass}), yet with limited stylization capabilities. Subsequent approaches massively parallelized the contour detection step (\sect{sec:gpu_extraction}) or the visibility computation (\sect{sec:gpu_visibility}) separately. They can be combined to render stylized lines at interactive framerates for complex 3D scenes.

\section{Two-pass hardware rendering} \label{sec:twopass}

 The basic idea of these approaches is to render the geometry twice: first, to fill the depth buffer, and then second, using modified geometry, to make the contours emerge from the rasterization.

\begin{figure}
	\centering
	\includegraphics[width=0.45\linewidth]{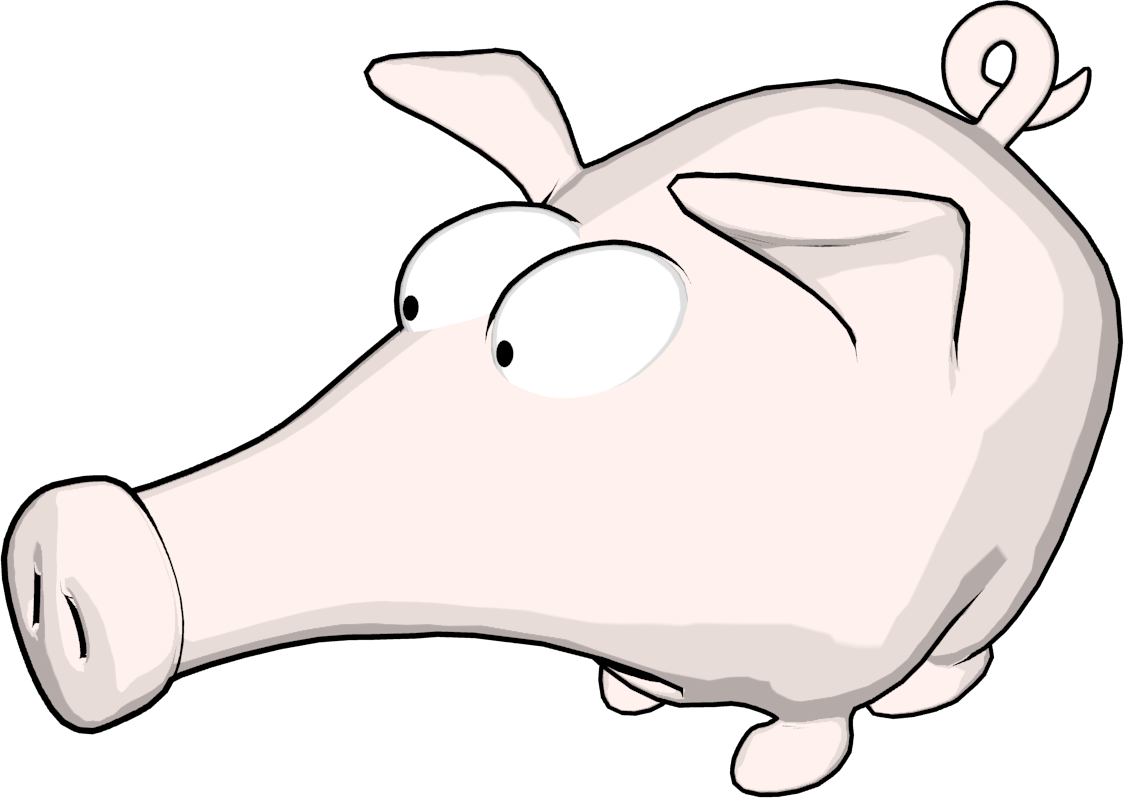}
	\quad
	\includegraphics[width=0.45\linewidth]{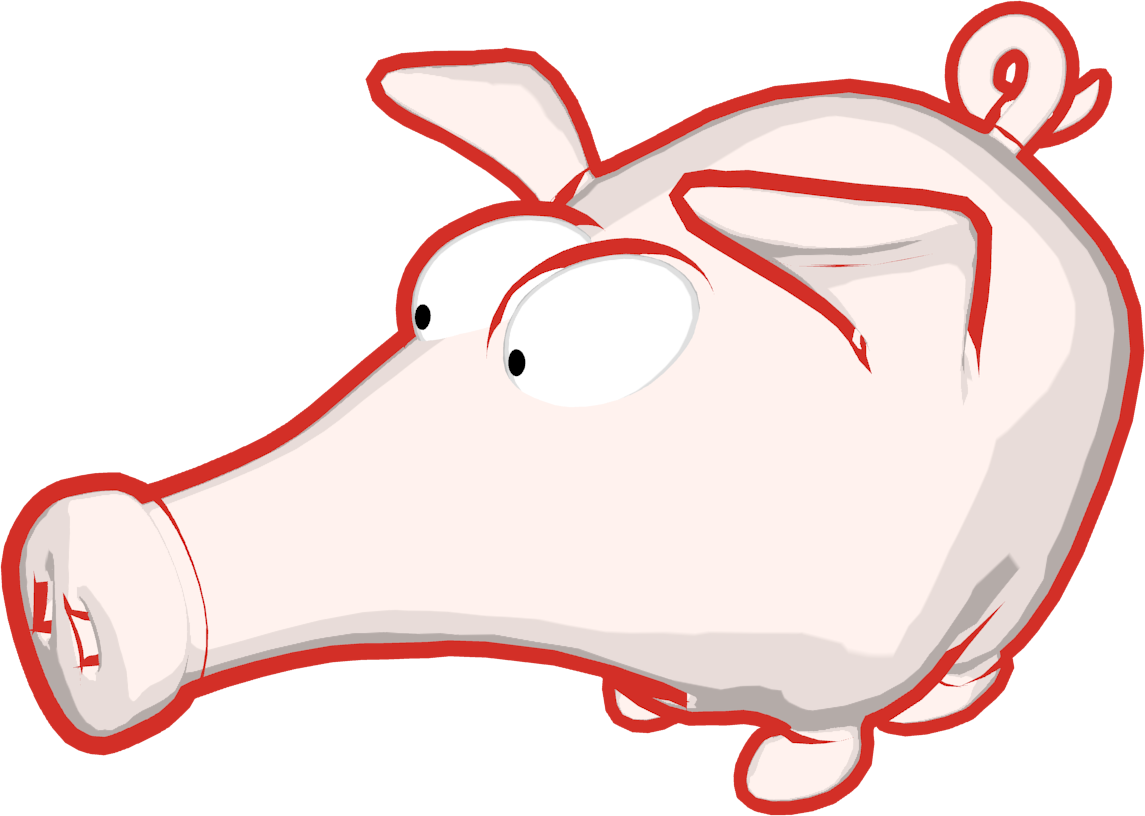}
	\caption{\textbf{Two-pass hardware rendering} --- After a first standard rendering pass, the scene is enlarged and only back-faces are rendered; the visible parts of the back-faces produce the black edges (left). The back-faces scaling factor and color allow to control the line width and color (right). Image computed with Blender Solidify modifier~\citep{BNPRSolidify}.}
	\label{fig:hybrid}
\end{figure}

For instance, \citet{Rossignac:1992} render the mesh in wireframe mode with thick lines, after a translation slightly away from the viewpoint. Alternatively, only back-facing polygons offset towards the camera can be rendered, such that they show through front-faces \citep{Raskar:1999,Gooch:1999}. Using the first generation of programmable hardware, a similar effect is achievable in one pass by enlarging all back-facing polygons in the Vertex Shader~\citep{Raskar:2001,Chen:2015:RTAS}.

Unlike the image-space filtering approaches (\chap{chap:image_space}), two-pass methods provide more stylization options, namely, more control over line thickness and color.  However, they do not provide stylization capabilities beyond these features. Furthermore, the need for two passes may make the method too slow for large models; nonetheless, they have been used in video games for simple models \citep{St-Amour:2010}. Finally, these methods are especially useful for rendering contours of unstructured geometric representation such as point clouds~\citep{Xu:2004}.

Subsequent approaches independently accelerate the contour extraction or the visibility computation using the graphics card.

\section{Contour extraction on the GPU} \label{sec:gpu_extraction}

Once the graphics pipeline offered programmable stages with Vertex and Fragment Shaders, GPU implementations mirroring the brute force CPU algorithm described in Section~\ref{sec:extraction} started to be possible. However, face adjacency information was not available initially. \citet{Card:2002,Brabec:2003,McGuire:2004b} circumvented this limitation by drawing every edge of the mesh as a quadrilateral fin, storing as vertex attributes their two adjacent face normals. The dot product of these normals with the view direction can then be performed in the Vertex Shader, and non-contour vertices can be discarded. 

With the introduction of the Geometry Shader stage, this is not required anymore~\citep{Stich:2007}. Regular mesh geometry with adjacency information can be sent to the GPU, and each face is then processed in parallel in the Geometry Shader. Some care must be taken not to detect the same contour edges twice, \eg{} by discarding back-faces~\citep{Hermosilla:2009}. \citet{Sander:2008} even obtained a speedup close to $2 \times$ compared to this naive GPU implementation using a scheme for efficient traversal of mesh edges. 

If the edges detected by the Geometry Shader are needed for a second rendering pass, a transform feedback operation can be used to read them back to the main GPU memory or even to the CPU.

\section{Hardware-accelerated visibility computation} \label{sec:gpu_visibility}

In this section, we present three hardware-based visibility techniques. In each case, the 3D curves are first detected either on the CPU or the GPU. 
The first two visibility techniques work at image-space pixel precision and thus tend to suffer from aliasing artifacts, whereas the last method is mostly resolution independent. However, the first method is extremely simple to implement; in the simplest version, it just requires adding some 3D line segments in a 3D renderer.

\subsection{Direct rendering with the depth buffer}
\label{sec:buffer}

Once the contour edges have been extracted either in software by one of the methods in Sections \ref{sec:extraction} and \ref{sec:interpolated_contours_detection} or on the GPU with the previous technique, the simplest and fastest solution to determine their visibility is to use the standard depth buffer algorithm. First, the 3D scene is rendered with depth writes enabled --- potentially disabling color writes if only lines should be rendered. The contour generator is then drawn as a 3D polyline with the less-or-equal depth function (\eg{} with \verb|glDepthFunc(GL_LEQUAL)| in OpenGL) and depth writes disabled. This way, line fragments occluded by the mesh geometry are automatically discarded during the depth test (\fig{depth_buffer}). To avoid ``depth fighting'' between the polyline and the underlying surface, a small offset can be applied to the fragments' depth  when rendering the mesh (\eg{} with \verb|glPolygonOffset| in OpenGL). Occlusion queries can be used if the result of the depth test needs to be read back on the CPU \citep{Eisemann:2008}.

\begin{figure}
	\centering
	\small
	\begin{subfigure}[b]{0.22\linewidth}
		\def\svgwidth{\textwidth}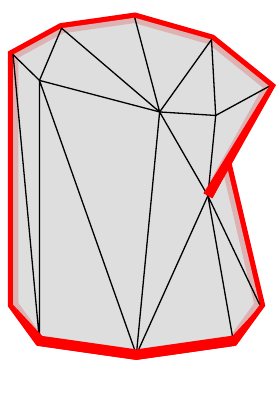\caption{depth test}\label{fig:depth_buffer}
	\end{subfigure}
	\hspace{1pt}
	\begin{subfigure}[b]{0.71\linewidth}
		\def\svgwidth{\textwidth}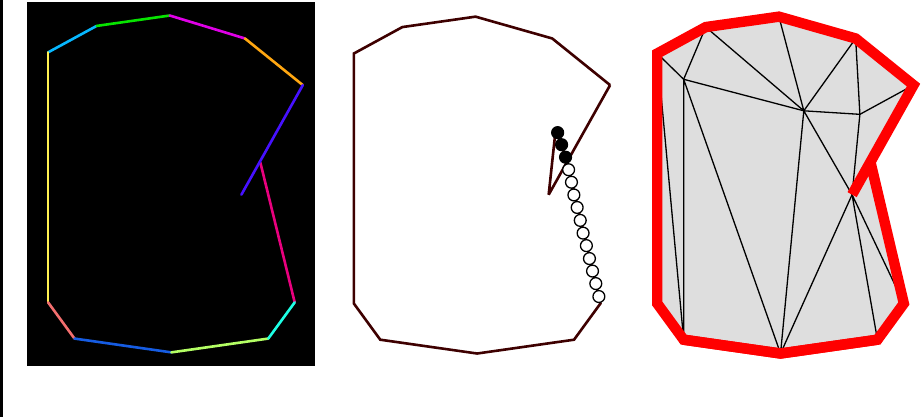\caption{ID test}\label{fig:item_buffer}
	\end{subfigure}
	\caption{\textbf{Buffer-based visibility} --- \textbf{(a)} Thick lines rendered with a simple depth test are irregularly occluded by the mesh geometry; \textbf{(b)} by first computing an item buffer with thin lines, and then probing visibility in this buffer (black and white circles) along the clipped lines, thick or stylized lines are correctly rendered. (The black wireframe is depicted for illustration purposes.)}\label{fig:buffers}
\end{figure}

This technique works well for pixel-wide line rendering, but thicker lines may sometimes partially disappear in the geometry. It is difficult to stylize curves this way, other than using thick lines. In addition, this depth test is unstable because contour edges are often adjacent to faces that are almost parallel to the viewing direction, and thus both edges and faces project to the same depth buffer pixels, but the faces, rendered first, wrote depth values closer to the viewer. To address this issue, \citet{Isenberg:2002} suggested modifying the depth test of a line fragment, by not only considering the single pixel of the depth buffer to which it projects, but also its $8 \times 8$ neighborhood. \citet{Cole:2010} proposed a full GPU implementation of such an approach, called the \emph{spine test}, also suggesting computation of the depth buffer at a higher resolution to reduce artifacts due to undersampling.
For detailed tutorial with a GLSL implementation of antialiased lines, see \citet{Rideout:2010}.

\section{Item buffer}

An alternative solution proposed by \citet{Northrup:2000} is based on an \emph{item buffer}, which had previously been used to accelerate ray-tracing~\citep{Weghorst:1984}. The idea is to render each line into an off-screen buffer (\eg{} a Framebuffer Object in OpenGL) with a thickness of one pixel and a unique color (ID). Each pixel of the item buffer eventually contains the unique color of a single visible line fragment at that pixel (\fig{item_buffer}). By scan-converting each line and reading the ID at the corresponding location in the item buffer, the visible portions of the line, called \emph{segments}, can be determined. Each segment, or chain of segments, can then be rendered with thick lines, or even more complex stylization effects (\sect{sec:stylization}). \citet{Kaplan:2007} showed that this approach can be extended to compute Quantitative Invisibility, thus allowing hidden line rendering with different styles, \eg~dotted lines. However, the item buffer suffers from two major limitations: it cannot be trivially anti-aliased, since line IDs cannot be averaged, and multiple lines cannot project to the same pixels even though they might all be partially visible. \citet{Cole:2008b} improved on those two aspects by computing a \emph{partial visibility} for each line, using super-sampling and ID peeling, but with a significant memory overhead (typically $\times 12$ to 16).

\section{Segment Atlas}

\begin{figure}
	\centering
	\small
	\def\svgwidth{\textwidth}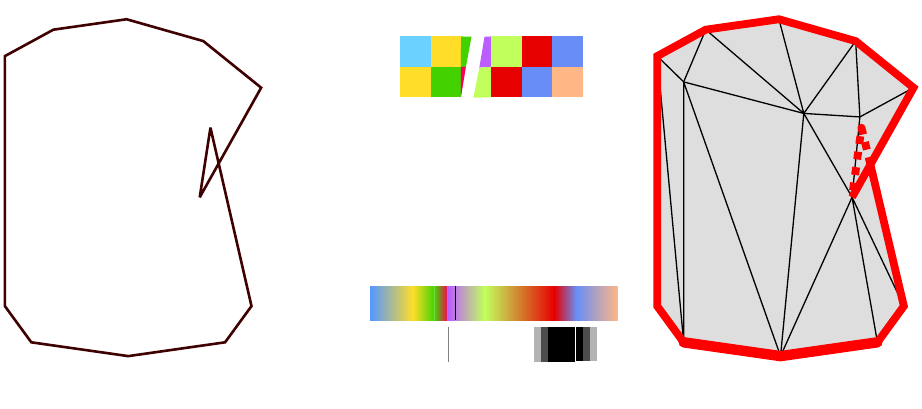\caption{\textbf{Segment atlas} --- Each input 3D line $\overline{\vec{p}_i\vec{q}_i}$ is projected and clipped by a fragment shader, which also computes its associated number of visibility samples $l_i$. These samples are then converted into segment atlas offsets $s_i$ by a running sum. The positions $(\vec{p}_i',\vec{q}_i')$ are eventually interpolated and the resulting sample positions $\vec{v}_j$ are used to compute the partial visibility $\alpha_j$ at this location by reading the depth buffer. The clipped segments $\overline{\vec{p}_i'\vec{q}_i'}$ can then be rendered in screen-space leveraging partial visibility to modulate the style of the line. (The black wireframe is depicted for illustration purposes.)}\label{fig:segment_atlas}
\end{figure}

\citet{Cole:2010} circumvented the limitations of the item buffer by introducing a novel data-structure, called the \emph{segment atlas}, that stores visibility samples along each line segment independently of their actual screen position. The segment atlas is created on the GPU in three steps (\fig{segment_atlas}). First, the input 3D lines are projected and clipped to the camera frustum with a dedicated fragment shader. For each 3D line $\overline{\vec{p}_i\vec{q}_i}$, the position of its endpoints $(\vec{p}_i',\vec{q}_i')$ in homogenous clip space are stored inside a GPU buffer along with a number $l_i$ of visibility samples proportional to the screen space length of the line (potentially equal for maximum precision). During a second pass, a running sum turns the sample counts $l_i$ into segment atlas offsets $s_i$. In a third step, the sample positions $\vec{v}_j$ are effectively created. Each clipped segment $\overline{\vec{p}_i'\vec{q}_i'}$ is discretized by generating a line from $s_i$ to $s_i + l_i$ in a geometry shader and letting the rasterizer interpolate the endpoint positions. For each generated fragment, a shader performs the perspective division and viewport transformation to produce the screen-space coordinate $\vec{v}_j$ of the sample. The depth buffer is then probed at this position and the returned value is compared with the sample own depth value; the partial visibility resulting from this test $\alpha_j$ is written in the segment atlas, by construction, at the proper location. Finally, the line segments (or chains of segments with little modifications) can be rendered with arbitrary thickness and style using the fragment-level visibility information provided by the segment atlas (Figure \ref{fig:segmentAtlas_stylization}). This method is up to $4 \times$ slower than direct OpenGL rendering (\sect{sec:buffer}) for small 3D models, but $2\times$ slower (or better) for complex meshes.

\begin{figure}
	\centering
	\includegraphics[width=0.45\linewidth]{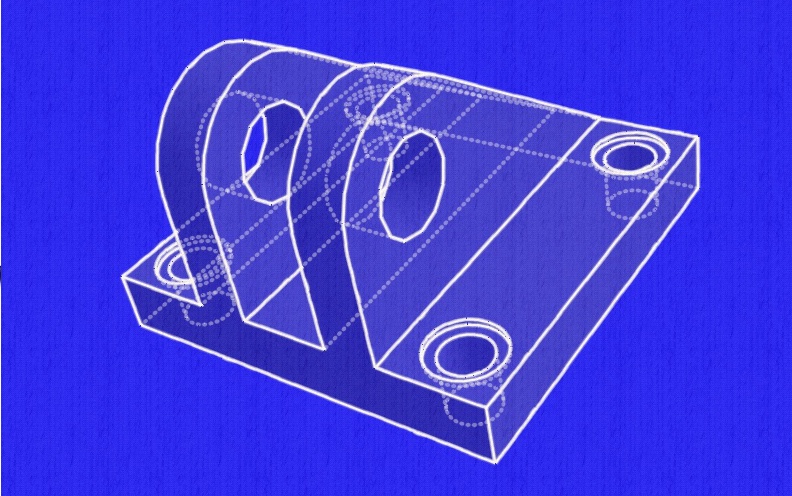}
	\quad
	\includegraphics[width=0.45\linewidth]{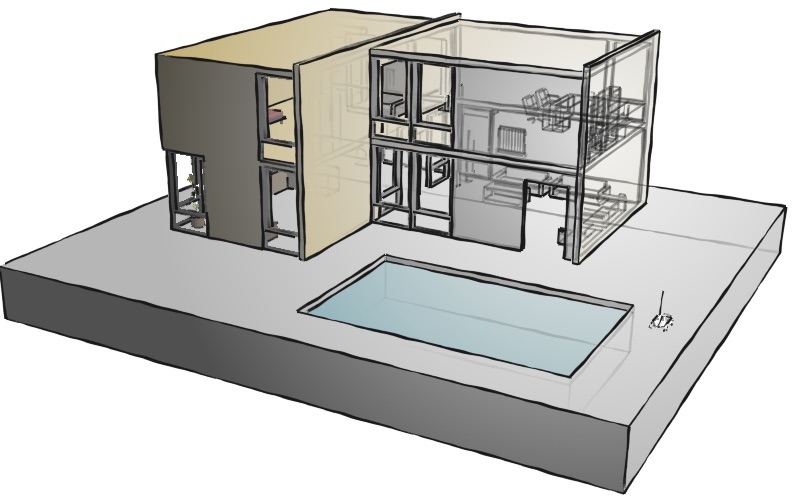}
	\caption{\textbf{Stylized line renderings using the Segment Atlas algorithm} \citep{Cole:2010} --- Hidden lines are included in these renderings. Images generated with ``dpix''~\citep{dpix}.}
	\label{fig:segmentAtlas_stylization}
\end{figure}


\chapter{Smooth Surfaces as Meshes}
\label{chap:smooth_as_meshes}

The algorithms described in the previous chapters work for polyhedral meshes. 
In this chapter, we describe heuristics for treating smooth surfaces as meshes for rendering.  In Chapter~\ref{chap:smooth_contours}, we will begin formal discussion of the theory and algorithms for smooth surfaces; using this theory avoids the problems with these heuristics.

\section{The ups and downs of mesh rendering for smooth surfaces}
\label{sec:ups_downs}

Many previous researchers have taken the approach of converting their smooth surface into a triangle mesh, and then computing the contours of that mesh.  This may seem like a sensible strategy, as it is common in computer graphics to tessellate a smooth surface and simply render the tessellation.
If the contours will be directly rendered as line segments, e.g., thick black lines (as in Chapters \ref{chap:image_space}), then this method produces good results.

Unfortunately, for stylized curves, this strategy leads to numerous artifacts, as illustrated in Figure \ref{fig:red_mesh}.  In fact, the topology of the contours will invariably get worse: smooth contours cannot exhibit bifurcations, but mesh contours can exhibit arbitrary branching. 

\begin{figure*}
  \centering
  \small
  \begin{subfigure}[b]{0.36\linewidth}
  \includegraphics[width=\textwidth]{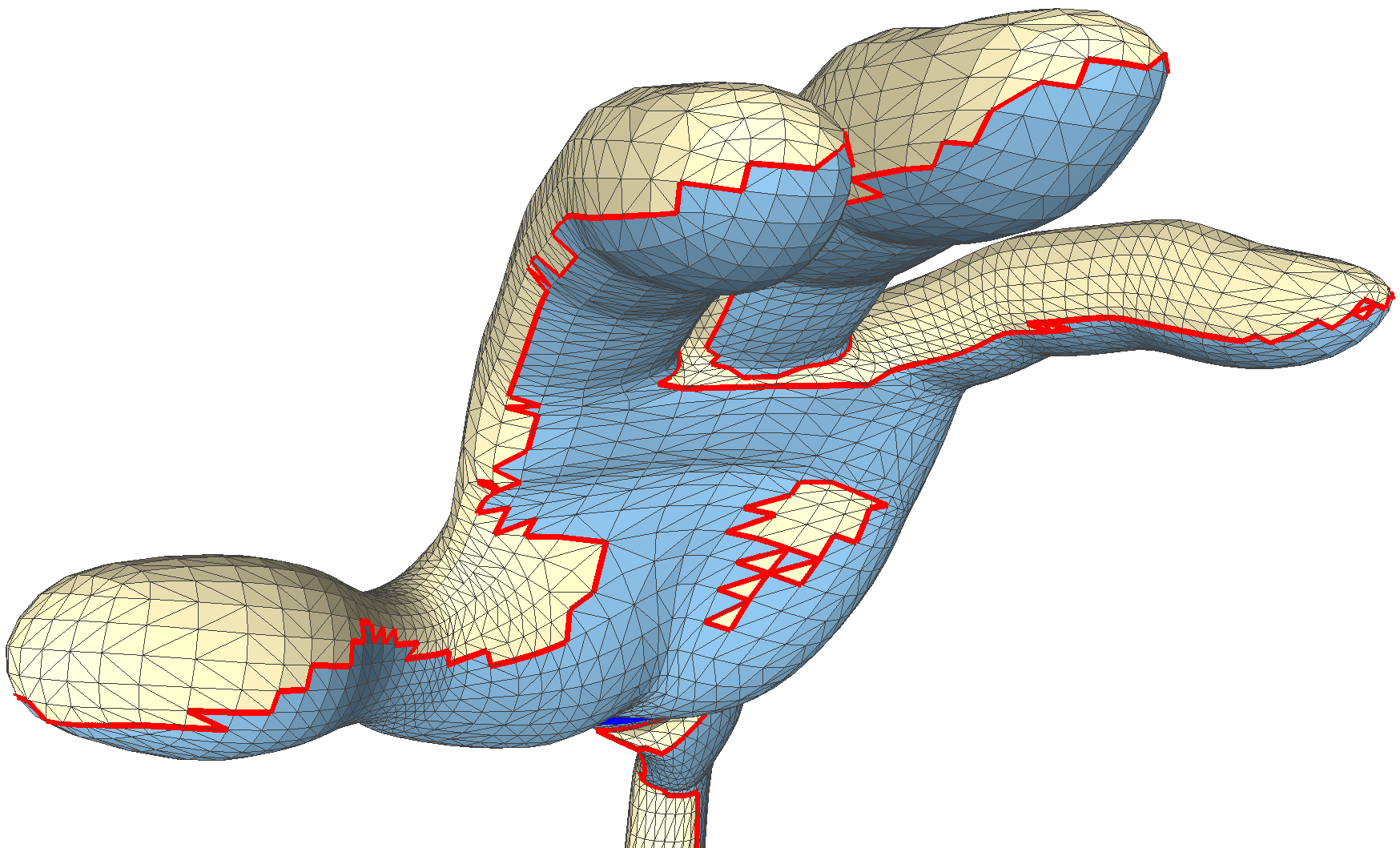}
	  \caption{Mesh contours}
  \end{subfigure}
  \begin{subfigure}[b]{0.3\linewidth}
  \includegraphics[width=\textwidth]{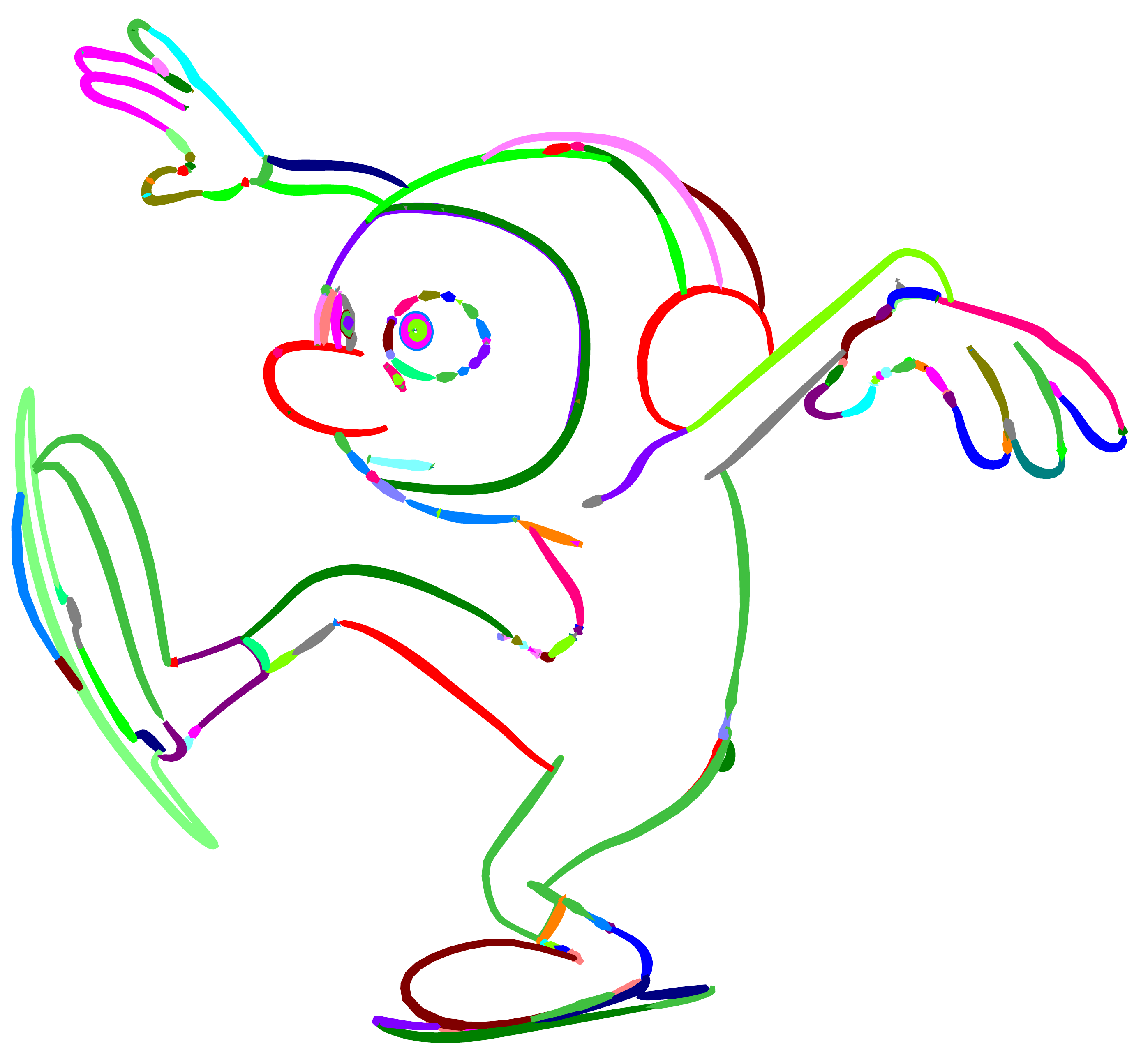}
  \caption{Chained curves}
  \end{subfigure}
  \begin{subfigure}[b]{0.32\linewidth}
    \includegraphics[width=\textwidth]{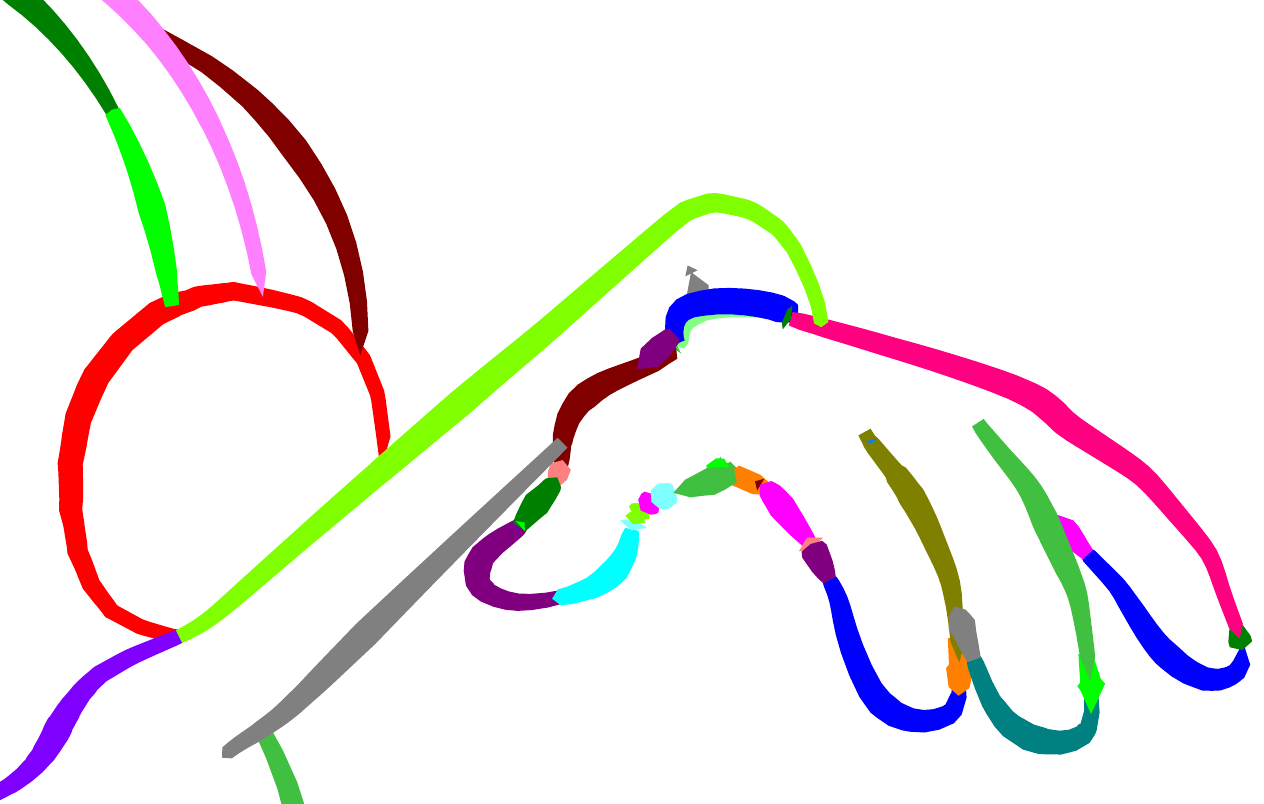}
    \caption{Closeup view on \textbf{(b)}}
  \end{subfigure}
  \caption{ \label{fig:red_mesh} \textbf{Mesh contours of a smooth surface} \citep{Benard:2014} --- 
  The original surface is shown in \fig{red_intersections}.
  \textbf{(a)} Converting the surface to a triangle mesh and extracting contours produces overly complex topology, including many bifurcations not present in the smooth surface's contours.
  \textbf{(b-c)} Connecting chains of edges that end at singularities produces many small chains, not directly suitable for stylization. Each chain is shown in a different color. ``Red'' \ccCopy Disney/Pixar
  }
\end{figure*}

Sometimes one starts from a mesh that begins polyhedral, such as geometry from web repositories or range scanners.  Nonetheless, the contours may be unexpectedly messy, because the underlying surface is smooth, even if the representation is not.  This may surprise the practitioner who is used to using triangle meshes for smooth objects.

Over the years, many researchers have observed the problems with mesh contours and developed heuristics to address them.   Using these kinds of heuristics can be effective when speed and simplicity is valued over perfection. This chapter describes heuristics for identifying cleaner contours, and Section \ref{sec:simplification} describes heuristics for cleaning up mesh contours as a post-process.

Because these methods are heuristic, animating these curves usually produces some flickering artifacts, where curve sections appear and disappear. It is often argued that these artifacts are allowable, because hand-drawn animation typically exhibits some flickering that gives it a sense of imperfection and life. However, the flickering artifacts are often qualitatively different from hand-drawn animation: these artifacts are errors that humans would not normally make.

In our experience, every researcher who encounters these problem expects there to be a simple fix.  They immediately propose their own clever solution. Efforts to turn these clever solutions into perfect results have always failed. 

We caution the reader that these heuristics will never produce perfectly correct curve topology. If you convert a smooth surface to a mesh, then you have discarded information about the true contour topology, and cannot recover it from the mesh. Using these heuristics, one should expect that curves will occasionally be connected incorrectly, and outlines might temporarily vanish. 
If you are working in an environment where visual perfectionism is valued over expediency, then trying eliminate errors from mesh contour rendering will lead to endless frustration. 

On the other hand, for most cases, these heuristics are usually good enough.
Accurate visibility, as discussed in Chapter~\ref{chap:smooth_contours}, is more difficult.

\section{Interpolated Contours}
\label{sec:interpolated_contours}

The Interpolated Contours approach \citep{Hertzmann:2000} produces smooth contours from meshes, by using some ideas from smooth surface contours, but applied on meshes. It produces cleaner contours by design, e.g., bifurcations are not possible.  The approach is analogous to Phong shading in computer graphics, in which a mesh is treated as if it has smoothly-varying normals.  Interpolated Contours has been used in many research projects, including Freestyle \citep{Grabli:2010}, now in Blender.

The general outline of the method is the same as for mesh contours: detect contours, detect singularities, and then compute visibility by ray tests and visibility propagation. The main difference is that Interpolated Contours pass within faces rather than on mesh edges.

\subsection {Contour definition and detection}
\label{sec:interpolated_contours_detection}

\begin{figure}
  \centering
  \small
  \def\svgwidth{0.4\textwidth}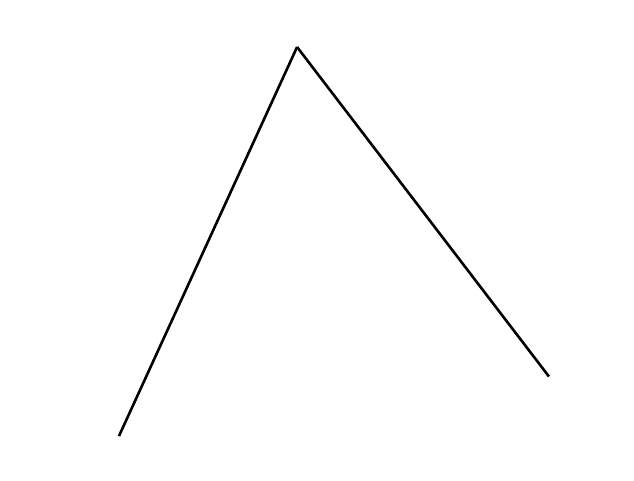\caption{\textbf{Linear interpolation} --- The function $g(\vec{p})$ is linearly interpolated along the oriented edges $(\vec{p}_3-\vec{p}_1)$ and $(\vec{p}_2-\vec{p}_3)$ of the triangle to find the endpoints of the line segment approximating the contour generator within the face.
  }\label{fig:linear_interpolation}
  \label{fig:linearContour}
\end{figure}

The approach is as follows. First, we assign a ``fake'' normal vector to each mesh vertex.  As in Phong shading, this normal vector is a weighted average of the normals $\vec{n}_j$ of the adjacent faces, weighted by triangle areas $A_j$. This vector should be normalized to be a unit vector:
\begin{align*}
\vec{\hat{n}}_i &= \frac{\sum_{j \in N(i)} A_j \vec{n}_j} {\sum_{j \in N(i)} A_j} \\
\vec{n}_i &= \vec{\hat{n}}_i / || \vec{\hat{n}}_i ||
\end{align*}
with $N(i)$ the one-ring face neighborhood of vertex $i$. One may also use the more robust weights of \citet{Max:1999}.

For a given camera center $\vec{c}$, we define the orientation function $g$ at a vertex $i$ as:
$$g(\vec{p}_i) = (\vec{p}_i - \vec{c}) \cdot \vec{n}_i.$$
A vertex with $g(\vec{p}) > 0$ is considered to be front-facing, and with $g(\vec{p}) < 0$ is back-facing. The generic position assumption implies that we cannot have $g(\vec{p}) = 0$ at a vertex.

We then linearly interpolate $g(\vec{p})$ within the face, which is equivalent to linearly interpolating the normals within a face.  This defines the orientation function over the entire surface as a piecewise linear function.

A face in which the orientation function has opposite signs at two vertices contains a contour.  This contour is a line segment within the face (Figure \ref{fig:linearContour}).  The endpoints of the line segment are found as follows.

On the edge between the vertex $i$ and the vertex $j$, the linear interpolation is:
$$g(t) = (1-t)g(\vec{p}_i) + t g(\vec{p}_j).$$
A contour crosses this edge when the sign of $g(\vec{p}_i)$ is opposite the sign of $g(\vec{p}_j)$. Solving for $g(t) = 0$ gives the contour point position as (\fig{linear_interpolation}):
\begin{align}\label{eq:interpolation}
  t          & = \frac{g(\vec{p}_i)}{g(\vec{p}_i)-g(\vec{p}_j)},                                                     \nonumber   \\
  \vec{p}(t) & = (1-t)\vec{p}_i + t \vec{p}_j = \frac{g(\vec{p}_i)\vec{p}_i - g(\vec{p}_j)\vec{p}_j}{g(\vec{p}_i)-g(\vec{p}_j)}.
\end{align}

As illustrated in Figures \ref{fig:interp_torus} and \ref{fig:interp_red}, these contours typically have much smoother and coherent topology than the mesh contours.

\begin{figure}
  \centering
  \small
  \begin{subfigure}[t]{0.48\linewidth}
    \includegraphics[width=\linewidth]{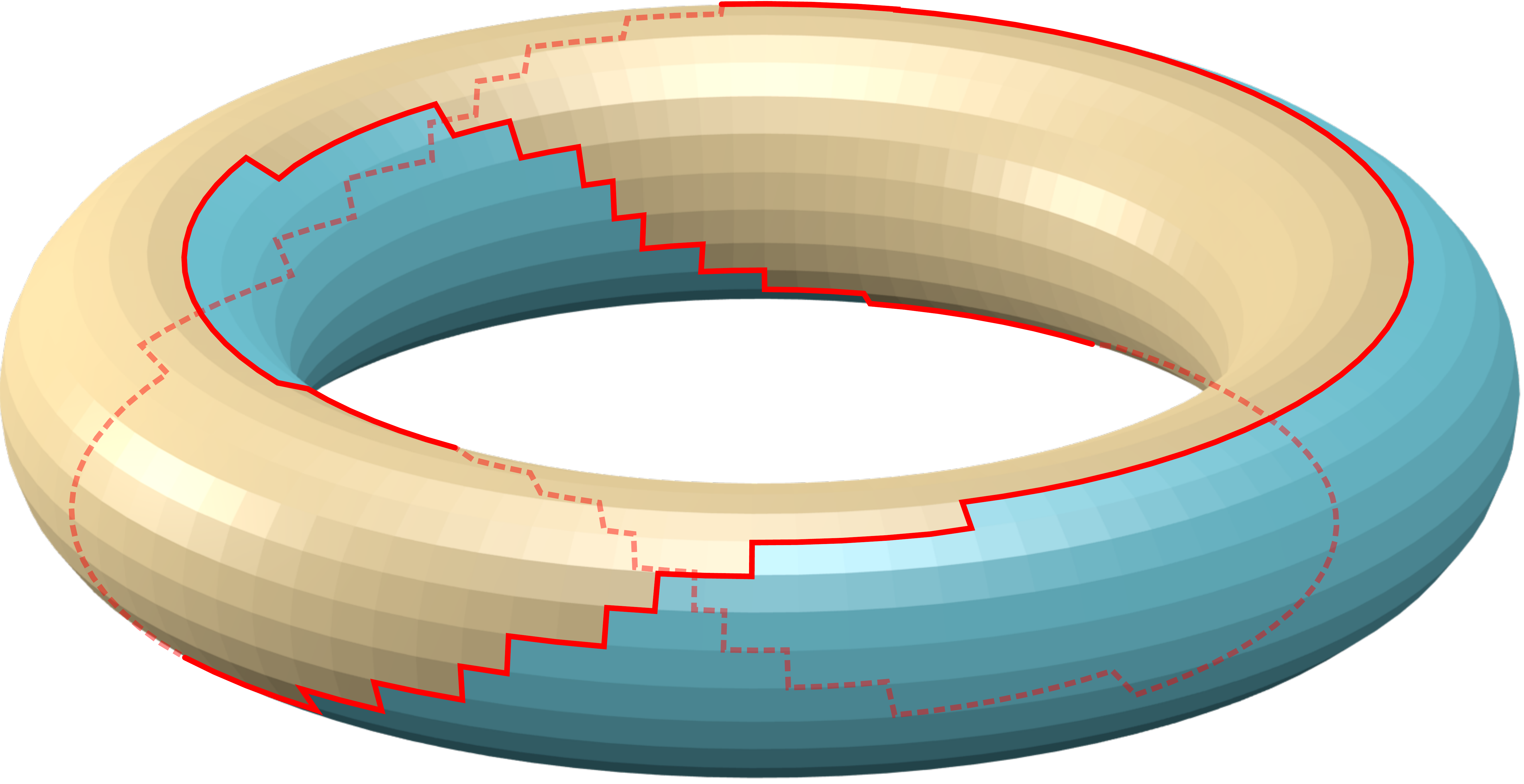}
    \caption{mesh contours}
  \end{subfigure}
  \quad
  \begin{subfigure}[t]{0.48\linewidth}
    \includegraphics[width=\linewidth]{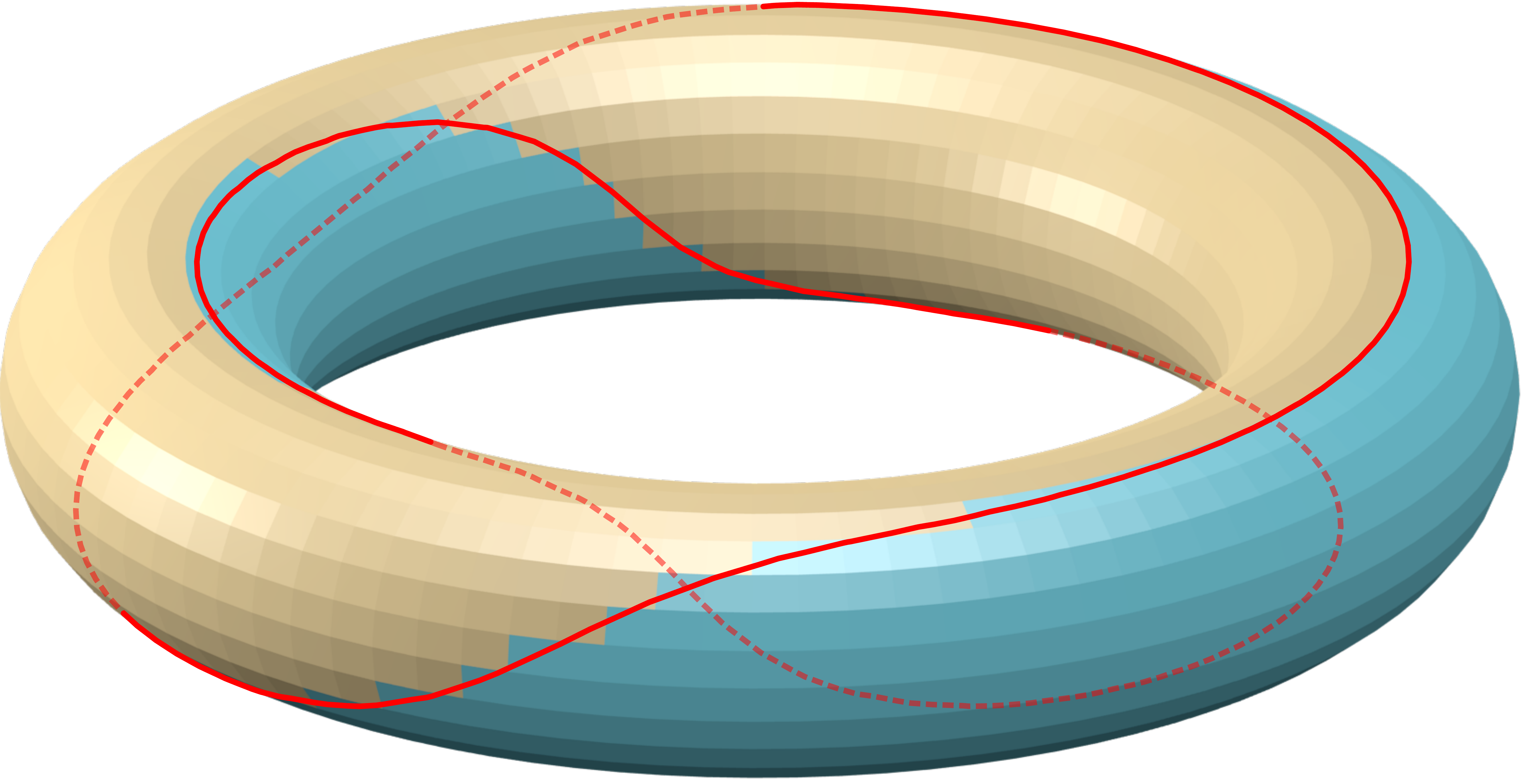}
    \caption{Interpolated Contours}\label{fig:interp_contours}
  \end{subfigure}
  \caption{\textbf{Mesh vs. Interpolated Contours} --- Unlike mesh contours \textbf{(a)}, the piecewise linear approximation of the contour generator \textbf{(b)} crosses back-faces of the polygonal mesh and may thus be hidden by front-faces closer to the camera.
  \label{fig:interp_torus}
  }
\end{figure}

\begin{figure*}
  \centering
  \begin{subfigure}[b]{0.36\linewidth}
    \includegraphics[width=\textwidth]{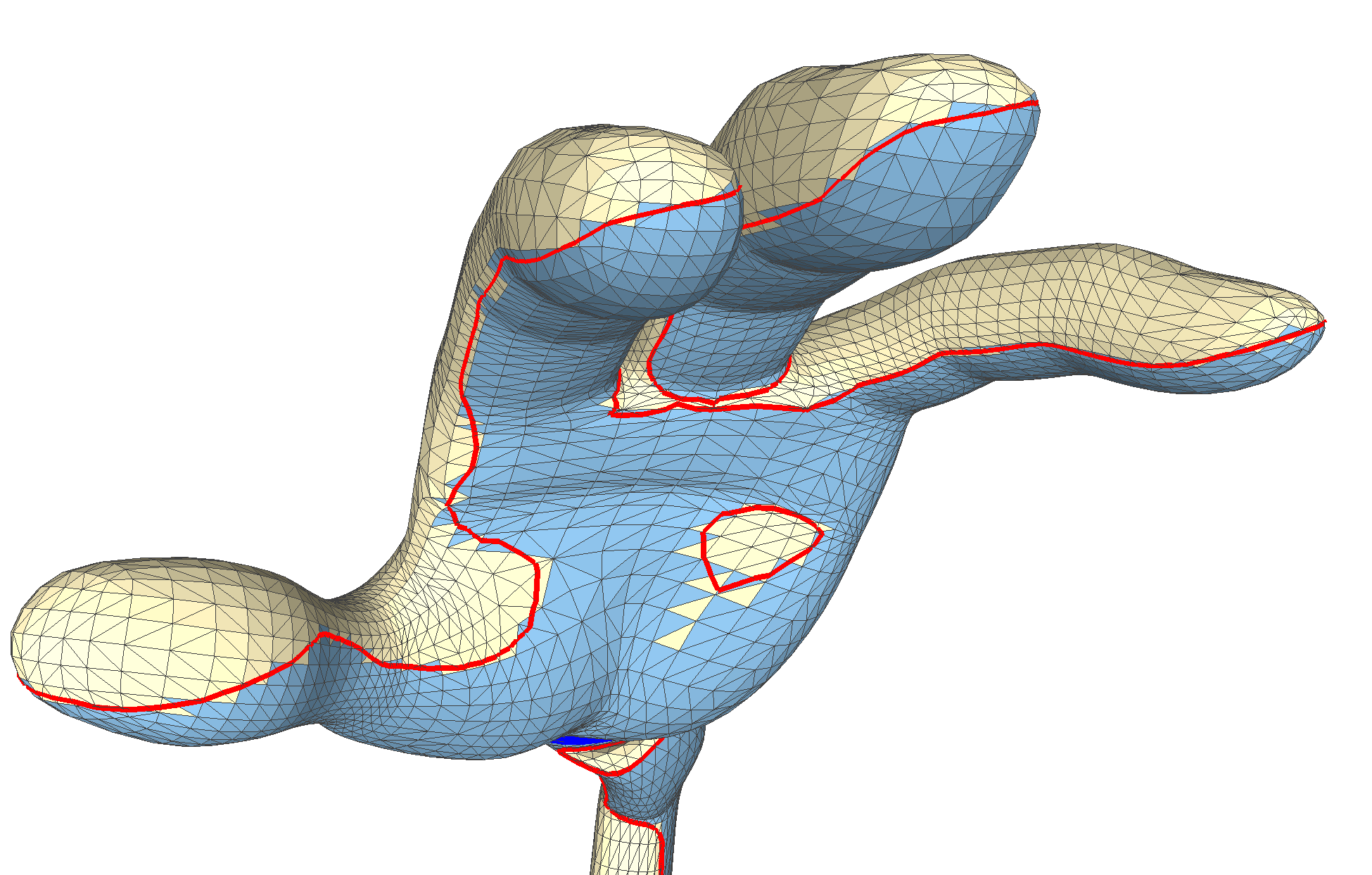}
    \caption{Interpolated Contours}
  \end{subfigure}
  \begin{subfigure}[b]{0.3\linewidth}
    \includegraphics[width=\textwidth]{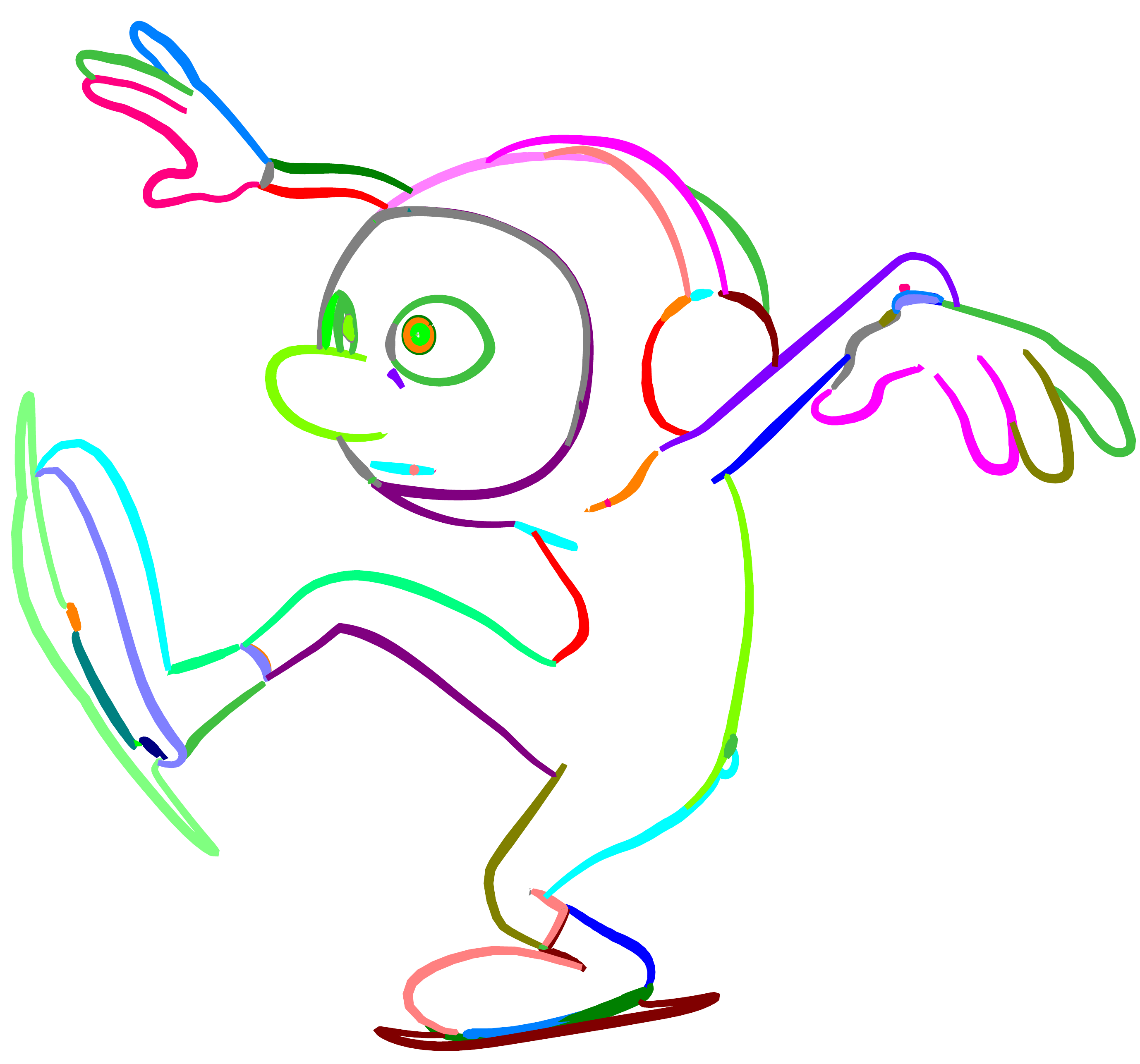}
    \caption{Chained curves}
  \end{subfigure}
  \begin{subfigure}[b]{0.32\linewidth}
    \includegraphics[width=\textwidth]{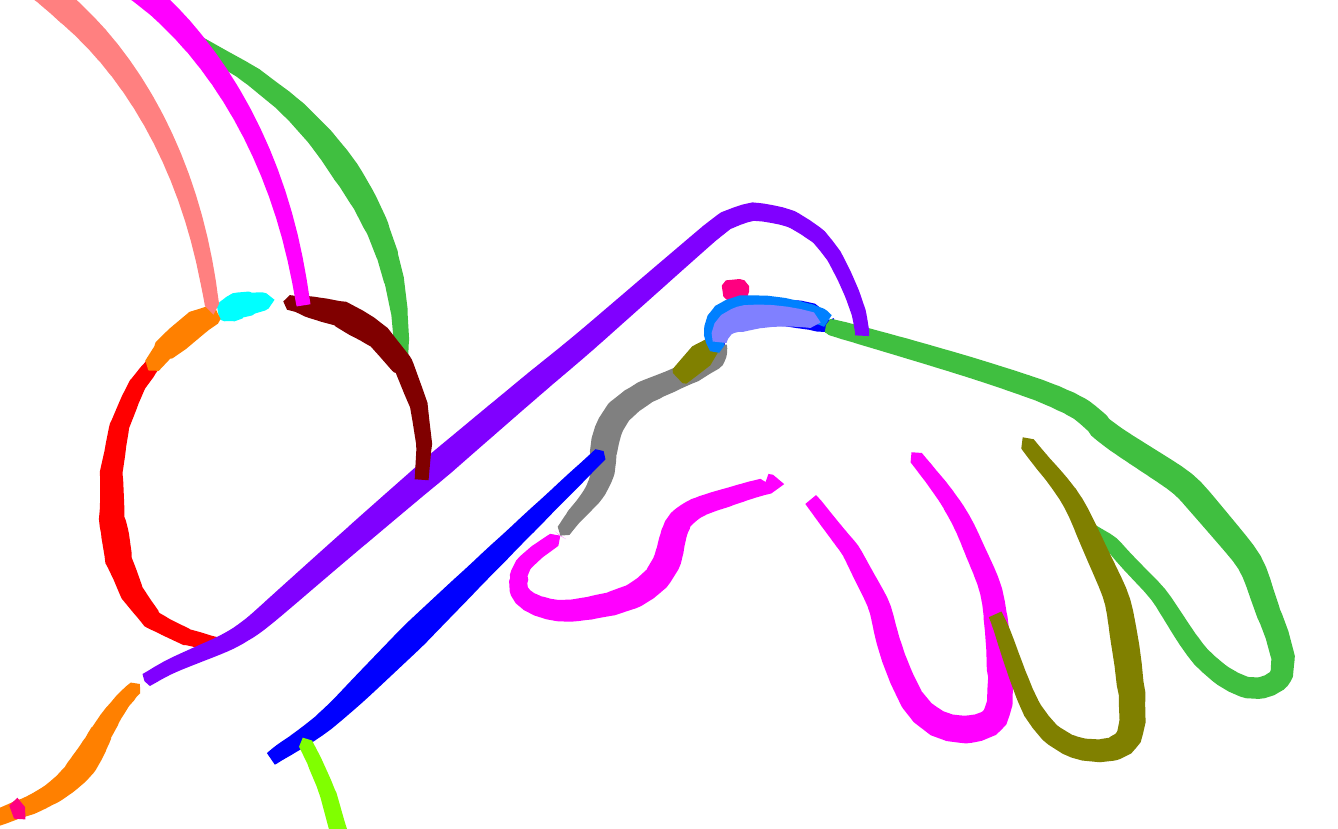}
    \caption{Closeup view on \textbf{(b)}}
  \end{subfigure}
  \caption{\textbf{Interpolated Contours of a smooth surface} \citep{Benard:2014} ---
  The original surface is shown in \fig{red_intersections}.
  \textbf{(a)} Interpolated Contours are much smoother and have approximately correct topology (compared with Figure \ref{fig:red_mesh}).
  \textbf{(b)} Chaining these curves gives much smoother, coherent curves.
  \textbf{(c)} Visibility is not well-defined for these curves, and gaps and other small errors may appear.
  ``Red'' \ccCopy~Disney/Pixar
    \label{fig:interp_red}
  }
\vspace{-0.75ex}
\end{figure*}

\subsection{Fast detection for static surfaces}

The dual space data structures of \sect{sec:dataStructures} can be adapted to find Interpolated Contours on static meshes.  For example, in the perspective dual space method of \citet{Hertzmann:2000}, each mesh vertex maps to a dual point $\vec{s}=(s_1,s_2,s_3,s_4)=(-n_x,-n_y,-n_z,\vec{n}\cdot\vec{p})$. As before, the camera maps to dual plane $g(\vec{s})=(c_x,c_y,c_z,1)\cdot{s}= 0$.  A dual edge is drawn between each pair of adjacent vertices; a dual edge contains a contour point if the edge crosses the camera dual plane. Hence, finding all edges with contour points is again reduced to intersecting a plane with a set of line segments.

\subsection{Singularities}
\label{sec:singularities}

The singularities of Interpolated Contours are similar to those of mesh contours. However, they behave somewhat differently. Interpolated Contours cannot exhibit bifurcations. Intersections on the surface lie within faces, since these contours lie within faces.

Defining curtain folds for Interpolated Contours requires the theory for smooth contours, which we will describe in the next chapter. For now, we will simply assume that we have a way to compute a function $\kappa_r$ at each mesh vertex.  This function, called the radial curvature, will be defined later in \sect{sec:radial_curvature}.
Linearly interpolating this function across each face gives a function $\kappa_r(\vec{p})$ over the face, and a line segment with $\kappa_r(\vec{p})=0$ can be computed by linear interpolation across the face edges, just as was done for the contour generator. For a mesh face that contains zero crossings in both $g(\vec{p})$ and $\kappa_r(\vec{p})$, the curtain fold lies at the intersection of these two line segments, if they intersect \citep{Hertzmann:2000,DeCarlo:2003}.
For methods to compute curvature from meshes, see \citep{Vasa:2016}.

When detecting a curtain fold this way, there will often be a spurious image-space intersection between the contour and itself near the curtain fold, and one may need a heuristic to clean up this case.  This can get tricky if other image-space intersections occur between these singularities.

\subsection{Visibility}

For ray tests, we use the original triangle mesh to determine when the smooth contour is occluded by the mesh \citep{Hertzmann:2000}.
These computations are necessarily heuristic, because the Interpolated Contours are not the contours of the mesh.  Some of the visibility techniques for mesh contours do not apply for Interpolated Contours; for example, it is unclear whether there is a useful analogue to ``concave'' and ``convex'' contours for Interpolated Contours, or whether QI can be propagated safely.  As a result, the simplest choice is to use multiple ray tests per curve for visibility, whenever possible.

Unfortunately, the approximate contour generator is not the mesh contour generator. About half the segments of the approximate contour generator lie on back-faces of the triangle mesh (whatever the tessellation density of the mesh), and they are thus hidden by front-faces closer to the camera (\fig{interp_contours}). This leads to a number of visibility errors (\fig{interp_red2}).
Several heuristics have been proposed to mitigate this problem. 

\begin{figure}
  \centering
  \includegraphics[width=0.6\textwidth]{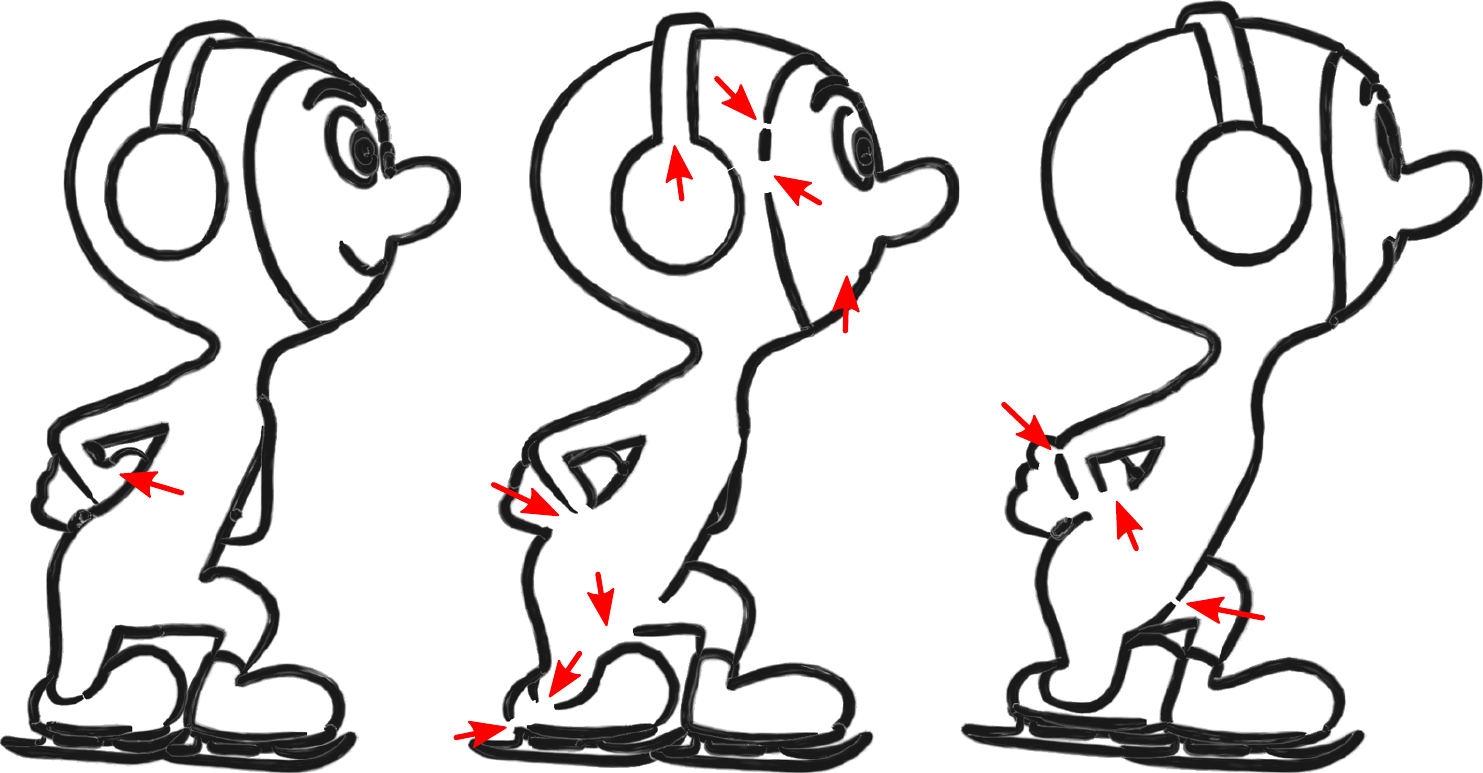}
  \caption{\textbf{Stylized interpolated contours of a smooth surface} \citep{Benard:2014} ---  
  The original smooth surface is shown in \fig{red_intersections}. After visibility computation, the contours exhibit many breaks and gaps (red arrows) which lead to objectionable temporal artifacts after stylization. ``Red'' \ccCopy Disney/Pixar}
  \label{fig:interp_red2}
\end{figure}

\citet{Hertzmann:2000} use a voting scheme.  Line segments between singularities are combined into chains, and then multiple ray tests are performed for each chain. These ray tests occur at the mesh vertices on faces with contours, using the vertices nearest to the viewer. The visibility is determined by a vote of these ray tests.
In addition, \citet{Grabli:2010} ignore occlusions from triangles adjacent to the contour's face. However these heuristics are not robust in every configuration. 

\begin{figure}
	\centering
	\small
	\def\svgwidth{0.55\linewidth}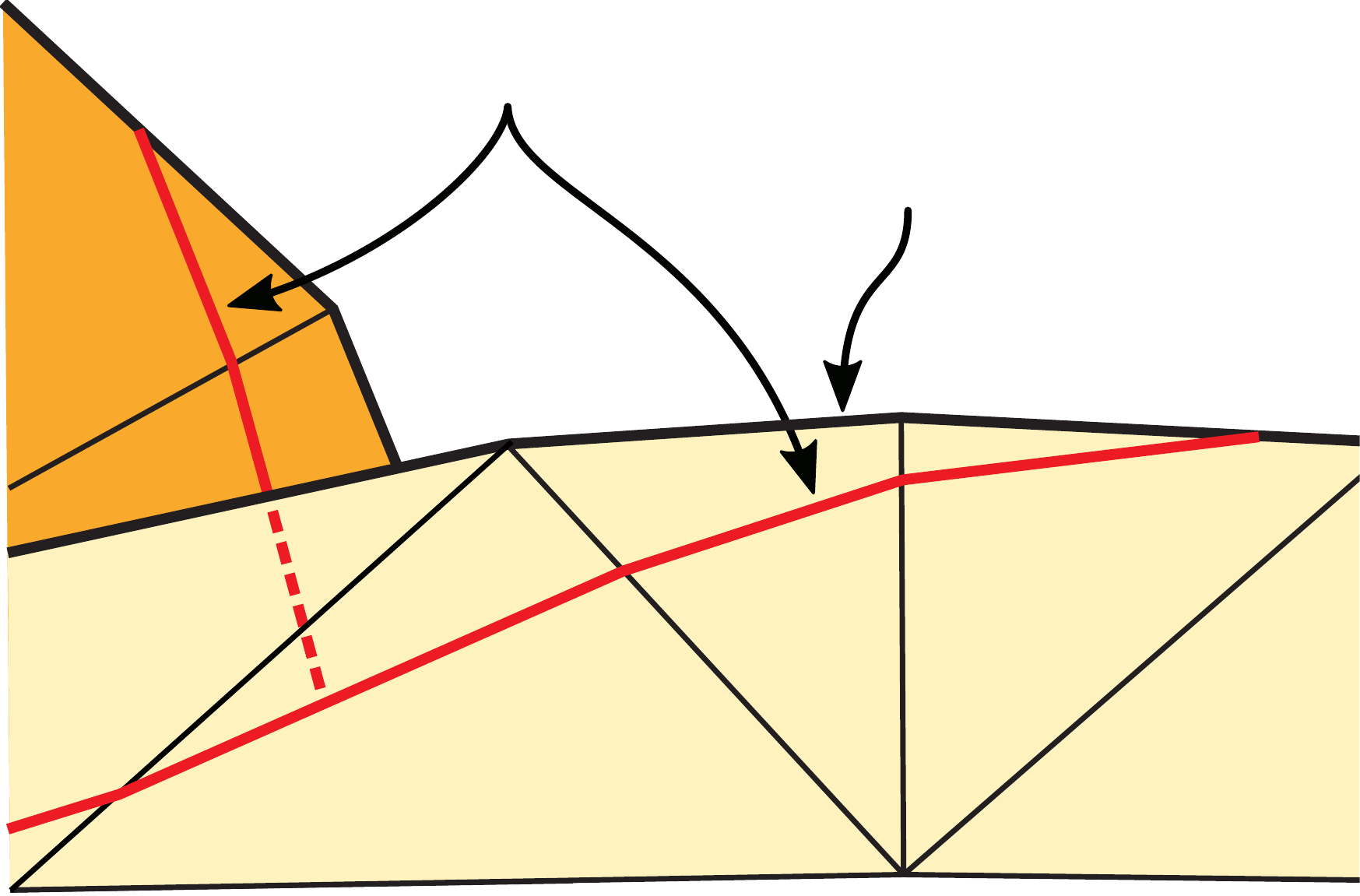\caption{
		\textbf{An example of the problem with interpolated contour visibility}  \citep{Benard:2014} ---
		In image space, the interpolated contour lies within the mesh contour.  This creates a ``halo region'' between the two contours in which the surface occludes other surfaces but the mesh is invisible. Ray tests to the rear surface in this region will say the rear surface is invisible. In other configurations, such as nearly-flat, bumpy surface, a surface can ``halo'' other curves nearby on the same surface.}\label{fig:haloing}
\end{figure}

A more subtle issue is that the mesh silhouette does not line up with the smooth contour \citep{Eisemann:2008,Benard:2014}. Ray tests are performed against the mesh, and the part of the mesh outside the interpolated contour can occlude other surfaces, making them incorrectly invisible (\fig{haloing}).  Other problem cases are discussed in \citet{Benard:2014}.

\section{Planar Maps} \label{sec:planar_map_interp}

Planar Maps, as introduced in Section \ref{sec:planar_map}, represent all the visible strokes and regions within a graph.  For mesh rendering from smooth surfaces, they offer the appeal that, even if there might be errors in the mesh approximation, the resulting drawing will still be internally consistent; it cannot have, e.g., giant holes in the outline.

\citet{Eisemann:2009} presented a method to construct a Planar Map of the visible contours; their method involves a hybrid of mesh contours and Interpolated Contours. They first compute the view graph of the input scene, and then backproject it onto the mesh to define regions of constant visibility. Performing a ray test through each region center, they build an adjacency-occlusion graph by inserting links between successively intersected region pairs, as well as between adjacent regions. Contour edges at the frontier of a visible and an occluded region in this graph should be invisible in the Planar Map. The backprojection step is inspired by 3D BSP tree construction and involves many plane-triangle intersections that may thus suffer from the same numerical issues, requiring special tolerancing.  

\begin{figure}
	\centering
	\small
	\def\svgwidth{\textwidth}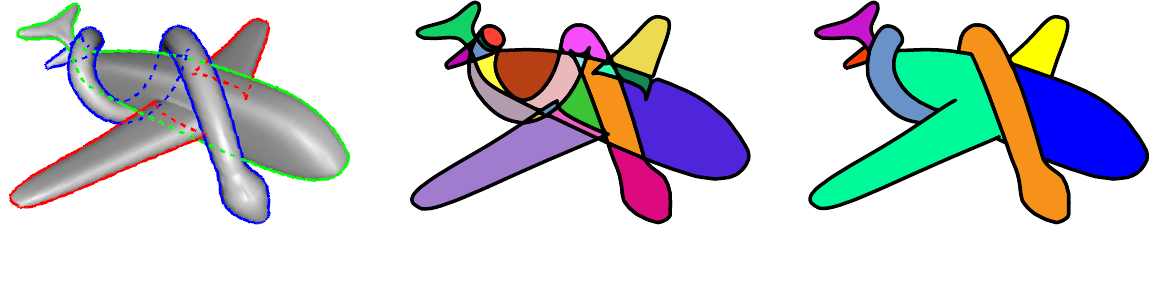\caption{\textbf{Planar Map} --- Starting from the projected mesh contours \textbf{(a)}, the Planar Map corresponds to the partition into cells of the 2D plane induced by these curves \textbf{(b)}. If only visible mesh faces are inserted in the Planar Map or snaxels are used, occluded contours will be discarded \textbf{(c).}}\label{fig:planar_map}
\end{figure}

\citet{Karsch:2011} avoided these robustness complications by computing jointly the partition on the mesh and in the image plane. They use \emph{snaxels}, \ie~active contours~\citep{Kass:1988} whose vertices lie on mesh edges, to delineate the zero sets of an implicit contour function defined over the mesh surface. While seeking to minimize the energy functional, the snaxels obey topology rules for splitting and merging. By detecting snaxels collisions both on the 3D surface and in the image plane, and designing proper merging rules, \citet{Karsch:2011} ensure that no contours are intersecting if they belong to separated regions of the mesh. The main problem of this approach is its initialization: one snaxel front per Planar Map region should be seeded on the 3D geometry, but these regions are not known a priori. Multiple passes, introducing additional fronts, might thus be required to converge to the correct solution.


\chapter{Parametric Surfaces: Contours and Visibility}
\label{chap:smooth_contours}

In the previous chapters, we only considered polygonal meshes as input. We will now present the theory of contours on smooth surfaces. Most of this theory mirrors that of mesh contours, but using the tools of differential geometry, whose fundamentals are briefly summarized in \app{app:diff_geom}. This chapter will focus on parametric surfaces; implicit surfaces and volumes will be described in \chap{chap:implicit_surfaces}. Algorithms to extract the apparent contour of smooth surfaces yield mostly-correct results for most surfaces. However, as we discuss, perfectly computing smooth surface contours remains an open research problem.

\section{Surface definition} \label{sec:parametric}

In the following, we will assume that all surfaces are at least $C^1$ smooth everywhere, though it is conceptually straightforward to generalize to surfaces with creases, since they behave like mesh edges (\chap{chap:mesh_contours}).
The theory in this chapter applies to any surface with a parametrization $\vec{u}$, and a surface function with position $\vec{p}(\vec{u})$ and normals $\vec{n}(\vec{u})$, but the algorithms are designed for spline patches and subdivision surfaces.  We briefly review these surfaces.

\paragraph{Spline patches.} 
Spline patches (also called ``freeform surfaces'') are parametric functions from a 2D domain to a surface in 3D: $f:\mathbb{R}^2 \rightarrow \mathbb{R}^3$.
Specifically, each input coordinate $\vec{u} = (u,v)$ maps to a 3D point:
$$\vec{p}(\vec{u}) = \left[
		\begin{array}{c}
			f_x(u,v) \\
			f_y(u,v) \\
			f_z(u,v)
		\end{array}\right]$$
on the surface. In a spline patch, these functions are defined as linear combinations of basis functions applied to control points.  
For example, in the nonuniform rational B-spline (NURBS) patch. The shape of such a patch is parameterized by a grid of $(m+1) \times (n+1)$ control points $\vec{p}_{i,j} \in \mathbb{R}^3$ and their associated scalar weights $w_{i,j} \in \mathbb{R}$. The 3D surface is then given by:
\begin{equation}\label{eq:NURBS}
	\vec{p}(u,v) = \frac{\sum_{i=0}^{n} \sum_{i=0}^{m} N_i^k(u)\, N_j^l(v)\, w_{i,j}\, \vec{p}_{i,j}}{\sum_{i=0}^{n} \sum_{i=0}^{m} N_i^k(u) N_j^l(v) \, w_{i,j}},
\end{equation}
where $N_i^k$ is the B-spline basis function of degree $k$ for the $i^{\mbox{th}}$ control point (\fig{NURBS}). Details can be found in most computer graphics textbooks. 

Regardless of the specific type of surface used,
the 
surface normal at a point $\vec{p}$ can be computed as follows. The two 3D vectors:
$$\vec{t}_u(\vec{u}) = \left.\frac{\partial \vec{p}}{\partial u}\right|_{\vec{u}} = \left[\frac{\partial f_x}{\partial u}, \frac{\partial f_y}{\partial u}, \frac{\partial f_z}{\partial u} \right]^\top
	\quad
	\vec{t}_v(\vec{u}) = \left.\frac{\partial \vec{p}}{\partial v}\right|_{\vec{u}} = \left[\frac{\partial f_x}{\partial v}, \frac{\partial f_y}{\partial v}, \frac{\partial f_z}{\partial v} \right]^\top$$
are tangent vectors at $\vec{p}$. A surface normal at that point is: $$\vec{n}(\vec{u}) =\vec{t}_u(\vec{u}) \times \vec{t}_v(\vec{u}),$$

\begin{figure}
	\centering
	\small
	\begin{subfigure}[b]{0.46\linewidth}
		\def\svgwidth{\hsize}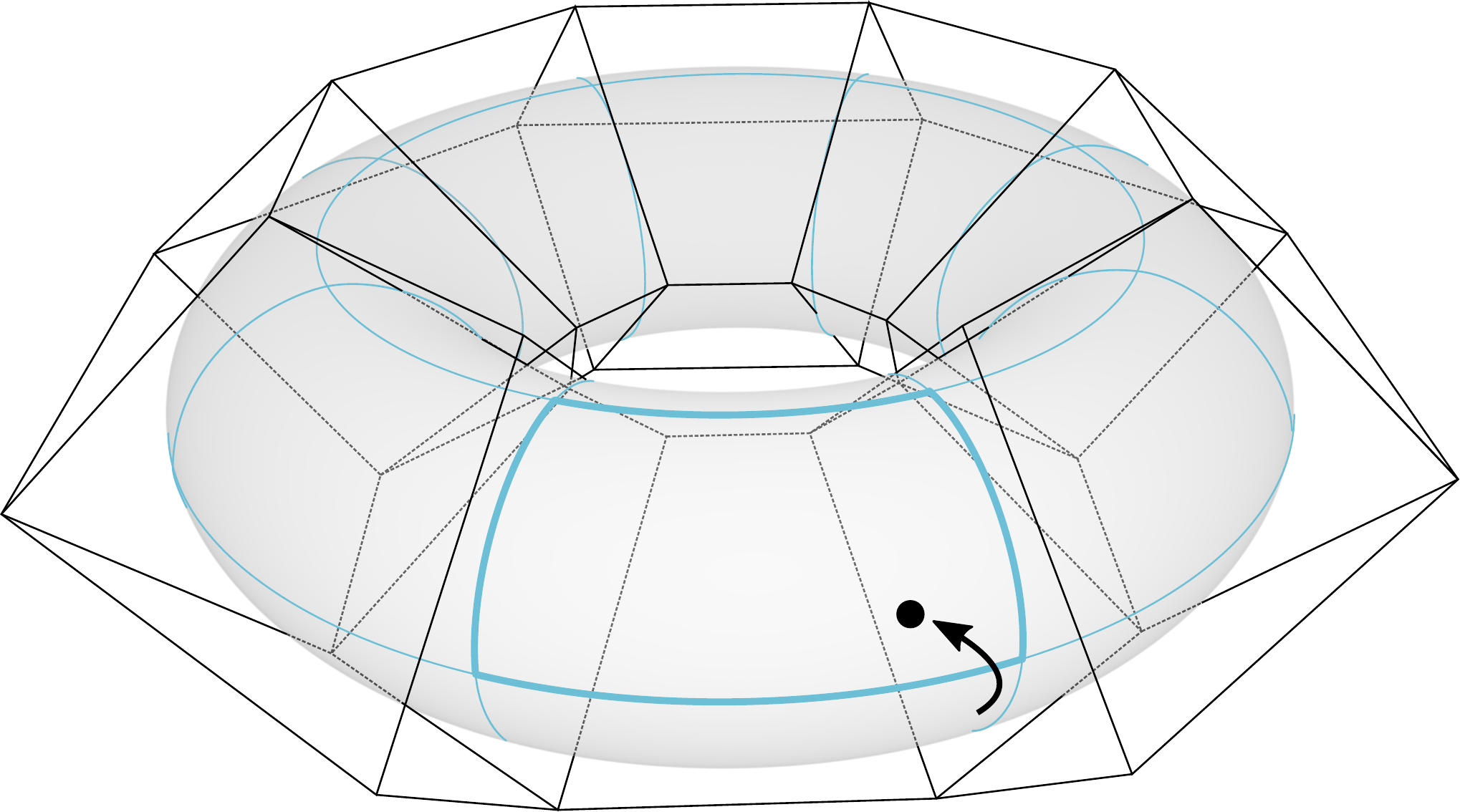\caption{NURBS surface (32 patches)}\label{fig:NURBS}
	\end{subfigure}
	\quad
	\begin{subfigure}[b]{0.5\linewidth}
		\def\svgwidth{\hsize}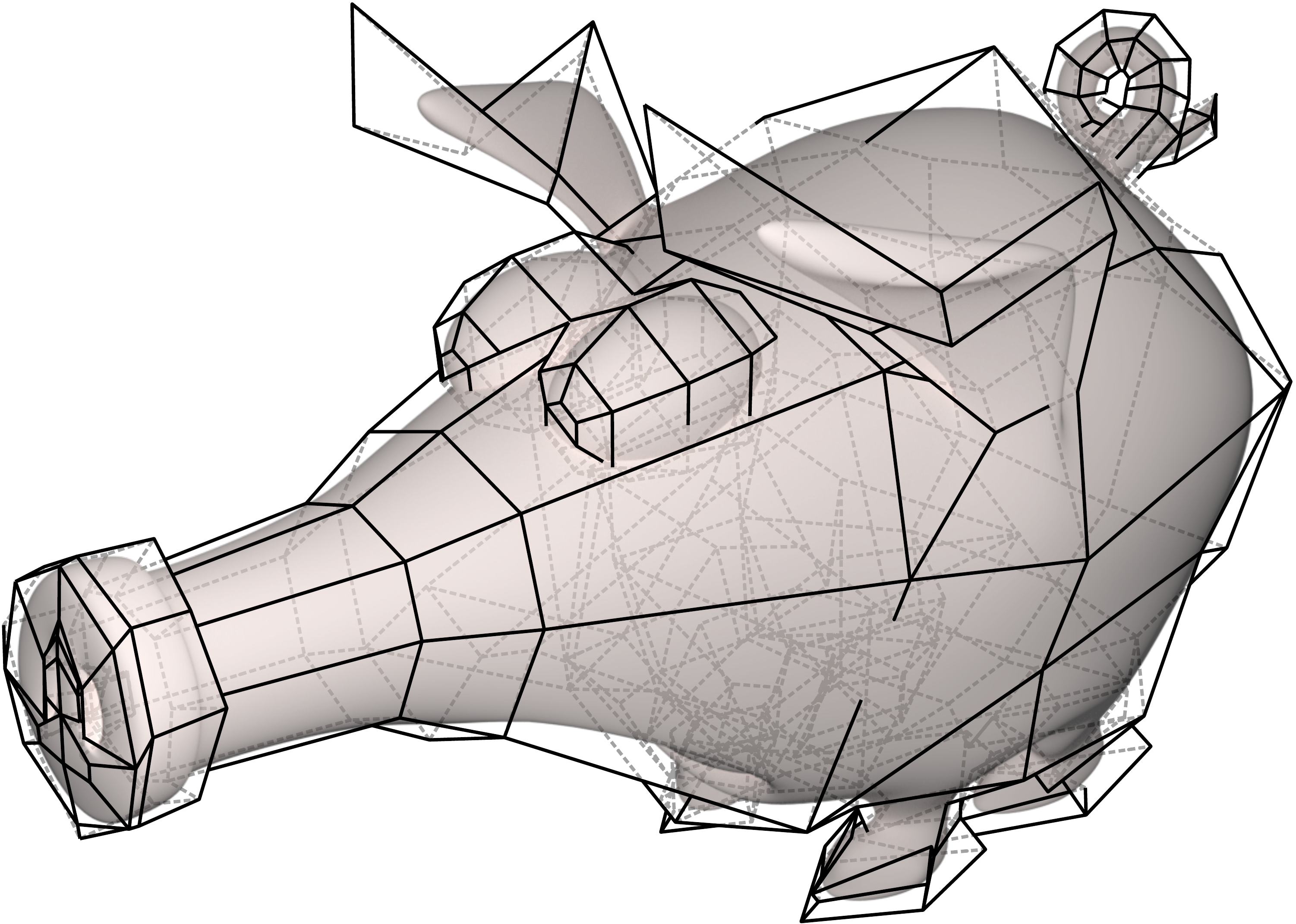\caption{Subdivision surface (381 faces)}\label{fig:subdiv}
	\end{subfigure}
	\caption{\textbf{Parametric surfaces} --- The map $f$ from the parameter plane $[0,1]^2$ or the control mesh surface $|M|$ to $\mathbb{R}^3$ defines the surface of a NURBS patch \textbf{(a)} or subdivision surface \textbf{(b)} respectively.}\label{fig:param_surfaces}
\end{figure}

\paragraph{Subdivision surfaces.} \label{sec:subd}
Modeling surfaces of general topology is quite difficult with patches. Subdivision surfaces are a generalization of splines that are popular for modeling surfaces of arbitrary topology \citep{Zorin:2000}.


A subdivision surface is defined by a polygonal mesh and a refinement scheme. The input polygonal mesh is called the control mesh. The corresponding smooth surface, called the \emph{limit surface}, is defined from the control mesh by recursively applying the refinement scheme an infinite number of times. 

The topology of the control mesh, denoted $M$, provides a piecewise parameterization of the limit surface. (Since the control mesh must be a simple polyhedron, it may be deformed or even lifted to $\mathbb{R}^4$ to remove all self-intersections.) In particular, let $\vec{u} \in M$ be a point on the control mesh. Then the subdivision surface may be viewed as a function $\vec{p}(\vec{u}) : M \rightarrow \mathbb{R}^3$, defined by the subdivision scheme and the positions of the control vertices. The point $\vec{u}$ is called the \emph{preimage} of a point $\vec{p}(\vec{u})$ on the surface.
Analytic representations of $\vec{p}(\vec{u})$ and its normals $\vec{n}(\vec{u})$ have been derived for popular schemes, such as Loop~\citep{Loop:1987,Stam:98} and Catmull-Clark~\citep{Catmull:1978,Stam:1998}.   
The open source library ``OpenSubdiv'' \citep{NieBner:2012,OpenSubdiv} supports direct evaluation of limit positions, normals and curvatures for both Loop and Catmull-Clark surfaces.



\section{Contour definition}

\begin{figure}
	\centering
	\small
	\def\svgwidth{0.85\textwidth}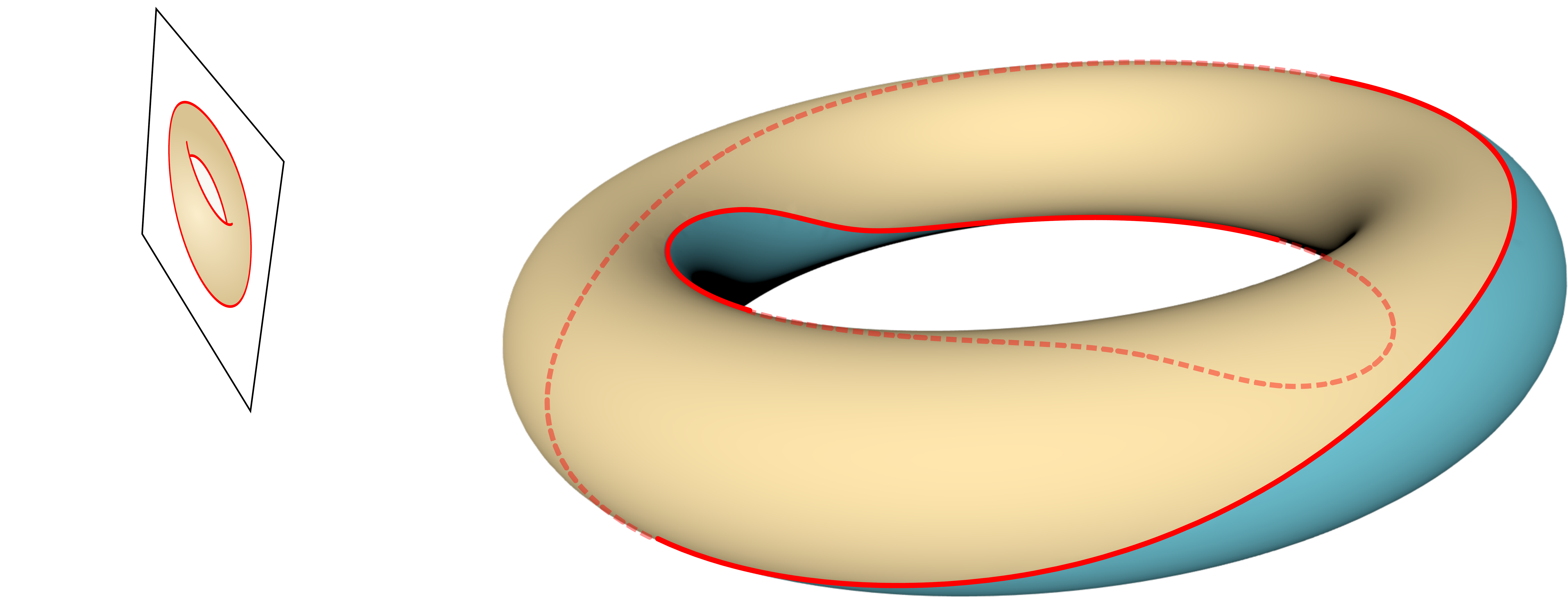\caption{\textbf{Smooth surface contour} --- The contour generator is the zero-set of the implicit orientation function $g(\vec{u})$ and thus the boundary between the front-facing and back-facing parts of a surface, as seen from a camera center $\vec{c}$. The apparent contour is the visible projection of the contour generator onto the image plane.}\label{fig:smooth_contour}
\end{figure}

As before, we assume that the surface is oriented, in generic position, and that only front-facing points may be visible. The surface is viewed from a camera center $\vec{c}$. We define the \emph{orientation function} (\fig{smooth_contour}):
$$g(\vec{u}) = (\vec{p}(\vec{u}) - \vec{c}) \cdot \vec{n}(\vec{u}).$$
A point with $g(\vec{u}) > 0$ is front-facing and a point with $g(\vec{u}) < 0$ is back-facing, the contour is the boundary between these regions: $g(\vec{u})=0$.  More formally, following Definition \ref{def:image-space-contours}. The contour is defined by:

\begin{definition}[parametric contour generator]\label{def:smooth_contours}
 The collection of all points $\vec{p}(\vec{u})$ for which the preimages $\vec{u}$ satisfy $g(\vec{u})=0$ is called the \emph{contour generator}~\citep{Marr:1977}. The visible projection of the contour generator onto the image plane is called the \emph{apparent contour}, or, simply, \emph{contour}.
\end{definition}

Interestingly, the smooth contour can also be interpreted in terms of shading, in two different ways. In the first way, we imagine photorealistic rendering of the surface with Lambertian shading ($\vec{n} \cdot \vec{v}$), with white texture, against a white background. The black pixels of this rendering ($\vec{n} \cdot \vec{v} \approx 0)$ are the contour (Figure \ref{fig:cow}).  Generalizing this idea of finding the darkest pixels (not necessarily black) of the Lambertian image motivates contour generalizations like the Suggestive Contours \citep{DeCarlo:2003,Lee:2007} and isophote stroke thickness \citep{Goodwin:2007}. And, an inverse interpretation is that the contours appear to have the same shape as rim lighting, in which an object is illuminated by a ring of lights perpendicular to the camera direction (\fig{rim_lights}(a)). 

\begin{figure}
\centering
\small
\def\svgwidth{\linewidth}\import{figures/smooth_contours/}{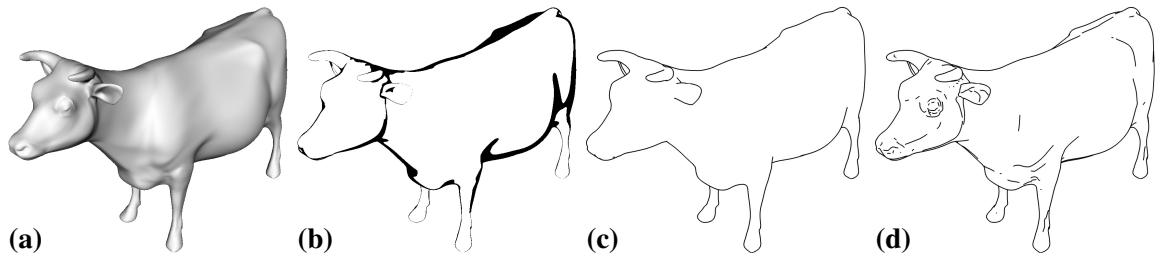}\caption{ \textbf{Shading interpretation of contours} --- \textbf{(a)} Lambertian shaded white object with light at viewpoint, so that shading is $\vec{n} \cdot \vec{v}$.
\textbf{(b)} Thresholded rendering, for visualization
\textbf{(c)} The contours are the black points in the shading image, where $\vec{n} \cdot\vec{v}=0$.
\textbf{(d)} Identifying dark ridges in the shading image produces the contours and Suggestive Contours \citep{DeCarlo:2003,Lee:2007}. Images generated with ``qrtsc'' \citep{qrtsc}.
}\label{fig:cow}
\end{figure}

\begin{figure}
	\centering 
	\small
	\textbf{(a)}
	\includegraphics[height=1.5in]{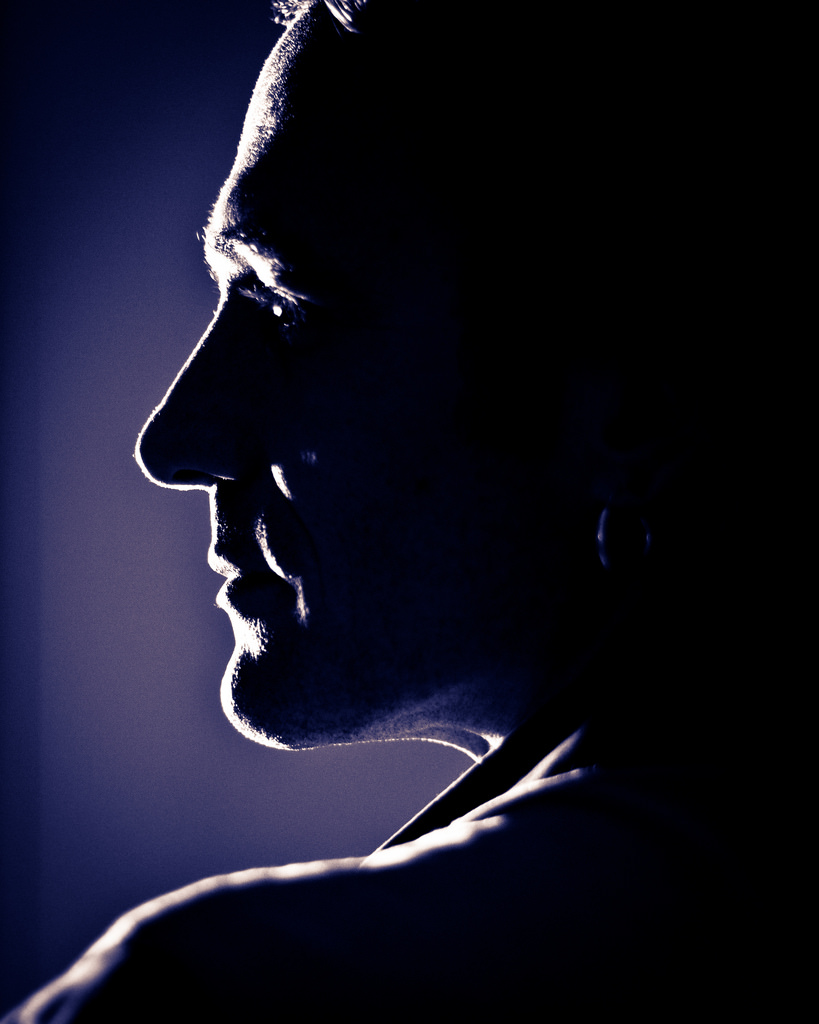}
	\qquad	
	\textbf{(b)}
	\includegraphics[height=1.5in]{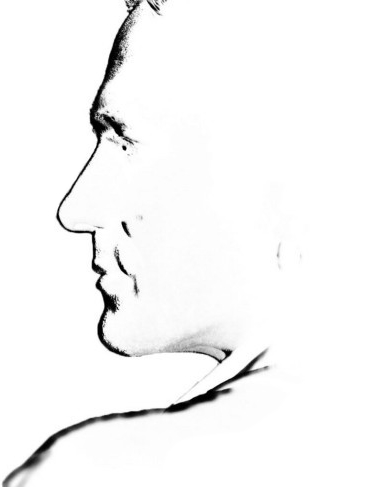}
	\qquad
	\textbf{(c)}\hspace{-15pt}
	\includegraphics[height=1.6in]{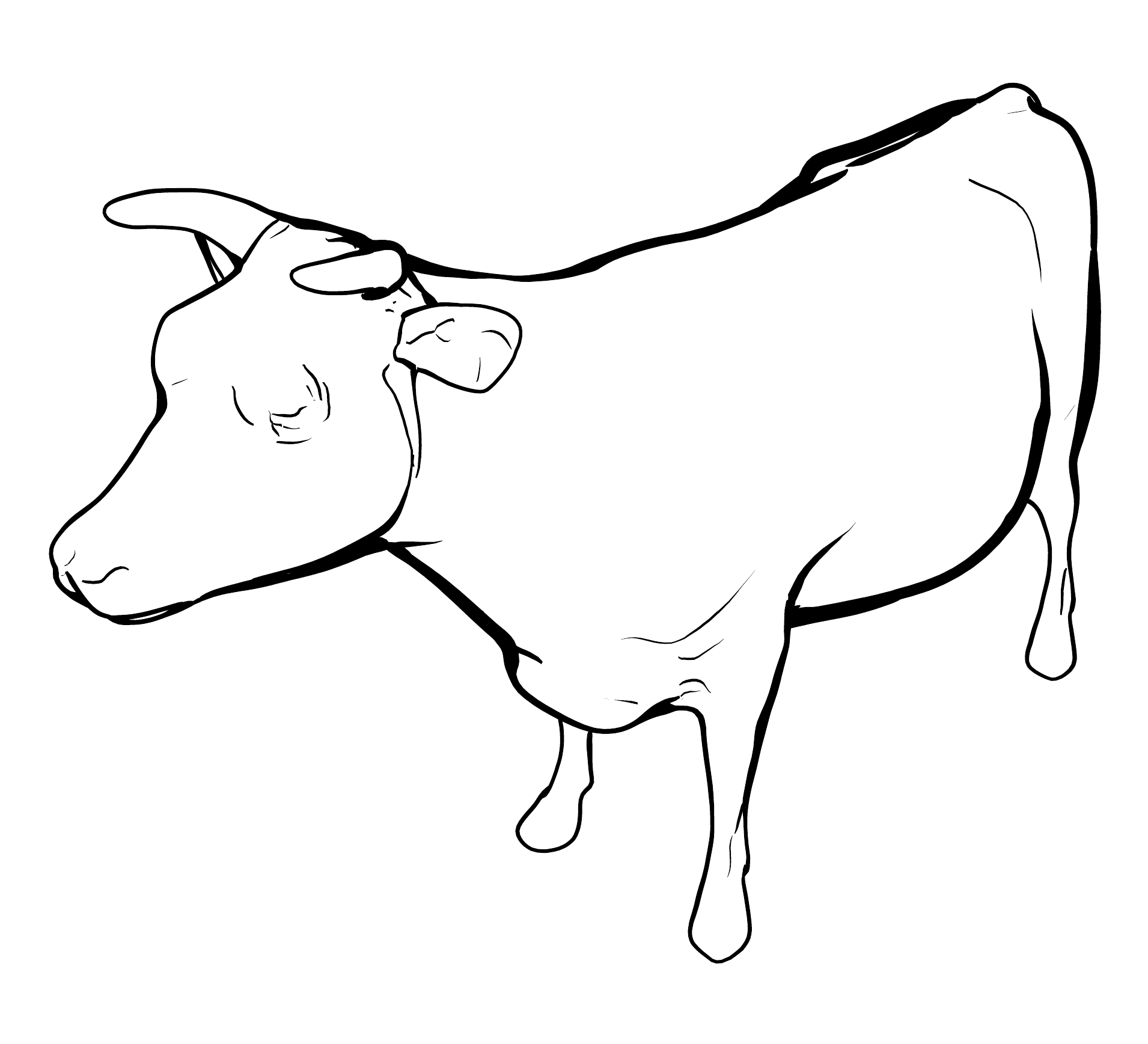}
	\caption{\textbf{Rim lighting and rendering} --- \textbf{(a)} Photograph taken  with rim lighting, \ie~a ring of lights perpendicular to the camera direction. Rim  lighting approximates the occluding contour. (Photo by Flickr user japrea \ccbysa)  
\textbf{(b)} Rim light photograph inverted and converted to grayscale, with the background removed.
	\textbf{(c)}~Computer-generated line drawing using contours, Suggestive Contours \citep{DeCarlo:2003} and isophote thickness \citep{Goodwin:2007}. The stroke thickness varies in the same way as it would for contours produced by rim lighting. 
	\label{fig:rim_lights} 
}
\end{figure}

\section{Contour extraction}

Because the contour generator is an implicit polynomial function, we cannot directly compute it. Instead, we must numerically approximate it; existing methods approximate it by piecewise linear curves.  While an early method proposed marching along the contour in parameter space \citep{Hornung:1985}, more modern methods identify patches with sign changes of $g(\vec{u})$ \citep{Elber:1990,Gooch:1998,Benard:2014}, similar to the treatment of Interpolated Contours in \chap{chap:smooth_as_meshes}.

\begin{figure}
	\centering
	\small
	\textbf{(a)} \def\svgwidth{0.4\textwidth}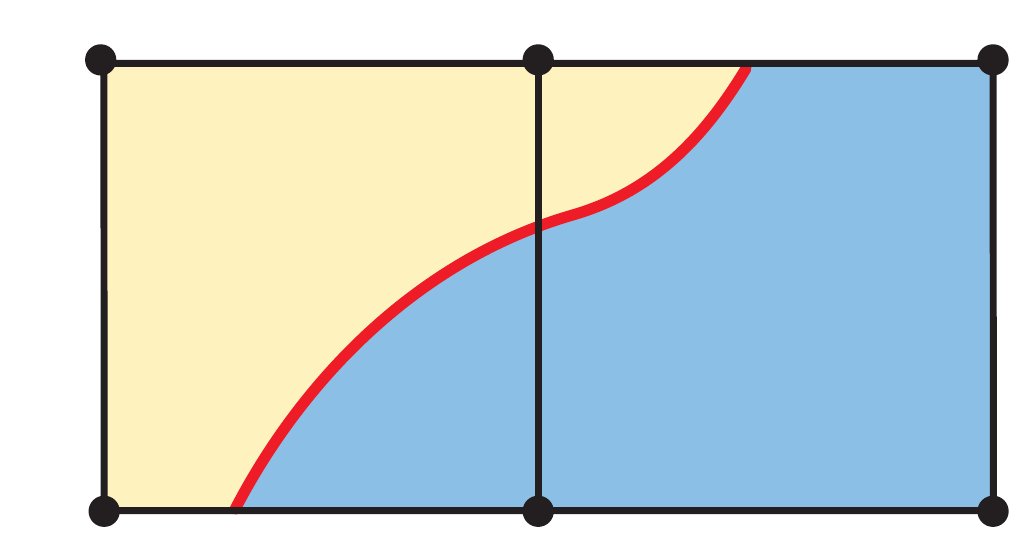
	\qquad
	\textbf{(b)} \def\svgwidth{0.4\textwidth}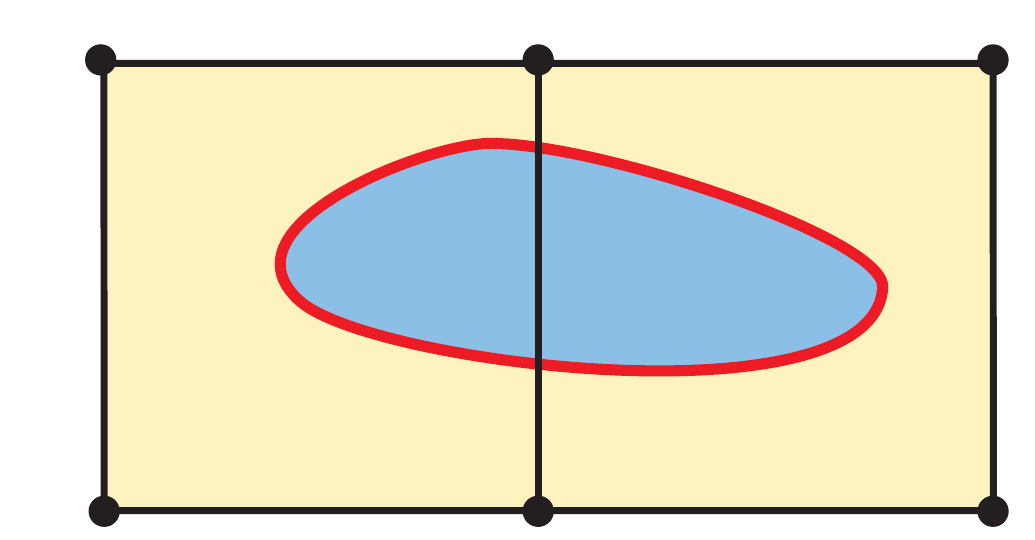
	\caption{\textbf{Orientation function evaluated at the control vertices} --- When $g(\vec{u})$ has opposite signs ($\sF$ and $\sB$) on both ends of a control polygon edge \textbf{(a)}, a contour point must exist on that edge. When there is no sign change ($\sF$ and $\sF$ or $\sB$ or $\sB$), there may be zero contour points, or a larger even number of contour points per edge \textbf{(b)}.
	}
	\label{fig:orientation_basemesh}
\end{figure}

Specifically, we first evaluate the orientation function $g(\vec{u})$ at all control vertices (\fig{orientation_basemesh}). For a spline patch, these control vertices can be visualized on a regular grid; for a subdivision surface, they live on the control polygon.  We denote points with $g(\vec{u})>0$ as $\sF$ and $g(\vec{u})<0$ as $\sB$. 

For any edge on the control polygon with opposite signs on the edge ($\sF$ and $\sB$), there must be a contour point somewhere on that edge (\fig{orientation_basemesh}a). For edges with the same sign ($\sF$ and $\sF$ or $\sB$ or $\sB$), there \textit{might} be contour points, in some even number, such as a small loop centered on this edge (\fig{orientation_basemesh}b).  In the following algorithms, we generally assume that no small loops like this occur, and assume that no sign change indicates that there is no contour on the edge. (For the special case of rational splines under orthographic projection, \citet{Elber:1990} showed  a sign test that may be used to quickly identify that some patches cannot contain contours.)

For each edge that must contain a contour, the edge can be parameterized as a 1D function $\vec{u}(t) = (1-t)\vec{u}_0 + t(\vec{u}_1)$ where $\vec{u}_0$ and $\vec{u}_1$ are the preimages of the control points. The preimage and position of the contour may be found using a root-finding algorithm on $g(\vec{u}(t))=0$ \citep{Elber:1990,Benard:2014}, such as the secant method or bisection search.  A simpler approach is to linearly interpolate to approximate the contour location \citep{Gooch:1998}, similar to \eqn{eq:interpolation}.

The identified contour locations may be connected to produce a piecewise linear approximation to the contour (Figure \ref{fig:comp_contours}). A more precise curve may be found by subdividing control faces and repeating the root-finding process.

\begin{figure}
	\centering
	\small
	\begin{subfigure}[t]{0.48\linewidth}
		\includegraphics[width=\linewidth]{figures/smooth_contours/torus_interpC_side.pdf}
		\caption{interpolated contours}
	\end{subfigure}
	\quad
	\begin{subfigure}[t]{0.48\linewidth}
		\includegraphics[width=\linewidth]{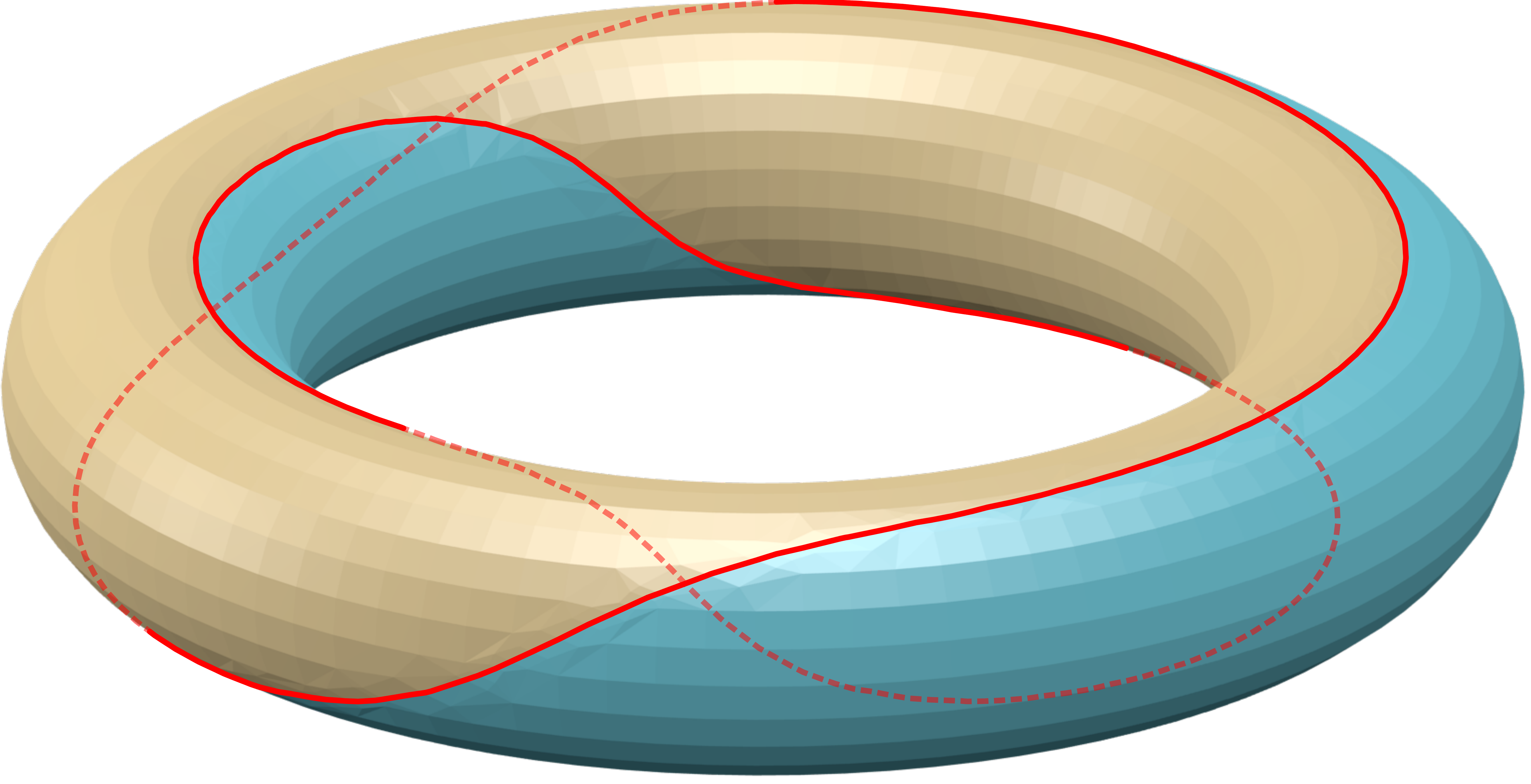}
		\caption{mesh contours after consistent tessellation}\label{fig:consistent_contours}
	\end{subfigure}
	\caption{\textbf{Contour generator approximation} ---  An input smooth torus represented as a Catmull-Clark subdivision surface is uniformly tessellated with one round of subdivision. With contour-consistent tessellation \textbf{(b)}, the mesh contours of the polygonal mesh is both topologically equivalent to the smooth surface contour and at the boundary of visible and invisible. Interpolated contours \textbf{(a)} do not have this property, leading to problems with visibility.
	}\label{fig:comp_contours}
\end{figure}

\section{Contour curvature} \label{sec:radial_curvature}

We now discuss the curvature of 2D contours and 3D contour generators.  
This analysis is necessary to identify curtain folds, and also gives insight into the relationship between surface curvature and apparent contour curvature.

In order to analyze image contours, it is useful to consider the following tangent direction at contour point $\vec{p}$. 
The direction $\vec{w}$ is defined as the (unnormalized) projection of the view vector $\vec{v} = \vec{p}-\vec{c}$ onto the tangent plane at $\vec{p}$. For contour points, $\vec{w}=\vec{v}$ since $\vec{v}$ is already in the tangent plane. The normal curvature along $\vec{w}$ is called the \emph{radial curvature} $\kappa_r(\vec{p})$ (\fig{radial_curvature}a)~\citep{DeCarlo:2003,Koenderink:1984}.  (See
Appendix \ref{app:curvature} for the definition of normal curvature.)

\begin{figure}
	\centering
	\small
	\def\svgwidth{\textwidth}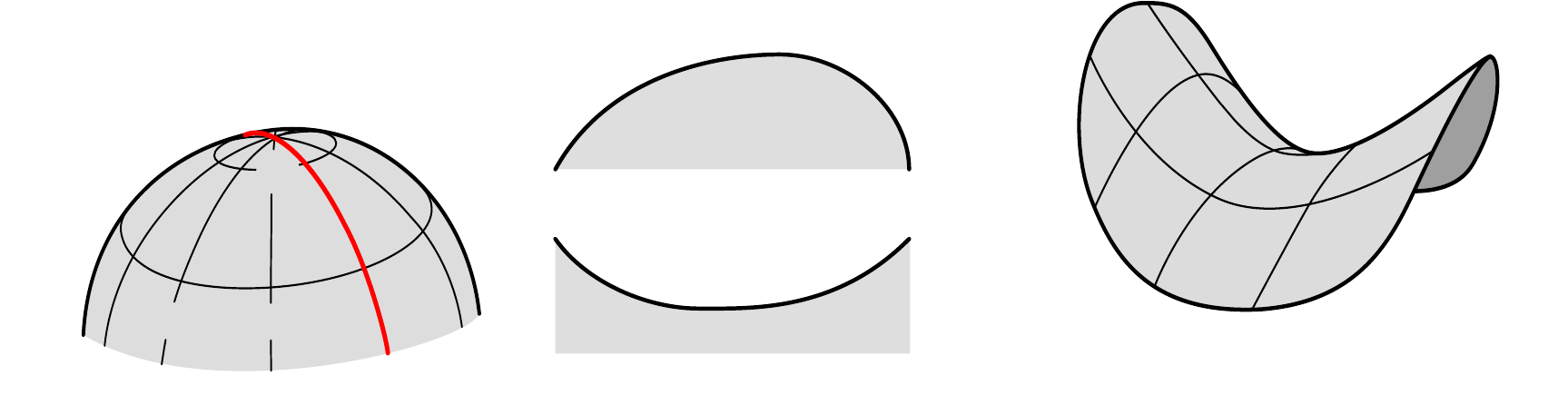\caption{\textbf{Radial curvature} --- \textbf{(a)} The radial curvature $\kappa_r(\vec{p})$ is the curvature of the radial curve at $\vec{p}$; \textbf{(b)} $\kappa_r$ is necessarily positive for visible contours (top) otherwise it would be locally occluded by the surface (bottom), and zero at curtain folds \textbf{(c)}.}\label{fig:radial_curvature}
\end{figure}

Another way to state the definition is as follows. The radial curvature is based on the \emph{radial plane}, the plane that contains the point $\vec{p}$, the surface normal $\vec{n}$, and the view vector $\vec{v}$. The \emph{radial curve} is the intersection of the surface with the radial plane. The radial curvature $\kappa_r(\vec{p})$ is then defined as the curvature of the radial curve at $\vec{p}$.

A contour can only be visible when it has a positive radial curvature ($\kappa_r > 0$) --- otherwise the contour generator would locally lie in a surface concavity (\fig{radial_curvature}(b)). This exactly parallels the concepts of \textit{concave} and \textit{convex} contours on meshes: a contour with positive radial curvature is a convex contour (in the radial direction).

The \emph{apparent curvature} $\kappa_p(\vec{p})$ of the contour curve is  the curvature of the apparent contour at $\vec{p}$. Under perspective projection, \citet{Koenderink:1984} demonstrated that the Gaussian curvature $K$ of the surface at $\vec{p}$ is related to the radial and apparent curvatures by:
$$K = \frac{\kappa_r(\vec{p}) \kappa_p(\vec{p})}{\norm{\vec{p}-\vec{c}}}.$$

The above observations allow us to relate the image-space curvature of the 2D contour with the corresponding 3D region.
For visible contours, the sign of the apparent curvature is thus the same as the sign of the Gaussian curvature. If the surface is elliptical ($K>0$), the fact that visible contours cannot have $\kappa_r < 0$, implies that $\kappa_p$ is necessarily positive, and thus the apparent contour displays a convexity. Conversely, if the surface is hyperbolic ($K<0$), $\kappa_p < 0$ and thus the apparent contour is concave. This leads to the following general rule illustrated in \fig{pig_curvature}: a convex apparent contour corresponds to a convex surface, a concave contour implies a saddle-shaped surface, and an inflection on the contour ($\kappa_p=0$) implies a parabolic point on the surface ($K=0$).  

\begin{figure}
	\centering
	\small
	\def\svgwidth{0.65\textwidth}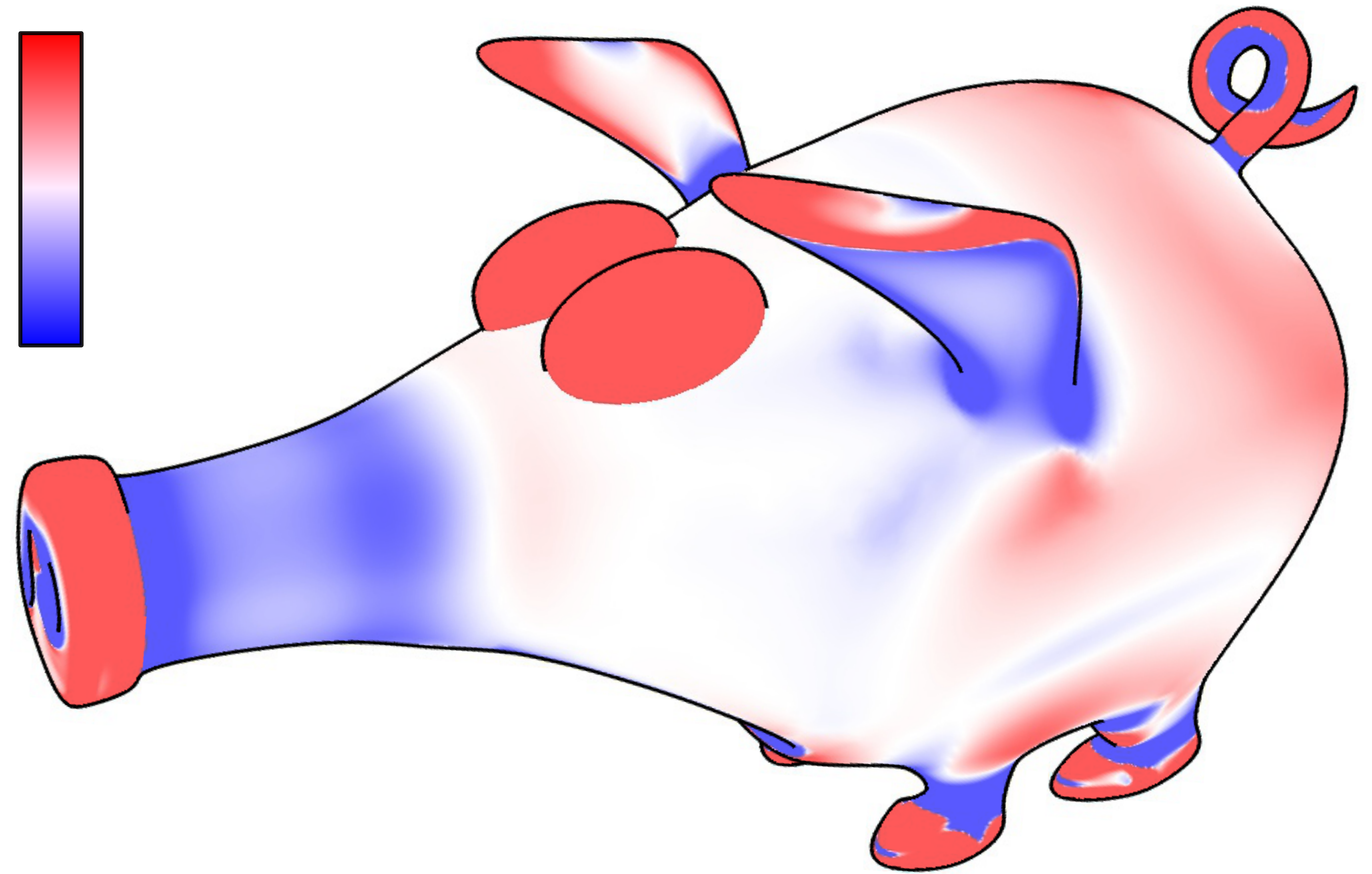\caption{\textbf{Relationship between the surface Gaussian curvature and the contour apparent curvature} --- Concave apparent contours originate from hyperbolic regions (in blue) and convex ones from elliptic regions (in red); their inflection coincides with parabolic points on the surface. Image generated with ``qrtsc'' \citep{qrtsc}.}\label{fig:pig_curvature}
\end{figure}

\section{Singular points}\label{sec:smooth_singular}

Curves on smooth surfaces exhibit similar \textit{singular points} --- that is, points where visibility might change --- as on meshes. Image-space intersections, intersections on the surface, and curtain folds are all possible singularities. However, smooth surface contours may not exhibit bifurcations in generic position \citep{Hertzmann:2000}.

\paragraph{Intersections.}
Image-space intersections and intersections on the surface create singularities for curves on smooth surfaces. 

Finding intersections on the surface between boundaries and other curves typically involve root-finding along each boundary edge, e.g., for boundary-contour intersection, find the boundary edge point with $g(\vec{p}(t))=0$. 

Image-space intersections involving contours are more difficult to find, because contours are implicitly defined.  These intersections must be detected numerically, and there is no simple data-structure for accurately accelerating the search without the possibility of missing some intersections.  
There are two general strategies one can take. First, one can convert the contours into polylines, and compute the intersections of these polylines. This is simple but may often be incorrect in some cases. Second, one may use an adaptive subdivision approach, in which bounding boxes for each curve are subdivided until either an intersection is found or the absence of an intersection can be proven \citep{Elber:1990}.

Computing intersections between smooth surfaces is also difficult (e.g., \citep{Houghton:1985}), and these intersection curves must then be intersected with other curves on the surface.

\paragraph{Contour curtain fold definition.}
\citet{Koenderink:1984} demonstrated that the radial curvature vanishes at a \emph{curtain fold cusp} ($\kappa_r = 0$), the contour transitioning from invisible to potentially visible. At a curtain fold, the 3D tangent of the contour generator exactly coincides with the view vector. As a result, the projection is not smooth ($\kappa_p$ is infinite); hence, curtain folds correspond to cusps in the apparent contour (\fig{radial_curvature}c \& \fig{smooth_cusp}). It can be shown that these points are the only generic cusps of smooth surface contours. For this reason, many previous authors use the term cusp instead of curtain fold; we use the latter terminology to emphasize the correspondence with curtain folds on mesh contours.

\begin{figure}
	\centering
	\small
	\def\svgwidth{0.85\textwidth}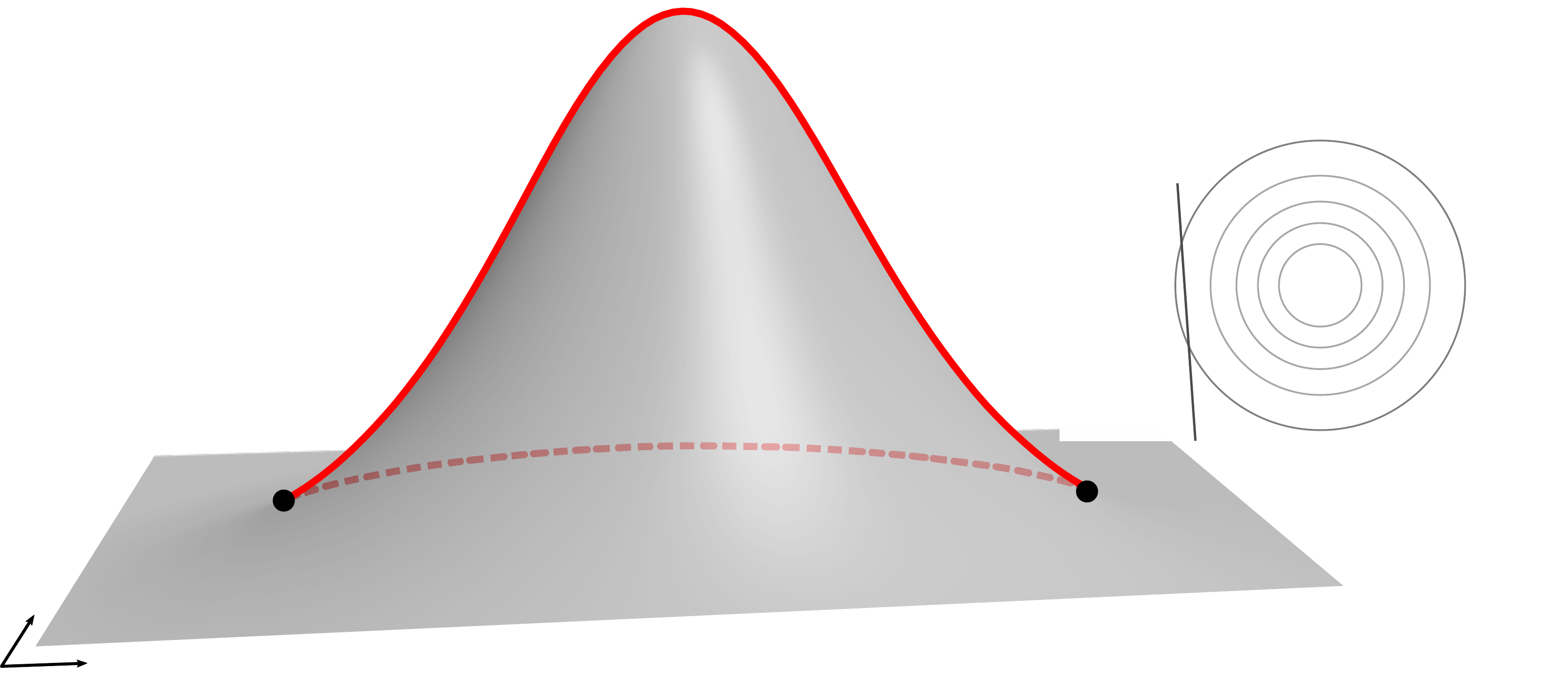\caption{\textbf{Contour generator curtain folds on a smooth surface} --- As shown on the top view, curtain folds are at the intersection of the contour generator and radial curvature zero-isocurve.  At each curtain fold, the curve tangent is aligned with the view direction.}\label{fig:smooth_cusp}
\end{figure}

As with meshes, curtain folds occur at the transition from concave ($\kappa_r < 0$) to convex ($\kappa_r > 0$) contours: when the contour transitions from locally occluded (concavity) to locally visible (convexity).
In the vicinity of curtain folds, the surface is necessarily hyperbolic; the visible branch of the apparent contours must thus be concave in image space \citep{Koenderink:1982}.

For parametric surfaces, radial curvature can be computed directly from the definition.

\paragraph{Contour curtain fold detection.} 
Detecting curtain folds on smooth surface contours entails finding surface points where both $g(\vec{u})=0$ and $\kappa_r(\vec{u}) = 0$. The simplest approach is to perform linear interpolation within a face, as in Section \ref{sec:singularities}. However, this may not be sufficiently accurate.


A more precise procedure is as follows \citep{Benard:2014}. The algorithm first searches for triangles where both the functions $g(\vec{u})$ and $\kappa_r(\vec{u})$ exhibit sign changes among the triangle vertices. When such a triangle is found, it is bisected along the edge that is longest in parameter ($\vec{u}$) space (Figure \ref{fig:cusp_detection}). This sign-change test and splitting process is repeated in the two new triangles. When the recursion detects a very small triangle with sign crossings in both functions, the centroid of that triangle is a  curtain fold preimage location.  (Note that the triangle splitting is not applied to the surface; the new triangles are stored only during this recursion and new vertex locations are computed by exact evaluation of $\vec{p}(\vec{u})$).

\begin{figure}
	\centering
	\small
	\def\svgwidth{\linewidth}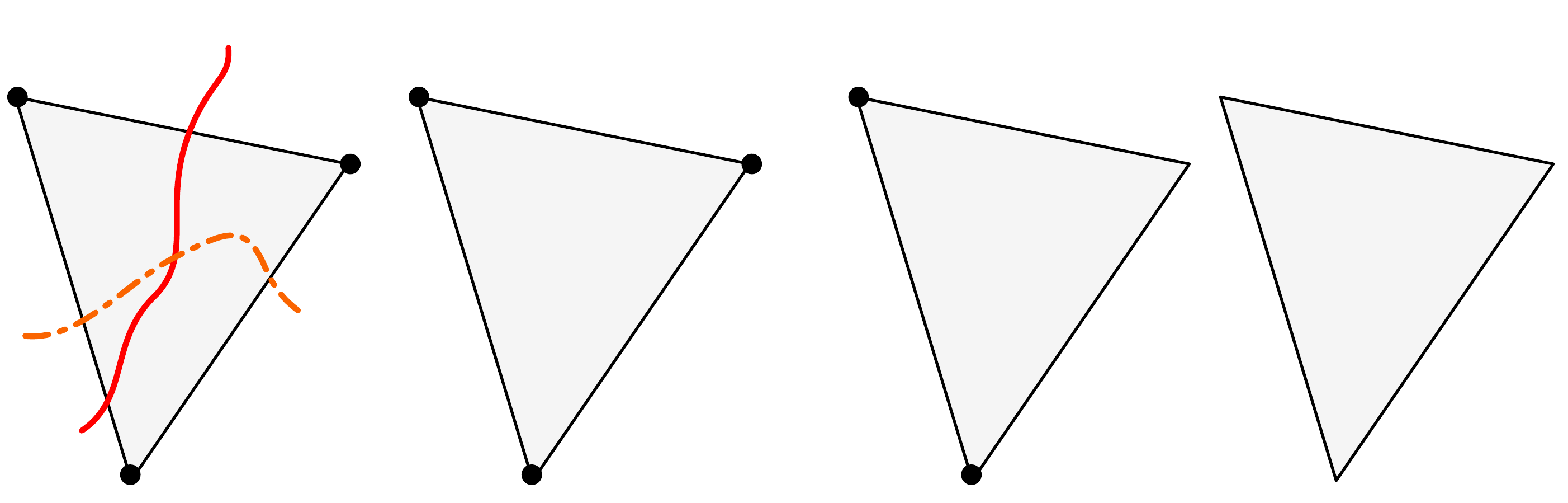\caption{
		\textbf{Contour curtain fold detection on smooth surfaces} ---
		\textbf{(a)} Given a base mesh triangle, we wish to find all points that satisfy both 
		$g(\vec{u})=0$ and $\kappa_r(\vec{u}) = 0$.
		\textbf{(b)} We find triangles that contain sign changes of both functions between the vertices.
		\textbf{(c)} For each such triangle, we split the triangle into two, by bisecting the long edge, and recursing into each of these triangles.
		\textbf{(d)} When this process leads to a very small triangle with sign changes, a contour curtain fold is marked at the centroid of this triangle.
	}\label{fig:cusp_detection}
\end{figure}

\paragraph{Curtain folds on other curves.}
Other types of curves may also have curtain folds.  In general, for all curves, a curtain fold occurs when the 3D tangent to the curve is aligned to the view vector.  This implies that the surface is normal to the view vector. Hence, at a curtain fold, the curve also intersects a contour generator.  Hence, curtain folds do not need to be specially handled for non-contour curves, because they will be detected as curve-contour intersections.

\section{Visibility computation}

Determining the visibility of the smooth surface contour is a significant, and very challenging problem. For mesh contours, the algorithms can rely on exact computations (\sect{sec:vis_algo}), up to numerical precision.  For example, a simple ray test can be used to determine the visibility of any point.  On the other hand, performing an exact ray test for a parametric surface's contour involves a  computation that is numerically unstable, since the true contour lies exactly on the boundary between visible and invisible. Hence, the mesh processing algorithms cannot be directly applied.

Previous authors have applied four different strategies to visibility for smooth surface contours.  The first, described in \chap{chap:smooth_as_meshes}, is to convert the surface to a mesh, and use heuristics (such as Interpolated Contours) to clean up the contours; this approach is simple, but can exhibit artifacts.  The other three strategies are described next.

We also note that, in the past few decades, new methods for ray-tracing subdivision surfaces have been developed, e.g., \citep{Kobbelt:1998,Tejima:2015,Benthin:2015}, and it may be time to revisit their usefulness for this problem.



\subsection{Ray-casting the smooth surface} \label{sec:smooth_ray_test}

The first approach is to directly apply ray-casting and singularity detection on the smooth surface, generalizing the procedure for meshes.  Because ray tests are unstable on the contour, one may perform them elsewhere on the surface, and then propagate visibility based on image-space relationships between curves. \citet{Elber:1990} perform those tests along a subset of isoparametric curves. While this method is demonstrated for simple surfaces (\fig{param_contours}), there are a few theoretical issues that suggest it may not work robustly for general surfaces. First, it assumes that all visible contours touch one of the isoparametric curves, which may not always be the case for complex models. Second, propagating visibility information depends on robustly computing image-space intersections between smooth curves. Once again, numerical instabilities are likely to arise, since curves may often be nearly tangent to each other in projection. 
As will all visibility propagation schemes, a single visibility computation error can propagate to create many erroneous visibility errors, such as a large silhouette curve that disappears. 

\begin{figure}
	\centering
	\small
	\includegraphics[width=0.43\linewidth]{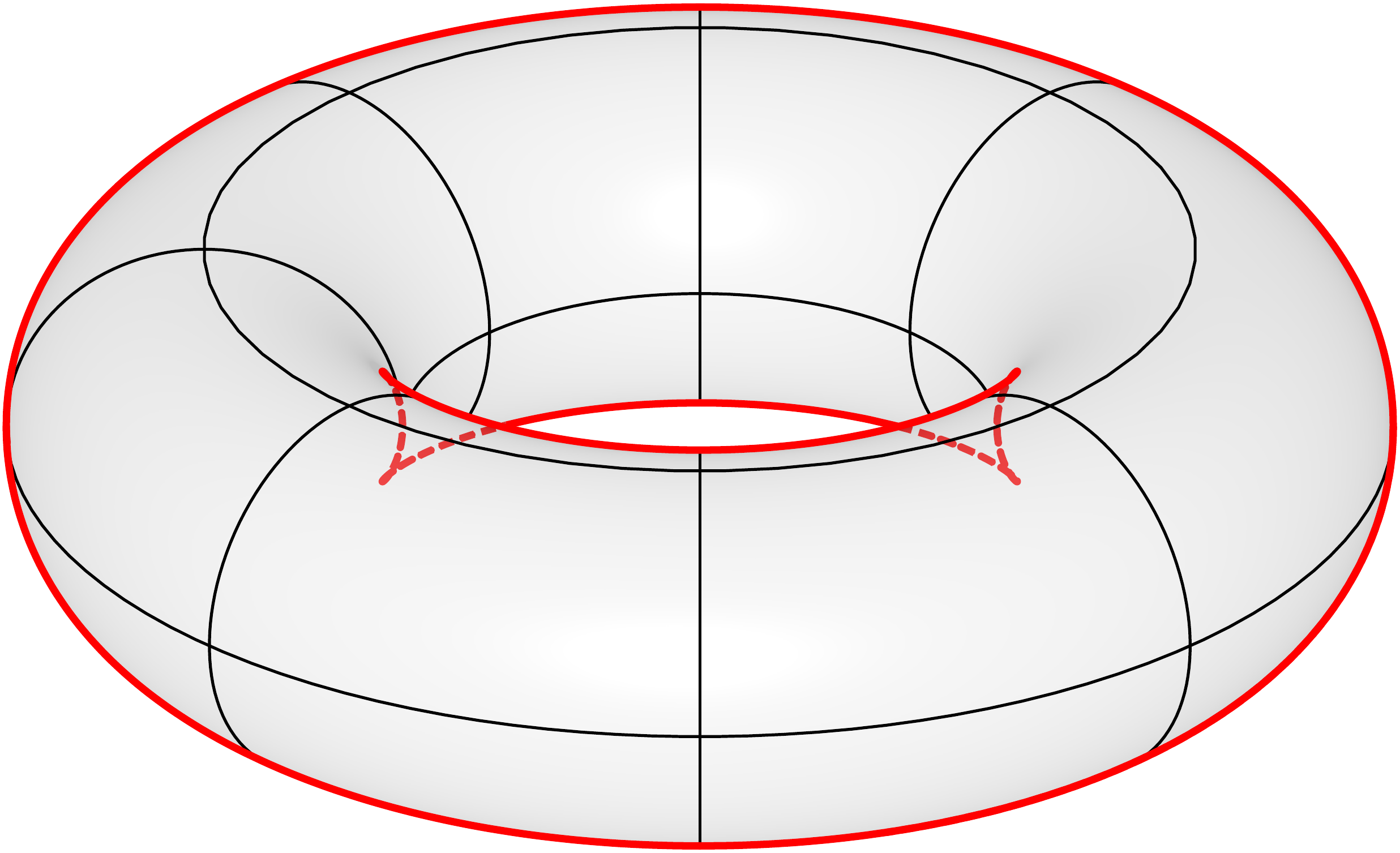}
	\quad
	\includegraphics[width=0.22\linewidth]{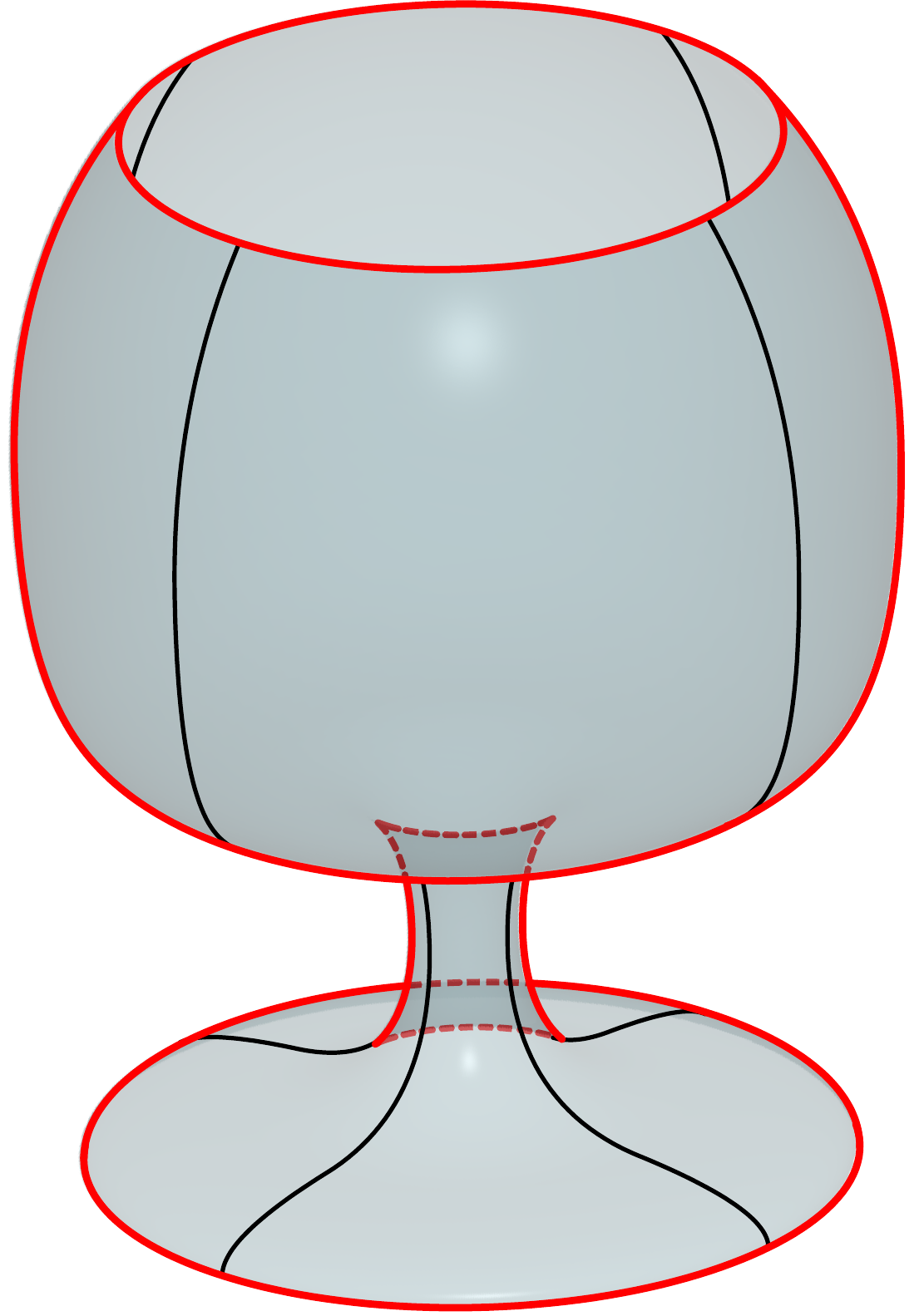}
	\quad
	\includegraphics[width=0.21\linewidth]{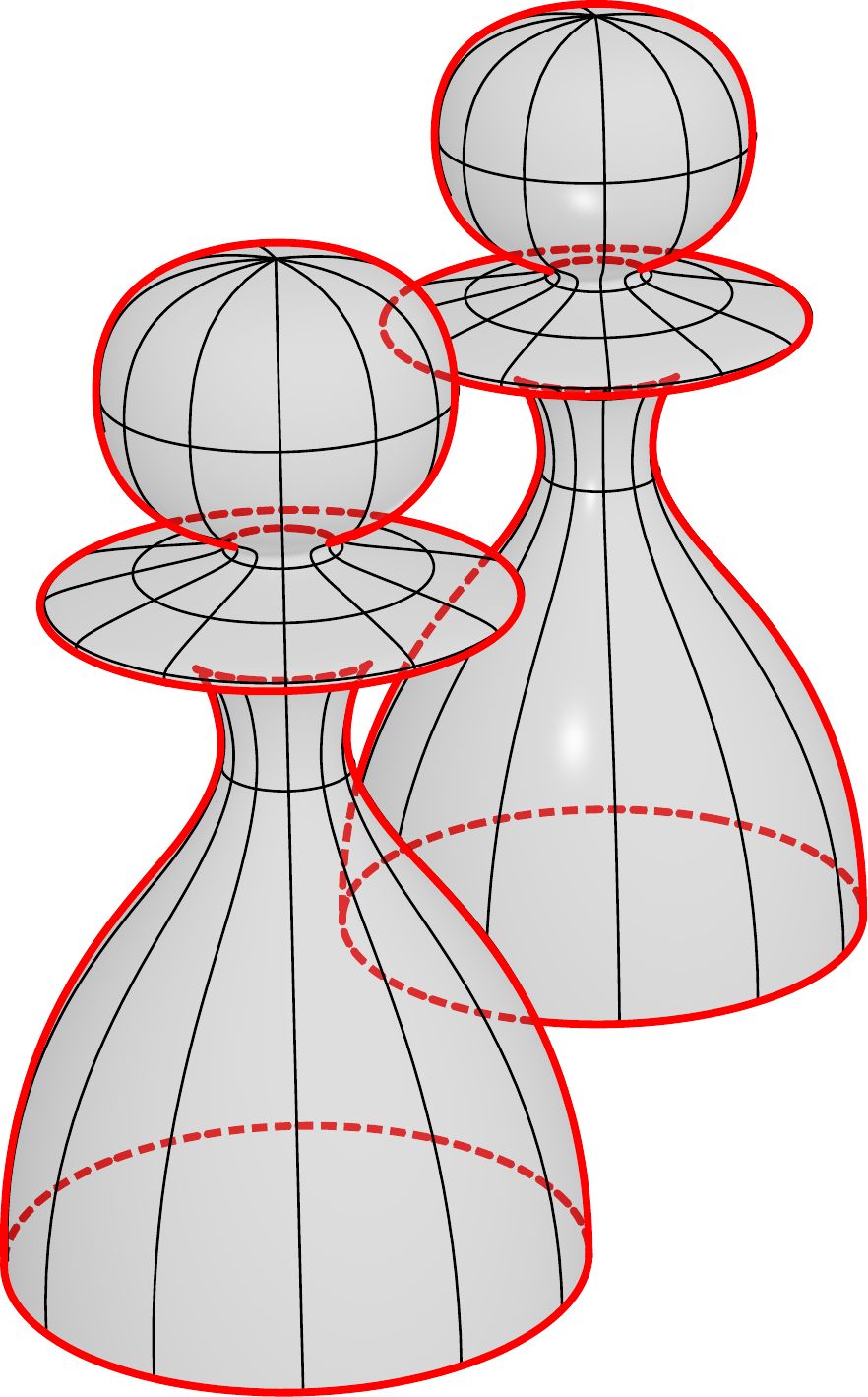}
	\caption{\textbf{Results of \citet{Elber:1990}'s patch-based contour extraction algorithm} --- Images generated with the IRIT modeling environment~\citep{IRIT}. }\label{fig:param_contours}
\end{figure}

One can imagine more robust versions of this procedure.  To our knowledge, this strategy has not been explored since \citet{Elber:1990}.

\subsection{Planar Maps}

\citet{Winkenbach:1996} used a Planar Map for visibility computations and stylization.  A Planar Map can potentially ensure a consistent topology, even in the presence of errors, and allows more control over stylization of the regions.

Their Planar Map was computed from a mesh tessellation of the surface. However, they also numerically refine the contours separately from the Planar Map, and their contours might not exactly match the Planar Map visibility, potentially causing small visibility errors.
Following the work of \citet{Gangnet:1989}, they restrict all the edge endpoints to have integer coordinates and use infinite-precision rational arithmetic to compute exact intersections which may mitigate some of the mismatch.

To our knowledge, Planar Maps for true smooth surface representations have not further been explored since \citet{Winkenbach:1996}.

\subsection{Contour-consistent tessellation} \label{sec:polygon_ray_test}

Finally, we summarize an approach that we developed, in collaboration with Michael Kass \citep{Benard:2014}.
Our approach is, for a given viewpoint, to tessellate the smooth surface into a triangle mesh for which the contour generators have the same topology as they do for the original smooth surface. The triangle mesh is also geometrically close to the smooth surface. Because we have effective and exact algorithms for contour rendering on meshes (Chapters \ref{chap:mesh_contours} and \ref{chap:visibility}), this guarantees a valid contour rendering, approximating that of the smooth surface.

The algorithm creates a new mesh initialized by copying the smooth surface's control mesh (\fig{consistent_tess}).
Throughout, the algorithm maintains a pointwise correspondence between the smooth surface and its polygonal approximation.  We provide  conditions that can be used to prove that the new mesh has achieved topologically-equivalent contours. The main goal is that front-faces on the mesh should correspond to front-facing regions on the surface, and back-faces should correspond to back-facing regions. The algorithm performs a sequence of local transformations on the mesh until these conditions are met.

\begin{figure}
	\centering
	\small
	\def\svgwidth{\hsize}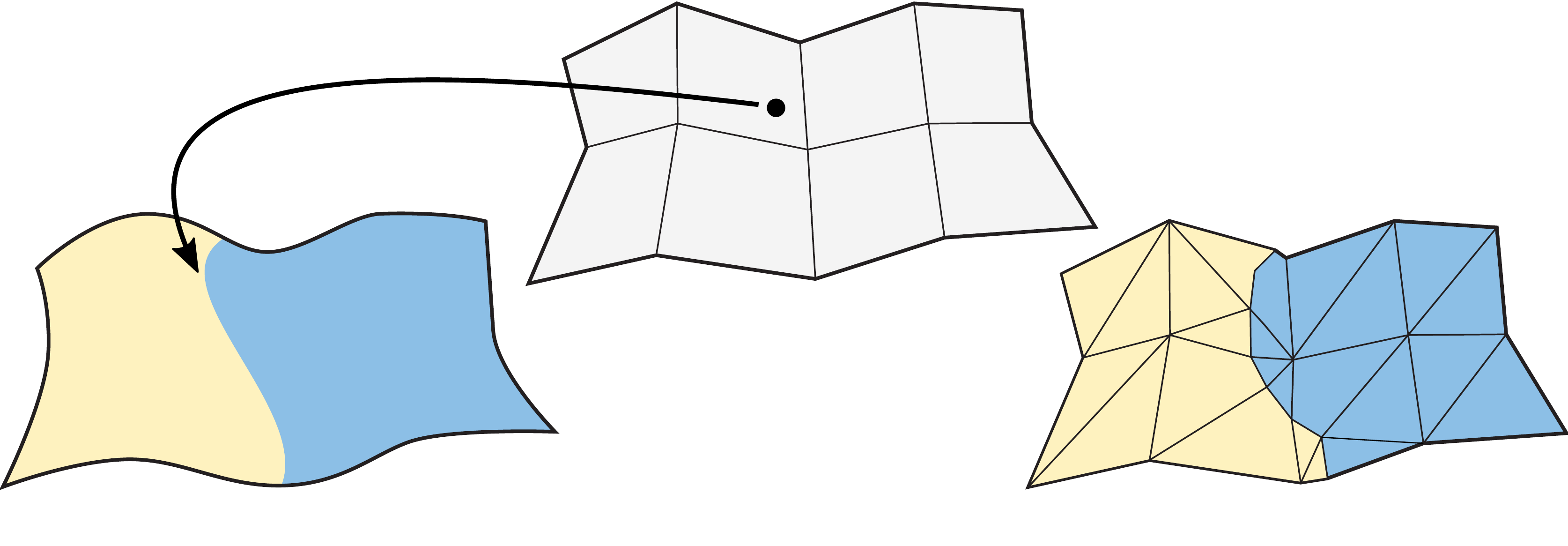\caption{\textbf{Contour-Consistent tessellation} --- Each vertex on the output triangle mesh produced by \citet{Benard:2014} corresponds to a preimage point $\vec{u}$ on the control mesh that maps to a point $f(\vec{u})$ on the smooth surface. The orientation (front- or back-facing) of each face of the triangle mesh is consistent with the smooth surface orientation $g(\vec{u})$.}\label{fig:consistent_tess}
\end{figure}

The method still does not have all the formal guarantees that one would like; in principle, there are a few ways we discuss where it could go wrong. These do not seem to be a problem in practice. However, it is complex to implement and the computation is slow.  Improving this is an area for future work.

This method produces a triangle mesh whose contour edges are topologically equivalent and geometrically close to the contour generators of the smooth surface (\fig{consistent_contours} \& \ref{fig:contour_consistent_red}).
Hence, it is the only current method that computes accurate contours and visibility for smooth surfaces, avoiding flickering artifacts when animated (\fig{contour_consistent_stylized}).


\begin{figure}
	\centering
	\begin{subfigure}[b]{0.36\linewidth}
		\includegraphics[width=\textwidth]{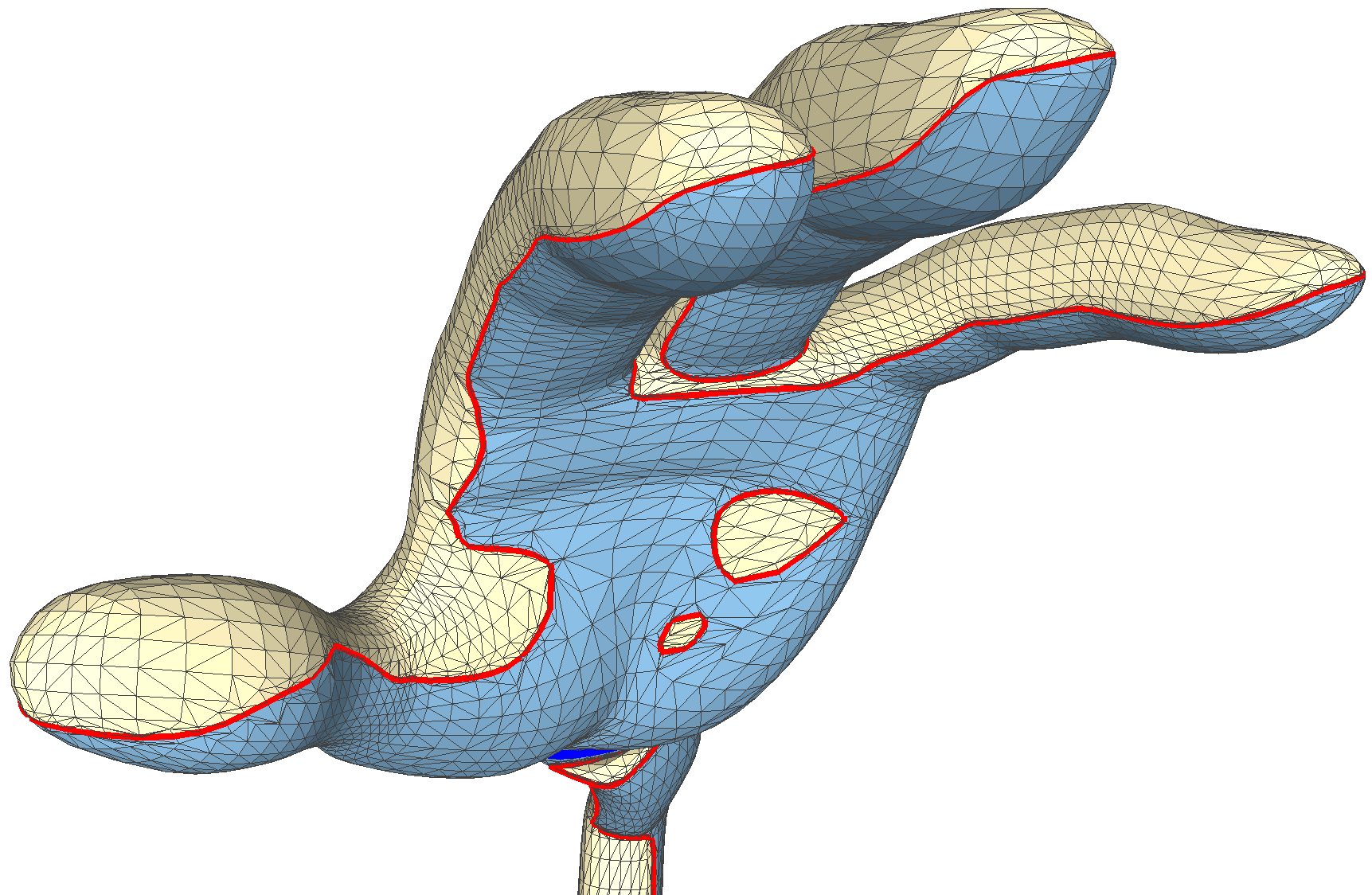}
		\caption{Contours}
	\end{subfigure}
	\begin{subfigure}[b]{0.3\linewidth}
		\includegraphics[width=\textwidth]{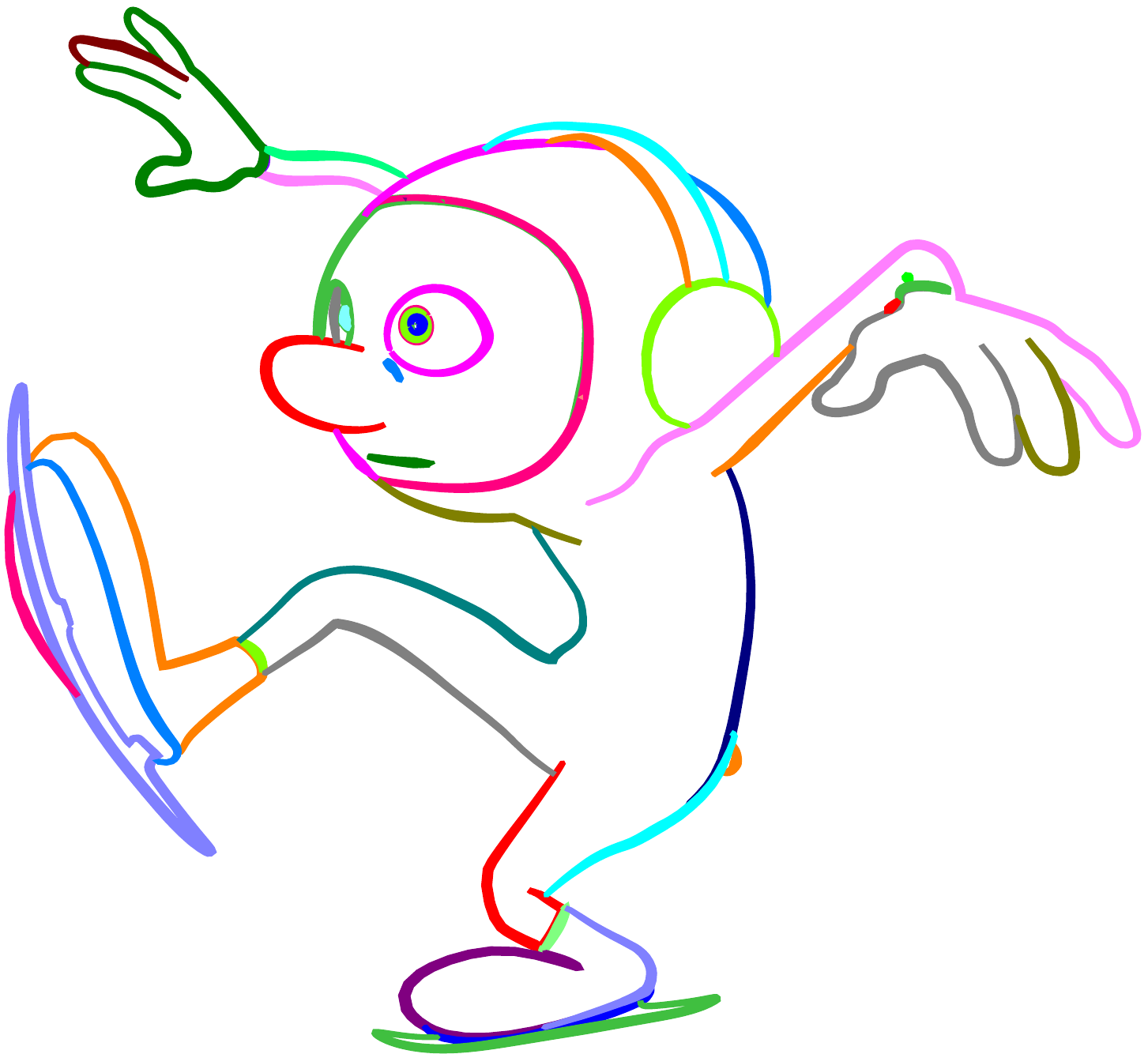}
		\caption{Chained curves}
	\end{subfigure}
	\begin{subfigure}[b]{0.32\linewidth}
		\includegraphics[width=\textwidth]{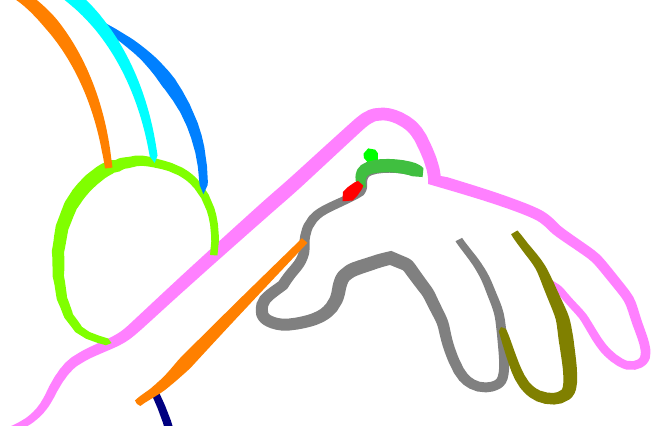}
		\caption{Closeup view on \textbf{(b)}}
	\end{subfigure}
\caption{\textbf{Contour-consistent contours of a smooth surface} \citep{Benard:2014} ---
	Compare with the contours in Figures \ref{fig:red_mesh} and \ref{fig:interp_red}.
 	\textbf{(a)} Contours computed with the contour-consistency algorithm correctly represent the contours of the original smooth surface.
  \textbf{(b)} Chaining these segments gives smooth, coherent curves.
	\textbf{(c)} Visibility is well-defined for these curves.
	``Red'' \ccCopy~Disney/Pixar
    \label{fig:contour_consistent_red}
  }
\end{figure}

\begin{figure}
  \centering
    \includegraphics[width=0.6\textwidth]{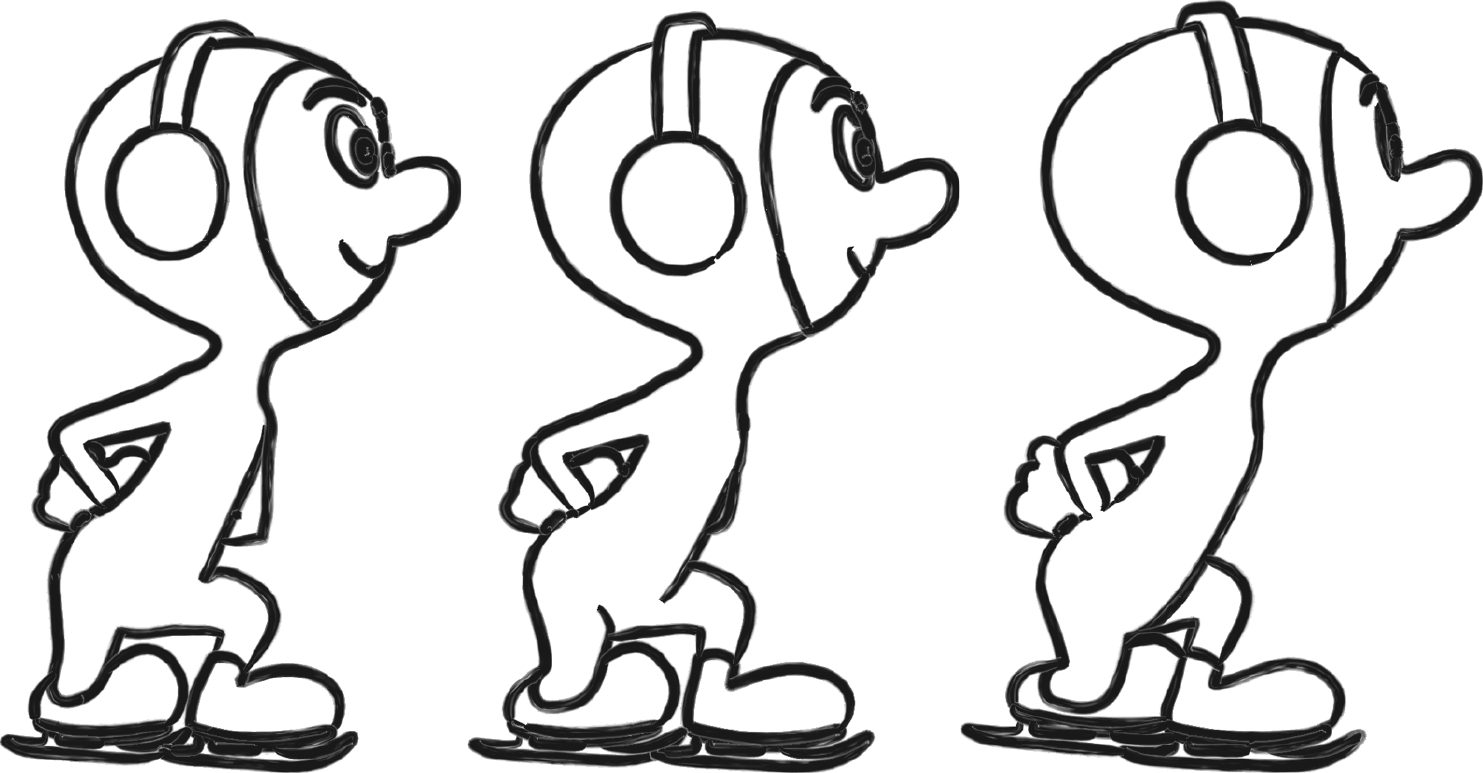}
	\caption{\textbf{Contour-consistent contours stylized with tapered strokes} \citep{Benard:2014} --- 
	Compare with the contours in Figure \ref{fig:interp_red2}. Contour-consistent contours do not suffer from breaks and gaps, producing more coherent animated strokes.
  ``Red'' \ccCopy Disney/Pixar}
  \label{fig:contour_consistent_stylized}
\end{figure}


\chapter{Implicit Surfaces: Contours and Visibility}
\label{chap:implicit_surfaces}

We now survey contour extraction and visibility algorithms for implicit surfaces. 
Implicit surfaces are smooth surfaces, so much of the theory from the previous section applies to them, but implementation is different.
 Implicit surfaces are used less frequently in computer graphics, but are often used in 3D imaging and for other kinds of volumetric data. 




\section{Surface definition}

An implicit surface $\mathcal{S}$ is defined as the isocontour of a scalar function $f: \mathbb{R}^3 \rightarrow \mathbb{R}$,
$$\mathcal{S} = \{ \vec{p} \in \mathbb{R}^3 | f(\vec{p}) = \rho \},$$
where $\rho \in \mathbb{R}$ is a target \emph{isovalue}. It is thus also called an \emph{isosurface}. An implicit surface is well-defined if $f$ does not have any \emph{singular} points, \ie{} its gradient $\nabla f$ is defined and non-zero everywhere. 
The isosurface forms a 2D manifold, partitioning the space into itself and two connected open sets: the \emph{interior}, where $(f - \rho) < 0$, and the \emph{exterior}, where $(f - \rho) > 0$ by convention.

The surface normal at $\vec{p}$ is the normalized gradient of $f$:
$$\vec{n} = \frac{\nabla f(\vec{p})}{\norm{\nabla f(\vec{p})}},$$
where $$\nabla f(\vec{p}) = \left(\left.\frac{\partial f}{\partial x}\right|_{\vec{p}}, \left.\frac{\partial f}{\partial y}\right|_{\vec{p}}, \left.\frac{\partial f}{\partial z}\right|_{\vec{p}} \right)^\top.$$
The normal curvature in any tangent space direction $\vec{t}$ is given by:
$$\kappa_\vec{n}(\vec{t}) = \frac{\vec{t}^\top (\nabla \vec{n}) \vec{t} }{ \norm{\vec{t}}^2 }$$
with $\nabla \vec{n}$ the gradient of the normal, \ie~the projection of the normalized Hessian (matrix of second partial derivatives) $\mathrm{H}f = \nabla^2 f$ onto the tangent plane.

Many approaches have been taken to compute the contour generator.

\section{Contour extraction} \label{sec:implicit_extraction}

The contour generator $\mathcal{C}$ of an implicit surface $\mathcal{S}$ seen from a camera center $\vec{c}$ is defined as the set of points $\vec{p}$ such that:
\begin{alignat*}{2}
	                          &  & (f(\vec{p}) - \vec{c}) \cdot \vec{n}           & = 0  \\
	\Leftrightarrow\mkern20mu &  & (f(\vec{p}) - \vec{c}) \cdot \nabla f(\vec{p}) & = 0,
\end{alignat*}
which is itself an implicit function. 

\subsection{Contour tracing} 
The most basic contour tracing algorithm is the generic algorithm described by \citet{Dobkin:1990} for tracing the contour of any smooth function from $\mathbb{R}^n$ to $\mathbb{R}^k$ ($k < n$). A drawback of this method is that it only extracts a fixed resolution piecewise-linear approximation of the contour. If the implicit function $f$ is at twice continuous, we can directly work with the function $f$ and trace an approximation of the contour based on its analytical tangent vector by numerical integration \citep{Bremer:1998,Foster:2005,Plantinga:2006}.

Assuming (for simplicity) an orthographic projection along the direction $\vec{v}$, a parametric curve $c: \mathbb{R} \rightarrow \mathbb{R}^3$ lies on the contour generator of an implicit surface $\mathcal{S}$ if:
\begin{equation}\label{eq:implicit_contour}
	\left\{\begin{aligned}
		f(c(t))                       & = 0                       \\
		\vec{v}^{\top} \nabla f(c(t)) & = 0. \text{\footnotemark}
	\end{aligned} \right.
\end{equation}
\footnotetext{Recalling that $\vec{a} \cdot \vec{b} = \vec{a}^\top \vec{b}$.}Denoting $\vec{w} = c'(t)$ the tangent vector of the curve and differentiating each equation with respect to $t$, we get:
\begin{equation*}
	\left\{ \begin{aligned}
		\nabla f(c(t)) \cdot \vec{w}             & = 0  \\
		\vec{v}^{\top} \mathrm{H}f(c(t)) \vec{w} & = 0.
	\end{aligned}\right.
\end{equation*}
This implies that $\vec{w}$ is proportional to the cross-product of the gradient at its basepoint and the product of the Hessian at the basepoint with the view direction, that is:
$$\vec{w} \propto \nabla f(c(t)) \times \vec{v}^\top \mathrm{H}f.$$
If we know a starting point $\vec{p}_0$ on the contour generator, we can progressively trace the full curve by taking small steps in the direction of the tangent $\vec{w}$. This corresponds to a numerical Euler integration scheme where, at each step:
\begin{alignat*}{2}
	                      &  & \vec{p}_{i+1} & = \vec{p}_i + \epsilon F(\vec{p}_i),                          \\
	\text{with}\mkern10mu &  & F(\vec{p})    & = \nabla f(\vec{p}) \times \vec{v}^\top \mathrm{H}f(\vec{p}).
\end{alignat*}
for small step size $\epsilon$.
Because 
the tangent vector at a curtain fold vanishes, tracing stagnates when it reaches a curtain fold.

The full process can be summarized as follows:
\begin{itemize}
	\item Find a starting point on the contour generator.
	\item Trace out the contour curve by Euler integration.
	\item Stop when the curve returns to the starting point or stagnates.
	\item If it stagnates, return to the starting point and trace in the opposite direction.
\end{itemize}

\paragraph{Stabilized integration.}
Euler integration is known to be unstable, that is, the position $\vec{p}$ might quickly leave the contour, even with a small step sizes, in critical configurations. To improve convergence, two correction terms can be added to the vector field $F(\vec{p})$. The first correction enforces the computed position to lie on the implicit surface by pointing towards the surface at all point of space:
$$F_{\text{surface}}(\vec{p}) = \frac{-f(\vec{p}) \nabla f(\vec{p}))}{\norm{\nabla f(\vec{p})^2}}.$$
The second correction ensures that the integration follows the contour generator. In the same way as $-f\nabla f$ tends to drive $f$ to zero, $-g \nabla g$ with $g(\vec{p}) = \vec{v} \cdot \nabla f(\vec{p})$ tends to drive $g$ to zero, \ie~towards the contour generator, leading to:
$$F_{\text{contour}}(\vec{p}) = \frac{-(\vec{v} \cdot \nabla f(\vec{p})) \vec{v}^\top \mathrm{H}f(\vec{p})}{\norm{\nabla f(\vec{p})^2}}.$$
The final Euler step is simply the weighted sum of the vector fields:
\begin{align*}
	\vec{p}_{i+1} & = \vec{p}_i + \epsilon \left( F(\vec{p}_i) + F_{\text{surface}}(\vec{p}_i) + k F_{\text{contour}}(\vec{p}_i) \right)                                            \\
	              & = \vec{p}_i + \frac{\epsilon}{\norm{\nabla f(\vec{p}_i)}^2} (\nabla f(\vec{p}_i) \times \vec{v}^\top \mathrm{H}f(\vec{p}_i) - f(\vec{p}_i)\nabla f(\vec{p}_i))) \\
	              & \hspace{.32in} -k(\vec{v} \cdot \nabla f(\vec{p}_i)) \vec{v}^\top \mathrm{H}f(\vec{p}_i))
\end{align*}
with $k$ a user-defined scalar value --- \citet{Bremer:1998} recommend choosing $k=0.5$. With those correction terms, the vector field does not vanish at curtain folds anymore, which may prove problematic if the tracer overshoots.

\begin{figure}
	\centering
	\small
	\def\svgwidth{0.9\textwidth}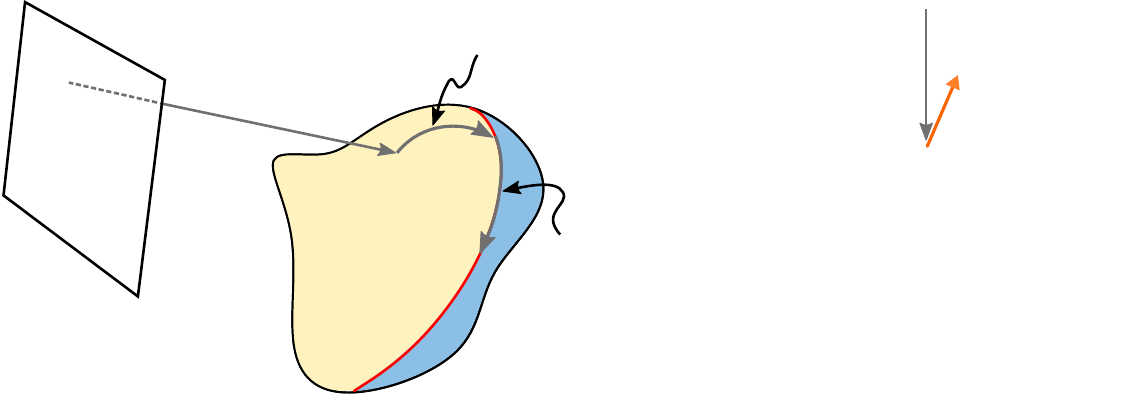\caption{\textbf{Contour finding} --- To find a starting position for tracing the occluding contour, \citet{Bremer:1998} shoot random rays in the view direction $\vec{v}$ from the image plane (left), and march along the surface in the tangent direction aligned with the projection of the gradient (right) until $\vec{v} \cdot \nabla f(\vec{p})$ changes sign.}\label{fig:implicit_tracing}
\end{figure}

\paragraph{Finding starting points.}
Different approaches have been proposed to find starting positions on the contour generator. \citet{Bremer:1998} shoot random rays from the orthographic camera and, when an intersection is found, they march along the surface in a tangent direction whose image-space projection is in the same direction as that of the gradient, \ie
$$F(\vec{p}) = \frac{\nabla f(\vec{p}) \times (\vec{v} \times \nabla f(\vec{p}))}{\norm{\nabla f(\vec{p})}^2},$$
by Euler integration, until the sign of $\vec{v} \cdot \nabla f(\vec{p})$ changes (\fig{implicit_tracing}). The term $F_\text{surface}$ can be added to the vector field to ensure that the integrated positions remain close to the surface.

Instead of casting rays each time the viewpoint changes, \citet{Foster:2005} precompute a dense set of seed points using the surface-constrained ``floater'' particles of \citet{Witkin:1994}, and select the points that are close enough to the contour generator, \ie~such as those where $\dotPabs{\vec{v}}{\nabla f(\vec{p})}$ is below a threshold.

\begin{figure}[t]
	\centering
	\small
	\begin{subfigure}[t]{0.5\linewidth}
		\def\svgwidth{\hsize}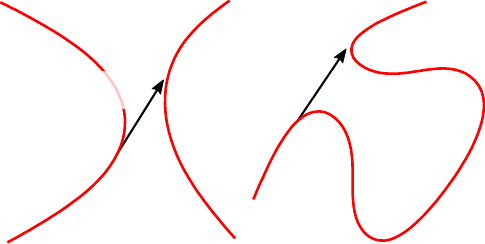\caption{fixed step size issues}\label{fig:implicit_topology}
	\end{subfigure}
	\qquad
	\begin{subfigure}[t]{0.3\linewidth}
		\includegraphics[width=\linewidth]{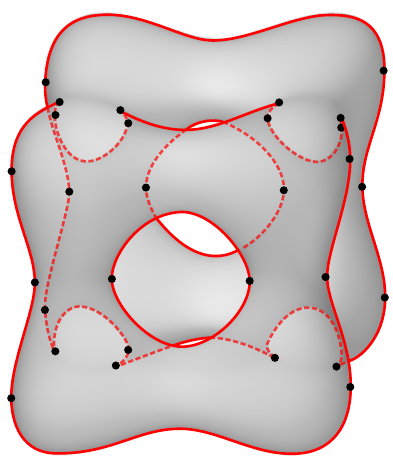}
		\caption{contours traced with a dynamic step size}\label{fig:tangle_cube}
	\end{subfigure}
	\caption{\textbf{Implicit surface contour generator tracing} --- With a fixed step size \textbf{(a)}, the traced contour may jump to another component of the contour (left) or skip a part of it (right). \textbf{(b)} With a dynamic step size integration scheme, \citet{Plantinga:2006} accurately trace the contour generator of a complex implicit tangle cube: $x^4-5x^2+y^4-5y^2+z^4-5z^2+10=0$. The starting points are indicated by the black dots.
	\label{fig:implicit_contours}
}
\end{figure}

\paragraph{Tracing with interval arithmetic.}
The above methods are neither guaranteed to find all contour generators, nor to trace them with accurate topology. With the fixed step size $\epsilon$, the tracing may accidentally jump to another component of the contour generator or skip a part of it (\fig{implicit_topology}). 

\citet{Plantinga:2006} describe a method that does offer topological guarantees under orthographic projection.
They provide a detailed theoretical analysis of the implicit surface contour.  They then use this to derive an
interval arithmetic algorithm to identify points on the contour, and then trace the contour generator with a dynamic step size. After each Euler integration step, they perform an interval test to check whether the segment $\vec{p}_i\vec{p}_{i+1}$ is a good approximation of the contour generator (\fig{tangle_cube}).  They further show how to accurately evolve the contour over time.

\subsection{Extraction as surface-surface intersection}

\eqn{eq:implicit_contour} can be interpreted slightly differently: the contour generator can be viewed as the curve at the intersection of two implicit surfaces, the object surface $\mathcal{S}$ and the \emph{contour surface} (\fig{implicit_intersection}) defined implicitly as the zero-set of $\nabla f(\vec{p}) \cdot \vec{v}$.
(\citet{Stroila:2008} called it the ``silhouette surface.'')
%

\begin{figure}[t]
	\centering
	\small
	front-view\hspace{-2.5em}
	\includegraphics[width=0.4\linewidth]{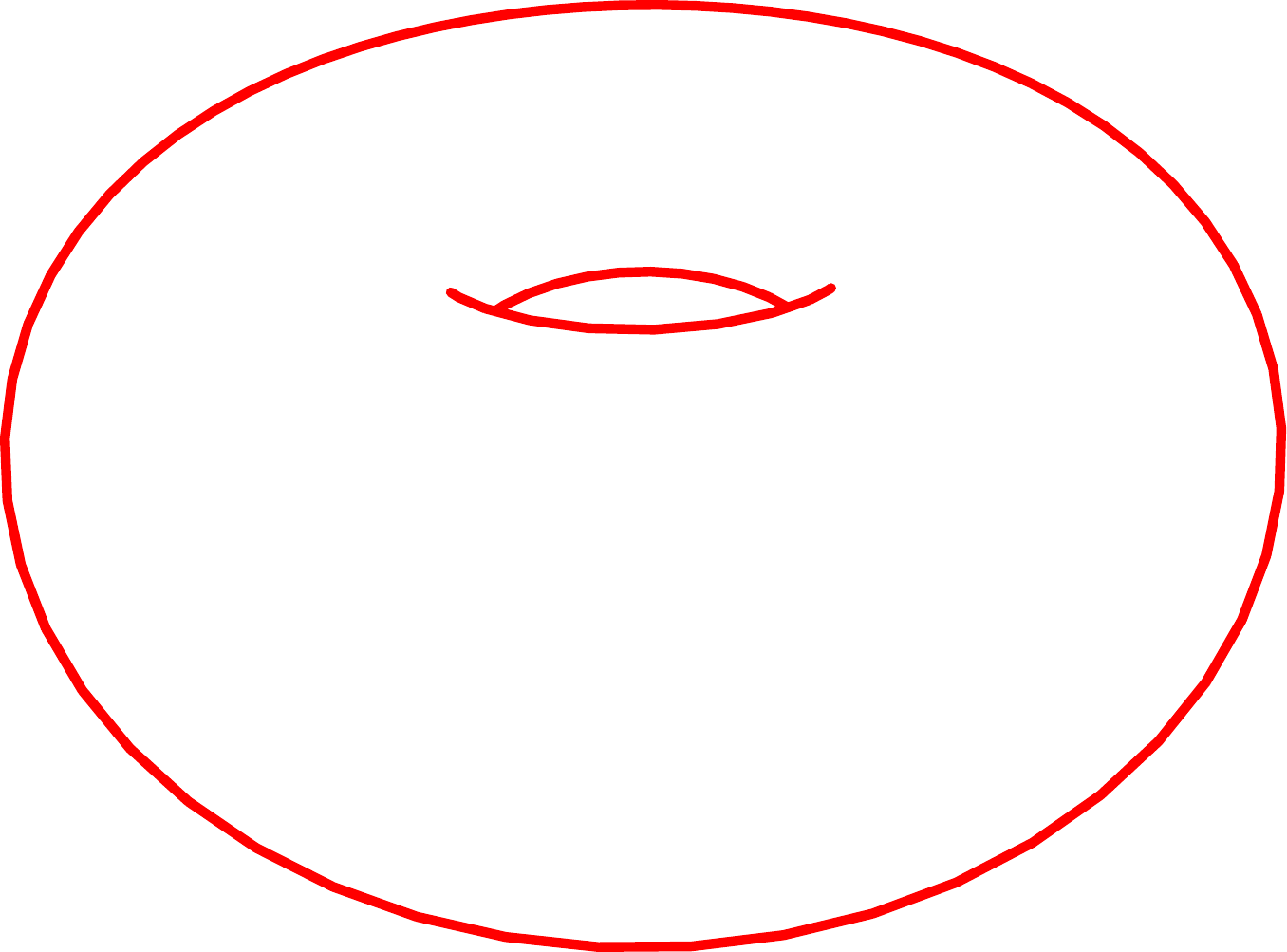}
	side-view
	\includegraphics[width=0.37\linewidth]{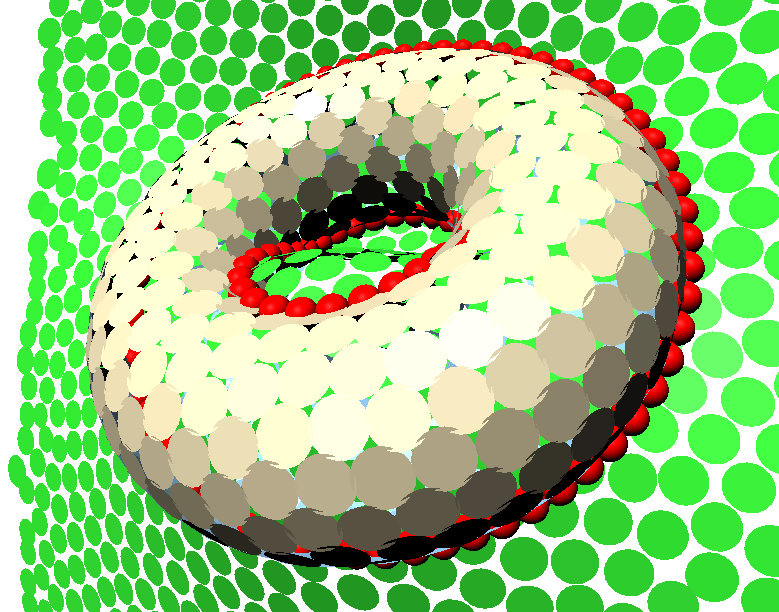}
	\caption{\textbf{Surface-surface intersection} --- The apparent contour (left) along the view direction $\vec{v}$ is the projection of the contour generator (in red), shown from a side-view (right), at the intersection of the implicit surface $\mathcal{S}$ (in yellow) and the implicit contour surface $\nabla f \cdot \vec{v} = 0$ (in green). Images generated with the ``Wickbert'' particles library~\citep{wickbert}.}\label{fig:implicit_intersection}
\end{figure}

To delineate this intersection, \citet{Stroila:2008} simultaneously constrain the ``floater'' particles of \citet{Witkin:1994} to lie on the implicit surface $\mathcal{S}$ and on the contour surface. After optimization, the particles can be connected together to form closed loop on the implicit surface. The differential properties of the contour curve can be leveraged to properly select each particle neighbors, although this does not guarantee accurately
tracing the contour, nor finding all contours.

\begin{figure}[t]
	\centering
	\small
	front-view\hspace{-2em}
	\includegraphics[width=0.35\linewidth]{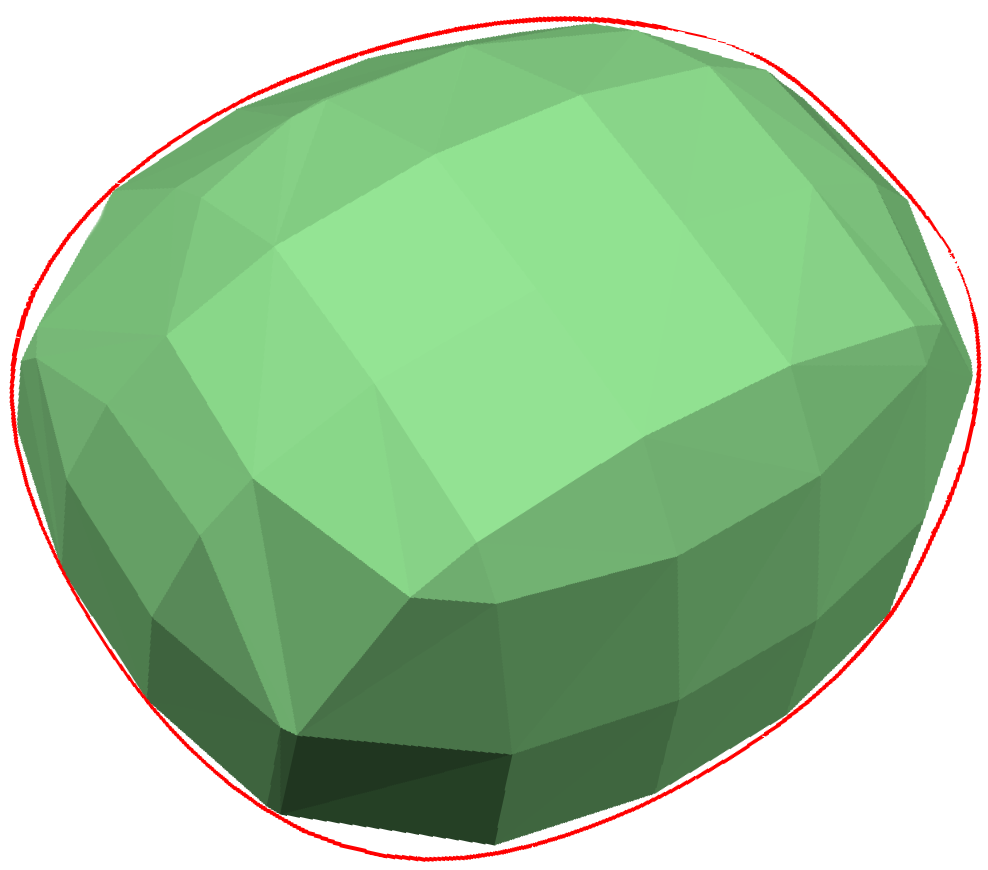}
	\qquad
	side-view\hspace{-1.75em}
	\includegraphics[width=0.35\linewidth]{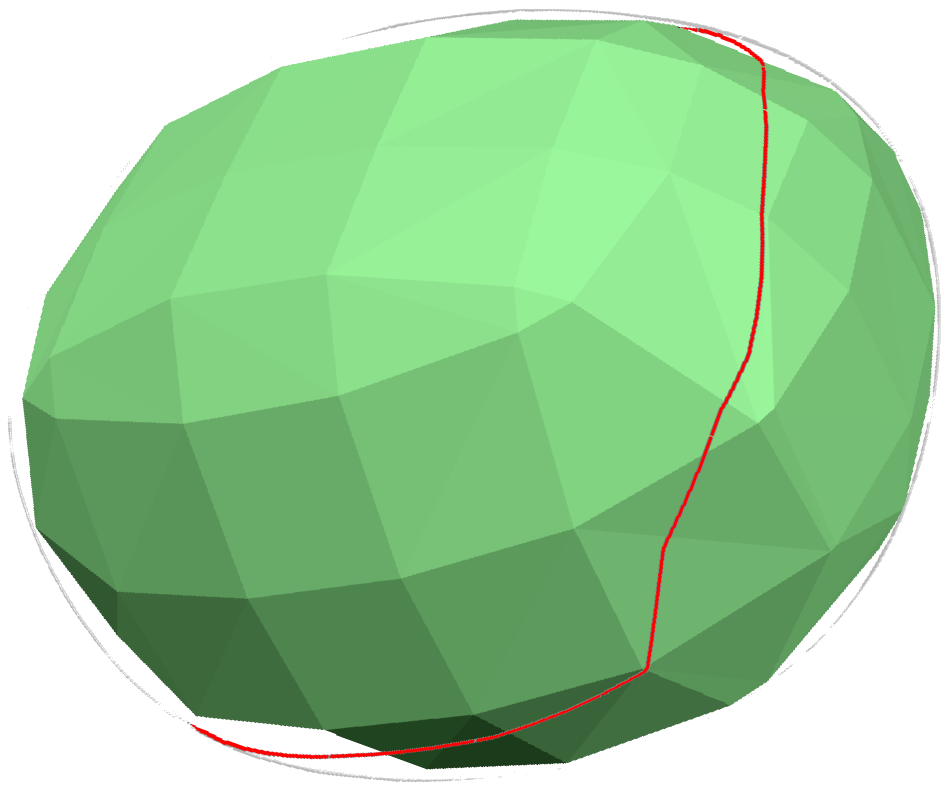}
	\caption{\textbf{Proxy-based method} --- The smooth surface occluding contour is approximated by computing the interpolated contours of on a coarse base mesh, projecting the vertices of this coarse contour onto the actual smooth surface and subdividing it. Images generated with ``ShapeShop''\citep{shapeshop}.
	}\label{fig:implicit_proxy}
\end{figure}

\subsection{Extraction on a mesh} 

A even simpler, but very approximate, approach consists in first extracting a polygonal mesh from the implicit surface, \eg{} with the Marching Cubes algorithm~\citep{Lorensen:1987} and its extensions~\citep{deAraujo:2015}, computing its interpolated contours (Section \ref{sec:interpolated_contours}), and projecting those onto the implicit surface after view-dependent subdivision~\citep{Schmidt:2007} (\fig{implicit_proxy}). Since the implicit function $f$ is defined everywhere in $\mathbb{R}^3$, the projection of a point $\vec{p}$ on the implicit surface $f(\vec{p}) = \rho$ can be computed by walking along the gradient of $f$, updating $\vec{p}$ with the following convergence iteration:
$$\vec{p} \leftarrow \vec{p} + \frac{(f(\vec{p}) - \rho)\nabla f(\vec{p})}{\norm{\nabla f(\vec{p})}}.$$
This scheme is a form of gradient descent on $(f(\vec{p})-\rho)^2$.
It usually leads to the iso-value after a few iterations if $f$ is smooth enough, though it may get stuck when there are discontinuities.

\section{Visibility} \label{sec:implicit_visibility}

As with mesh contours, two families of approaches can be used to determine the visibility of the extracted contour generators. The first one is based on ray-casting and thus more accurate but only suitable for offline computation. The second one uses the depth buffer algorithm and is well suited for real-time applications, even though the implicit nature of the surface makes the creation of the depth buffer more complex.

\paragraph{Ray-casting and propagation.} The simplest, but most time-consuming, method consists in testing the visibility of every point $\vec{p}_i$ on the discretized occluding contour by casting a ray from the camera towards $\vec{p}_i$ and performing a ray-surface intersection test. Due to numerical issues, $\vec{p}_i$ may not exactly lie on the implicit surface; therefore it may be locally occluded by the surface leading to a spurious ray intersection. To mitigate this problem, \citet{Bremer:1998} discard surface intersections that are too close to $\vec{p}_i$ which, in turn, may reveal contours that are barely obscured by nearby pieces of surface.

To reduce the number of ray tests, one can first check every $n^\text{th}$ point for occlusion, and then refine by testing the intermediate points when the visibility changes between $\vec{p}_i$ and $\vec{p}_{i+n-1}$ \citep{Bremer:1998,Foster:2005}. To reduce the number of tests even further, a view graph can be built and the visibility information can be propagated along the contour chains (\sect{sec:ray_casting}). However the propagation rules slightly differ from those for polygonal meshes since implicit surfaces are closed by construction. To reduce the number of ray-tests further, \citet{Stroila:2008} describe methods to propagate Quantitative Invisibility (Section \ref{sec:vis_algo}) on contours of implicit surfaces.

\paragraph{Depth buffer.} The alternative solution is to render the 3D scene into a depth buffer, and then to use this buffer to determine the visibility of the apparent contour. Besides the problems mentioned in \sect{sec:buffer}, an additional difficulty of such an approach is that, unlike polygonal meshes, implicit surfaces cannot be rasterized directly. One could again extract a polygonal mesh from the implicit surface, but a highly refined tessellation is required to avoid visual artifacts -- otherwise the mesh contours will largely disagree with the smooth contours -- which is computationally expensive.

Instead, we can discretize the implicit surface using \emph{surfels} \citep{Pfister:2000}, oriented ellipses that are traditionally used to render 3D point clouds. For instance, \citet{Foster:2005} generate a surfel for every ``floater'' particle distributed on the implicit surface and orient them according to the surface gradient. Each surfel is then projected into the 2D image plane and rasterized with depth writes activated. To provide accurate results, the implicit surface needs to be suitably covered by surfels, which may require a large number of particles in areas of high curvature~\citep{Meyer:2005}.

On deforming surfaces, dynamically recomputing such particle distribution is too costly for real-time applications. In the same spirit as the painterly rendering technique of \citet{Meier:1996}, \citet{Schmidt:2007} first distribute surfels on a low resolution mesh (previously used for contour extraction), and then project them onto the implicit surface. However, after projection, the surfels may not properly cover the surface. Even thought a simple heuristic is proposed to non-uniformly scale them, it cannot guarantee that the surface will be accurately approximated by the distribution of surfels.

\section{Volumetric data} \label{sec:volumes}

Volumetric data can be seen as a special case of implicit surfaces, one where the implicit function $f$ is discretized on a regular 3D grid (voxel grid) into density values:
\begin{equation}
v_{ijk} = f(\vec{p}_{ijk}) = f(x_i,y_j,z_k).
\end{equation}
Different methods make different assumptions on the density between these values. In this case, the local curvature and scale of small features is limited by the discretization, making it possible to make better guarantees about contour and intersection detections, unlike with arbitrary implicit surfaces.

There are two general approaches to contour visualization. The first entails defining an isosurface, and then extracting contours of the isosurface, similar to what was done in the general implicit surface case.  The second approach is analogous to image-space contour rendering (\chap{chap:image_space}): it uses a conventional volume rendering method, but with transfer functions designed to emphasize contours

\subsection{Isosurface extraction} 

To compute the contours corresponding to a given iso-value $\rho$, we could first extract the corresponding iso-surface $f(\vec{p}_{ijk}) = \rho$ with, \eg~the Marching Cubes algorithm \citep{Lorensen:1987}, compute its normals as $\vec{n} = \nabla f$, and extract the surface Interpolated Contours (\sect{sec:interpolated_contours}). However, as discussed previously for implicit surfaces, occluding contours can be seen as the zero-set of two implicit functions:
\begin{equation*}
	\left\{ \begin{aligned}
		f(\vec{p}_{ijk}) - \rho               & = 0 \\
		\nabla f(\vec{p}_{ijk}) \cdot \vec{v} & = 0
	\end{aligned}\right.
\end{equation*}
or, geometrically, as the intersection of two implicit surfaces.

\begin{figure}[t]
	\centering
	\small
	\def\svgwidth{\hsize}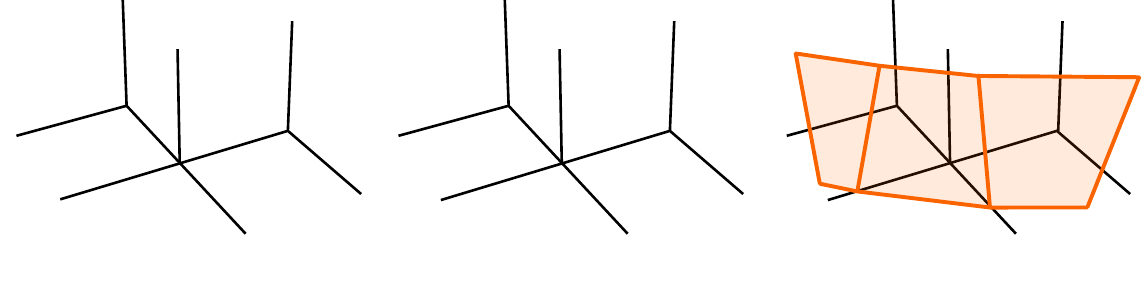\caption{\textbf{Marching line algorithm} --- The intersections of the isosurface (left) and contour surface (center) with each face of every voxel are first computed. The resulting segments are intersecting on the faces at contour points (right).}\label{fig:volume_marching}
\end{figure}

The simplest solution to extract a piecewise linear approximation of the intersection curves of these implicit functions is the marching lines algorithm \citep{Thirion:1996,Burns:2005}. For each voxel, this method first computes the line segments at the intersection of both functions with the voxel faces, using linear interpolation of the density and gradient values at the voxel corners (\fig{volume_marching}, left and center). It then finds the intersection points between these two set of lines on each face. These contour points are eventually connected to produce contour chains (\fig{volume_marching}, right).

The piecewise linear approximation of both the implicit functions and their intersection might be too crude. Assuming that the input scalar field is smooth enough, \citet{Schein:2004} use trivariate B-spline functions to model the volumetric data. The scalar values are used as control points of a trivariate tensor product B-spline function $D :\mathbb{R}^3 \rightarrow \mathbb{R}$:
$$D(u,v,w) = \sum_i \sum_j \sum_k f(\vec{p}_{ijk}) N_i^d(u) N_j^d(v) N_k^d(w),$$
where $N_i^d(u)$, $N_j^d(v)$ and $N_k^d(u)$ are the B-spline basis functions of degree $d$ that controls the smoothness of the representation. Extracting the contour then boils down to resolve the following system of equations:
\begin{equation*}
	\left\{ \begin{aligned}
		D - \rho               & = 0 \\
		\nabla D \cdot \vec{v} & = 0
	\end{aligned}\right.
\end{equation*}
The multidimensional Newton-Raphson solver of \citet{Elber:2001} can be used to solve this system to a desired accuracy starting from a dense set of seed points. If quadratic or higher basis functions are used, the gradient field of the trivariate B-spline function $\nabla D$ is continuous, and thus the contours are smooth and continuous. In addition, since the solver can refine the data at any parametric location $(u,v,w)$ and not just discrete grid points, the contour approximation is better adapted to the input data. However, with this method, the connectivity of the extracted contour points cannot be trivially inferred for the voxel grid anymore, limiting the rendering and stylization possibilities.

\paragraph{Acceleration strategies.}
A naive implementation of the two previous approaches would have an $O(n^3)$ complexity for a dataset of size $n \times n \times n$ voxels. In practice, a typical isosurface will have surface area $O(n^2)$ and thus its contour will have size $O(n)$ edges (\sect{sec:sparse}). 

To speed-up contour extraction in volumes, strategies similar to those used for mesh contours can be employed.

\citet{Burns:2005} adapt randomized search with temporal coherence (Section \ref{sec:randomized}). They first test random voxels until a contour is found. Then, they move to the voxel adjacent to the face containing one of the intersection points and repeat until the initial voxel is reached, forming a contour loop. The random sampling strategy to find a starting voxel can be further improved by leveraging temporal coherence, \ie~searching nearby contour-containing voxels from the previous frame, and by gradient descent, alternatively walking along the gradient of the isosurface and contour functions until a new contour-containing voxel is found.

An acceleration data-structure can be built during a preprocessing step. \citet{Elber:2001} construct a 2D lookup table whose first dimension corresponds to isovalue ranges and second dimension corresponds to bounding cones covering the unit sphere. The trivariate B-spline function $D$ is then clustered into ``singletons'' based on its isovalue and gradient, and stored them in the table. At runtime, given a view direction $\vec{v}$ and isovalue $\rho$, only the relevant singletons can be efficiently retrieved from the table, defining the seed points for the Newton-Raphson solver.

\paragraph{Visibility.}
Finally, the visibility of the contour can be computed with respect to the target isosurface by tracing a ray from each contour point towards the camera, similarly to \citet{Bremer:1998}. For each voxel traversed by the ray, we need to test whether the isosurface is intersected. It is achieved by checking if the sign of $f$ changes at the intersection of the ray with the face by which it leaves the voxel. If it does, the ray passes from outside to inside the isosurface, and the contour point is thus occluded.

\subsection{Direct volume rendering} 

\begin{figure}[t]
	\centering
	\small
	\def\svgwidth{0.9\linewidth}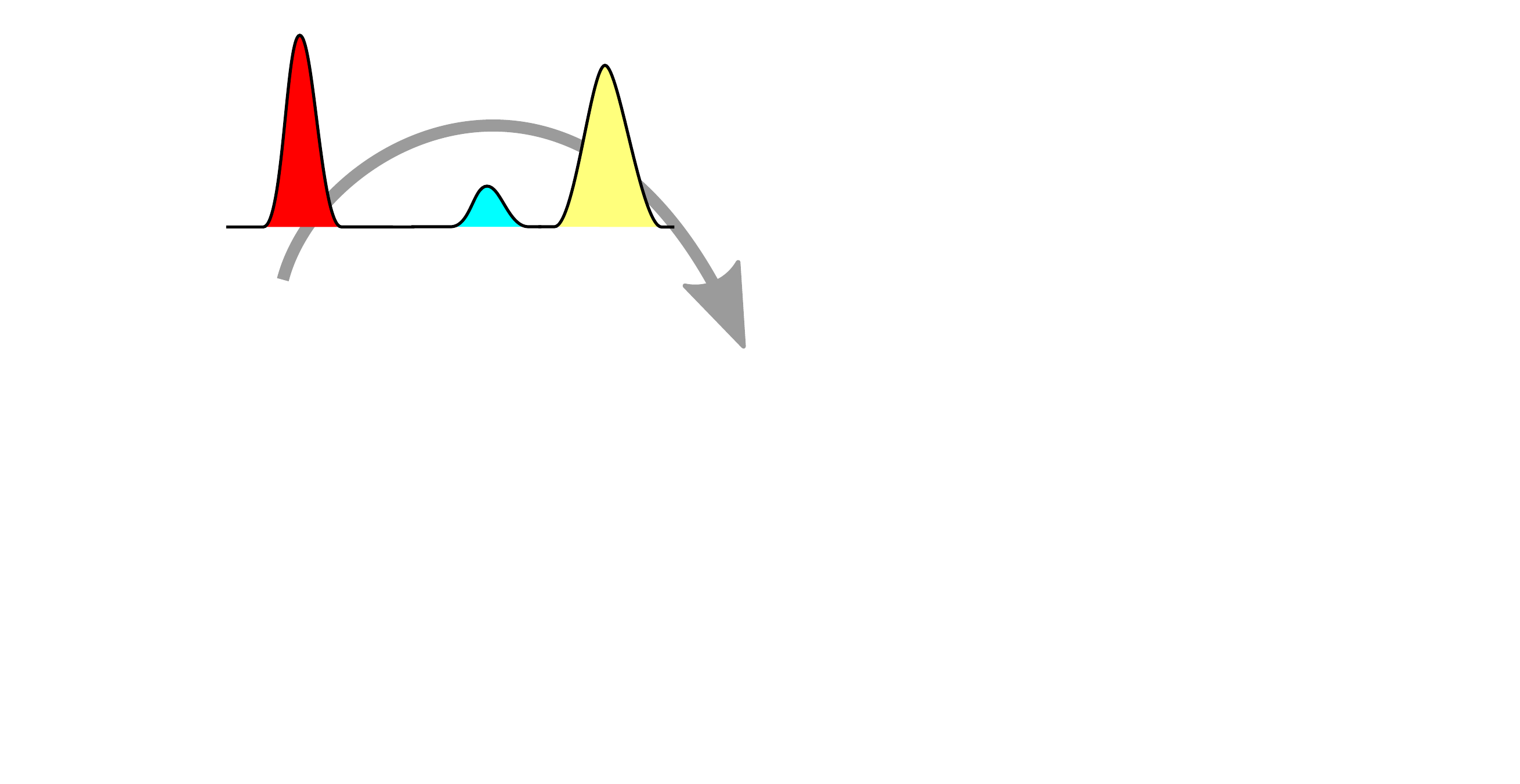\caption{\textbf{Direct volume rendering} \citep{Ikits:2004} --- The density value $f$ of each voxel $p_{ijk}$ is mapped through a transfer function (top) into a color $c$ and an opacity $\alpha$; those are then accumulated along each view ray to produce an image of the data.}\label{fig:volume_direct}
\end{figure}

In this approach, we use conventional volume rendering, and modify the transfer function to visualize contours.

In conventional volume visualization \citep{Kaufman:2005}, the initial density value $v_{ijk} = f(\vec{p}_{ijk})$ of each voxel $\vec{p}_{ijk}$ of the discrete volume is first mapped through a transfer function into a color value $c_{ijk}$ and an opacity $\alpha_{ijk}$. 
Rays are then cast from each pixel of the camera, along which colors and opacities are regularly sampled with appropriate interpolation (often tri-linear) in the voxel grid. Those samples are eventually composited in front-to-back order to yield a single color per pixel. For shading computation, a normal at each voxel is computed as the normalized gradient of $f$, which is usually approximated with central differences:
\begin{equation*}
	\nabla f(\vec{p}_{ijk}) = \nabla f(x_i,y_j,z_k) \approx \frac{1}{2} \left[
		\begin{aligned}
			 & f(x_{i+1},y_j,z_k) - f(x_{i-1},y_j,z_k) \\
			 & f(x_i,y_{j+1},z_k) - f(x_i,y_{j-1},z_k) \\
			 & f(x_i,y_j,z_{k+1}) - f(x_i,y_j,z_{k-1})
		\end{aligned} \right],
\end{equation*}
although more advanced gradient estimation operators have been proposed~\citep{Lichtenbelt:1998}. This approach is called the ``pre-classified model'' since voxel densities are mapped to colors and opacities prior to interpolation. An alternative solution is the ``post-classified model'' that interpolates the voxel densities before mapping the resulting values to colors and opacities, and thus tends to produce sharper results.

To enhance occluding contours, we can increase the opacity of voxels whose gradient is near perpendicular to the view direction~$\vec{v}$. Given input opacity $\alpha_{ijk}$, \citet{Ebert:2000} suggest  using the following opacity transfer function:
$$\alpha_{ijk}' = \alpha_{ijk} (k_{sc} + k_{ss}(1-\abs{\nabla f(\vec{p}_{ijk}) \cdot \vec{v}}^{k_{se}}), $$
where $k_{sc}$ controls the scaling of non-contour regions, $k_{ss}$ controls the amount of contour enhancement, and $k_{se}$ controls the sharpness of the contour ``curve''.

When one wishes to only visualize data variations and ignore scalar values, voxel color can directly be defined based on the gradient magnitude. For instance, \citet{Csebfalvi:2001} propose the following color transfer function:
$$c_{ijk} = w(\abs{\nabla f(\vec{p}_{ijk})}) (1-\abs{\nabla f(\vec{p}_{ijk}) \cdot \vec{v}})^{k_{se}},$$
with $w$ a windowing function selecting the range of interest in the gradients. Opacity modulations can also be computed with the same transfer function and combined with colors with standard front-to-back compositing (\fig{direct_stylized1}).

Alternatively, maximum intensity projection can be used. It consists in only keeping the sample with the highest intensity (color times opacity) along each ray. This tends to reduce the visual overload since a single (but potentially different) iso-value is represented per pixel, but discards any depth information, which may make the interpretation of the resulting image difficult. The depth ordering perception can be improved by local maximum intensity projection that selects the closest sample to the camera which is above a threshold.

\begin{figure}[t]
	\centering
	\small
	\begin{subfigure}[t]{0.3\linewidth}
		\includegraphics[width=\linewidth]{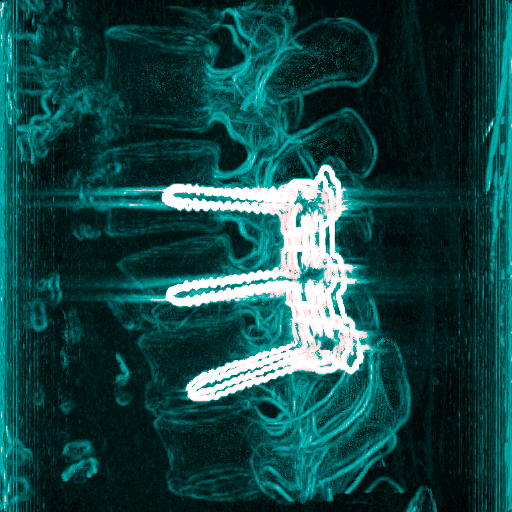}
		\caption{Vertebra rendered with the transfer function of \citet{Csebfalvi:2001}.}
		\label{fig:direct_stylized1}
	\end{subfigure}
	\quad
	\begin{subfigure}[t]{0.6\linewidth}
		\includegraphics[width=\linewidth]{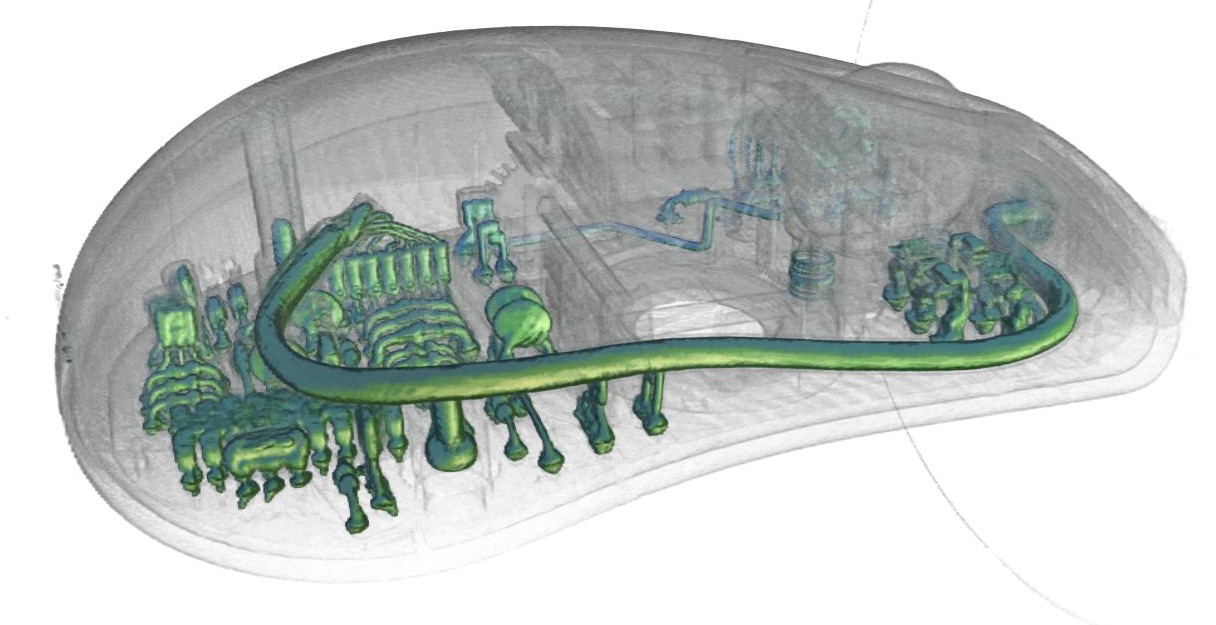}
		\caption{Computer mouse visualized with the interactive technique of \citet{Lum:2002}.}
		\label{fig:direct_stylized2}
	\end{subfigure}
	\caption{\textbf{Contour enhanced direct volume rendering.}}\label{fig:direct_stylized}
\end{figure}

\paragraph{Graphics hardware acceleration.}
The above methods involve expensive ray marching, which prohibits real-time rendering of large volumetric data sets. Graphics hardware can be used for acceleration \citep{Lum:2002,Nagy:2002}. In particular, the volume is sliced with polygons aligned with the view, and progressively accumulated in the image plane. Significant speed-ups are obtained by storing the voxel grid into a 3D texture and leveraging  graphics hardware for the slicing and re-sampling operations. For contour rendering, the normalized gradient is stored as an additional 3D texture. It is accessed at each fragment of every slice and used to modulate the color and/or opacity of the fragment by its dot product with the view direction (\fig{direct_stylized2}).

\begin{figure}[t]
	\centering
	\small
	\includegraphics[width=0.31\linewidth]{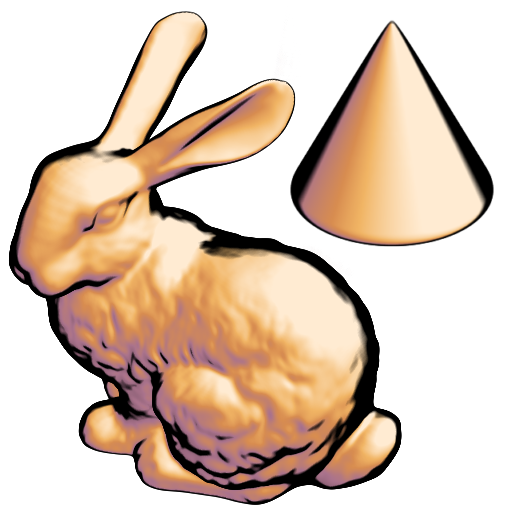}
	\includegraphics[width=0.31\linewidth]{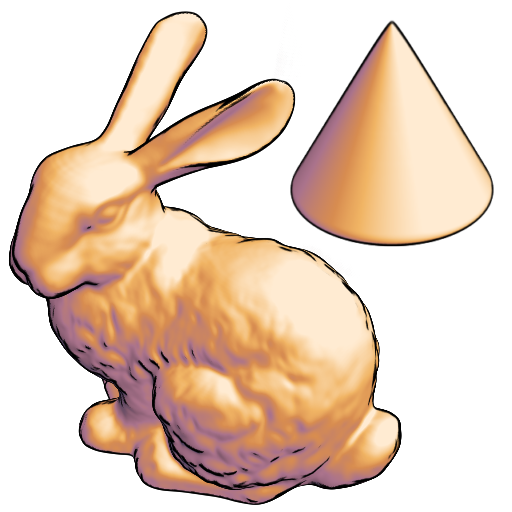}
	\includegraphics[width=0.31\linewidth]{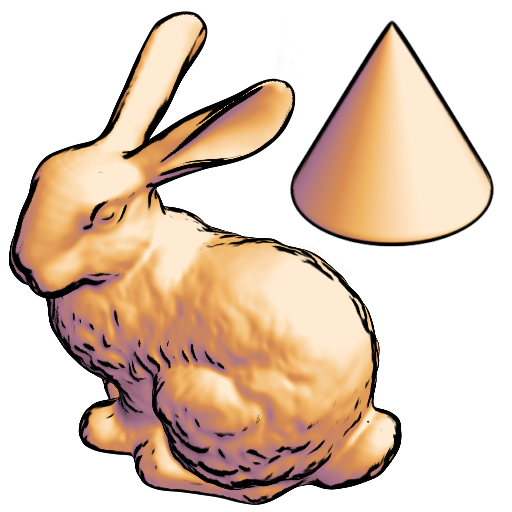}
	\caption{\textbf{Contour thickness variations in volume rendering} --- From left to right, apparent contours based solely on $\vec{v} \cdot \vec{n}$, with controlled thickness using $T=1$ and $T=2.5$. Images rendered with ``miter'' courtesy of Gordon Kindlmann, more information at \url{http://www.sci.utah.edu/~gk/vis03/}.}\label{fig:implicit_contour}
\end{figure}

\paragraph{Line thickness control.}
In the above methods, the thickness of the apparent contour varies in image-space (\fig{implicit_contour} left). These thickness variations are related to the radial curvature $\kappa_r$ of the selected iso-surfaces. For a given threshold on $\abs{\vec{n} \cdot \vec{v}}$, the resulting contour is thicker on areas of low radial curvatures since the normal is nearly perpendicular to the view direction in this large, almost flat region. Conversely, in areas of high curvature, the normal quickly changes resulting on a very thin contour.

To solve this problem, \citet{Kindlmann:2003} derive a 2D transfer function that takes the radial curvature $\kappa_r$ into account and guarantees that the apparent contour will approximately have a constant, controllable thickness $T$ in image-space (\fig{implicit_contour} right). 


To use this transfer function, we need to accurately estimate the curvature in the voxel grid and thus the Hessian of the 3D scalar field. Simply computing finite differences over the previously estimated normals would provide a very crude approximation. Instead, \citet{Kindlmann:2003} use axis-aligned 1D continuous convolution filters for zero-, first- and second-derivative estimation. They show that the cubic B-spline and its derivatives provide a good tradeoff between accuracy and robustness to the noise.

\subsection{Particle distribution} 

The particle-based implicit surface visualization technique of \citet{Foster:2005} can also be adapted to volumes~\citep{Busking:2008}. During a pre-process, particles are distributed on a given isosurface by sampling the volume data with linear interpolation on a user-defined grid and applying some relaxation steps~\citep{Meyer:2005}. At runtime, particles whose normal is almost orthogonal to the view direction are selected for rendering. From each of those, a short line segment is traced in a direction locally parallel to the contour. A linear approximation of this direction can be computed as follows. In the local coordinate system formed by the two principal directions $\vec{e}_1$ and $\vec{e}_2$ and the the normal $\vec{n}$, the behavior of the normal can be linearly approximated by:
$$\tilde{\vec{n}}(u,v) = (-\kappa_1 u, -\kappa_2 v, 1)^\top,$$
with $\kappa_1$ and $\kappa_2$ the principal curvatures in the corresponding principal directions.
The view vector expressed in the same coordinate frame is:
$$\tilde{\vec{v}} = (\vec{e_1} \cdot \vec{v}, \vec{e_2} \cdot \vec{v},\vec{n} \cdot \vec{v})^\top.$$
The contour lime in this frame can thus be defined as the set of parametric locations $(u,v)$ such as:
$$\tilde{\vec{v}} \cdot \tilde{\vec{n}}(u,v) = 0 \Leftrightarrow -\kappa_1(\vec{e_1} \cdot \vec{v})u - \kappa_2(\vec{e_2} \cdot \vec{v}) v + \vec{n} \cdot \vec{v} = 0,$$
from which a parallel direction in world space can be derived:
$$\vec{d} = -\kappa_2(\vec{e_2} \cdot \vec{v})\vec{e}_1 + \kappa_1(\vec{e_1} \cdot \vec{v})\vec{e}_2.$$

\begin{figure}[t]
	\centering
	\small
	\def\svgwidth{\linewidth}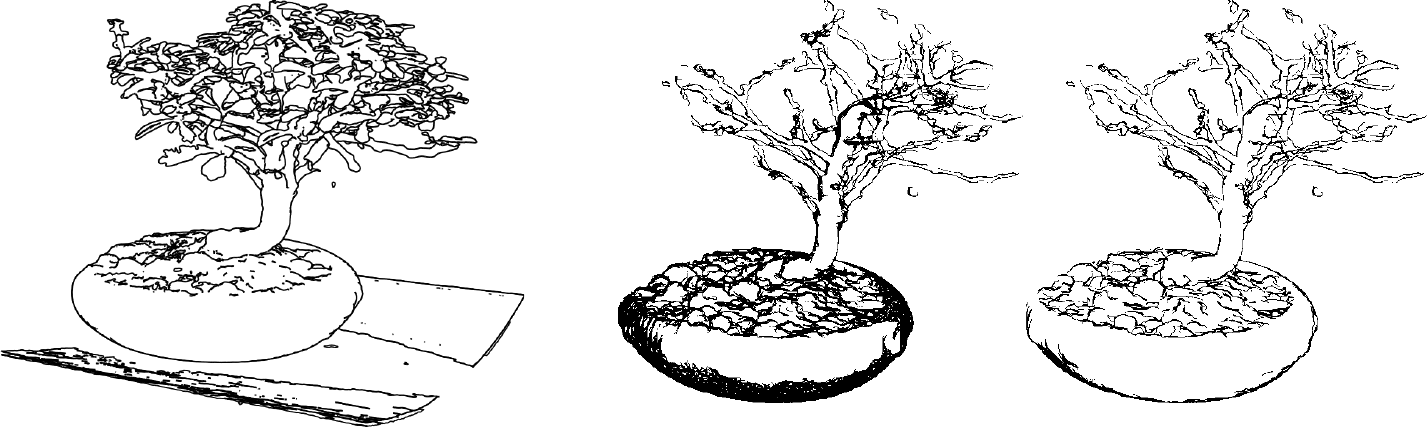\caption{\textbf{Bonsai dataset ($\mathbf{256^3}$ voxels)} --- Contours extracted with the marching line method of~\citet{Burns:2005} \textbf{(a)} and the particle system of~\citet{Busking:2008} without \textbf{(b)} and with \textbf{(c)} thickness control. Volumetric dataset courtesy of Stefan~\citet{TheVolumeLibrary}.}\label{fig:volume_bonsai}
\end{figure}

In addition, to avoid thickness variations of the apparent contour (\fig{volume_bonsai}), \citet{Busking:2008} use a thresholding function that depends on the image-space distance $T$ of the particle to the contour line approximation, \ie~assuming orthogonal projection:
$$T = \frac{(\vec{n} \cdot \vec{v})^2}{\sqrt{(\kappa_1(\vec{e_1} \cdot \vec{v}))^2 + (\kappa_2(\vec{e_2} \cdot \vec{v}))^2}}.$$
Since principal curvature information is view-independent, it can be pre-computed for all particles.


\chapter{Stylized Rendering and Animation}
\label{chap:rendering}

We now come to the reward for all the hard work of curve extraction: rendering these curves in an attractive, artistic style. This chapter describes algorithms for stylized rendering and animation.

The basic steps are to combine visible contour line segments into longer curves (\sect{sec:view_graph_planar_map}), and then to render those curves with stroke textures, such as a pen, pencil, or paint strokes (\sect{sec:stylization}).  In many cases, it will be necessary to simplify the curves before rendering, for example, to remove unnecessarily detail that an artist would never draw (\sect{sec:simplification}).

Together with the strokes, we usually wish to draw the model with shading or texturing. We briefly survey shading and texturing methods in \sect{sec:shading}.

Finally, when rendering these models in animation, we often want to render strokes with coherent stylization over time. We survey coherent stylization in \sect{sec:animation}.

\begin{figure}
	\centering
	\small
	\def\svgwidth{\linewidth}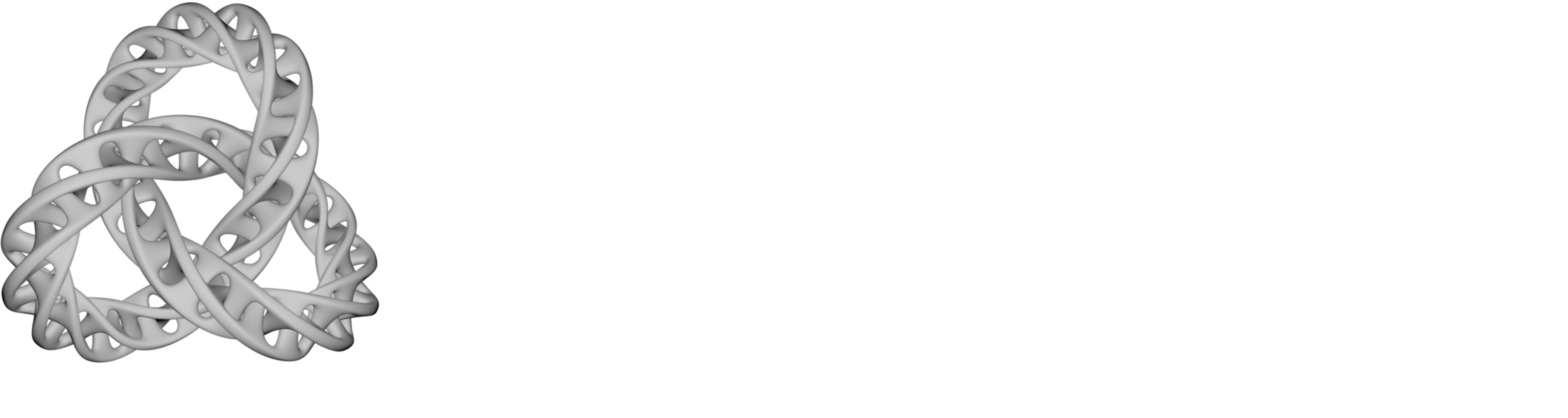\caption{ \textbf{Stylized contour rendering pipeline} --- Starting from a 3D geometry, the contour line segments are first extracted. Their visibility is then computed by building their View Graph; here the visible chains are shown, each with a separate color. Topological simplification can be optionally applied for legibility or artistic purposes. Finally, each chain of the View Graph is rendered with stylized strokes. Images generated with Blender Freestyle.
	}\label{fig:knot_pipeline}
\end{figure}

\section{Stroke extraction}\label{sec:view_graph_planar_map}

The first stage of processing is to extract smooth curves from the View Graph (\sect{sec:view_graph}).  As a reminder, the View Graph refers to the connected line segments extracted from the 3D model. Most line segments have pointers at each ends indicated which line segments they connect to. Some line segments are connected at singular points (\sect{sec:singular}), such as a T-junctions, where three visible segments and one invisible segment connect. Line segments between junctions must all have the same visibility.  The line segment and singularity data structures  record information about where they came from. For example, a line segment may record that it was extracted from a contour edges, and would include a pointer to the mesh edge that it came from.
Some topological filtering on the View Graph may be necessary in order to extract clean strokes, such as by removing tiny loops. This filtering is described in \sect{sec:simplification}.


One may also use a Planar Map instead of the View Graph \citep{Winkenbach:1994,Winkenbach:1996,Eisemann:2008}.  The Planar Map allows stylization to take into account the relationship between the curves and the shapes of the regions they enclose.

\paragraph{Chaining.}
From the View Graph, we need to extract a set of curves for stylization. These curves are simply chains of line segments, which can be smoothed and rendered.  Normally, we will extract only the visible curves, but invisible curves can also be extracted for hidden-line renderings, in which the invisible curves are rendered in a different style from the visible curves. In the rest of this chapter, we will assume that only the visible curves are being extracted.

The basic approach to extracting these curves is greedy.  We call this process ``chaining.''  We pick some visible line segment at random, and follow pointers from one end of the line segment, following pointers from line segment to line segment, concatenating them into a list of line segments. If this chaining process returns to the original segment, then it is recorded as a closed loop. If the chaining process reaches a singularity, then the behavior depends on the singularity (refer to Figure \ref{fig:singular_points}). At a curtain fold cusp, chaining stops. At a T-junction, if the chain is in the foreground, then the chaining continues through the junction, otherwise it stops (\fig{chaining}(a)). At a Y-junction, the chaining process continues through the chain to connect silhouette edges, and stops for boundary edges (\fig{chaining}(b)).

\begin{figure}[t]
	\centering
	\small
	\def\svgwidth{0.9\linewidth}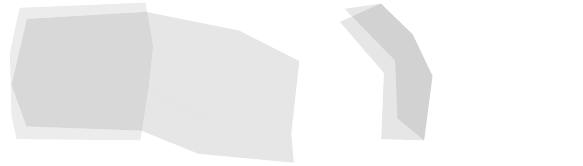\caption{\textbf{Segment chaining at junctions} --- \textbf{(a)} At a T-junction, a foreground chain continues through the junction, whereas a background chains stops. \textbf{(b)} At a Y-junction, a silhouette chain crosses the junction, whereas a boundary chain stops.}\label{fig:chaining}
\end{figure}

The above process is repeated until all visible line segments have been added to a chain. Then, each chain can be rendered separately.

Depending on the rendering style, we may wish to use different chaining rules. For example, a style that renders a single thick silhouette around the entire object would extract a silhouette chain that follows the image-space object boundary. To do this, at each singularity, the chain follows whichever outgoing edge is on the silhouette. Determining whether a mesh edge is a silhouette edge could be determined with a ray test. \citet{Sousa:2003} explored other chaining strategies, and \citet{Grabli:2010} allow the user to programmatically define how edges should be chained together. With this approach, a given edge of the View Graph can even belong to multiple chains, allowing to produce sketchy drawing with overlapping strokes (\fig{stroke_style}).

\paragraph{Smoothing.}
A chain is a polyline, i.e., a set of connected line segments. Converting a chain into a stroke typically entails smoothing polyline in some way. For example, one may merge redundant control points (e.g., control points that are fewer than 2 pixels apart), set a number of control points proportional to the stroke's arc-length, and then perform a least-squares B-spline fit. The spline can then be converted back to a list of finely-spaced control points. The specific amount of smoothing to apply is a stylistic choice.

The topology of the View Graph also must be taken into account when smoothing curves.
For example, a common stylization would be to smooth every stroke independently. However, this can cause a T-junction to break, with the far stroke either penetrating the near stroke, or separating from it; note how the junction lines overlap in Figure \ref{fig:stroke_style}(lower-right).  This disconnection can be prevented either by constraining the smoothing algorithm or postprocessing it.

\section{Stroke rendering} \label{sec:stylization}

\begin{figure}
	\centering
	\small
	{\def\svgwidth{0.5\linewidth}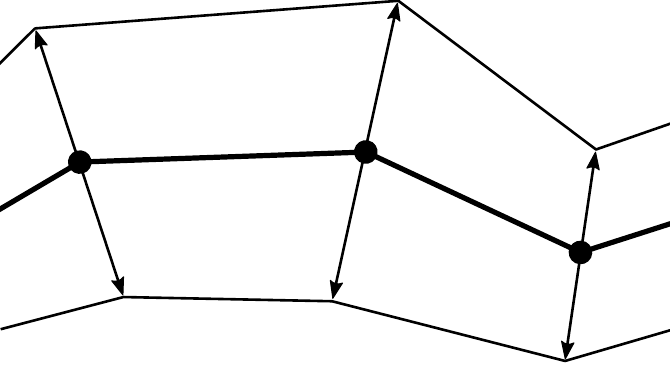}
	\quad
	\includegraphics[width=0.42\linewidth]{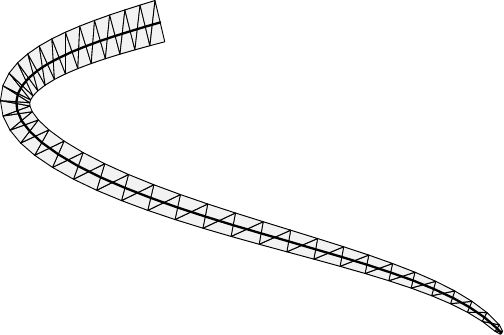}
	\caption{\textbf{Skeletal stroke} --- For each vertex $\vec{v}_i$ of the stroke path, a rib vector $\vec{r}_i$ is generated along the angle bisector to give breadth to the stroke (left). The skeletal stroke can then be rendered as a triangle strip (right).}\label{fig:skeletal_stroke}
\end{figure}

Once smooth strokes have been computed, they can be rendered. Strokes are typically parameterized with a skeletal stroke representation \citep{Hsu:1994}. Skeletal strokes describe a parameterization of the region around the stroke. As illustrated in \fig{skeletal_stroke}, for each 2D vertex $\vec{v}_i$ of the stroke path, a ``rib'' vector $\vec{r}_i$ is constructed orthogonally to the direction defined by its previous $\vec{v}_{i-1}$ and next $\vec{v}_{i+1}$ vertices (if they exist): $$\vec{r}_i = w_i \begin{bmatrix}
		0 & -1 \\
		1 & 0
	\end{bmatrix}\frac{\vec{v}_{i+1} - \vec{v}_{i-1}}{\norm{\vec{v}_{i+1}-\vec{v}_{i-1}}},$$
scaled by the half thickness of the stroke $w_i$. Special treatments are required at places where the radius of curvature of the strokes is smaller than its half thickness, otherwise the ribs will cross each other producing folds, which is especially problematic at sharp corners since the radius of curvature is zero \citep{Asente:2010}. The skeletal strokes can then be rendered as a series of triangular strips~\citep{Northrup:2000}, or as quads with caps~\citep{McGuire:2004b}.

\begin{figure}
	\centering
	\small
	\begin{tabular}{cc}
		\includegraphics[width=0.46\linewidth]{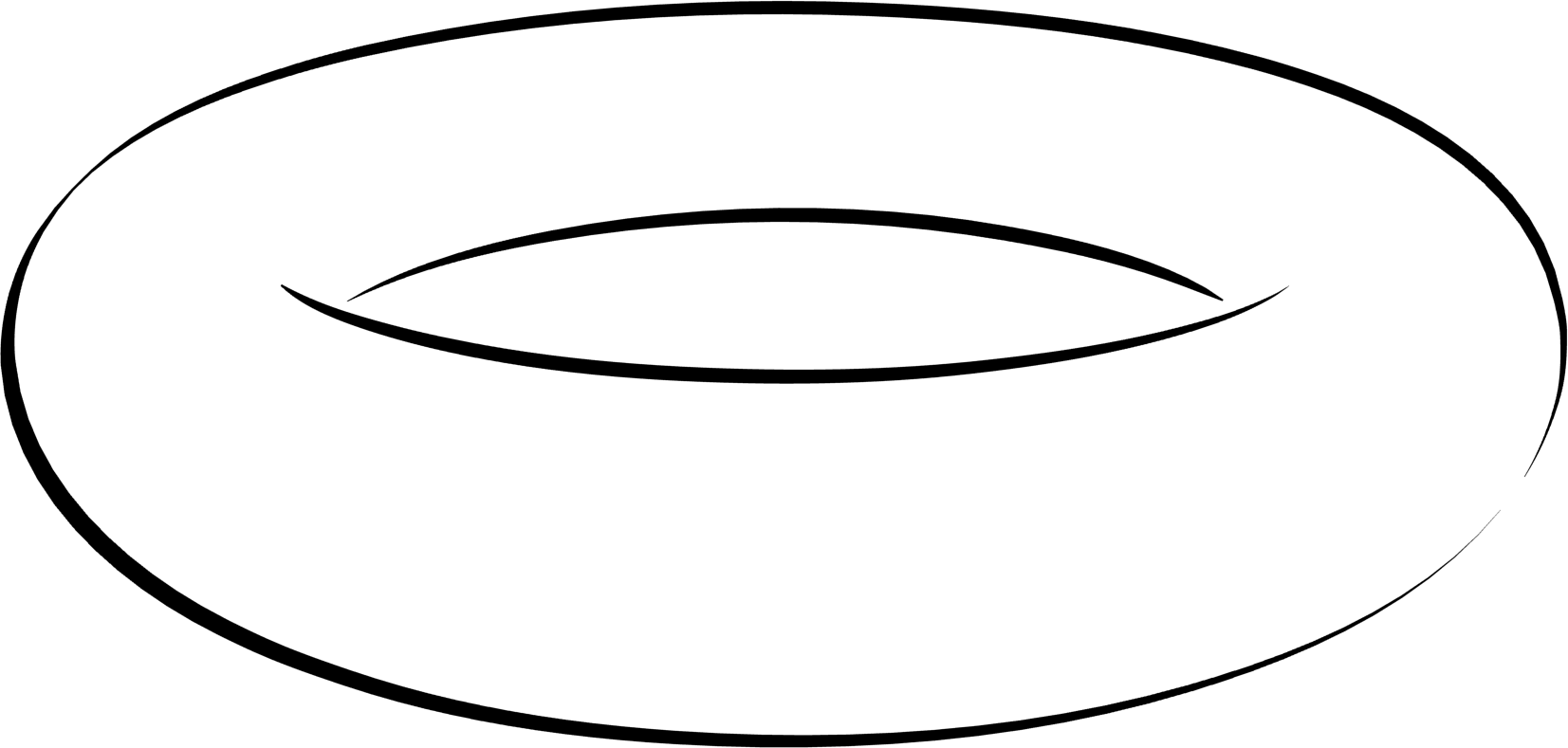} &
		\includegraphics[width=0.47\linewidth]{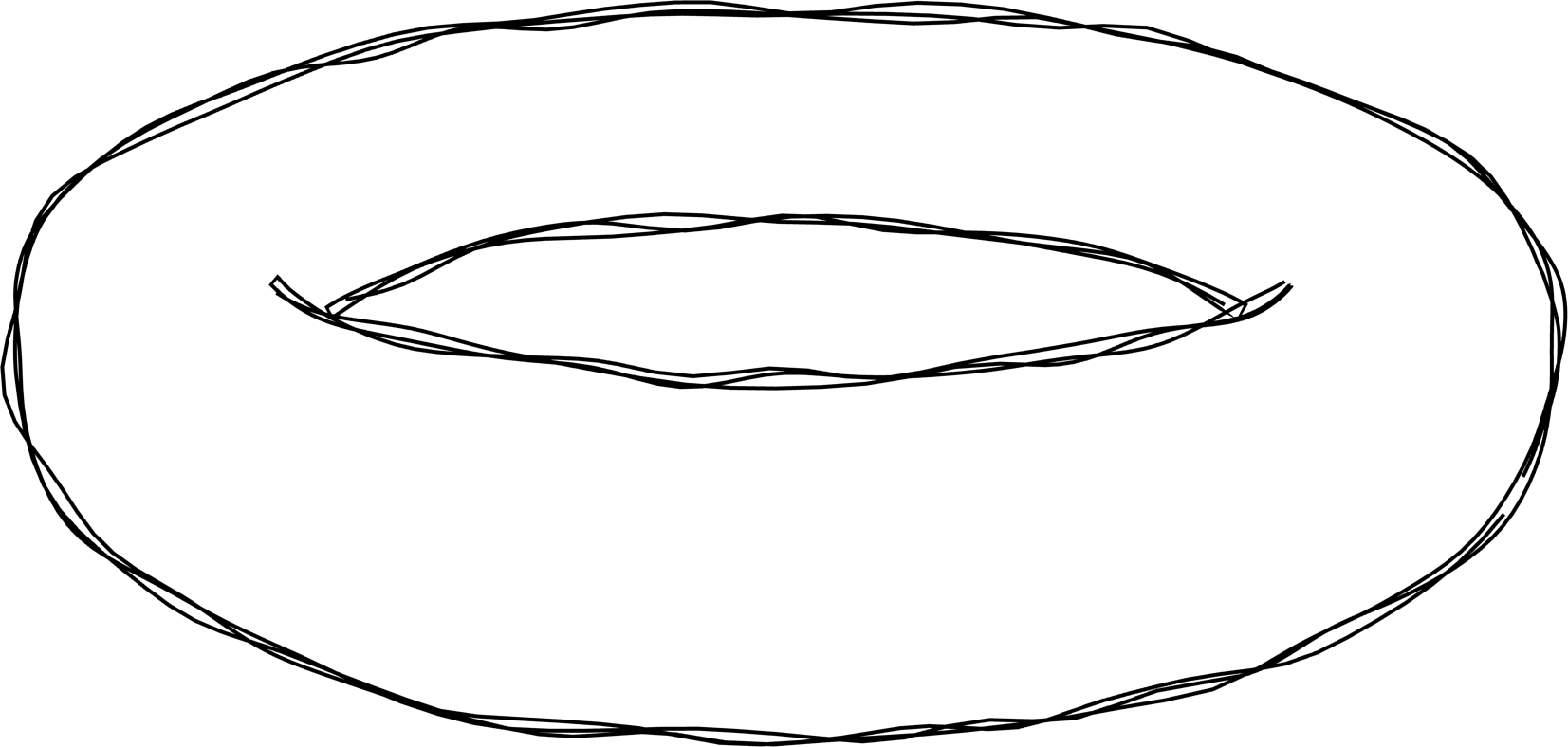}       \\
		\includegraphics[width=0.46\linewidth]{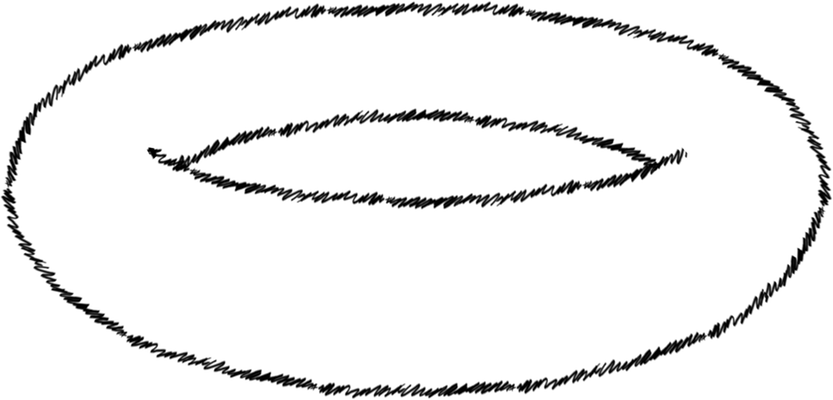}     &
		\includegraphics[width=0.47\linewidth]{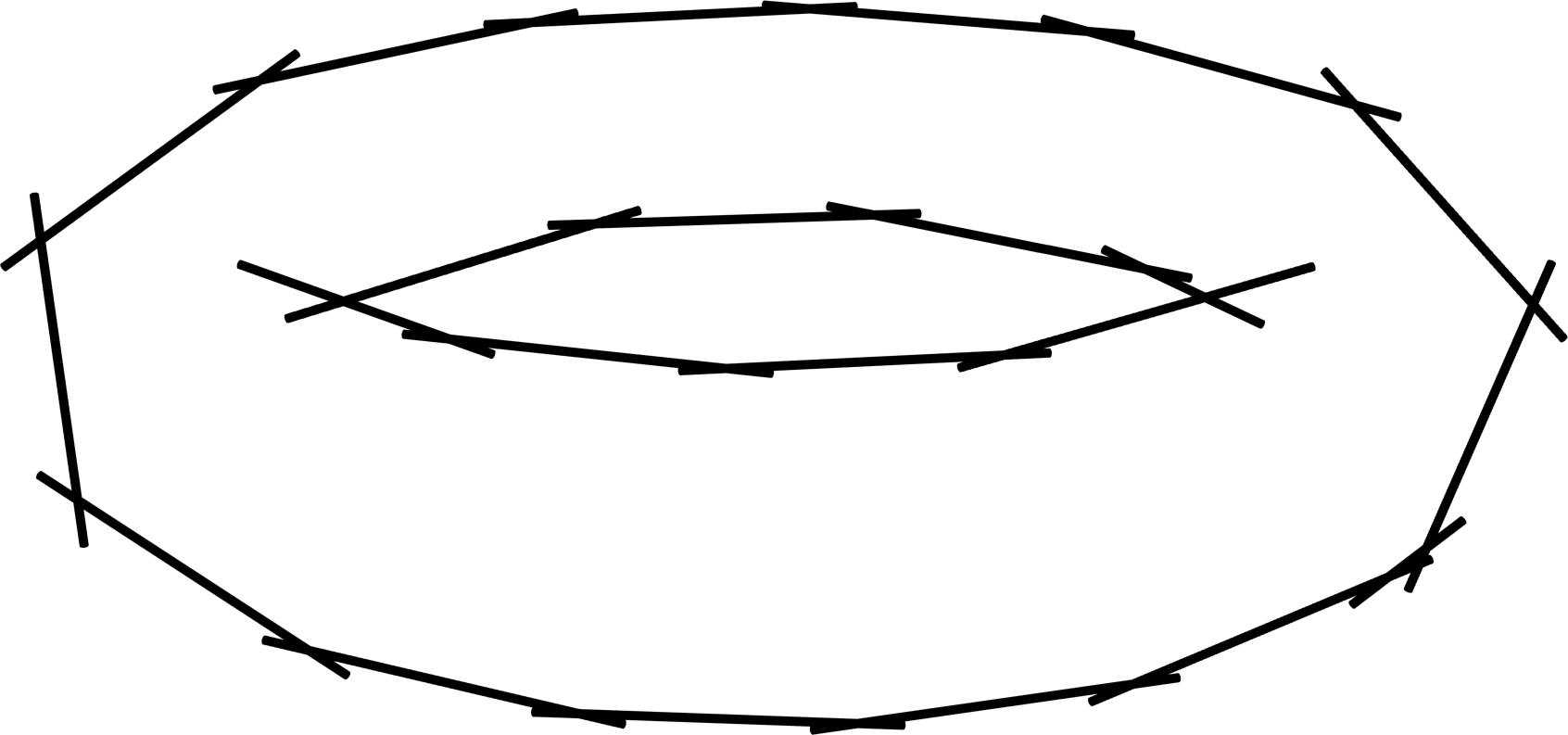}
	\end{tabular}
	\caption{\textbf{Stroke style} --- Various stroke styles applied to the contours of \fig{consistent_contours} using Blender Freestyle. From left to right, top to bottom: calligraphy, sketchy overdraw, textured scribbles and guiding lines.}\label{fig:stroke_style}
\end{figure}

Stroke stylization may be defined by a texture map, such as a scanned pen stroke \citep{Hsu:1994}, or procedurally by making attributes such as color, transparency, dashes \citep{Dooley:1990,Elber:1995b,Grabli:2010,Northrup:2000} vary along the stroke (\fig{stroke_style}). Artistic perturbations of the stroke by analytic (\eg~sine function) or noise functions and offsets \citep{Markosian:1997} can further add high-frequency variations to the stroke. To avoid tiling and stretching artifacts, texture synthesis can be used to generate arbitrary-length stroke textures \citep{Benard:2010} and offsets \citep{Hertzmann:2002,Kalnins:2002,Benard:2012,Lang:2015} from examples.
Stroke styles can depend on the underlying source geometry, e.g., thicker strokes for nearer objects or less-curved objects \citep{Goodwin:2007}, and stroke styles can also vary for different objects and materials. Stroke visibility may be ``haloed'' for greater clarity \citep{Elber:1995b}.
More sophisticated texture-based stroke models include ``RealBrush'', which uses multiple scanned paint strokes to model paint mixing \citep{Lu:2013} and ``DecoBrush'' \citep{Lu:2014}, which uses procedurally-defined line art textures.
These stroke textures may also be drawn directly in a WYSIWYG interface \citep{Kalnins:2002,Cardona:2015,Cardona:2016}. See the video accompanying the work of \citet{Kalnins:2002} for a particularly inspiring interface.

An alternative approach is to reproduce the appearance of natural media, such as ink, paint, watercolor, or charcoal, using using physical models. These methods simulate the pigments deposited by a drawing tool (\eg~pen, pencil, brush) on a substrate (paper or canvas). For example, \citet{Curtis:1997} reproduce watercolor strokes using fluid simulation to compute the motion of water and pigments deposited by a brush on a textured paper. Examples of similar simulations include oil paint \citep{Baxter:2004,Chen:2015:Wetbrush} and graphite pencils \citep{Sousa:1999}. 
The order in which strokes are drawn is usually important and an adequate blending model is thus required. If an accurate painting simulation is available, the Kubelka-Munk \citep{Kubelka:1948,Haase:1992} model of pigment layering can be used. Otherwise, simpler approximations can be used such as the OpenGL blending modes. To emulate thick media such as oil paint, a simple ``replace'' mode is usually sufficient. Substractive blending gives a decent approximation of the behavior of wet materials such as ink and watercolor. Finally the ``minimum'' blending mode can imitate dry media such as graphite and crayons.

\section{Topological simplification} \label{sec:simplification}

There are three reasons to perform a topological simplification step. First, computing smooth contours from meshes  often creates overly complex contours, as discussed in Section \ref{sec:ups_downs}. Second, even the correct contours may exhibit very tiny loops or other topological features that we would like to remove. Third, smoothing and filtering strokes is an important stylization step in creating artistic images. This includes removing strokes in overly dense regions.  These different simplification and stylization steps are each, potentially, operations on the View Graph (or Planar Map), one cannot apply them to individual strokes in isolation.

\paragraph{Clean-up heuristics for mesh contours.} 
When the underlying surface is smooth, each chain should be topologically equivalent to a line, which is essential for artifact-free stylized rendering. Unfortunately, due to numerical instabilities, topological errors might have been introduced during the contour extraction step, especially with polygonal mesh approximations. After projection onto the image plane, this leads to two main artifacts: overlapping contour edges and small ``zig-zags''. For example, when a mesh with low curvature is tangent to the view direction, multiple adjacent triangles may alternate between front- and back-facing orientation, producing a cluster of contour edges (\fig{bif}). 

Fast heuristics may be used to clean up these artifacts.  
For example, \citet{Northrup:2000} describe a method that identifies image-space line segments that overlap and are nearly parallel; they eliminate redundant segments that are similar to another very nearby segment (parallel and close-by), and smaller than the other segment. This allows them to merge many small edges with complex topology into long, simple paths. In addition, \citet{Isenberg:2002} perform a first simplification step in object-space to merge adjacent edges connecting at acute angles, and to remove spurious bifurcations produced by clusters of contour edges.

\citet{Foster:2007} proposed an alternative approach based on multi-resolution filtering through reverse subdivision. Each contour chain is first decomposed into a coarse base path and a representation of its high-frequency details. It is then reconstructed to its original resolution with a scaled-down version of the details to remove errors.

Though these methods do not provide any topological guarantees, they can be fast and simple and produce appealing results.

\paragraph{Removing tiny details.}
Even if the input contour generators have correct topology, they may exhibit unappealing topological details such as tiny loops or breaks due to cusps (\fig{redShoulder}a). To improve the appearance of the final rendered strokes, topological simplification can be applied to the View Graph chains (\fig{redShoulder}b). Unlike the previous heuristics, this simplification is purely a stylistic control; one that depends on the scale of the objects as well as the rendering style, not an attempt to estimate and fix topology from noisy curves.

In particular, we proposed the following topological simplifications (Figure \ref{fig:topsimp}) \citep{Benard:2014}. First, we categorized View Graph vertices (i.e., singularities) by the number of visible curves they connect: \emph{a dead-end} vertex is adjacent to a single visible curve (e.g., a visible curtain fold); a \emph{connector} vertex is adjacent to two visible curves, and a \emph{junction} vertex is adjacent to more than two vertices, i.e., bifurcations and image-space intersection vertices. Then, we defined a candidate chain for simplification as any connected sequence of visible curves that do not contain any junctions, but with image-space arc length less than a user-specified threshold (between 10 and 20 pixels in our experiments). Eventually the algorithm marks as invisible any candidate chain that (a) connects a junction to a dead-end, (b) connects a dead-end to a dead-end, (c) connects a vertex to itself, or (d) is overlapped in 2D by another chain (\fig{topsimp}). This process is iterated until there are no more changes to be made.

\begin{figure}
	\centering
	\small
	\def\svgwidth{0.6\linewidth}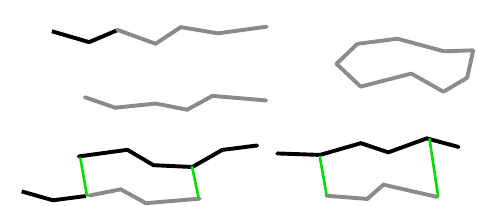\caption{\textbf{Topological simplification} \citep{Benard:2014} ---
		Four cases are considered for simplification (candidate chain depicted in grey):
		\textbf{(a)} junction to dead-end connection;
		\textbf{(b)} dead-end to dead-end connection;
		\textbf{(c)} small closed loop;
		\textbf{(d)} small overlapping pieces of curve between two junctions.
	}\label{fig:topsimp}
\end{figure}

\begin{figure}
	\centering
	\small
	\begin{tabular}{c@{\hspace{-1.2em}}c@{\hspace{-2.5em}}cc@{\hspace{-1.2em}}c@{\hspace{-2.5em}}c}
		\textbf{(a)}                                                                           &
		\includegraphics[width=0.25\linewidth]{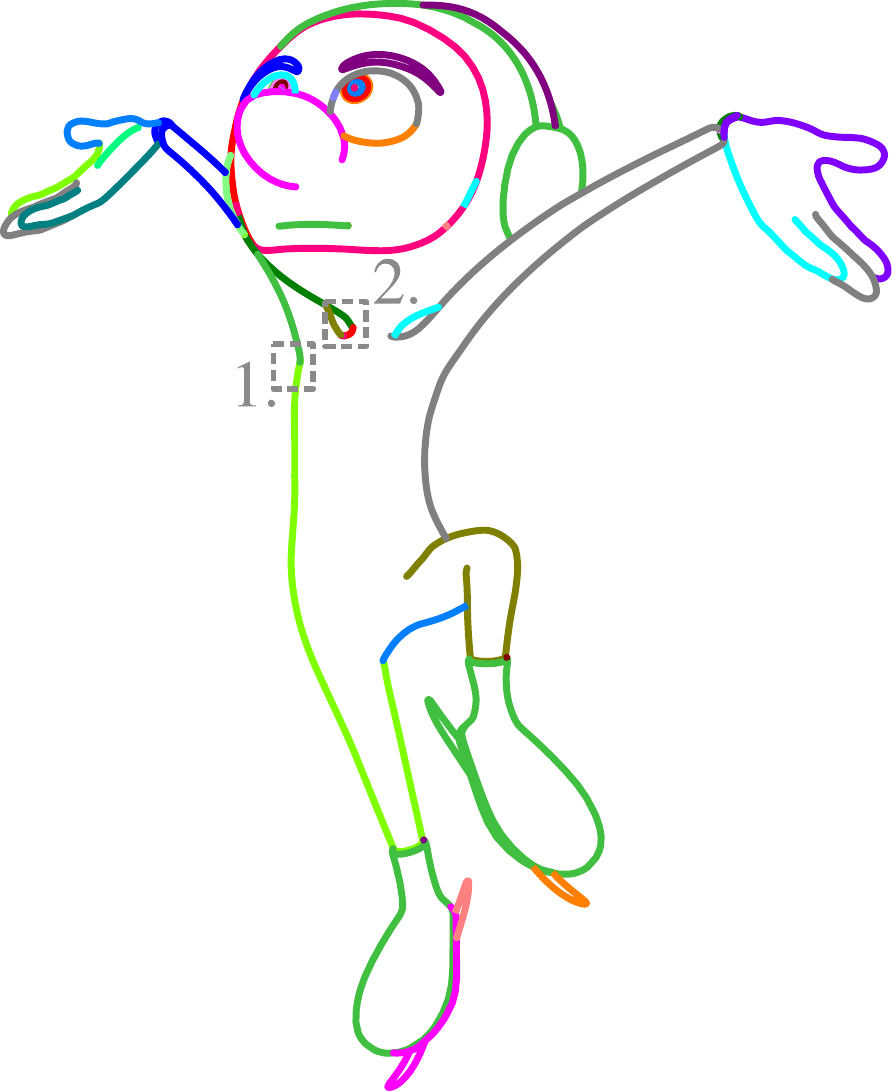}      &
		\includegraphics[width=0.28\linewidth]{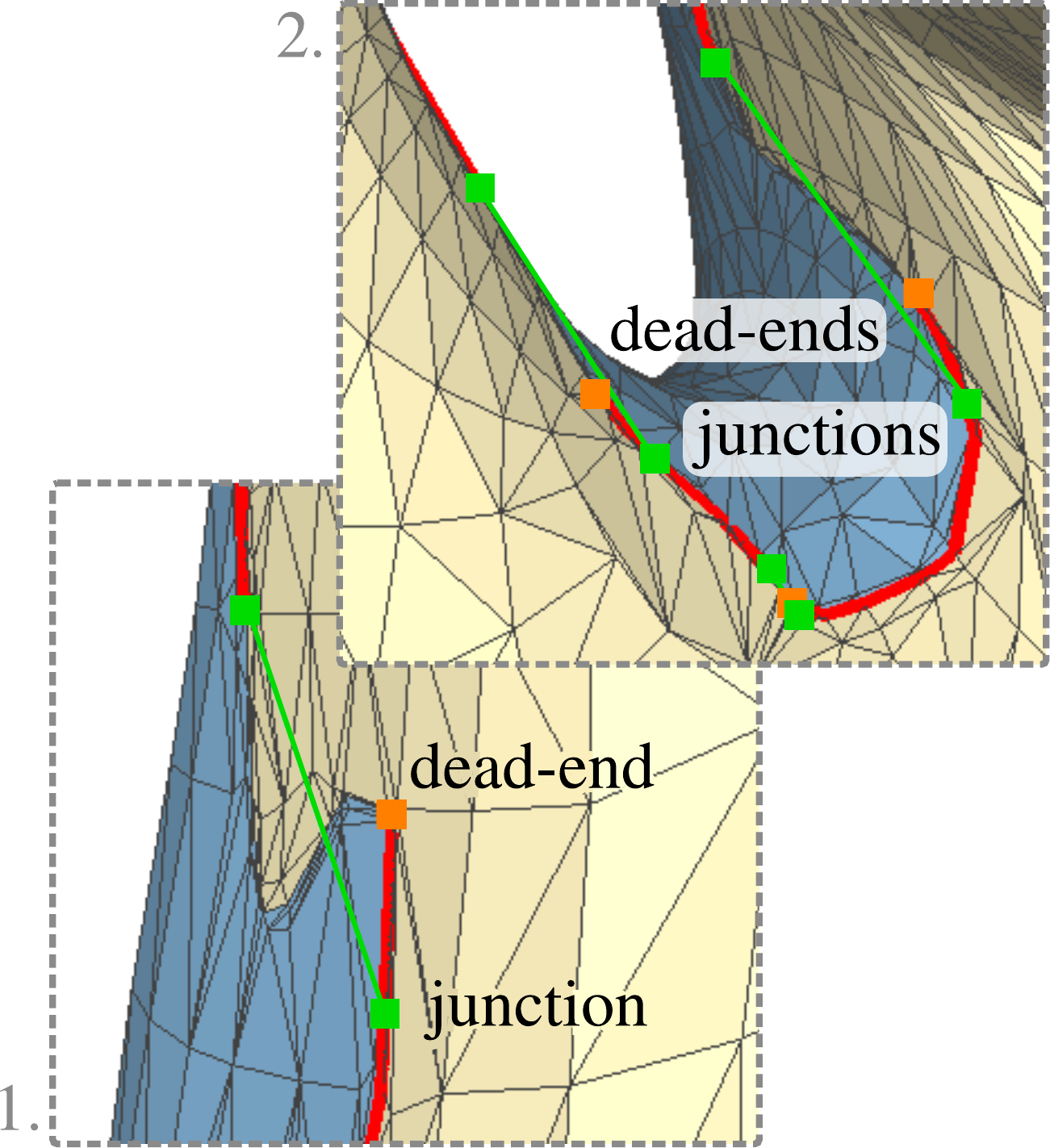}            &
		\textbf{(b)}                                                                           &
		\includegraphics[width=0.25\linewidth]{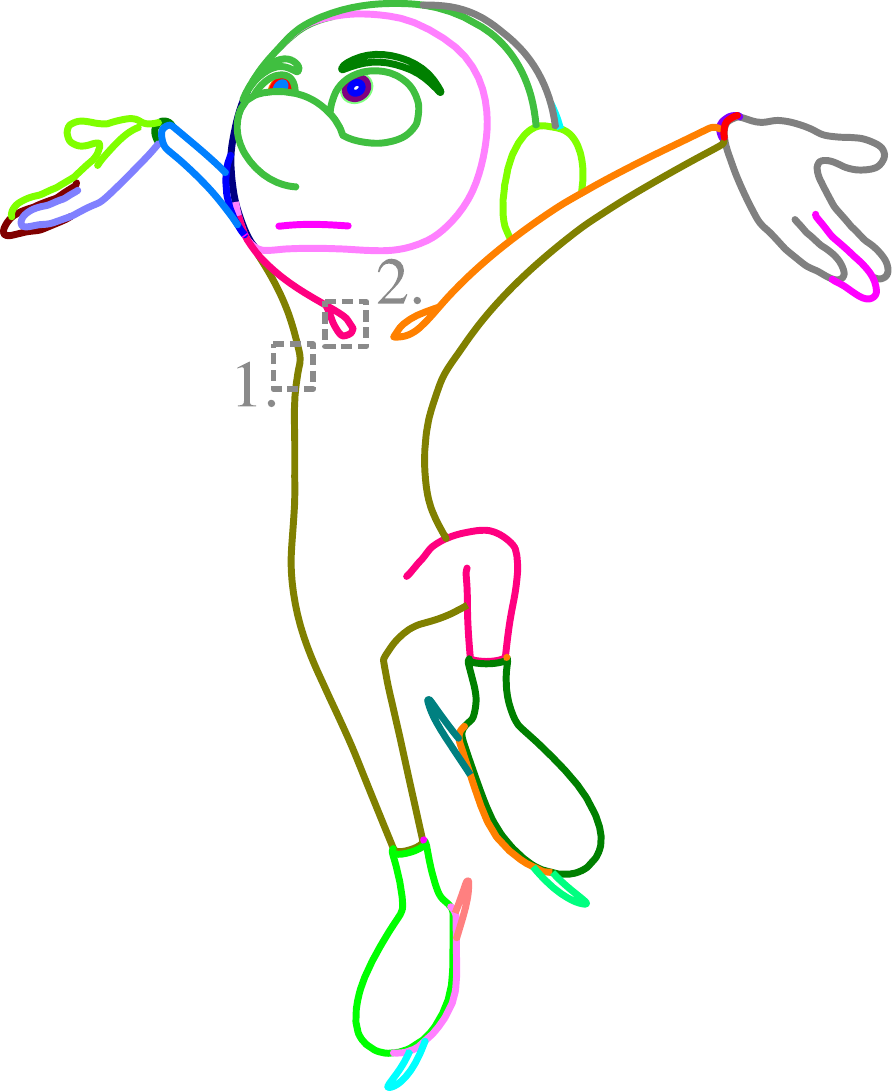} &
		\includegraphics[width=0.28\linewidth]{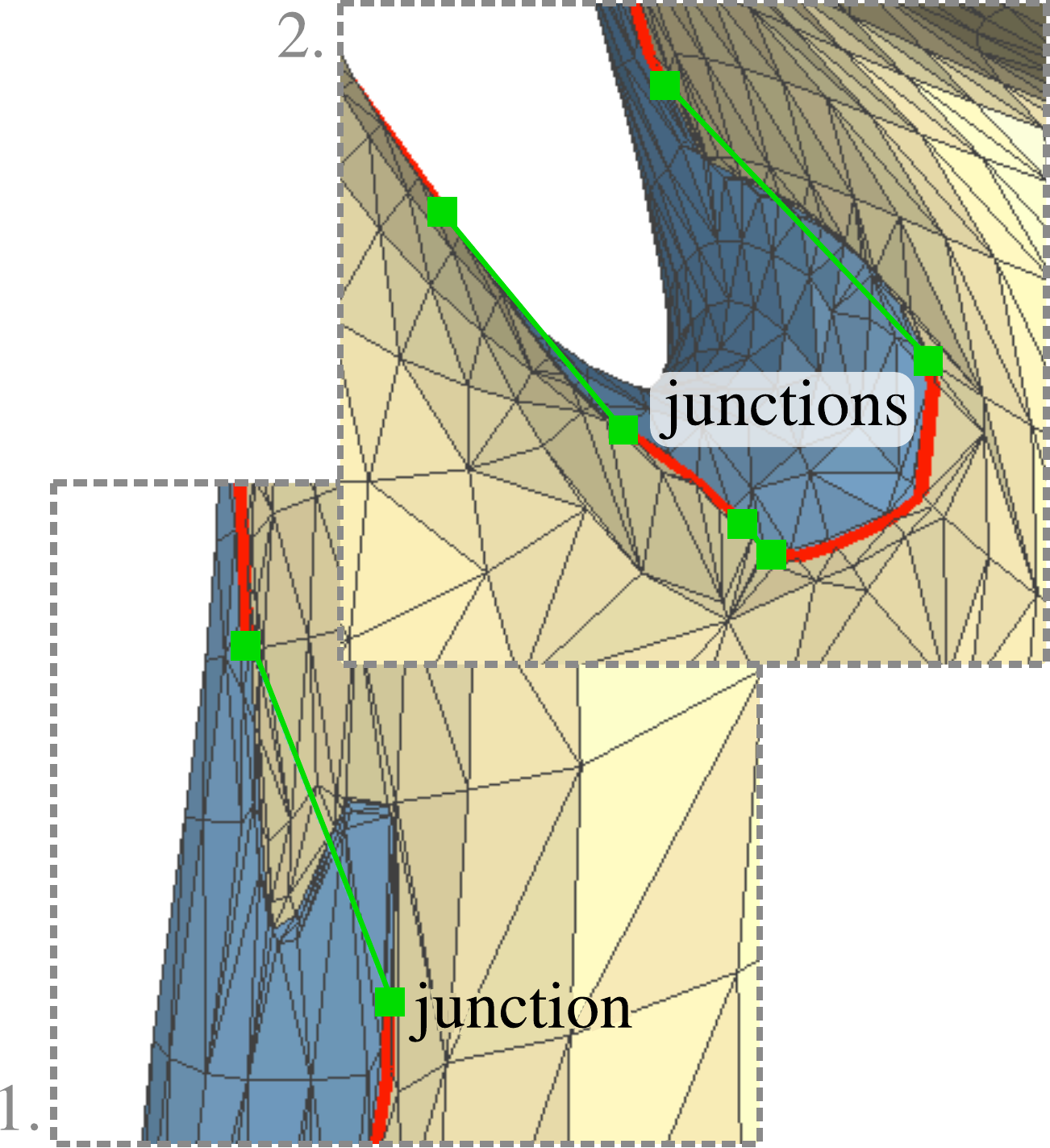}
	\end{tabular}
	\caption{\textbf{Result of topological simplification on Red} \citep{Benard:2014} ---
		Closeup on Red's shoulder and armpit, with genuine cusps, \textbf{(a)} before and \textbf{(b)} after	topological simplification. ``Red'' \ccCopy~Disney/Pixar}
	\label{fig:redShoulder}
\end{figure}

\paragraph{Controlling image-based density.} 
An artist drawing a small or distant object would draw only a few of its curves. However, directly computing all curves on a detailed surface from far away produces an overly-dense set of curves. \citet{Wilson:2004} propose a first method to omit excess curves from a drawing, based on the density of strokes in image space. \citet{Grabli:2004} further distinguish between two kinds of line density measures: an \emph{a priori} density and a \emph{casual} density. The former estimates, at a given scale and for a given direction, the geometric complexity of the line drawing that would be created if all visible lines were drawn without stylization. The latter measures the actual visual complexity of the stylized drawing while it is rendered one curve at a time. 

As noted by \citet{Winkenbach:1994,Preim::1995}, drawing simplification by line omission also requires ordering the curves by \emph{relevance}, the least relevant ones being omitted first. Various definition of curve relevance have been proposed. For instance, to preserve curves that separate objects that are far apart in depth, the curve relevance can be defined as the maximum depth difference measured along its segments.

\paragraph{Scale-dependent contours.} 
Another approach to curve simplification is to use a range of representations of the original surface, using a coarse version for rendering at a distance, and a detailed version for close-up viewing \citep{Deussen:2000,Jeong:2005,Kirsanov:2003,Ni:2006}.  Image-space density can be used as a criteria for selecting the object level-of-details. An advantage of this approach is that it can ensure a coherent topology of the drawing; however, it may not adapt to naturally varying density as well object-based stylization, for example, for highly-foreshortened objects.

\section{Object shading and texturing}
\label{sec:shading}

In addition to line drawing, we usually wish to shade or texture the object.  A thorough discussion of shading is beyond the scope of this tutorial. Generally, most methods compute shading independently from contours and other lines. In principle, doing so could create inconsistency between the line drawing and the shaded rendering, e.g., if the lines are smoothed or simplified. While this is not usually a problem, using a Planar Map to maintain a consistent representation can fix this issue \citep{Eisemann:2008,Winkenbach:1994}.

Once a shaded image has been rendered from the 3D scene, any of the methods in the book of \citet{Rosin:2013} may be applied as post-processes to stylize images as well. Nevertheless, more dedicated shading primitives will allow to produce more stylized results.

\paragraph{Toon shading.} A simple and popular shading algorithm is ``toon shading''. In toon shading, the shading is simply a function of the view vector and light direction: $\vec{n} \cdot \vec{l}$. In its simplest form, the user specifies two colors: a light color and a dark color. For shading, points where the dot product is below a threshold get rendered with the dark color; other points get rendered with the light color (\fig{contour_basic}).  This generalizes contour rendering, which discretizes $\vec{n} \cdot \vec{v}$ into black at the contour and white everywhere else.

Several generalizations of toon shading have been developed, typically mapping $\vec{n}$ to color with a more complex function, \eg~\citep{Lake:2000,Sloan:2001,Gooch:1998,Barla:2006,Mitchell:2007:IRT:1274871.1274883,Eisemann:2008,Vanderhaeghe:2011}.

\paragraph{Hatching and texturing.} Surface hatching (\fig{hatching}) often involves drawing hatching curves \citep{,Winkenbach:1994,Winkenbach:1996,Hertzmann:2000,Singh:2010,Kalogerakis:2012,Gerl:2013}.
These are curves on the surface, and their visibility is computed together with the other curves, using the algorithms described in this tutorial.

A wide variety of algorithms for texturing have also been developed without drawing texture curves, instead using texture maps in some way, \eg~\citep{Klein:2000,Praun:2001,Webb:2002,Breslav:2007}.

Two methods for stylizing objects and animations by example using the Image Analogies framework of \citet{Hertzmann:2001} have been developed \citep{Benard:2013,Fiser:2016}.

\section{Animation}
\label{sec:animation}

To produce an animated line drawing, one can simply extract and stylize the contours at every frame independently. However, as first noted by \citet{Masuch:1997}, the coherence between frames largely depends on the chosen stroke style. A ``calm'' style with only small deviations from the base path leads to a rather smooth animation, whereas a ``wild'' style with strong geometric distortions or using strong textures may lead to visual artifacts such as popping or sliding. This is a recurrent but still mostly open problem in non-photorealistic rendering \citep{Benard:2011}. In the following, we will summarize the main solutions to improve temporal coherence of line drawing animations.

\paragraph{2D curve tracking.} The general objective is to establish correspondences between stokes of subsequent frames and derive from them a coherent space-time parameterization. View-independent lines, such as creases or ridges and valleys, that are fixed on the 3D surface can leverage the underlying surface parameterization to ensure such correspondences.
In contrast, view-dependent lines such as contours move on the surface and their geometry and even topology change from frame to frame. Consequently most approaches directly compute correspondences in image-space. This is related to computer-assisted rotoscoping which aims to track edges in videos \citep{Agarwala:2004,ODonovan:2011} with the benefit that perfect motion information can be computed from 3D animations. Disney ``Paperman''~\citep{Whited:2012} precisely follows such an approach to propagate hand-drawn strokes over CG renders. The main drawback of this system is that artists need to manually merge or split strokes when their topology should change to adapt to the animation. 

\begin{figure}

\includegraphics[width=\linewidth]{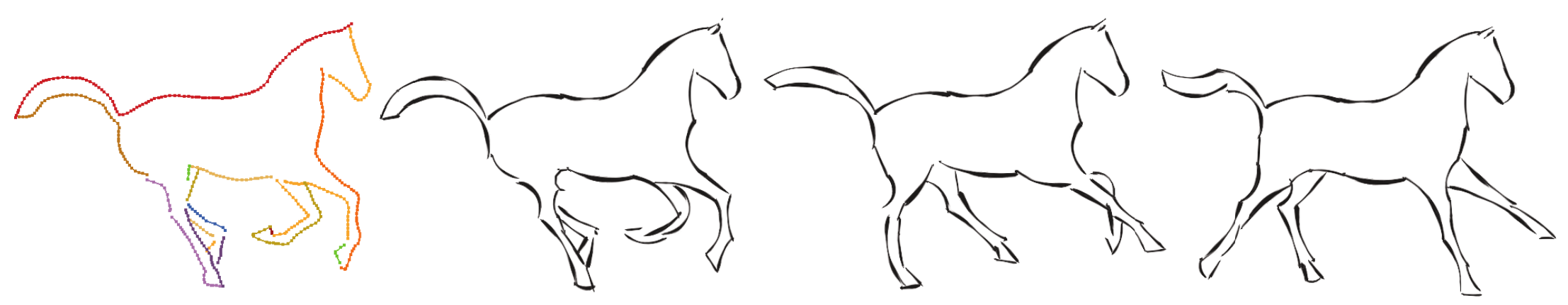}
\caption{\textbf{2D curve tracking by active contours} \citep{Benard:2012} --- The contour curves extracted from a galloping horse are tracked by active contours (left) ensuring a coherent parameterization of the strokes (right) during the animation.} \label{fig:horses}
\end{figure}

The curve tracking method of \citet{Benard:2012} lifts this limitation. It uses active contours (a.k.a. snakes) that automatically update their position, arrangement and topology to match the contour animation (\fig{horses}). However, based on heuristics, this method cannot guarantee that the strokes are faithfully depicting the contours, especially at junctions.

\begin{figure}
	\centering
	\small
	\def\svgwidth{0.6\linewidth}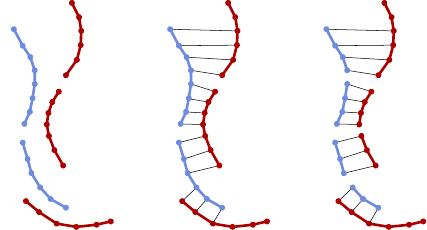\caption{\textbf{2D curve matching} \citep{Ben-Zvi:2015} ---	
	Starting from the strokes at frame $f_i$ in blue and $f_{i+1}$ in red \textbf{(a)},
	point correspondences are computed \textbf{(b)} by a constrained optimization,
	from which consistent sub-strokes can be extracted \textbf{(c)} .
}\label{fig:matching}
\end{figure}

For more robustness, \citet{Ben-Zvi:2015} turn the problem of curve tracking into one of matching. Since matching full curves would be impractical due to topological changes, they seek to find a mapping between the vertices of the strokes in frame $f_i$ and $f_{i+1}$ (\fig{matching}). This mapping aims at minimizing the distance between matched vertices, after moving the points in $f_i$ to $f_{i+1}$ according to the animation motion field. The mapping should also maximize the number of matched vertices, favor one-to-one matches, and maintain the spatial ordering of the vertices on the strokes as much as possible. These objectives can be expressed as a constrained optimization problem on a bipartite graph whose nodes on each side of the graph are the vertices in each frame and whose edges connect any pair of vertices from different frames. Coherent sub-strokes are eventually constructed from these point-wise correspondences by splitting inconsistent curves (\fig{matching}c). A major drawback of this method is the computation cost of the point-matching algorithm which does not scale well with the number of curves.

\begin{figure}
	
	\begin{subfigure}[t]{0.31\linewidth}
		\includegraphics[width=\linewidth]{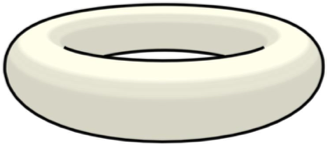}
		\caption{contours} \label{fig:css_0}
	\end{subfigure}
	\hspace{0.5em}
	\begin{subfigure}[t]{0.31\linewidth}
		\includegraphics[width=\linewidth]{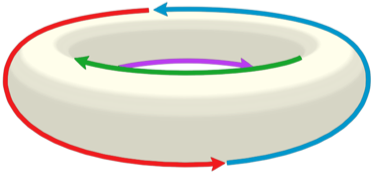}
		\caption{strokes in frame $f_{i}$} \label{fig:css_1}
	\end{subfigure}
	\hspace{0.5em}
	\begin{subfigure}[t]{0.31\linewidth}
		\includegraphics[width=\linewidth]{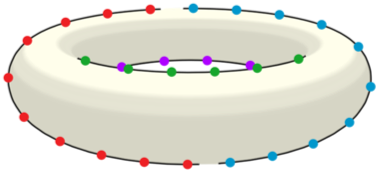}
		\caption{samples on \textbf{(b)}} \label{fig:css_2}
	\end{subfigure}
	\hspace{0.5em}
	\begin{subfigure}[t]{0.31\linewidth}
		\includegraphics[width=\linewidth]{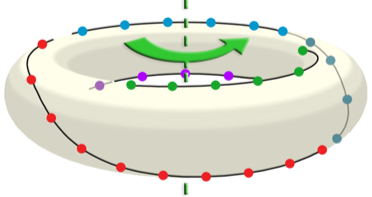}
		\caption{new view in frame $f_{i+1}$} \label{fig:css_3}
	\end{subfigure}
	\hspace{0.5em}
	\begin{subfigure}[t]{0.31\linewidth}
		\includegraphics[width=\linewidth]{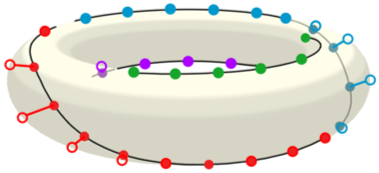}
		\caption{propagation} \label{fig:css_4}
	\end{subfigure}
	\hspace{0.5em}
	\begin{subfigure}[t]{0.31\linewidth}
		\includegraphics[width=\linewidth]{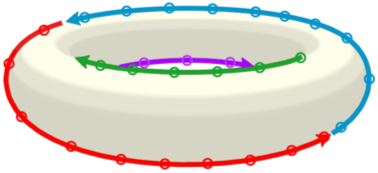}
		\caption{new strokes in $f_{i+1}$} \label{fig:css_5}
	\end{subfigure}
	\caption{\textbf{Parameterization propagation} \citep{Kalnins:2003} --- The strokes at frame $f_{i}$ \textbf{(b)} are sampled uniformly along their arc-length \textbf{(c)} and propagated by reprojection in the new camera \textbf{(d)} and local search in image-space \textbf{(e)}. Coherently parameterized strokes are created by potentially splitting new contour curves and fitting a continuous function to the samples balancing uniformity in 2D with coherence in 3D \textbf{(f)}.} \label{fig:css}
\end{figure}

\paragraph{Parameterization propagation.}  Another solution to these topological issues is to propagate the parameterization of the strokes instead of their geometry. Building upon the work of \citet{Bourdev:1998}, this is the key idea of \citet{Kalnins:2003}. They sample the parametrization of the strokes at frame $f_i$ uniformly along their arc-length (\fig{css_2}), and reproject those samples in the camera of the next frame $f_{i+1}$ following the 3D animation to approximate the contour motion (\fig{css_3}). Then, they locally search in image-space the location of the closest contour paths in the new view (\fig{css_4}). Samples from different brushes of frame $f_i$ may end up on the same contour path in frame $f_{i+1}$. In such a case, the contour path needs to be split into multiple strokes with consistent parameterization samples. Eventually, each stroke parameterization is computed by optimizing an energy function that balances the competing goals of uniform image-space arc-length parameterization and coherence on the object surface. This method ensures a temporally coherent parameterization for contours with simple topology at interactive framerates. However, since more complex objects, such as the Stanford Bunny, generate contours made of many tiny fragments, topological simplification (\sect{sec:simplification}) needs to be applied first, otherwise the strokes may get increasingly fragmented over time.

\begin{figure}
	\centering
	\small
	\def\svgwidth{0.75\linewidth}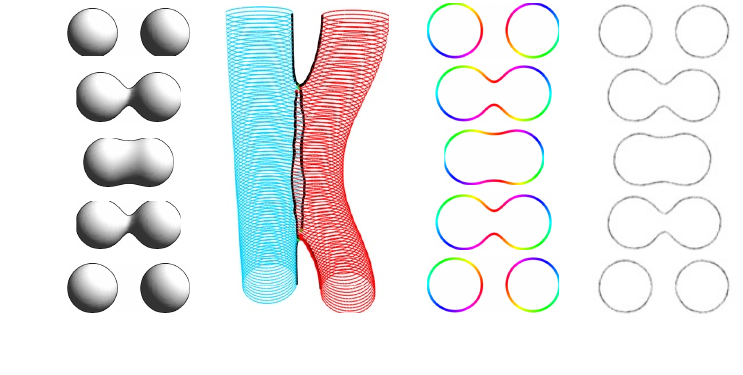\caption{\textbf{Space-time parameterization} \citep{Buchholz:2011} ---	
	Temporally coherent textured strokes are computed by parameterizing the space-time surface swept by the
	contours over time.
	}\label{fig:tcoco}
\end{figure}

\paragraph{Space-time parameterization.} The above methods consider only two animation frames at a time. \citet{Buchholz:2011} proposed considering the entire animation as a whole, building a space-time parameterization of the strokes over time (\fig{tcoco}). To do so, they build the space-time surface swept by the contours over time, taking into account merging and splitting events. This allows minimizing distortions and ``popping'' artifacts over the whole sequence, rather than greedily processing each frame one-at-a-time. This comes however at the price of expensive computations (up to several minutes for few seconds of animation). 

\paragraph{Motion of contours across time.} 
Some theoretical analysis of how contours evolve over time is provided by \citet{Cipolla:2000,Plantinga:2006}. \citet{Plantinga:2006} additionally provide a fast algorithm for tracking contours of implicit surfaces over time, with guaranteed topological correctness, under orthographic projection.


\chapter{Conclusion}
\label{chap:conclusion}

The techniques we have described here provide an account for how to make line drawings from 3D models.  In organizing the information from the past few decades of research on this topic, we hope that this material will be useful for future practitioners and researchers.

The methods described here present a range of design choices.  On one extreme, hardware rendering methods allow real-time performance and relatively simple implementation, but only a narrow range of rendering styles.  Interpolated contours are much more flexible and allow for many more rendering styles, though with some implementation effort, and some potential artifacts appearing on smooth surfaces. In the other extreme, methods that attempt to correctly compute all curves for smooth surfaces are currently the most complex; providing full correctness guarantees remains a research problem. Within this range of options, there are more design choices to be made, such as which visibility tests to use besides ray tests, which strategies to use to propagate visibility, what numerical robustness strategies to try.  These choices make tradeoffs between stylistic control, accuracy, efficiency, and complexity of implementation.  Our community does not yet have the experience of building real systems that would allow us to make recommendations about many of these choices.  However, at a high level, knowing one's requirements for stylistic variation, real-time performance, and accuracy can guide one to one of the three main approaches listed above.

\section{Open research problems}

The work presented here has been developed for applications in entertainment, art, and scientific visualization. There are several games that have used hardware line-drawing methods, and 3D methods have shown up in a few films here and there, such as in Disney's Paperman, and the Freestyle NPR line drawing algorithms have been incorporated in the free Blender package.  Still, many of the most sophisticated methods here have not made it into commercial use.

These research areas were very active in the 90's and 2000's, and now there is little effort. This is particularly evident at SIGGRAPH, the flagship venue for this research; now most research, when it appears, is at more specialized, less-impactful venues.  There are many possible reasons for this stagnation.  We believe that it does not reflect a lack of interest in these problems, but, rather, the difficulty for many researchers outside the area in identifying known research projects.   It should be obvious from this tutorial that there are some clear, open research problems that are purely geometric in nature.

For professional stylization applications (e.g., films), one would ideally like to have a space-time planar map of a scene, including space-time correspondences and complete curve topology, and no existing system for this has even been attempted, to our knowledge. The engineering effort involved may be well beyond what is achievable by a single graduate student working alone, without truly heroic effort and understanding of the geometric, numerical, and artistic considerations involved.

Once developed, such a system would create a variety of new engineering, authoring, and workflow challenges and opportunities.

Since these problems first arose, there has been tremendous progress in other areas of geometric modeling and simulation. The fundamental difficulty for correct curve topology is that a single incorrect visibility test at a seemingly-insignificant little triangle can cause enormous visibility errors. This parallels problems that can occur in other areas of modeling and simulation. For example, when simulating a hanging piece of cloth, a single little missed collision can lead to massive interpenetrations that ruin the simulation: a single non-robust test causes topological catastrophe. The geometric modeling and simulation communities have developed robust geometric tools and techniques to prevent these problems. It may now be time to revisit the line drawing problem with this new knowledge in hand.

A separate problem is to build artist-friendly tools for authoring artistic styles. There have been various approaches explored in the literature, such as rotoscoping \citep{Kalnins:2002,Sabiston:2001,Whited:2012,Cardona:2015} and procedural authoring \citep{Grabli:2010}.  

Machine learning and example-based rendering could provide a way to author styles. So far, there has been very little effort combining machine learning and 3D NPR; exceptions include \citep{Kalogerakis:2012,Benard:2013,Fiser:2016}.

Whatever the future work, the methods described in this tutorial provide the first part of the story, but the rest is yet to be written.

\section{Acknowledgements}

We thank Doug DeCarlo for comments on a draft of this tutorial, and a reviewer for extremely detailed comments.

\appendix


\chapter{Fundamentals of Differential Geometry}
\label{app:diff_geom}

This chapter presents the fundamentals of differential geometry that are useful to define smooth surface contours. It is based on the books and courses of \citet{Cipolla:2000,Rusinkiewicz:2008,Crane:2013}.

\section{Geometry of surfaces}
The geometry of a 3D surface can be described using a map $f : M \subset \mathbb{R}^2 \rightarrow \mathbb{R}^3$ from a region $M$ in the Euclidean plane $\mathbb{R}^2$ to a subset $f(M)$ of $\mathbb{R}^3$. It is called a parametric surface if $f$ is an \emph{immersion}, that is its partial derivatives $\frac{\partial f}{\partial u}$ and $\frac{\partial f}{\partial v}$ are injective at each point of $M$. In this case, $f$ defines an \emph{immersed surface}; a surface where, for every point $\vec{u}$ of $M$, a definite 2-dimensional tangent plane is associated at $\vec{p} = f(\vec{u})$, or, simply, $\vec{p}(\vec{u})$ (\fig{param_surface}).

\begin{figure}
	\centering
	\small
	\def\svgwidth{\textwidth}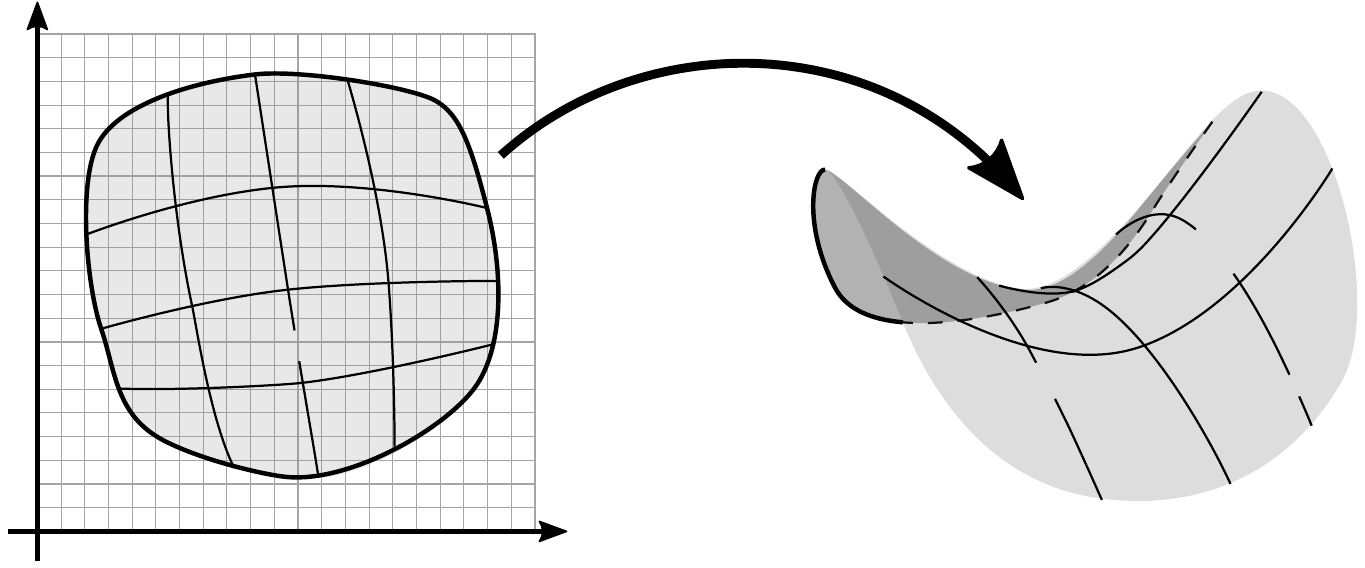\caption{\textbf{Parametric surface} --- The immersive map $f : M \subset \mathbb{R}^2 \rightarrow \mathbb{R}^3$ associates to each parametric location $\vec{u}$ a unique 3D point $\vec{p}$; the normal $\vec{n}$ at that point is defined as the cross-product of the partial derivatives of the parameterization.}\label{fig:param_surface}
\end{figure}

A vector $\vec{n} \in \mathbb{R}^3$ is \emph{normal} to the tangent plane at $\vec{p}$ if, for all tangent vectors $\vec{t}$ at $\vec{p}$, $\vec{n} \cdot \vec{t} = 0$, where $\cdot$ is the canonical Euclidean dot product on $\mathbb{R}^3$. Since the tangent plane at $\vec{p}$ can be defined as a linear combination of the partial derivatives of the immersion $f$, the \emph{unit surface normal} $\vec{n}$ at $\vec{p}$ is defined as:
$$\vec{n}(\vec{u}) = \frac{\left.\frac{\partial f}{\partial u}\right|_{\vec{x}=\vec{u}} \times \left.\frac{\partial f}{\partial v}\right|_{\vec{x}=\vec{u}}}{\left\lVert \left.\frac{\partial f}{\partial u}\right|_{\vec{x}=\vec{u}} \times \left.\frac{\partial f}{\partial v}\right|_{\vec{x}=\vec{u}} \right\rVert}.$$
Note that interchanging the two partial derivatives takes $\vec{n}$ to $-\vec{n}$. If we can choose a consistent direction for $\vec{n}$ for all points $\vec{p}$, $M$ is \emph{orientable}. For orientable surfaces, $\vec{n}$ can be seen as a continuous map, called the \emph{Gauss map}, which associates each point of $M$ with its unit normal, viewed as a point on the unit sphere $S^2$.

\section{Curvature} 
\label{app:curvature}

Informally, the curvature describes how much a surface bends at a certain point and in a particular direction. For instance, an (infinite) cylinder curves around in a circle along one direction, and is completely flat along another direction. It is common to treat surface curvature in terms of 3D curves contained in the surface. We thus first need to define the curvature of a 3D curve.

\paragraph{Curvature of a curve.} Let $c : I \subset \mathbb{R} \rightarrow \mathbb{R}^3$ be a 3-dimensional parametric curve with \emph{unit speed}, \ie with \emph{arc-length} or \emph{natural} parameterization: $\left\lVert \frac{dc}{dt} \right\rVert = 1$. The curvature of $c$ is measured by the rate at which the \emph{unit} tangent vector changes as we move along $c$. This change is split into two pieces: the unit vector $\vec{n}$, called the \emph{principal normal}, which describes the direction of change, and a scalar $\kappa \in \mathbb{R}$, called the \emph{curvature}, which expresses the magnitude of change:
$$\frac{d^2 c}{dt^2} = \vec{t}' = -\kappa \vec{n}.$$
Assuming that $\kappa$ is never zero, the plane spanned by $\vec{t}$ and $\vec{n}$ is the \emph{osculating plane}. The vector $\vec{b}$ orthogonal to the osculating plane is called the \emph{binormal} of the curve: $\vec{b} = \vec{t} \times \vec{n}$. The orthonormal coordinate frame made of $\vec{t},\vec{n},\vec{b}$ is called the \emph{Frenet frame} (\fig{curve_curvature}).

\begin{figure}
	\centering
	\small
	\def\svgwidth{0.4\textwidth}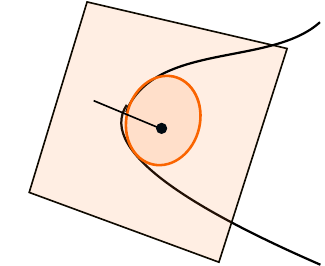\caption{\textbf{Curvature of a curve} --- The curvature $\kappa$ of the 3D parametric curve $c$ at $\vec{p}=c(t)$ corresponds to the inverse of the radius of the osculating circle $S$ in the osculating plane (in gray) defined by the tangent $\vec{t}$ and principal normal $\vec{n}$.}\label{fig:curve_curvature}
\end{figure}

For any point $\vec{p} = c(t)$, the circle $S$ in the osculating plane with centre $\vec{p} + \vec{n}/\kappa$ is called the \emph{osculating circle} of $c$ at $\vec{p}$. This circle best approximates $c$ at $\vec{p}$, meaning that it has the same tangent direction $\vec{t}$ and curvature vector $\kappa \vec{n}$, or that their first and second derivatives agree. It corresponds to the circle passing through $\vec{p}$ and two infinitely close points on $c$, one approaching from the left and one from the right of $\vec{p}$. The radius and center of the osculating circle are often referred to as the \emph{radius of curvature} and \emph{center of curvature}, respectively.

The \emph{torsion} $\tau$ of $c$ measures the tendency of the curve to leave its osculating plane, \ie~the way the normal and binormal
twist around the curve. The \emph{Frenet-Serret formula} describes how the $\vec{t},\vec{n},\vec{b}$ frame changes along the curve:
$$\vec{t}' = -\kappa \vec{n}, \qquad \vec{n}' = \kappa \vec{t} - \tau \vec{b}, \qquad \vec{b}' = \tau \vec{n}.$$
For our purpose, the important formula is the second one, which shows that we can get the curvature by extracting the tangential part of $\vec{n}'$.

\begin{figure}
	\centering
	\small
	\def\svgwidth{\textwidth}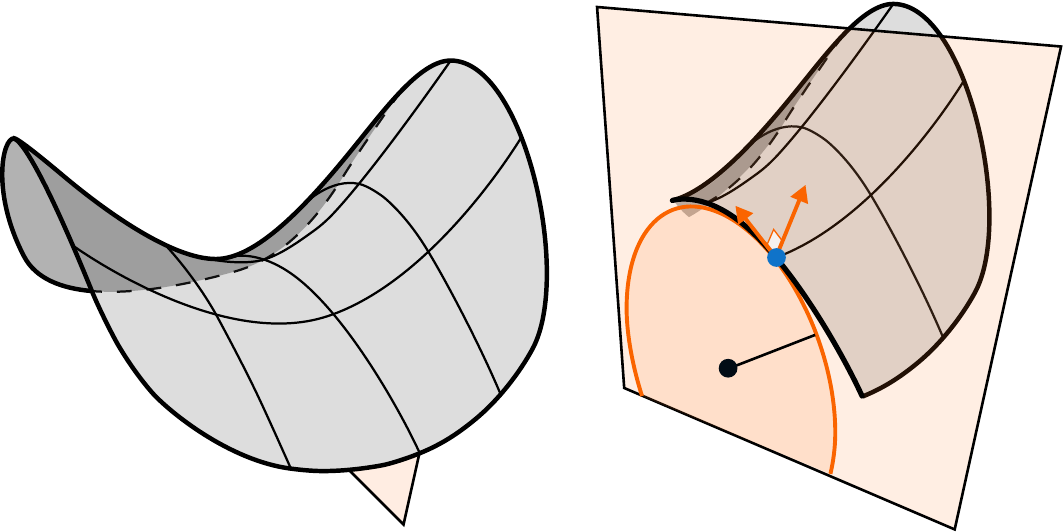\caption{\textbf{Normal curvature} --- The curvature $\kappa_n$ of the curve at the intersection of the surface $f(M)$ and the plane spanned by the normal $\vec{n}$ and tangent direction $\vec{t}$ is the normal curvature of $M$ at $\vec{p}$ in the direction $\vec{t}$.}\label{fig:surface_curvature}
\end{figure}

\paragraph{Normal curvature of a surface.} Going back to surfaces, consider the plane containing a given point $\vec{p}$ on the surface $f(M)$, a vector $\vec{t}$ in the tangent plane at that point and the associated normal $\vec{n}$. This plane intersects the surface in a curve, and the curvature $\kappa_{n}$ of this curve is called the \emph{normal (or sectional) curvature} in the direction~$\vec{t}$ (\fig{surface_curvature}). Using the Frenet-Serret formula, we can get the normal curvature along $\vec{t}$ by extracting the tangential part of $\vec{n}'$:
$$\kappa_{n} (\vec{t}) = \frac{\vec{t} \cdot \vec{n'}}{\lVert \vec{t} \rVert^2}.$$
This means that the normal curvature is a measure of how much the normal changes in the direction $\vec{t}$. Note that it is signed, meaning that the surface can bend toward or away from the normal, but it is not affected by the length of $\vec{t}$. Since $\vec{n}$ is a unit vector, its derivative $\vec{n}'$ is perpendicular to $\vec{n}$, hence in the tangent plane of $\vec{p}$.
$\vec{n}'$ is also called the \emph{shape operator} $S(\vec{t})$.

\paragraph{Principal, Gaussian and Mean curvature.} For a given point $\vec{p}$, the unit vectors $\vec{e}_1$ and $\vec{e}_2$ along which the normal curvature is maximal and minimal, respectively, are called the \emph{principle directions} at $\vec{p}$; the associated curvature values $\kappa_1$ and $\kappa_2$ are called the \emph{principal curvatures}. If $\kappa_1 = \kappa_2$, every direction is principal and the point is called an \emph{umbilic} (\fig{curvature}a). Otherwise, there are two orthogonal principle directions (\ie{} $\vec{e}_1 \cdot \vec{e}_2 = 0$). Principal directions and curvatures respectively corresponds to eigenvectors and eigenvalues of the shape operator: $$S(\vec{e_i}) = \kappa_i \vec{e_i}.$$

\begin{figure}
	\centering
	\small
	\def\svgwidth{\textwidth}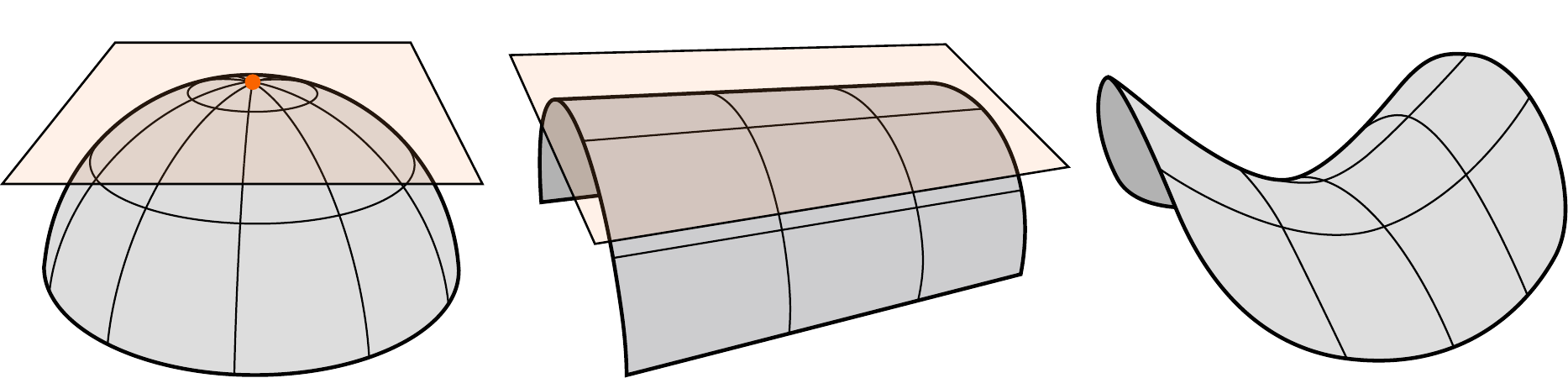\caption{\textbf{Surface categorization based on Gaussian curvature.}}\label{fig:curvature}
\end{figure}

The \emph{Gaussian} curvature $K$ is equal to the product of the principal curvatures: $$K = \kappa_1 \kappa_2,$$
and the mean curvature $H$ is their average: $$H = \frac{\kappa_1 + \kappa_2}{2}.$$
Based on the sign of its Gaussian curvature, a surface point is called (\fig{curvature}):
\begin{itemize}
	\item \emph{elliptic} if $K > 0$, \ie{} $\kappa_1$ and $\kappa_2$ have the same sign;
	\item \emph{parabolic} if $K = 0 \Rightarrow \kappa_1 = 0$ or $\kappa_2 = 0$;
	\item \emph{hyperbolic} if $K < 0$, \ie{} $\kappa_1$ and $\kappa_2$ have opposite signs.
\end{itemize}
Surfaces with zero Gaussian curvature are called \emph{developable surfaces}, because they can be flattened out into a plane without distortions. For example, a cylinder always has one of its principal curvature equal to zero. Surfaces with zero mean curvature are called \emph{minimal surfaces} because they locally minimize surface area. Since $\kappa_1 = -\kappa_2$ on minimal surfaces, they tend to look like saddles, which is also a good example of surfaces with negative Gaussian surfaces. On the other hand, surfaces with positive Gaussian curvature tend to look like hemispheres.

\paragraph{The fundamental forms.} Even though they do not introduce new geometric ideas, the fundamental forms are important for historical reasons. The first fundamental form $\mathbf{I}$ corresponds to the metric induced by the map $f$, which measures the inner product between any two vectors $\vec{x}$, $\vec{y}$ in the tangent plane at $\vec{p}$:
$$\mathbf{I}(\vec{x},\vec{y}) = \vec{x} \cdot \vec{y}.$$
The second fundamental form $\mathbf{II}$ at $\vec{p}$ is a symmetric bilinear form specified by:
$$\mathbf{II} (\vec{x},\vec{y}) = S(\vec{x}) \cdot \vec{y} = S(\vec{y}) \cdot \vec{x}.$$
With those notations, the normal curvature in the tangent direction $\vec{t}$ can be re-written as: 
$$\kappa_n = \frac{\mathbf{II}(\vec{t},\vec{t})}{\mathbf{I}(\vec{t},\vec{t})}.$$

\paragraph{Principal coordinates.} Using the principal directions $\vec{e}_1, \vec{e}_2$ as a local basis for the tangent plane at $\vec{p}$ leads to the \emph{principal coordinates}. When the vectors $\vec{x}$ and $\vec{y}$ are expressed in principal coordinates, the second fundamental form corresponds to the following diagonal matrix: 
$$\mathbf{II}(\vec{x},\vec{y}) = \vec{x}^\top \left[ 
    \begin{array}{cc}
        \kappa_1 & 0 \\
        0 & \kappa_2
    \end{array} 
\right]\vec{y}.$$
This leads to Euler formula stating that the normal curvature in the direction $[\cos\theta, \sin\theta]^\top$ expressed in principal coordinates, where $\theta$ is the angle measured between this direction and $\vec{e}_1$, is: $$\kappa_n(\theta) = \kappa_1 \cos^2\theta + \kappa_2 \sin^2\theta.$$


\chapter{Convex and Concave Contours}
\label{app:convex}

We now show that concave contour edges cannot be visible.

Any triangle lies in a supporting plane that cuts space into two half-spaces.  The surface normal points to one of these half-spaces.  We can say that a shape is \textbf{in front of} the face when it is in the half-space that the normal points to.  For example, a face is \textbf{front-facing} when the camera $\vec{c}$ is on the front side of the face. Likewise, a face is back-facing when the camera is \textbf{in back of the face}.

\begin{definition}[Concave/convex edge]\label{def:concave_edge}
A mesh edge on an orientable mesh is \emph{concave} if each triangle is in front of the other. (That is, triangle $A$ is in front of triangle $B$ and vice-versa.) Otherwise, if they are each in back of the other, the face is \emph{convex}.
\end{definition}

\begin{figure}
	\small
	\centering
	\textbf{(a)}~
	\def\svgwidth{0.22\textwidth}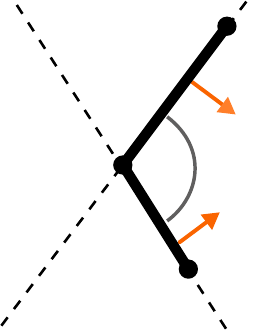
	\qquad
	\textbf{(b)}~
	\def\svgwidth{0.22\textwidth}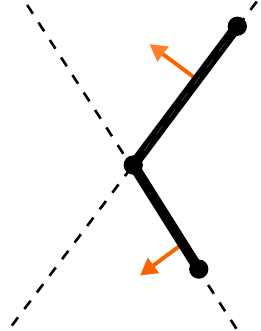
	\caption{\textbf{Concave/convex edge} ---
	Cross-section plane for a mesh edge between triangles $\vec{t}_1$ and $\vec{t}_2$, with normals $\vec{n}_1$ and $\vec{n}_2$. The cross-section is some 3D plane perpendicular to the edge between the two triangles.
	The two faces divide space in four quadrants: I,II,III, and IV. For example, I is the set of points in front of both faces, and II is in front of $\vec{t}_1$ and behind $\vec{t}_2$.
	\textbf{(a)} Concave case: $\theta < \pi$. Each face is in front of the other. In this case, the edge cannot be a visible contour.
	\textbf{(b)} Convex case: $\theta > \pi$. Each face is behind the other.}
	\label{fig:concave_edge}
\end{figure}

These cases are visualized in Figure \ref{fig:concave_edge}. It is easier to think about convexity in terms of the angle $\theta$ between the two oriented faces: the edge is convex if $\theta<\pi$, and concave otherwise.  (If $\theta=\pi$, than the edge can never be a contour.) However, computing $\theta$ is somewhat tricky.

Intuitively, concave edges should not be visible when they are contours, because, from the camera's point-of-view, they are hidden inside the surface. 

\begin{theorem}
On an orientable mesh where back-faces are never visible, 
a contour on a concave edge is never visible.
\end{theorem}

\begin{proof}
The two faces on the edge divide 3D space into four quadrants (Figure \ref{fig:concave_edge}a) where the viewpoint $\vec{c}$ can lie. For the edge to be a contour, the viewpoint must be in either quadrant II or IV; otherwise, both faces are either front-facing or back-facing. If the camera is in quadrant II, triangle $\vec{t}_2$ is back-facing, and thus must be invisible. Moreover, $\vec{t}_2$ occludes $\vec{t}_1$, at least in the vicinity of the edge. Hence, the edge is invisible.  The same reasoning directly applies to case IV.
\end{proof}

Furthermore, a contour on a convex edge is, locally, visible. While a convex contour could be occluded by some other surface far away, it lies on a front-face and so may be visible and is not locally occluded.


\paragraph{Front-facing and back-facing edges.}

\citet{Markosian:1997} first introduced a version of these ideas called \textit{front-facing edges} and \textit{back-facing edges}. They defined a contour edge as \textit{front-facing} if its adjacent face nearest the camera is front-facing, otherwise it is back-facing (\fig{internal_external}).  They mention in a footnote that they use convex/concavity instead of this definition. 


\begin{figure}
	\centering
	\small
	\def\svgwidth{0.95\textwidth}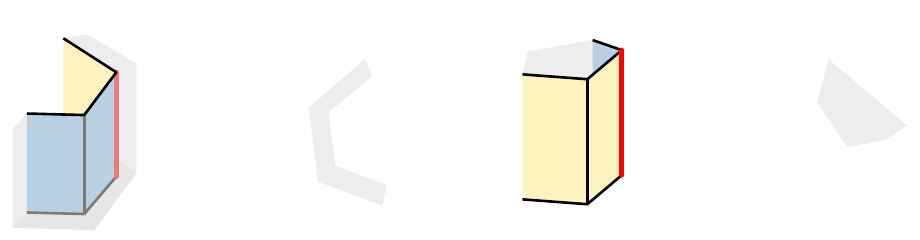\caption{\textbf{Front-facing/back-facing edges} --- A contour generator edge is \emph{back-facing}, resp.~\emph{front-facing}, when its adjacent face nearest to the camera position $\vec{c}$ is back-facing ($\sB$), resp.~front-facing ($\sF$). The ``interior'' of the mesh (assumed closed) is depicted in grey for illustration purpose.
	Front-facing contour edges are always convex edges on the surface, and back-facing contours are always concave.}\label{fig:internal_external}
\end{figure}



Unfortunately, this definition does not always work.
The distance between the camera and a face is the distance to the nearest point on the face.  Figure \ref{fig:markosian_counterexample} shows a counterexample where their definition fails. In many situations, however, their method would be  equivalent to computing convex/concave.


\begin{figure}
\centering
\small
\def\svgwidth{0.28\textwidth}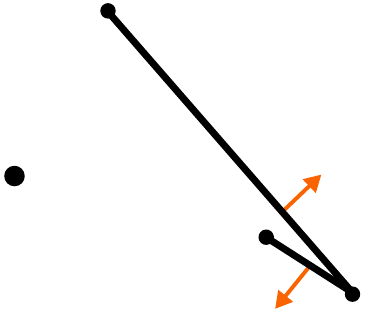\caption{
	\textbf{Counterexample for Markosian's definition of front/back-facing edges} --- 
	In this example, the face nearest to the camera is $\vec{t}_2$, which is back-facing. Therefore the contour edge between the triangles would be marked as back-facing and thus invisible.  In fact, $\vec{t}_1$ occludes $\vec{t}_2$, even though $\vec{t}_2$ is nearer to the camera. This is a convex edge, and thus potentially visible.
}\label{fig:markosian_counterexample}
\end{figure}

That said, the definition only fails in extreme cases, and the definition is useful for intuition when looking at diagrams.



\chapter{Accurate Numerical Computation} 
\label{app:numerical}

The visibility computations determine the topology of the View Graph, and minor errors can cause major topological errors. Hence, it is important to use robust geometric computations wherever possible.

One simple approach to improving precision is to use high-precision arithmetic (e.g., \texttt{float128}), infinite precision arithmetic, and/or rational arithmetic (since rational numbers are closed under perspective and intersection operations).

\section{Logical intersections}
\label{sec:intersect}

During visibility, there are many different intersection computations.  Whenever possible, these intersections should be done using pointers and logic, rather than numerics.  In principle, all intersections could be detected in the image-space intersection step. However, numerical error could cause intersections to be missed this way.  For example, a contour edge intersecting a boundary edges can be detected simply by finding edges that share a vertex.  Smooth curves, on the other hand, lie within mesh faces; their end-points lie with mesh edges. Their intersections with mesh boundaries amount to determining if they share the same mesh edge.

Keeping track of pointers is also necessary for avoiding spurious intersections. For example, when performing a ray test on a contour edge, it is important that neither of the adjacent triangles unintentionally ``occludes'' the edge.

\section{Orientation test} \label{sec:orientation}

 The orientation test is a useful building block of computational geometry. Many of the tests described here can be implemented in terms of the orientation test, and robust libraries exist for this test, such as the predicates of \citet{Shewchuk:1997}\footnote{\url{http://www.cs.cmu.edu/~quake/robust.html}}.

In 3D, the orientation test determines whether a point $\vec{d}$ lies to the left of, to the right of, or on the oriented plane defined by three other points $\vec{a}$, $\vec{b}$ and $\vec{c}$, appearing in counter-clockwise order when viewed from above the plane (\fig{orientation}). Two points are thus on the opposite sides of a triangular face if the results of their orientation tests with the triangle's supporting plane have opposite signs.  

\newcommand{\AD}{\vv{\vec{a}\vec{d}}}
\newcommand{\BD}{\vv{\vec{b}\vec{d}}}
\newcommand{\CD}{\vv{\vec{c}\vec{d}}}

In any dimension, the orientation test can be implemented as a matrix determinant. In 3D, this is equivalent to the signed volume of the parallelepiped spanned by the vectors $\AD = (\vec{a}-\vec{d})$, $\BD = (\vec{b}-\vec{d})$, and $\CD = (\vec{c}-\vec{d})$. That is, the side of the triangle $\bigtriangleup\vec{abc}$ that $\vec{d}$ lies on is determined by the sign of the scalar triple product:
\begin{align*}
\mbox{\textsc{Orient3D}}(\vec{a},\vec{b},\vec{c},\vec{d}) 
&= \AD \cdot (\BD \times \CD) \\
&= \det(\AD, \BD, \CD)  \\
&= \left | \begin{array}{ccc} 
a_x - d_x & a_y - d_y & a_z - d_z \\
b_x - d_x & b_y - d_y & b_z - d_z \\
c_x - d_x & c_y - d_y & c_z - d_z
\end{array} \right | \\
&= \left | \begin{array}{cccc} 
a_x & a_y & a_z & 1 \\
b_x & b_y & b_z & 1 \\
c_x & c_y & c_z & 1 \\
d_x & d_y & d_z & 1
\end{array} \right |
\end{align*}
Highly-accurate libraries exist for these routines. The second determinant formula shows that swapping input parameters will flip the sign of the output, e.g., 
$\mbox{\textsc{Orient3D}}(\vec{a},\vec{b},\vec{c},\vec{d}) =
-\mbox{\textsc{Orient3D}}(\vec{b},\vec{a},\vec{c},\vec{d})$.

The main use for orientation tests is to determine whether two points are on the same or the opposite side of a triangle. For this, we can define a function $\mbox{\textsc{SameSide}}(\vec{a},\vec{b},\vec{c},\vec{d}, \vec{e})$ that returns \textbf{true} if $\vec{d}$ and $\vec{e}$ are on the same side of $\bigtriangleup\vec{abc}$, that is:
\begin{align*}&\mbox{\textsc{SameSide}}(\vec{a},\vec{b},\vec{c},\vec{d}, \vec{e}) \\
	&= (\mbox{\textsc{Orient3D}}(\vec{a},\vec{b},\vec{c},\vec{d}) >0) == (\mbox{\textsc{Orient3D}}(\vec{a},\vec{b},\vec{c},\vec{e})>0)
\end{align*}

\newcommand{\BA}{\vv{\vec{b}\vec{a}}}
\newcommand{\CA}{\vv{\vec{c}\vec{a}}}
\newcommand{\DA}{\vv{\vec{d}\vec{a}}}

In some cases, we also need to evaluate whether a point is on the front-facing side of a triangle, or the back-facing side.  For this, we can define a function
$\mbox{\textsc{FrontSide}}(\vec{a},\vec{b},\vec{c},\vec{d})$ that returns a 
positive value if $\vec{d}$ is on the front side of $\bigtriangleup\vec{abc}$. The surface normal (unnormalized), is given by $\vec{n} = \BA \times \CA$, and so:
\begin{align*}
\mbox{\textsc{FrontSide}}(\vec{a},\vec{b},\vec{c},\vec{d})
&= \DA \cdot \vec{n} \\
&= \DA \cdot (\BA \times \CA) \\
&= \det(\DA,\BA,\CA) \\
&= \mbox{\textsc{Orient3D}}(\vec{d},\vec{b},\vec{c},\vec{a}) \\
&= -\mbox{\textsc{Orient3D}}(\vec{a},\vec{b},\vec{c},\vec{d})
\end{align*}


\begin{figure}
	\centering
	\small
	\def\svgwidth{0.4\textwidth}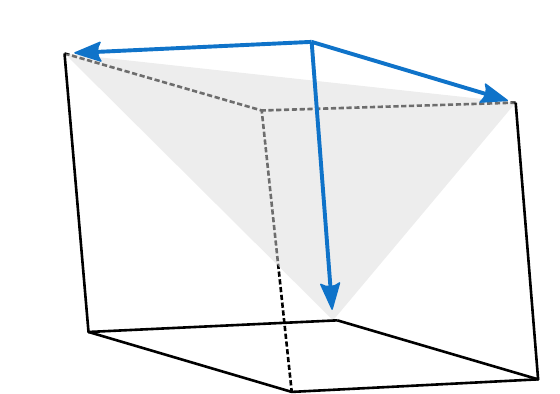
	\qquad
	\def\svgwidth{0.4\textwidth}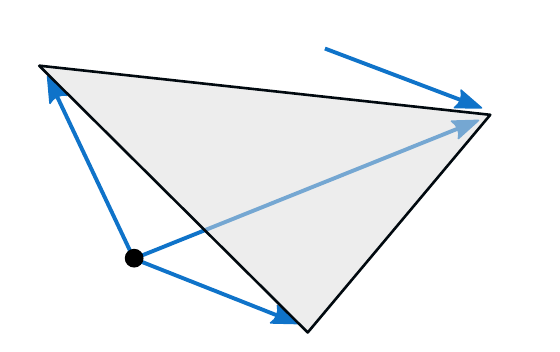
	\caption{\textbf{3D orientation and sidedness tests} --- The 3D orientation test \textbf{(a)} determines whether the point $\vec{d}$ is above, below or on the supporting plane of the oriented triangle $\bigtriangleup\vec{abc}$, \ie~whether the signed volume of the parallelepiped spanned by the vectors $\protect\AD$, $\protect\BD$ and $\protect\CD$ is positive, negative or zero. Two points $\vec{d}$ and $\vec{e}$ are then on the opposite sides of the triangle $\bigtriangleup\vec{abc}$ if the their orientation tests have opposite signs \textbf{(b)}.}\label{fig:orientation}
\end{figure}

\subsection{Applications of the orientation test}
\label{app:orient_apps}

Several of the tests used here can be implemented with the orientation test:

\paragraph{Front-facing.}
Given camera position $\vec{c}$, the face $\bigtriangleup\vec{abd}$ is front-facing if $\mbox{\textsc{FrontSide}}(\vec{a},\vec{b},\vec{d},\vec{c}) > 0$.

\paragraph{Concave edge.} (\sect{sec:concave_edge}).
A mesh edge with vertices $\vec{a}$ and $\vec{b}$ with triangles $\bigtriangleup\vec{abd}$ and $\bigtriangleup\vec{bae}$ is concave if: $\mbox{\textsc{FrontSide}}(\vec{a},\vec{b},\vec{d},\vec{e}) > 0$, or, equivalently, if $\mbox{\textsc{FrontSide}}(\vec{b},\vec{a},\vec{e},\vec{d}) > 0$.


\paragraph{Image-space intersection.}
Image-space intersections are normally detected by the sweep-line algorithm, for efficiency. However, they can also be expressed in terms of 2D orientation tests. The 2D orientation test  \textsc{SameSide2D} is a 2D version of the one described above.
  Using image-space coordinates, two line segments $\vec{ab}$ and $\vec{de}$ intersect if 
$\mathbf{not}~\textsc{SameSide2D}(\vec{d},\vec{e},\vec{a},\vec{b})~\mathbf{and}~\mathbf{not}~\textsc{SameSide2D}(\vec{a},\vec{b},\vec{d},\vec{e})$.

Alternately, one can test for intersection without computing the 2D projection at all. The problem is equivalent to testing whether the 3D triangles $\bigtriangleup\vec{abc}$ and $\bigtriangleup\vec{cde}$ intersect, where $\vec{c}$ is the camera position, which amounts to two 3D \textsc{SameSide} tests.

\paragraph{Overlap test for image-space intersection.} (\sect{sec:singular}). Suppose we have determined that 3D segments $\vec{ab}$ and $\vec{de}$ intersect in image space, viewed from camera $\vec{c}$. Suppose line segment $\vec{ab}$ is a contour or boundary, and it is in front of the other segment at this point. Let $\bigtriangleup\vec{abf}$ be an adjacent triangle on the mesh.  We need to determine which side of $\vec{de}$ is occluded by the surface. This amounts to determining whether 
$\mbox{\textsc{SameSide}}(\vec{a},\vec{b},\vec{c},\vec{d}, \vec{f})$ is true.

\paragraph{Boundary curtain fold.} (\sect{sec:singular}).
Curtain folds can occur on mesh boundaries (\fig{singular_points} {\large\textcircled{\footnotesize 4}}). Since there is no analogue of concave/convex edges for mesh boundaries, a different test is required. It consists in checking whether faces adjacent to a boundary vertex overlap in image-space, which can be computed by checking whether any non-adjacent face of the one-ring neighborhood of the boundary vertex (brown triangles in \fig{curtain_fold}) occludes the boundary edge which is the farthest from the camera.  

This can be reduced to three clipping tests (\fig{curtain_fold}); the edge $\vec{pe}$ is occluded by the face $\bigtriangleup\vec{pqr}$ if and only if:
(1) $\vec{c}$ and $\vec{e}$ are on opposite sides of this triangle, \ie{} $\mbox{\textsc{SameSide}}(\vec{p},\vec{q},\vec{r},\vec{c}, \vec{e})$ is false,
(2) $\vec{e}$ and $\vec{r}$ are on the same side of the triangle $\bigtriangleup\vec{cpq}$, \ie{} $\mbox{\textsc{SameSide}}(\vec{c},\vec{p},\vec{q},\vec{e}, \vec{r})$ is true,
(3) $\vec{e}$ and $\vec{q}$ are on the same side of the triangle $\bigtriangleup\vec{cpr}$, \ie{} $\mbox{\textsc{SameSide}}(\vec{c},\vec{p},\vec{r},\vec{e}, \vec{q})$ is true.
(This same test can also be used to detect curtain folds on contours, but it a simpler test exists for that case.)

It is possible, though unlikely, that a vertex has two boundary edges emerging from it, and both are locally-occluded. In this case, the above test produces a spurious curtain fold. These cases can be detected by performing the above local-overlap test on both boundary edges.  This additional test is optional, since spurious curtain folds should not affect the final visibility results.

\bibliography{contour_tutorial}

\end{document}